\documentclass[singlecolumn,10pt,aps,pra,superscriptaddress,floatfix,tightenlines,nofootinbib,nobibnotes]{revtex4-2}

\usepackage[colorlinks,linkcolor=blue,urlcolor=blue,citecolor=blue]{hyperref}
\usepackage{xurl}
\usepackage{booktabs}
\usepackage{comment}
\usepackage{xcolor}

\usepackage{dsfont}
\usepackage{physics}
\usepackage{nicefrac}
\usepackage{mathtools}
\usepackage{amsmath, amssymb, amsfonts, amsthm}
\usepackage{thmtools, thm-restate}

\usepackage{algpseudocodex}
\usepackage{algorithm}

\usepackage[whole]{bxcjkjatype}
\usepackage{accents}

\newtheorem{theorem}{Theorem}
\newtheorem*{theorem*}{Theorem}

\newtheorem{corollary}[theorem]{Corollary}
\newtheorem{lemma}[theorem]{Lemma}

\newtheorem{proposition}[theorem]{Proposition}

\newtheorem{remark}[theorem]{Remark}

\DeclareMathOperator{\id}{id}

\allowdisplaybreaks

\begin{document}

\title{Quantum Information Decoupling Beyond Finite Dimensions}

\author{Hayata Yamasaki}
\email{hayata.yamasaki@gmail.com}
\affiliation{
Department of Computer Science, Graduate School of Information Science and Technology, The University of Tokyo, 7--3--1 Hongo, Bunkyo-ku, Tokyo, 113--8656, Japan
}

\begin{abstract}
Quantum information decoupling is a pillar of quantum information theory, underlying quantum communication, error correction, and information recovery. However, its existing formulations rely on random unitary operations tied to finite-dimensional assumptions on quantum systems. Here we establish a decoupling framework for arbitrary separable, possibly infinite-dimensional systems. For finite-dimensional inputs and arbitrary separable reference/output systems, we derive error bounds on one-shot decoupling for completely positive maps in terms of sandwiched R\'enyi conditional entropies. For partial-trace decoupling with finite-dimensional references, the error exponent is optimal up to the known critical rate. To handle infinite-dimensional inputs, we assume finite entropy of the manipulated system. We construct finite-rank projections on independent and identically distributed (IID) states to restrict randomization to a projected finite-dimensional subspace, while ensuring high success probability by including auxiliary components in the discarded subsystem at asymptotically negligible cost. This leads to an infinite-dimensional IID partial-trace decoupling protocol achieving optimal asymptotic dimension rates. As an application, we construct an infinite-dimensional quantum-state-merging protocol, a mother protocol of quantum information theory. Under finite entropy of Alice's marginal, it achieves the same quantum-communication and total-cost rates as in finite dimensions, governed by mutual information and conditional entropy. These results show that the operational interpretation of entropic quantities through these achievable rates constitutes a universal principle beyond finite dimensions. More broadly, our framework provides foundational tools for exploring quantum information regardless of whether we model the physical world using finite- or infinite-dimensional spaces.
\end{abstract}

\maketitle

\tableofcontents

\section{Introduction}

\subsection{Background}

Many foundational techniques in quantum information theory are
formulated under the assumption that the Hilbert spaces representing
quantum systems are finite-dimensional
\cite{nielsen2010quantum,wilde2013quantum,watrous2018theory}.
This assumption is mathematically convenient, but it is not a
fundamental restriction imposed by quantum mechanics.
Several major platforms for quantum information processing, including
cavity modes, microwave resonators, and continuous-variable optical
systems, are naturally described by infinite-dimensional Hilbert
spaces.
Even systems used as qubits are often only approximately two-dimensional. For example, a superconducting qubit is typically encoded in the ground and first excited states of a superconducting electrical circuit that possesses infinitely many higher energy levels.
Finite-dimensional descriptions of such systems should therefore be
understood as effective truncations rather than exact descriptions of
the underlying physical systems.

Extending quantum information theory beyond finite dimensions is
important for both operational and foundational reasons.
From an operational perspective, restricting a physical system to a
prescribed finite-dimensional subspace may exclude protocols that
exploit additional degrees of freedom and may therefore alter the achievable
performance of information-processing tasks.
At a more fundamental level, quantum field theory and candidate
descriptions of quantum gravity naturally involve infinitely many
degrees of freedom that cannot, in
general, be reduced to finite-dimensional quantum systems
\cite{goto2026rethinkingquantuminformationgravity}.
Consequently, information-theoretic tools established only in finite
dimensions do not automatically apply to continuous-variable
platforms, field-theoretic settings, or other physical models with
infinite-dimensional state spaces.
A systematic infinite-dimensional formulation is therefore needed to
determine which principles and operational interpretations of quantum
information theory are genuinely independent of finite-dimensional assumptions, and which physically motivated conditions, such as
finite entropy or energy constraints, are required in place of finite
dimensionality.

Quantum information decoupling~\cite{horodecki2005partial,horodecki2007quantum,
abeyesinghe2009mother,berta2009singleshotquantumstatemerging,
dupuis2010decouplingapproachquantuminformation,dupuis2014one,
sharma2015randomcodingexponentsgalore,8811607,10232924,
wakakuwa2021one,colomer2024decoupling,
cheng2024jointstatechanneldecouplingoneshot,li2024operational,
he2026quantum,berta2026tightanyshotquantumdecoupling} is a pillar of
modern quantum information theory.
In its original formulation in Refs.~\cite{horodecki2005partial,horodecki2007quantum}, one considers a quantum state \(\rho^{AR}\), where the manipulatable system \(A\) is correlated with an inaccessible reference system \(R\).
After taking \(n\) independent and identically distributed (IID) copies,
one applies a random unitary transformation to \(A^n\), decomposes the
transformed system into two subsystems, and discards one of them.
The objective is to make the retained subsystem uncorrelated with the reference system up to a vanishingly small error. 
For finite-dimensional systems \(A\), the random unitary transformation is commonly drawn from a uniform probability distribution, that is, the normalized Haar measure on the unitary group~\cite{horodecki2005partial,horodecki2007quantum,
abeyesinghe2009mother,berta2009singleshotquantumstatemerging,
dupuis2010decouplingapproachquantuminformation,dupuis2014one,
sharma2015randomcodingexponentsgalore,8811607,10232924,
wakakuwa2021one,colomer2024decoupling,
cheng2024jointstatechanneldecouplingoneshot,li2024operational,
he2026quantum,berta2026tightanyshotquantumdecoupling}.
The Haar-random unitary transformations can be replaced by more structured
ensembles, including unitary designs and random quantum circuits~\cite{Szehr_2013,3179473.3179474,brown2015decoupling,
Nakata2017decouplingrandom,harrow2023approximate}.
If one additionally assumes the availability of a catalyst initially uncorrelated with $R$, constructions based on the convex split lemma~\cite{PhysRevLett.119.120506,PhysRevLett.118.080503,
anshu2022efficient,li2024reliability,11173697,
gour2025convexsplitlemmainequalities} can also realize decoupling using substantially more structured random operations such as random permutations.

Decoupling is fundamental because the elimination of correlations with an
inaccessible reference system is equivalent, through purification and information-recovery arguments, to the existence of a coherent reconstruction operation on the complementary system. This principle underlies a broad range of protocols in quantum information theory.
A canonical example is quantum state merging
\cite{horodecki2005partial,horodecki2007quantum,abeyesinghe2009mother,
berta2009singleshotquantumstatemerging,
dupuis2010decouplingapproachquantuminformation,dupuis2014one,
sharma2015randomcodingexponentsgalore,
berta2026tightanyshotquantumdecoupling,
10.1007/978-3-642-10698-9_8,10.1063/1.4795243,
sharma2014strongconversequantumstate,leditzky2016strong,8590809,
yamasaki2019spreadquantuminformationoneshot,
yamasaki2019entanglementtheorydistributedquantum,9040656,
yeh2025alphabitstatemerging}.
In this task, Alice and Bob initially hold the \(A\) and \(B\) parts,
respectively, of a pure state \(\ket{\psi}^{ABR}\) correlated with an inaccessible reference system \(R\).
The objective is to transfer Alice's share coherently to Bob while preserving the state.
Quantum state merging is itself a basic component of more general protocols in
distributed quantum information processing~\cite{yamasaki2019entanglementtheorydistributedquantum,
oppenheim2008uncommoninformationthecost,
oppenheim2008stateredistributionmergingintroducing,
PhysRevLett.103.220501,dutil2010oneshotmultipartystatemerging,
7935494,PhysRevLett.122.190502,8234697,
https://doi.org/10.1002/qute.201800066,
PhysRevLett.122.010502,PhysRevA.100.042306,
PhysRevA.103.062613,lee2025improved}.
A fully quantum version of state merging, known as the fully quantum Slepian--Wolf protocol, was identified as a mother protocol of quantum information theory, from which several fundamental communication primitives, including channel coding and channel simulation, can be derived
\cite{abeyesinghe2009mother,
dupuis2010decouplingapproachquantuminformation,
PhysRevLett.93.230504,4626055,hayden2008decoupling,
berta2011quantum,Datta_2011,6690133,9956855,
Wakakuwa2022oneshothybridstate,PhysRevA.104.012408}.
State merging also provides operational interpretations of basic entropic quantities; in the IID limit, its optimal resource rates are characterized by the quantum mutual information and conditional quantum entropy~\cite{horodecki2005partial,horodecki2007quantum}.
In particular, the conditional quantum entropy can take a negative value, and in such a case, its absolute value can be interpreted as the optimal rate at which state merging generates entanglement rather than consuming it.
Beyond quantum information theory, decoupling and the corresponding information-recovery formulations have become important in the study of various phenomena in quantum physics, such as quantum scrambling~\cite{4626055}, black-hole information recovery~\cite{hayden2007black}, thermalization~\cite{PhysRevE.94.022104}, and Landauer erasure~\cite{rio2011thermodynamic}.

Despite this broad scope, the decoupling framework underlying
these applications remains predominantly finite-dimensional.
The principal obstruction is not merely that infinite-dimensional
systems require additional technical analysis; rather, the standard random-unitary argument fundamentally loses one of its essential mathematical ingredients: a canonical probability distribution describing uniformly random unitary transformations.
The unitary group of an infinite-dimensional separable Hilbert space
is not locally compact and therefore does not admit a normalized Haar
probability measure analogous to that available in finite dimensions.
Although probability distributions over Gaussian unitary transformations or other
restricted classes of transformations can be introduced for particular continuous-variable settings, such ensembles are tailored to specific applications and do not provide a general substitute for Haar-random-unitary decoupling~\cite{Hahn2025classicalsimulation,PhysRevA.110.042402,
turner2014curious,PhysRevX.14.011013,dy4m-gq5c}.

Further difficulties arise in infinite dimensions because quantities that are automatically
finite and continuous on finite-dimensional state spaces may
become infinite or discontinuous; that is, dimension-dependent bounds used in conventional finite-dimensional proofs may cease to yield meaningful bounds in infinite dimensions.
This is a recurring obstacle in extending quantum-information results
beyond finite dimensions.
For example, in entanglement theory,
Ref.~\cite{hayden2001asymptotic} established, under
finite-dimensional assumptions, that the entanglement cost
\cite{PhysRevA.53.2046,PhysRevA.54.3824}
is given by the regularized entanglement of formation; however, as shown in Ref.
\cite{yamasaki2025entanglement}, proving such a characterization applicable to both finite- and infinite-dimensional systems requires alternative techniques for both the achievability and converse bounds.

These issues are particularly important in decoupling because the
manipulated system and the reference system play fundamentally
different operational roles.
The system \(A\) is accessible to the protocol designer, so one may
restrict its effective dimension or impose suitable approximation,
entropy, or energy conditions.
The reference system \(R\), by contrast, represents arbitrary quantum
information external to the protocol and is not available for
manipulation.
A satisfactory extension of the decoupling framework should therefore avoid
imposing finite-dimensionality, or comparably restrictive structural
assumptions, on \(R\).
This requirement already arises when the manipulated system \(A\) is finite-dimensional but the inaccessible reference system \(R\) is separable and possibly infinite-dimensional.

\subsection{Results}

In this work, we develop the framework for decoupling and quantum state merging without assuming finite-dimensionality.
Our main results are as follows.

\paragraph*{One-shot decoupling for finite-dimensional input and
separable reference and output systems}
We first study one-shot decoupling for a finite-dimensional input
system \(A\) and separable, possibly infinite-dimensional reference
and output systems \(R\) and \(E\), respectively.
Given a bipartite state \(\rho^{AR}\), a nonzero bounded completely
positive map \(\mathcal E^{A\to E}\), and a Haar-random unitary \(U^A\)
acting on \(A\), define the resulting, possibly unnormalized, output operator by
\(\widetilde\tau_{\mathcal E,U}^{ER}\coloneqq\left(\mathcal E^{A\to E}\otimes\id^{R}\right)\left(\left(U^A\otimes I^R\right)\rho^{AR}\left(U^{A\dagger}\otimes I^R\right)\right)\), where $I^R$ and $\id^{R}$ are the identity operator and channel, respectively, on $R$.
Its Haar average is
\(\omega_{\mathcal E}^{E}\otimes\rho^R\), where
\(\omega_{\mathcal E}^{E}\coloneqq\mathcal E^{A\to E}\left(\pi^A\right)/\operatorname{Tr}\left[\mathcal E^{A\to E}\left(\pi^A\right)\right]\) is the normalized output state of \(\mathcal E^{A\to E}\) on the maximally mixed state $\pi^A$ of $A$, and \(\rho^R\coloneqq\operatorname{Tr}_A\left[\rho^{AR}\right]\) is the reference marginal of the original state $\rho^{AR}$.
We quantify decoupling by the Haar-averaged relative entropy
\begin{align}
\mathbb E_U\left[
D\left(
\widetilde\tau_{\mathcal E,U}^{ER}
\middle\|
\omega_{\mathcal E}^{E}\otimes\rho^R
\right)\right].
\end{align}
When \(\mathcal E^{A\to E}\) is trace-preserving, the output is normalized for every
\(U\), and relative entropy provides a stronger error criterion than
the conventional trace-distance criterion
\cite{berta2009singleshotquantumstatemerging,
dupuis2010decouplingapproachquantuminformation,
dupuis2014one}.
It also admits an operational interpretation as an optimal error exponent in quantum hypothesis testing
\cite{hiai1991proper,
887855,
brandao2010generalization,
hayashi2025generalized,
10898013}.

Relative-entropy decoupling was previously analyzed in
Refs.~\cite{he2026quantum,berta2026tightanyshotquantumdecoupling},
but those analyses were formulated for finite-dimensional reference
and output systems.
To remove these dimensional assumptions, we extend the optimal
logarithmic trace inequality of
Ref.~\cite{gour2026optimaltraceinequalitiessingleshot}
from finite-dimensional operators to positive-semidefinite
trace-class operators on arbitrary separable Hilbert spaces
(Proposition~\ref{prop:infinite_dimensional_trace_inequality}).
This sharpens the earlier infinite-dimensional trace inequality of
Ref.~\cite{cheng2025sharpestimatesquantumcovering}, which served as a
key ingredient for the relative-entropy decoupling analysis in 
Ref.~\cite{berta2026tightanyshotquantumdecoupling}.

We also show that, in the operator-concavity step of the analysis in
Ref.~\cite{berta2026tightanyshotquantumdecoupling}, operator concavity
is asserted and used after extending the logarithm to the kernel of a singular
operator by setting \(\log 0=0\), but indeed, the resulting extension
is not an operator concave function
(see Remark~\ref{rem:tight_any_shot_invalid_log_concavity_step}).
We resolve this issue by adding strictly positive perturbations to the
arguments of the logarithms and representing the Haar average as an
isometric compression of a strictly positive operator; then, Jensen's trace inequality~\cite{hansen2003jensen} applies on its standard domain, yielding an alternative trace inequality needed to derive the relative-entropy decoupling bounds
(Proposition~\ref{prop:decoupling_operator_Jensen_reduction}).

Combining these ingredients, we obtain error bounds on one-shot decoupling (Theorem~\ref{thm:finite_input_infinite_reference_decoupling} and Corollary~\ref{cor:finite_input_infinite_reference_decoupling_exponent}).
Informally, these results can be stated as follows:
\begin{align}
&\mathbb E_U
\left[
D\left(
\widetilde\tau_{\mathcal E,U}^{ER}
\middle\|
\omega_{\mathcal E}^{E}\otimes\rho^R
\right)
\right]
\nonumber\\
&\lesssim
\inf_{0<s\le1}
\exp\left[
-s\left(
\widetilde H_{1+s}(A|R)_\rho
+
\widetilde H_{1+s}(A'|E)_{\omega_{\mathcal E}}
\right)
\right],
\label{eq:informal_one_shot_decoupling_bound}
\end{align}
where \(A'\) is a copy of \(A\),
\(\omega_{\mathcal E}^{A'E}\) is the normalized Choi state associated
with \(\mathcal E^{A\to E}\), \(\widetilde H_{1+s}\) denotes the sandwiched R\'enyi conditional entropy, and the symbol \(\lesssim\) suppresses explicit prefactors depending
only on \(s\) and the input dimension \(|A|\).
Compared with Ref.~\cite{berta2026tightanyshotquantumdecoupling}, replacing the operator-concavity step used therein with the argument above introduces an additional dimension-dependent factor \(C_{|A|}\). For \(n\) IID copies, however, the corresponding factor \(C_{|A|^n}\) for \(|A|\ge 2\) converges to \(2\) as \(n\to\infty\) and therefore remains $O(1)$. Consequently, this factor does not affect the achievable asymptotic error exponent.

Thus, we provide a complete proof of the achievable relative-entropy decoupling bound and the corresponding asymptotic error exponent without requiring either \(R\) or \(E\) to be finite-dimensional. For finite-dimensional reference systems and partial-trace maps, this achievable exponent matches the converse bound shown in Ref.~\cite{berta2026tightanyshotquantumdecoupling} up to the critical rate identified therein and is therefore optimal in this regime.
The relative-entropy one-shot decoupling theorem also supplies the finite-dimensional random-unitary
building block used below to establish decoupling for
infinite-dimensional input systems.

\paragraph*{Infinite-dimensional IID decoupling}

Building on the one-shot theorem above, we formulate infinite-dimensional IID decoupling for a state
\(\left(\rho^{AR}\right)^{\otimes n}\) on separable, possibly infinite-dimensional systems,
where \(A\) is the system under our control and \(R\) is an
inaccessible reference system.
For every \(n\), we introduce a finite-rank projection \(\Pi_n^{A^n}\) on \(A^n\) to define a finite-dimensional subspace
\(\mathcal H^{A_{\Pi_n}^n}\coloneqq\Pi_n^{A^n}\mathcal H^{A^n}\).
Its orthogonal complement \(\mathcal H^{A_{\Pi_n^\perp}^n}\) may remain infinite-dimensional, and its weight with respect to
\(\left(\rho^A\right)^{\otimes n}\) is
\(\delta_{\Pi_n}
\coloneqq
\operatorname{Tr}\left[
(I^{A^n}-\Pi_n^{A^n})(\rho^A)^{\otimes n}
\right]\), where \(I^{A^n}\) is the identity operator on \(A^n\).
For the normalized projected state on \(\mathcal H^{A_{\Pi_n}^n}\), we apply a unitary \(U_n\), identify
the projected space as
\(\mathcal H^{A_{\Pi_n}^n}
\simeq
\mathcal H^{E_n}\otimes\mathcal H^{M_n}\), and discard \(M_n\).
Writing \(\tau_n^{E_nR^n}\) for the resulting state and
\(\pi^{E_n}\) for the maximally mixed state on \(E_n\), we require
\begin{align}
\delta_{\Pi_n}\longrightarrow0,
\qquad
D\!\left(
\tau_n^{E_nR^n}
\middle\|
\pi^{E_n}\otimes\left(\rho^R\right)^{\otimes n}
\right)
\longrightarrow0.
\label{eq:criteria}
\end{align}
In particular, decoupling is required from the reference marginal of
the original state, rather than merely from that of the projected state.
Relative entropy provides a strictly stronger decoupling criterion than trace distance: convergence in relative entropy implies convergence in trace distance, but in infinite dimensions, states can converge in trace distance while their relative entropy remains bounded away from zero or even diverges.
We assume that the manipulated system has finite quantum entropy \(H(A)_\rho<\infty\), which may be viewed as an energy constraint (see also Eq.~\eqref{eq:finite_entropy_finite_energy_condition});
notably, no finite-dimensionality, entropy, or energy assumption is
imposed on \(R\).

The main technical challenge is to reconcile the finite-dimensional
subspaces required for Haar-random-unitary decoupling with an
asymptotically unit projection success probability in infinite
dimensions.
Under the assumption \(H(A)_\rho<\infty\), we first construct, for
each block length \(m\), a finite-rank spectral cutoff
\(\widehat{\Pi}_m^{A^m}\) on \(A^m\).
This cutoff selects a finite-dimensional subspace whose asymptotic
size is governed by \(H(A)_\rho\).
To obtain the growing number of IID blocks, we fix \(m\)
and write \(n=mb+c\), where \(b\) and \(c\) are respectively the
quotient and remainder obtained by dividing \(n\) by \(m\).
However, a naive projection requiring all \(b\) complete blocks to lie
in \(\mathcal H^{A_{\widehat{\Pi}_m}^m}\) succeeds only with
probability
\(
\left(1-\delta_{\widehat{\Pi}_m}\right)^b
\),
which converges to zero unless \(\delta_{\widehat{\Pi}_m}=0\).

To address this issue, we instead retain the direct sum of all weakly
typical common--rare patterns.
Such a pattern contains approximately
\((1-\delta_{\widehat{\Pi}_m})b\) common blocks lying in
\(\mathcal H^{A_{\widehat{\Pi}_m}^m}\) and approximately
\(\delta_{\widehat{\Pi}_m}b\) rare blocks lying in
\(\mathcal H^{A_{\widehat{\Pi}_m^\perp}^m}\).
From each retained pattern, we extract a fixed number of common blocks, chosen slightly below the typical common-block count so that this many blocks are available in every retained pattern.
These extracted blocks define a common finite-dimensional subsystem across all retained patterns, allowing the same Haar-random unitary transformation to act on that subsystem in every pattern sector.
Moreover, the fraction of common blocks left unextracted can be made arbitrarily small.
On the other hand, the rare-block space may be infinite-dimensional, and unless \(\delta_{\widehat{\Pi}_m}=0\), the number of rare blocks grows with \(n\).
Nevertheless, the normalized rare-block input marginal has finite entropy, allowing weak typicality to be applied to the repeated rare blocks.
This yields a finite-rank projection whose success probability converges to one.
The residual \(c<m\) copies are treated separately by a finite-rank
spectral cutoff whose rank is subexponential in \(n\) and whose success
probability also converges to one.

Consequently, all these components can be incorporated into a single
finite-rank projection \(\Pi_n\) on the full \(n\)-copy input space while
retaining asymptotically unit success probability.
All auxiliary components other than the extracted common subsystem are included
in the discarded system \(M_n\).
For a fixed block length \(m\), including these auxiliary components may increase the size of the discarded system.
However, the resulting additional dimension cost vanishes at first order when the limit \(n\to\infty\) is taken first and the block length \(m\) is subsequently sent to infinity
(Theorem~\ref{thm:IID_decoupling_high_probability_global_cutoffs}).

Combining this high-probability projection with the preceding
finite-dimensional-input one-shot decoupling theorem, we construct
infinite-dimensional decoupling protocols satisfying
\eqref{eq:criteria} and achieving the following rates
(Theorem~\ref{thm:IID_decoupling_original_direct}):
\begin{align}
\lim_{n\to\infty}
\frac1n\log|A_{\Pi_n}^n|
=
H(A)_\rho,
\qquad
\lim_{n\to\infty}
\frac1n\log|M_n|
=
\frac12 I(A:R)_\rho.
\end{align}
These rates coincide with the optimal finite-dimensional
rates~\cite{horodecki2005partial,horodecki2007quantum}.
Moreover, we prove matching converse bounds for an arbitrary sequence
of finite-rank projections with \(\delta_{\Pi_n}\to0\) and for
arbitrary unitaries and partial-trace factorizations on the
projected spaces (Theorem~\ref{thm:IID_decoupling_converse}).
Thus, the rates \(H(A)_\rho\) and
\(\frac12 I(A:R)_\rho\) are optimal within the present
infinite-dimensional decoupling formulation.

\paragraph*{Infinite-dimensional IID quantum state merging}
Finally, we formulate infinite-dimensional IID quantum state merging
on separable, possibly infinite-dimensional systems in the style of the fully quantum Slepian--Wolf protocol
\cite{horodecki2005partial,horodecki2007quantum,
abeyesinghe2009mother}.
Starting from the pure IID input state
\(\left(\ket{\psi}^{ABR}\right)^{\otimes n}\),
Alice coherently encodes her input \(A^n\) into a finite-dimensional quantum
message \(M_n\), which she transmits to Bob, and a finite-dimensional
system \(E_n\), which she retains.
Bob then applies a decoding operation to \(M_nB^n\).
The protocol succeeds when its final state converges in purified
distance to the ideal state in which Alice's share of
\(\left(\ket{\psi}^{ABR}\right)^{\otimes n}\) has been coherently transferred to Bob, and \(E_n\) is maximally entangled with a corresponding system held by Bob.
Purified-distance convergence is the conventional criterion in the
decoupling approach to quantum state merging: once \(E_n\) is
approximately decoupled from the inaccessible reference \(R^n\),
Uhlmann's theorem
\cite{UHLMANN1976273}
provides a decoding isometry on Bob's side that reconstructs both the
transferred input and Bob's share of the generated entanglement
\cite{horodecki2005partial,horodecki2007quantum}.
Because relative entropy controls purified distance, the stronger relative-entropy decoupling theorem established above directly provides the approximation required for the construction of the decoding operation from Uhlmann's theorem.

To construct the protocol, we apply the infinite-dimensional IID
decoupling theorem to the \(A^nR^n\)-marginal of the input state,
assuming only \(H(A)_\psi<\infty\), without imposing
finite-dimensionality or entropy assumptions on \(B\) or \(R\).
The resulting decoupling protocol decomposes the finite-dimensional
projected part of Alice's input into two subsystems: \(M_n\), which is
sent to Bob, and \(E_n\), which Alice retains.
We extend the finite-rank decoupling operation to a trace-preserving
encoder on the full infinite-dimensional input space by routing the complementary projection branch into a local auxiliary system
discarded by Alice, without communicating an additional success flag.
The auxiliary components used to construct the high-probability projected subspace are included in \(M_n\), and their dimensions are controlled so that they do not affect the first-order
quantum-communication rate.
The asymptotically unit projection success probability and the
vanishing relative-entropy decoupling error together imply a vanishing
purified-distance error for quantum state merging, while Uhlmann's
theorem supplies the required decoding operation.

With \(q\) denoting the asymptotic quantum-communication rate and
\(e\) denoting the asymptotic entanglement-generation rate, the
protocol achieves the following rates
(Theorem~\ref{thm:infinite_dimensional_IID_state_merging_direct}):
\begin{align}
q
&=
\frac{1}{2}I(A:R)_\psi,
&
q-e
&=
H(A|B)_\psi.
\end{align}
The quantity \(q-e\) represents the total cost of the protocol.
Indeed, if quantum communication at rate \(q\) is implemented by
quantum teleportation
\cite{PhysRevLett.70.1895},
with classical communication treated as free, it consumes entanglement
at rate \(q\), while the protocol generates entanglement at rate \(e\).
Thus, \(q-e\) is the net entanglement-consumption rate, with a negative
value corresponding to a net entanglement gain.

These rates coincide with the optimal rates of the fully quantum
Slepian--Wolf protocol in finite dimensions
\cite{horodecki2005partial,horodecki2007quantum,
abeyesinghe2009mother,Datta_2011}.
Thus, under the sole finite-entropy assumption on Alice's system, the
same entropic rates remain operationally achievable when \(A\), \(B\),
and \(R\) are all allowed to be infinite-dimensional.
In particular, the result provides direct operational interpretations
of \(I(A:R)_\psi\) as twice an achievable quantum-communication rate
and of \(H(A|B)_\psi\) as an achievable total cost.
It thereby motivates the use of these information-theoretic quantities
in the analysis of quantum information processing and quantum physics
beyond finite dimensions.

\subsection{Outlook}

This work shows the universality of the core operational principles underlying quantum decoupling and quantum state merging across physical systems modeled by finite- or infinite-dimensional separable Hilbert spaces, provided that the system being manipulated has finite entropy.

When the manipulated system is finite-dimensional, our relative-entropy error bounds for one-shot decoupling remain valid even when the inaccessible reference and output systems are infinite-dimensional.
For partial-trace decoupling with a finite-dimensional reference
system, the resulting achievable asymptotic error exponent matches the
converse bound in
Ref.~\cite{berta2026tightanyshotquantumdecoupling}
up to the critical rate identified therein and is therefore optimal in
this regime.

When the manipulated system is infinite-dimensional, its unitary group
does not admit a normalized Haar probability measure.
Nevertheless, our construction circumvents this obstruction by
introducing finite-rank projections whose success probabilities
converge to one and applying Haar-random-unitary decoupling only to
suitably chosen finite-dimensional subsystems.
Under the sole assumption \(H(A)<\infty\), we recover the optimal
system-dimension rates of IID decoupling known from finite-dimensional
settings, without imposing finite-dimensionality, entropy, or energy
assumptions on the inaccessible reference system.

Even when all input systems are infinite-dimensional, the resulting
fully quantum Slepian--Wolf protocol likewise achieves the same
quantum-communication and total-cost rates as the optimal
finite-dimensional protocol.
These achievability results provide direct operational interpretations of the mutual information as twice an achievable quantum-communication rate and of the conditional entropy as an achievable total cost, thereby highlighting their operational significance in quantum information processing and quantum physics beyond finite dimensions.

An interesting open problem is to establish matching
infinite-dimensional converse bounds for the decoupling error exponent
and quantum-state-merging cost rates under assumptions comparable to those
used for achievability.
For relative-entropy decoupling, the known converse bound on the error
exponent relies on finite-dimensional spectral pinching and does not
directly extend to an infinite-dimensional reference system
(see Remark~\ref{remark:converse_one_shot_decoupling}).
For quantum state merging, conventional converse arguments may give rise to indeterminate differences of infinite entropies (see Remark~\ref{rem:converse_state_merging}).
Moreover, even under the assumption \(H(A)<\infty\), convergence in
purified distance does not by itself impose a finite-entropy or energy
constraint on the actual output of a protocol.
One-sided semicontinuity methods
\cite{shirokov2023close,shirokov2025alicki},
similar to those used in the study of entanglement cost
\cite{yamasaki2025entanglement},
may lead to converse bounds under additional assumptions such as
\(H(AB)<\infty\).
However, determining whether matching converse bounds hold under the sole
condition \(H(A)<\infty\) may require fundamentally
different techniques.

Another important direction is to extend the infinite-dimensional IID
decoupling protocol from partial-trace channels to general bounded
completely positive maps with infinite-dimensional inputs.
Our one-shot theorem already applies to arbitrary bounded completely
positive maps when \(A\) is finite-dimensional, while allowing the
reference and output systems to be arbitrary separable systems.
This setting is directly relevant to decoupling-based analyses of
channel coding
\cite{dupuis2010decouplingapproachquantuminformation,
dupuis2014one,sharma2015randomcodingexponentsgalore,
wakakuwa2021one,9956855,
cheng2024jointstatechanneldecouplingoneshot,
berta2026tightanyshotquantumdecoupling}.
For infinite-dimensional \(A\), partial trace has a special feature:
all auxiliary components introduced by the finite-rank projection can
be absorbed into the subsystem discarded by the partial trace.
For a general completely positive map, there is no analogous canonical
prescription for treating these auxiliary components or for specifying
how their outputs should enter the decoupling criterion.
Thus, even the appropriate formulation of decoupling for general
completely positive maps with infinite-dimensional inputs remains
nontrivial.

Nevertheless, the mother-protocol interpretation of the fully quantum
Slepian--Wolf protocol suggests a complementary route
\cite{abeyesinghe2009mother,PhysRevLett.93.230504,4626055}.
In finite dimensions, resource reductions based on the mother protocol
provide a unified way to derive protocols for channel coding, channel
simulation, entanglement-assisted communication, and several
distributed tasks.
Establishing analogous reductions beyond finite dimensions could
therefore yield a systematic framework for deriving and analyzing a
broader family of protocols under finite-entropy or energy constraints.
The infinite-dimensional quantum-state-merging protocol developed here
initiates this line of exploration.

Finally, extending the present framework from separable Hilbert spaces
to operator-algebraic settings is a fundamental longer-term objective,
with potential applications to quantum field theory and quantum gravity
\cite{goto2026rethinkingquantuminformationgravity}.
Taken together, pursuing these directions would help clarify which operational principles of quantum information theory are universal across different mathematical frameworks used to model the underlying physical world.

\subsection{Organization of the paper}
The remainder of the paper is organized as follows.
Section~\ref{sec:preliminaries} introduces the notation and
information-theoretic quantities used throughout the paper for
separable, possibly infinite-dimensional quantum systems.
Section~\ref{sec:finite_input_infinite_reference_one_shot_decoupling}
establishes one-shot decoupling bounds for finite-dimensional input
systems and separable reference and output systems, including the
trace inequalities underlying the relative-entropy error bounds.
Section~\ref{sec:infinite_dimensional_IID_decoupling} develops
infinite-dimensional IID decoupling through high-probability
finite-rank projections and establishes the corresponding achievable
rates and converse bounds.
Finally,
Section~\ref{sec:infinite_dimensional_IID_quantum_state_merging}
applies the IID decoupling result to infinite-dimensional quantum state merging and derives
the achievable quantum-communication and total-cost rates.

\begin{acknowledgments}
H.Y. acknowledges Oliver Hahn for discussions.
H.Y. was supported by JST PRESTO Grant No. JPMJPR23FC, JST CREST Grant No. JPMJCR25I5, JST Moonshot R\&D Grant No. JPMJMS256J, and Faculty Research Funding from Google Quantum AI\@.
\end{acknowledgments}

\section{Preliminaries}
\label{sec:preliminaries}

In this section, we introduce the notation used throughout this work and summarize the entropy quantities required in the subsequent analyses of decoupling and quantum state merging.
Subsection~\ref{subsec:quantum_mechanics_in_infinite_dimensions}
fixes notation for Hilbert spaces, operators, states, maps, Haar-random unitaries, spectral calculus, relative entropy, and distance measures.
Subsection~\ref{subsec:quantum_entropies} introduces the quantum
entropy, conditional quantum entropy, and mutual information under a
finite-marginal-entropy assumption, and explains its relation to finite-energy conditions.
Subsection~\ref{subsec:sandwiched_renyi_relative_entropy} defines the
sandwiched R\'enyi relative entropy in infinite dimensions, explains the domain condition required for its well-definedness, and provides the corresponding conditional sandwiched R\'enyi entropy.

\subsection{Quantum mechanics in infinite dimensions}
\label{subsec:quantum_mechanics_in_infinite_dimensions}

We fix notation for quantum systems, operators, states, maps,
spectral calculus, and distance measures, while referring to
Ref.~\cite{conway2025course} for details of the operator theory.
All Hilbert spaces in this paper are complex and separable, and may be
infinite-dimensional unless stated otherwise.
A quantum system labeled \(A\) is represented by a Hilbert space
\(\mathcal H^A\), whose inner product is denoted in braket notation by
\begin{align}
\bra{\psi}\ket{\psi'}
\end{align}
for
\(\ket{\psi},\ket{\psi'}\in\mathcal H^A\).

For a Hilbert space \(\mathcal H^A\), a linear operator \(X\) on
\(\mathcal H^A\) is bounded if its operator norm is finite, i.e.,
\begin{align}
\|X\|
\coloneqq
\sup
\left\{
\|X\ket{\psi}\|:
\ket{\psi}\in\mathcal H^A,\ 
\|\ket{\psi}\|\le1
\right\}
<
\infty.
\label{eq:operator_norm_definition}
\end{align}
We write
\(\operatorname{B}(\mathcal H^A)\)
for the Banach space of bounded operators on
\(\mathcal H^A\), equipped with the operator norm
\(\|\cdot\|\).
For
\(X\in\operatorname{B}(\mathcal H^A)\),
its adjoint is the unique bounded operator
\(X^\dagger\in\operatorname{B}(\mathcal H^A)\)
satisfying
\begin{align}
\bra{\psi}X\ket{\psi'}
=
\bra{X^\dagger\psi}\ket{\psi'}
\end{align}
for all
\(\ket{\psi},\ket{\psi'}\in\mathcal H^A\).

A bounded operator
\(X\in\operatorname{B}(\mathcal H^A)\)
is positive semidefinite, written \(X\ge0\), if
\begin{align}
\bra{\psi}X\ket{\psi}
\ge
0
\end{align}
for every
\(\ket{\psi}\in\mathcal H^A\).
We write
\begin{align}
\operatorname{B}_+(\mathcal H^A)
\coloneqq
\left\{
X\in\operatorname{B}(\mathcal H^A):
X\ge0
\right\}
\end{align}
for the cone of positive-semidefinite bounded operators.

Let
\(\{\ket{\psi_i}\}_{i\in\mathcal I}\)
be an orthonormal basis of
\(\mathcal H^A\), where
\(\mathcal I\subseteq\{0,1,\ldots\}\)
is finite or countably infinite.
For
\(X\in\operatorname{B}_+(\mathcal H^A)\),
the trace is defined by
\begin{align}
\operatorname{Tr}[X]
\coloneqq
\sum_{i\in\mathcal I}
\bra{\psi_i}X\ket{\psi_i}
\in
[0,+\infty].
\label{eq:extended_trace_positive_operator}
\end{align}
The value of the nonnegative series in
\eqref{eq:extended_trace_positive_operator}, including whether it is
finite, is independent of the choice of orthonormal basis.

A bounded operator
\(X\in\operatorname{B}(\mathcal H^A)\)
is trace class if its trace norm is finite, i.e.,
\begin{align}
\|X\|_1
\coloneqq
\operatorname{Tr}[|X|]
<
\infty,
\qquad
|X|
\coloneqq
\sqrt{X^\dagger X}.
\label{eq:trace_norm_definition}
\end{align}
We write
\(\operatorname{T}(\mathcal H^A)\)
for the Banach space of trace-class operators on
\(\mathcal H^A\), equipped with the trace norm
\(\|\cdot\|_1\).
For
\(X\in\operatorname{T}(\mathcal H^A)\),
its trace is defined by
\begin{align}
\operatorname{Tr}[X]
\coloneqq
\sum_{i\in\mathcal I}
\bra{\psi_i}X\ket{\psi_i}
\in
\mathbb C,
\label{eq:trace_trace_class_operator}
\end{align}
where the series is absolutely convergent and independent of the
choice of orthonormal basis.
We write
\begin{align}
\operatorname{T}_+(\mathcal H^A)
\coloneqq
\operatorname{T}(\mathcal H^A)
\cap
\operatorname{B}_+(\mathcal H^A)
\end{align}
for the cone of positive-semidefinite trace-class operators.

The set of quantum states on
\(\mathcal H^A\)
is
\begin{align}
\operatorname{D}(\mathcal H^A)
\coloneqq
\left\{
\rho^A\in\operatorname{T}_+(\mathcal H^A):
\operatorname{Tr}[\rho^A]=1
\right\}.
\label{eq:density_operators_definition}
\end{align}
When useful, possibly unnormalized output operators are denoted with
a tilde, as in
\(\widetilde\tau\).
This output-operator notation is distinct from the standard tilde used
below for sandwiched R\'enyi quantities.
A normalized vector
\(\ket{\psi}^A\in\mathcal H^A\)
is identified with the corresponding pure state
\(\psi^A=\ket{\psi}\bra{\psi}^A\)
when no ambiguity arises.

For \(X,Y\in\operatorname{T}(\mathcal H^A)\), the trace-norm distance is
\begin{align}
\frac12
\|X-Y\|_1.
\label{eq:trace_distance_definition}
\end{align}
For quantum states, the trace-norm distance is also called the trace
distance.
Unless explicitly stated otherwise, convergence of trace-class
operators is understood with respect to the trace norm.

A bounded operator
\(U\in\operatorname{B}(\mathcal H^A)\)
is unitary if
\begin{align}
U^\dagger U
=
UU^\dagger
=
I^A.
\end{align}
The group of unitary operators on
\(\mathcal H^A\)
is denoted by
\begin{align}
\operatorname{U}(\mathcal H^A).
\label{eq:unitary_group_HA}
\end{align}

When
\(\mathcal H^A\)
is finite-dimensional,
\(\operatorname{U}(\mathcal H^A)\)
is a compact group and admits a unique normalized left- and
right-invariant Borel probability measure, called the normalized Haar
measure and denoted by
\begin{align}
\mathrm dU.
\end{align}
For every scalar- or Banach-space-valued Haar-integrable function
\(f\) on
\(\operatorname{U}(\mathcal H^A)\),
we write
\begin{align}
&\mathbb E_{
U\sim
\mathrm{Haar}
\left(
\operatorname{U}
\left(
\mathcal H^A
\right)
\right)
}
\left[
f(U)
\right]
\nonumber\\
&\coloneqq
\int_{
\operatorname{U}
\left(
\mathcal H^A
\right)
}
f(U)
\,\mathrm dU.
\label{eq:normalized_Haar_expectation}
\end{align}
Operator-valued Haar integrals are understood as Bochner integrals in
the relevant Banach space whenever this is the topology used in the
given context.

A bounded linear operator
\begin{align}
V:
\mathcal H^A
\longrightarrow
\mathcal H^E
\end{align}
is an isometry if
\begin{align}
V^\dagger V
=
I^A.
\end{align}
A surjective isometry is a unitary isomorphism.

The identity operator on
\(\mathcal H^A\)
is denoted by \(I^A\).
The identity map on
\(\operatorname{T}(\mathcal H^A)\)
is denoted by
\(\id^A\).
For
\(k\in\{1,2,\ldots\}\),
we let
\(\id_k\)
denote the identity map on
\(\operatorname{T}(\mathbb C^k)\).

If
\(\mathcal H^A\)
is finite-dimensional, we write
\begin{align}
|A|
\coloneqq
\dim\mathcal H^A.
\label{eq:dimension_shorthand_A}
\end{align}
The maximally mixed state on
\(\mathcal H^A\)
is denoted by
\begin{align}
\pi^A
\coloneqq
\frac{I^A}{|A|}.
\label{eq:maximally_mixed_state_A}
\end{align}
If
\(\mathcal H^{A'}\simeq\mathcal H^A\)
is a fixed copy of
\(\mathcal H^A\),
the normalized maximally entangled state is denoted by
\begin{align}
\Phi^{AA'}
\coloneqq
\ket{\Phi}\bra{\Phi}^{AA'},
\qquad
\ket{\Phi}^{AA'}
\coloneqq
\frac{1}{\sqrt{|A|}}
\sum_{i=0}^{|A|-1}
\ket{i}^{A}
\otimes
\ket{i}^{A'},
\label{eq:normalized_maximally_entangled_state}
\end{align}
where the orthonormal bases are fixed by the identification between
\(A\) and \(A'\).

For isomorphic systems \(A\) and \(\widetilde A\), the swap operator
is denoted by
\begin{align}
F^{A\widetilde A}
:
\mathcal H^A
\otimes
\mathcal H^{\widetilde A}
\longrightarrow
\mathcal H^A
\otimes
\mathcal H^{\widetilde A}
\end{align}
and is defined by
\begin{align}
F^{A\widetilde A}
\left(
\ket{\varphi}^A
\otimes
\ket{\psi}^{\widetilde A}
\right)
=
\ket{\psi}^A
\otimes
\ket{\varphi}^{\widetilde A},
\label{eq:swap_operator_definition}
\end{align}
where the vectors on \(A\) and \(\widetilde A\) are identified using
the fixed identification between the two copies.

For composite systems, we write, for example,
\begin{align}
\rho^{AR}
\in
\operatorname{D}
\left(
\mathcal H^A
\otimes
\mathcal H^R
\right).
\end{align}
We use superscripts to denote reduced states:
\begin{align}
\rho^A
&\coloneqq
\operatorname{Tr}_R[\rho^{AR}],
\\
\rho^R
&\coloneqq
\operatorname{Tr}_A[\rho^{AR}].
\label{eq:reduced_states_notation_prelim}
\end{align}
The \(n\)-fold independent and identically distributed state is
written as
\begin{align}
\left(
\rho^{AR}
\right)^{\otimes n}
\in
\operatorname{D}
\left(
\mathcal H^{A^n}
\otimes
\mathcal H^{R^n}
\right),
\label{eq:n_copy_state_notation_prelim}
\end{align}
where
\begin{align}
\mathcal H^{A^n}
&\coloneqq
\left(
\mathcal H^A
\right)^{\otimes n},
&
\mathcal H^{R^n}
&\coloneqq
\left(
\mathcal H^R
\right)^{\otimes n}.
\end{align}
Tensor products that differ only by the canonical ordering of their
tensor factors are identified throughout this work.

A linear map
\begin{align}
\mathcal E:
\operatorname{T}(\mathcal H^A)
\longrightarrow
\operatorname{T}(\mathcal H^E)
\end{align}
is bounded if
\begin{align}
\|\mathcal E\|_{1\to1}
\coloneqq
\sup
\left\{
\|\mathcal E(X)\|_1:
X\in\operatorname{T}(\mathcal H^A),\
\|X\|_1\le1
\right\}
<
\infty.
\label{eq:trace_class_map_norm}
\end{align}
It is positive if
\begin{align}
X\in\operatorname{T}_+(\mathcal H^A)
\quad\Longrightarrow\quad
\mathcal E(X)\in\operatorname{T}_+(\mathcal H^E).
\end{align}
It is completely positive if, for every
\(k\in\{1,2,\ldots\}\), the map
\begin{align}
\id_k\otimes\mathcal E:
\operatorname{T}
\left(
\mathbb C^k\otimes\mathcal H^A
\right)
\longrightarrow
\operatorname{T}
\left(
\mathbb C^k\otimes\mathcal H^E
\right)
\end{align}
is positive.
It is trace-preserving if
\begin{align}
\operatorname{Tr}[\mathcal E(X)]
=
\operatorname{Tr}[X]
\end{align}
for every
\(X\in\operatorname{T}(\mathcal H^A)\).
A bounded completely positive trace-preserving linear map is called a
quantum channel.

For
\(X\in\operatorname{B}_+(\mathcal H^A)\),
let \(E_X\) denote its projection-valued spectral measure, so that
\begin{align}
X
=
\int_{[0,\|X\|]}
\lambda\,
\mathrm{d}E_X(\lambda).
\end{align}
For every Borel set
\(\Delta\subseteq[0,\|X\|]\),
\(E_X(\Delta)\)
is an orthogonal projection on
\(\mathcal H^A\).
We define the support projection of \(X\) by
\begin{align}
\Pi_X
\coloneqq
E_X((0,\|X\|]),
\end{align}
and its support as the closed subspace
\begin{align}
\operatorname{supp}(X)
\coloneqq
\Pi_X\mathcal H^A.
\end{align}

For a real-valued Borel function \(f\) on
\((0,\|X\|]\),
we define the generalized functional calculus by
\begin{align}
f(X)
\coloneqq
\int_{(0,\|X\|]}
f(\lambda)\,
\mathrm{d}E_X(\lambda),
\label{eq:generalized_functional_calculus}
\end{align}
with natural domain
\begin{align}
\operatorname{dom}(f(X))
\coloneqq
\left\{
\ket{\psi}\in\mathcal H^A:
\int_{(0,\|X\|]}
|f(\lambda)|^2
\,
\mathrm{d}
\bra{\psi}E_X(\lambda)\ket{\psi}
<
\infty
\right\}.
\label{eq:generalized_functional_calculus_domain}
\end{align}
By convention,
\(f(X)\)
acts as zero on
\(\ker(X)\).
If \(f\) is bounded on
\((0,\|X\|]\),
then \(f(X)\) is a bounded operator on all of
\(\mathcal H^A\).
For a general real-valued Borel function \(f\),
\(f(X)\) is the corresponding densely defined, possibly unbounded
self-adjoint operator with the domain in
\eqref{eq:generalized_functional_calculus_domain}.

For nonzero
\(X\in\operatorname{T}_+(\mathcal H^A)\),
there exist
\begin{align}
r_X
\in
\{1,2,\ldots\}
\cup
\{\infty\}
\end{align}
and an index set
\begin{align}
\mathcal I_X
\coloneqq
\begin{cases}
\{0,1,\ldots,r_X-1\},
&
r_X<\infty,
\\
\{0,1,\ldots\},
&
r_X=\infty,
\end{cases}
\end{align}
together with positive eigenvalues
\begin{align}
\lambda_i>0,
\qquad
i\in\mathcal I_X,
\end{align}
counted with multiplicity, and a corresponding orthonormal basis
\begin{align}
\left\{
\ket{\psi_i}
\right\}_{i\in\mathcal I_X}
\end{align}
of
\(\operatorname{supp}(X)\),
such that
\begin{align}
X
=
\sum_{i\in\mathcal I_X}
\lambda_i
\ket{\psi_i}\bra{\psi_i}.
\label{eq:spectral_decomposition_trace_class_positive}
\end{align}
The series in
\eqref{eq:spectral_decomposition_trace_class_positive}
converges in trace norm, and
\begin{align}
\sum_{i\in\mathcal I_X}
\lambda_i
=
\operatorname{Tr}[X]
<
\infty.
\end{align}
For \(X=0\), the corresponding spectral decomposition is understood
as the empty sum.
For the real-valued Borel function \(f\) above, the generalized
functional calculus takes the form
\begin{align}
f(X)
=
\sum_{i\in\mathcal I_X}
f(\lambda_i)
\ket{\psi_i}\bra{\psi_i}
\label{eq:generalized_functional_calculus_discrete}
\end{align}
on its natural domain
\begin{align}
\operatorname{dom}(f(X))
=
\left\{
\ket{\psi}\in\mathcal H^A:
\sum_{i\in\mathcal I_X}
|f(\lambda_i)|^2
\left|
\bra{\psi_i}\ket{\psi}
\right|^2
<
\infty
\right\},
\end{align}
and \(f(X)\) acts as zero on
\(\ker(X)\).

For
\(\alpha\in\mathbb R\),
we define the generalized power by
\begin{align}
X^\alpha
\coloneqq
\int_{(0,\|X\|]}
\lambda^\alpha\,
\mathrm{d}E_X(\lambda).
\label{eq:generalized_real_power}
\end{align}
In particular,
\begin{align}
X^0
=
\Pi_X,
\end{align}
and, by the convention above, negative powers act as zero on
\(\ker(X)\).

Let \(\log\) denote the natural logarithm.
For
\(X\in\operatorname{B}_+(\mathcal H^A)\),
we define the generalized logarithm by
\begin{align}
\log X
\coloneqq
\int_{(0,\|X\|]}
\log\lambda\,
\mathrm{d}E_X(\lambda),
\label{eq:generalized_operator_logarithm}
\end{align}
where, by the convention above,
\(\log X\)
acts as zero on
\(\ker(X)\).
It is a densely defined self-adjoint operator and need not be bounded,
even when \(X\) is trace class.

For positive-semidefinite trace-class operators
\(X,Y\in\operatorname{T}_+(\mathcal H^A)\),
we use the Umegaki relative entropy~\cite{umegaki1962conditional},
formally written as
\begin{align}
D(X\|Y)
=
\operatorname{Tr}
\left[
X
\left(
\log X-\log Y
\right)
\right].
\label{eq:formal_umegaki_trace_expression}
\end{align}
While
Ref.~\cite{umegaki1962conditional}
introduced this trace expression under suitable finiteness
conditions,
Refs.~\cite{araki1976relative,araki1977relative}
extended the definition without requiring the two logarithmic
contributions to be separately finite.
In infinite dimensions,
\eqref{eq:formal_umegaki_trace_expression}
cannot in general be interpreted as the difference of the separately
evaluated traces
\(\operatorname{Tr}[X\log X]\)
and
\(\operatorname{Tr}[X\log Y]\).
Indeed,
\(\operatorname{Tr}[X\log X]\)
may be equal to
\(-\infty\),
whereas the product
\(X\log Y\),
and hence its trace, need not be well defined.
We therefore specify the meaning of
\eqref{eq:formal_umegaki_trace_expression}
through a single extended spectral expression.

If \(X=0\), we set
\begin{align}
D(0\|Y)
\coloneqq
0.
\label{eq:umegaki_zero_first_argument}
\end{align}
If \(X\ne0\) and
\begin{align}
\operatorname{supp}(X)
\not\subseteq
\operatorname{supp}(Y),
\end{align}
we set
\begin{align}
D(X\|Y)
\coloneqq
+\infty.
\label{eq:umegaki_support_failure}
\end{align}

Suppose henceforth that
\begin{align}
X\ne0,
\qquad
\operatorname{supp}(X)
\subseteq
\operatorname{supp}(Y).
\label{eq:umegaki_support_condition}
\end{align}
For each
\(i\in\mathcal I_X\),
define the finite positive measure
\(\mu_i^Y\)
on
\([0,\|Y\|]\)
by
\begin{align}
\mu_i^Y(\Delta)
\coloneqq
\bra{\psi_i}
E_Y(\Delta)
\ket{\psi_i}
\label{eq:umegaki_Y_spectral_measure}
\end{align}
for every Borel set
\(\Delta\subseteq[0,\|Y\|]\).
In fact,
\(\mu_i^Y\)
is a probability measure, since
\begin{align}
\mu_i^Y([0,\|Y\|])
=
1.
\end{align}
Moreover,
\eqref{eq:umegaki_support_condition}
implies
\begin{align}
\mu_i^Y(\{0\})
&=
\bra{\psi_i}
E_Y(\{0\})
\ket{\psi_i}
\nonumber\\
&=
0,
\label{eq:umegaki_zero_spectral_mass}
\end{align}
because
\(E_Y(\{0\})=I^A-\Pi_Y\)
and
\(\ket{\psi_i}\in\operatorname{supp}(X)
\subseteq\operatorname{supp}(Y)\).

The rigorous meaning of
\eqref{eq:formal_umegaki_trace_expression}
is the extended spectral expression
\begin{align}
D(X\|Y)
\coloneqq
\sum_{i\in\mathcal I_X}
\lambda_i
\int_{(0,\|Y\|]}
\log
\frac{\lambda_i}{\lambda}
\,
\mathrm{d}\mu_i^Y(\lambda).
\label{eq:umegaki_relative_entropy_positive_operators}
\end{align}
This expression is well defined in
\(\mathbb R\cup\{+\infty\}\); to see this, defining
\begin{align}
[a]_-
\coloneqq
\max\{-a,0\},
\end{align}
one has, for every
\(i\in\mathcal I_X\),
\begin{align}
&
\left[
\lambda_i
\int_{(0,\|Y\|]}
\log
\frac{\lambda_i}{\lambda}
\,
\mathrm{d}\mu_i^Y(\lambda)
\right]_-
\nonumber\\
&\le
\lambda_i
\int_{(0,\|Y\|]}
\left[
\log
\frac{\lambda_i}{\lambda}
\right]_-
\,
\mathrm{d}\mu_i^Y(\lambda)
\nonumber\\
&=
\int_{(\lambda_i,\|Y\|]}
\lambda_i
\log
\frac{\lambda}{\lambda_i}
\,
\mathrm{d}\mu_i^Y(\lambda)
\nonumber\\
&\le
\int_{(\lambda_i,\|Y\|]}
\lambda
\,
\mathrm{d}\mu_i^Y(\lambda)
\nonumber\\
&\le
\int_{(0,\|Y\|]}
\lambda
\,
\mathrm{d}\mu_i^Y(\lambda),
\label{eq:umegaki_single_eigenvector_negative_part_bound}
\end{align}
where we used
\begin{align}
\lambda_i
\log
\frac{\lambda}{\lambda_i}
\le
\lambda-\lambda_i
\le
\lambda
\qquad
\text{for \(\lambda>\lambda_i\)}.
\end{align}
Consequently,
\begin{align}
&
\sum_{i\in\mathcal I_X}
\left[
\lambda_i
\int_{(0,\|Y\|]}
\log
\frac{\lambda_i}{\lambda}
\,
\mathrm{d}\mu_i^Y(\lambda)
\right]_-
\nonumber\\
&\le
\sum_{i\in\mathcal I_X}
\int_{(0,\|Y\|]}
\lambda
\,
\mathrm{d}\mu_i^Y(\lambda)
\nonumber\\
&=
\sum_{i\in\mathcal I_X}
\bra{\psi_i}
Y
\ket{\psi_i}
\nonumber\\
&=
\operatorname{Tr}
\left[
\Pi_XY\Pi_X
\right]
\nonumber\\
&\le
\operatorname{Tr}[Y]
<
\infty.
\label{eq:umegaki_negative_part_summable}
\end{align}
Thus, the total negative part is finite, whereas the positive part may
diverge to
\(+\infty\).
It follows that
\eqref{eq:umegaki_relative_entropy_positive_operators}
is unambiguously defined.
Although the individual measures
\(\mu_i^Y\)
may depend on the choice of eigenbasis within a degenerate eigenspace,
their total contribution to
\eqref{eq:umegaki_relative_entropy_positive_operators}
does not.
Hence, the expression is independent of the chosen eigenbasis and of
the enumeration of the eigenvalues.

For general positive-semidefinite trace-class operators,
\(D(X\|Y)\)
need not be nonnegative.
For every \(c>0\),
\begin{align}
D(X\|cY)
=
D(X\|Y)
-
\operatorname{Tr}[X]\log c.
\label{eq:homogeneous_relative_entropy_scaling}
\end{align}
Hence, when \(X\ne0\), the relative entropy can be made arbitrarily
negative by increasing \(c\).

For states
\(\rho,\sigma\in\operatorname{D}(\mathcal H^A)\),
the fidelity~\cite{UHLMANN1976273,Jozsa01121994}
is defined by
\begin{align}
F(\rho,\sigma)
\coloneqq
\left\|
\sqrt{\rho}\sqrt{\sigma}
\right\|_1^2,
\label{eq:fidelity_definition_prelim}
\end{align}
and the purified
distance~\cite{PhysRevA.66.042304,PhysRevA.71.062310,
rastegin2006sinedistancequantumstates,5550419}
is defined by
\begin{align}
P(\rho,\sigma)
\coloneqq
\sqrt{
1-F(\rho,\sigma)
}.
\label{eq:purified_distance_definition_prelim}
\end{align}
For normalized states,
Ref.~\cite[Theorems~4 and~6]{2600498.2600500}
shows that
\begin{align}
D(\rho\|\sigma)
&\ge
-\log
\left[
\left(
\operatorname{Tr}
\left[
\sqrt{\rho}\sqrt{\sigma}
\right]
\right)^2
\right]
\nonumber\\
&\ge
-\log
F(\rho,\sigma).
\label{eq:relative_entropy_fidelity_bound_prelim}
\end{align}
Since
\begin{align}
-\log x
\ge
1-x
\qquad
\text{for every \(x\in(0,1]\)},
\end{align}
it follows that
\begin{align}
P(\rho,\sigma)^2
\le
D(\rho\|\sigma).
\label{eq:decoupling_relative_entropy_controls_PD}
\end{align}
Moreover, the Fuchs--van de Graaf
inequality~\cite[Theorem~1]{761271}
gives
\begin{align}
\frac12
\left\|
\rho-\sigma
\right\|_1
\le
P(\rho,\sigma).
\label{eq:trace_distance_purified_distance_bound_prelim}
\end{align}
Consequently,
\begin{align}
\frac14
\left\|
\rho-\sigma
\right\|_1^2
\le
D(\rho\|\sigma).
\label{eq:decoupling_relative_entropy_controls_trace_distance}
\end{align}

\subsection{Quantum entropy}
\label{subsec:quantum_entropies}

We define the quantum entropy, binary entropy, conditional quantum
entropy, and mutual information used in the IID analysis.
We also recall that the finite-entropy assumption imposed below can be
interpreted as a finite-energy condition with respect to a suitable
Hamiltonian satisfying the Gibbs hypothesis.

For a state
\begin{align}
\rho^A
\in
\operatorname{D}(\mathcal H^A),
\end{align}
define
\begin{align}
\eta(x)
&\coloneqq
-x\log x,
\qquad
x\in(0,1],
&
\eta(0)
&\coloneqq
0.
\end{align}
Since
\(\eta\)
is nonnegative and bounded on
\([0,1]\),
the generalized functional calculus defined in
\eqref{eq:generalized_functional_calculus}
gives a bounded positive-semidefinite operator
\(\eta(\rho^A)\).
The quantum entropy of
\(\rho^A\)
is defined by
\begin{align}
H(A)_\rho
\coloneqq
\operatorname{Tr}
\left[
\eta(\rho^A)
\right]
\in
[0,+\infty],
\label{eq:von_neumann_entropy_definition}
\end{align}
where the trace is understood in the extended sense of
\eqref{eq:extended_trace_positive_operator}.
Although
\(\log\rho^A\)
may be unbounded,
the operator
\(-\rho^A\log\rho^A=\eta(\rho^A)\),
understood through functional calculus, is bounded.
It may not, however, be trace class.
If
\begin{align}
H(A)_\rho
<
+\infty,
\end{align}
then
\(\eta(\rho^A)\)
is trace class and
\begin{align}
H(A)_\rho
=
-
\operatorname{Tr}
\left[
\rho^A\log\rho^A
\right].
\label{eq:von_neumann_entropy_trace_expression}
\end{align}

For the spectral decomposition
\begin{align}
\rho^A
=
\sum_{i\in\mathcal I_{\rho^A}}
\lambda_i^A
\ket{\psi_i}\bra{\psi_i}^A
\label{eq:rho_A_spectral_decomposition_entropy}
\end{align}
introduced in
\eqref{eq:spectral_decomposition_trace_class_positive},
one has
\begin{align}
H(A)_\rho
=
-
\sum_{i\in\mathcal I_{\rho^A}}
\lambda_i^A
\log\lambda_i^A.
\label{eq:entropy_eigenvalue_formula}
\end{align}
The series in
\eqref{eq:entropy_eigenvalue_formula}
takes values in
\([0,+\infty]\).
For
\(n\in\{1,2,\ldots\}\),
if
\(H(A)_\rho<+\infty\),
then
\begin{align}
H
\left(
A^n
\right)_{\rho^{\otimes n}}
=
nH(A)_\rho
<
+\infty.
\label{eq:entropy_tensor_power_additivity}
\end{align}

For
\(t\in[0,1]\),
the binary entropy is defined by
\begin{align}
h_2(t)
\coloneqq
-t\log t
-
(1-t)\log(1-t),
\label{eq:binary_entropy_definition}
\end{align}
with the convention
\begin{align}
0\log0
\coloneqq
0.
\end{align}
In particular,
\begin{align}
h_2(0)
=
h_2(1)
=
0.
\end{align}

We next formulate the finite-energy characterization of finite
entropy.
A possibly unbounded self-adjoint operator \(O\) on
\(\mathcal H^A\)
is called positive semidefinite if
\begin{align}
\operatorname{spec}(O)
\subseteq
[0,+\infty).
\end{align}
Let
\(E_O\)
denote its projection-valued spectral measure and define its bounded
spectral truncations by
\begin{align}
O_m
&\coloneqq
O E_O([0,m])
\nonumber\\
&=
\int_{[0,m]}
\lambda\,
\mathrm{d}E_O(\lambda),
\qquad
m\in\{1,2,\ldots\}.
\label{eq:positive_operator_spectral_truncation}
\end{align}
The extended trace of \(O\) is defined by
\begin{align}
\operatorname{Tr}[O]
\coloneqq
\sup_{m\in\{1,2,\ldots\}}
\operatorname{Tr}[O_m]
\in
[0,+\infty].
\label{eq:extended_trace_unbounded_positive_operator}
\end{align}
If \(O\) is bounded, this definition agrees with
\eqref{eq:extended_trace_positive_operator}.

Let \(H\) be a possibly unbounded positive-semidefinite self-adjoint
operator on
\(\mathcal H^A\),
and let
\begin{align}
H_m
\coloneqq
H E_H([0,m])
\end{align}
denote its bounded spectral truncations.
For
\(\rho^A\in\operatorname{T}_+(\mathcal H^A)\),
the energy of
\(\rho^A\)
with respect to \(H\) is defined by
\begin{align}
\operatorname{Tr}[H\rho^A]
\coloneqq
\sup_{m\in\{1,2,\ldots\}}
\operatorname{Tr}
\left[
(\rho^A)^{1/2}
H_m
(\rho^A)^{1/2}
\right]
\in
[0,+\infty].
\label{eq:extended_energy_expectation}
\end{align}

For
\(\beta>0\),
the bounded positive-semidefinite operator
\(e^{-\beta H}\)
is defined by
\begin{align}
e^{-\beta H}
\coloneqq
\int_{[0,+\infty)}
e^{-\beta\lambda}
\,
\mathrm{d}E_H(\lambda).
\label{eq:exponential_unbounded_hamiltonian}
\end{align}
The partition function of \(H\) at inverse temperature
\(\beta\)
is
\begin{align}
Z_H(\beta)
\coloneqq
\operatorname{Tr}
\left[
e^{-\beta H}
\right]
\in
(0,+\infty].
\label{eq:partition_function_definition}
\end{align}
If
\(Z_H(\beta)<+\infty\),
then
\begin{align}
\gamma_H(\beta)
\coloneqq
\frac{
e^{-\beta H}
}{
Z_H(\beta)
}
\label{eq:gibbs_state_definition}
\end{align}
is the Gibbs state of \(H\) at inverse temperature
\(\beta\).

Reference~\cite[Proposition~1]{shirokov2006entropy}
shows that
\begin{align}
H(A)_\rho
<
+\infty
\end{align}
if and only if there exists a positive-semidefinite self-adjoint
operator \(H\) on
\(\mathcal H^A\)
such that
\begin{align}
\operatorname{Tr}[H\rho^A]
<
+\infty
\label{eq:finite_entropy_finite_energy_condition}
\end{align}
and
\begin{align}
Z_H(\beta)
=
\operatorname{Tr}
\left[
e^{-\beta H}
\right]
<
+\infty
\qquad
\text{for every \(\beta>0\)}.
\label{eq:finite_entropy_partition_condition}
\end{align}
Condition
\eqref{eq:finite_entropy_partition_condition}
is also known as the Gibbs
hypothesis~\cite{winter2016tight}.
It ensures that the Gibbs state in
\eqref{eq:gibbs_state_definition}
is well defined at every positive inverse temperature.

The Hamiltonian in this characterization may be chosen diagonal in a
complete orthonormal eigenbasis of
\(\rho^A\).
More precisely, let
\(\mathcal J_{\rho^A}\)
be a finite or countably infinite index set and write
\begin{align}
\rho^A
=
\sum_{i\in\mathcal J_{\rho^A}}
\lambda_i^A
\ket{\psi_i}\bra{\psi_i}^A,
\qquad
\lambda_i^A\ge0,
\qquad
\sum_{i\in\mathcal J_{\rho^A}}
\lambda_i^A
=
1,
\label{eq:rho_complete_eigenbasis_energy}
\end{align}
where
\(\{\ket{\psi_i}^A\}_{i\in\mathcal J_{\rho^A}}\)
is an orthonormal basis of all of
\(\mathcal H^A\).
Unlike
\(\mathcal I_{\rho^A}\),
the set
\(\mathcal J_{\rho^A}\)
also includes basis vectors in
\(\ker(\rho^A)\).

For a sequence
\begin{align}
\{g_i\}_{i\in\mathcal J_{\rho^A}}
\subseteq
[0,+\infty),
\end{align}
define
\begin{align}
H
=
\sum_{i\in\mathcal J_{\rho^A}}
g_i
\ket{\psi_i}\bra{\psi_i}^A
\label{eq:diagonal_effective_hamiltonian}
\end{align}
on the natural domain
\begin{align}
\operatorname{dom}(H)
=
\left\{
\ket{\varphi}\in\mathcal H^A:
\sum_{i\in\mathcal J_{\rho^A}}
g_i^2
\left|
\bra{\psi_i}\ket{\varphi}
\right|^2
<
+\infty
\right\}.
\end{align}
Then
\begin{align}
\operatorname{Tr}[H\rho^A]
=
\sum_{i\in\mathcal J_{\rho^A}}
\lambda_i^A g_i
\label{eq:diagonal_hamiltonian_energy}
\end{align}
and
\begin{align}
Z_H(\beta)
=
\sum_{i\in\mathcal J_{\rho^A}}
e^{-\beta g_i}.
\label{eq:diagonal_hamiltonian_partition_function}
\end{align}
Consequently, the finite-entropy characterization above is equivalent
to the existence of a sequence
\(\{g_i\}_{i\in\mathcal J_{\rho^A}}\)
such that
\begin{align}
\sum_{i\in\mathcal J_{\rho^A}}
\lambda_i^A g_i
&<
+\infty,
\\
\sum_{i\in\mathcal J_{\rho^A}}
e^{-\beta g_i}
&<
+\infty
\qquad
\text{for every \(\beta>0\)}.
\label{eq:finite_entropy_diagonal_hamiltonian_conditions}
\end{align}

Representative continuous-variable systems possess natural
Hamiltonians satisfying the Gibbs hypothesis, justifying the finite-entropy assumptions.
For example, on a single-mode bosonic Fock space, the number operator
is
\begin{align}
N
\coloneqq
\sum_{n=0}^{\infty}
n
\ket{n}\bra{n},
\label{eq:number_operator_definition}
\end{align}
where the sum defines \(N\) on its natural domain.
Its partition function is
\begin{align}
Z_N(\beta)
&=
\operatorname{Tr}
\left[
e^{-\beta N}
\right]
\nonumber\\
&=
\sum_{n=0}^{\infty}
e^{-\beta n}
\nonumber\\
&=
\frac{1}{
1-e^{-\beta}
}
<
+\infty
\qquad
(\beta>0).
\label{eq:number_operator_partition_function}
\end{align}
Consequently, every state satisfying
\begin{align}
\operatorname{Tr}[N\rho^A]
<
+\infty,
\end{align}
where the energy is understood in the sense of
\eqref{eq:extended_energy_expectation},
has finite entropy.
More generally, the standard Hamiltonian of any finite collection of
harmonic oscillators, as well as the corresponding total number
operator, satisfies the Gibbs hypothesis.
In particular, every finite-mode bosonic Gaussian state has finite
total mean photon number and hence finite entropy.

For a state
\begin{align}
\rho^{AR}
\in
\operatorname{D}
\left(
\mathcal H^A
\otimes
\mathcal H^R
\right)
\end{align}
satisfying
\begin{align}
H(A)_\rho
<
+\infty,
\end{align}
the conditional quantum entropy is defined by
Ref.~\cite[Definition~3]{doi:10.1137/S0040585X97985121}
as
\begin{align}
H(A|R)_\rho
\coloneqq
H(A)_\rho
-
D
\left(
\rho^{AR}
\middle\|
\rho^A\otimes\rho^R
\right).
\label{eq:extended_conditional_entropy}
\end{align}
Under this assumption,
\begin{align}
D
\left(
\rho^{AR}
\middle\|
\rho^A\otimes\rho^R
\right)
<
+\infty,
\end{align}
and
\(H(A|R)_\rho\)
is finite and satisfies
\begin{align}
-H(A)_\rho
\le
H(A|R)_\rho
\le
H(A)_\rho.
\label{eq:conditional_entropy_bounds_finite_marginal}
\end{align}

The quantum mutual information is defined by
\begin{align}
I(A:R)_\rho
&\coloneqq
D
\left(
\rho^{AR}
\middle\|
\rho^A\otimes\rho^R
\right)
\nonumber\\
&=
H(A)_\rho
-
H(A|R)_\rho.
\label{eq:I_AR_definition_finite_A_general}
\end{align}
It is finite under the assumption
\(H(A)_\rho<+\infty\)
and satisfies
\begin{align}
0
\le
I(A:R)_\rho
\le
2H(A)_\rho.
\label{eq:mutual_information_bounds_finite_marginal}
\end{align}

For a normalized pure state
\(\ket{\psi}^{ABR}\)
satisfying
\begin{align}
H(A)_\psi
<
+\infty,
\end{align}
pure-state duality
\cite{doi:10.1137/S0040585X97985121,holevo2010mutual}
gives
\begin{align}
H(A|B)_\psi
=
-H(A|R)_\psi.
\label{eq:conditional_entropy_duality_pure_state_finite_entropy_A}
\end{align}
Accordingly,
\begin{align}
I(A:R)_\psi
&=
H(A)_\psi
-
H(A|R)_\psi,
\label{eq:I_AR_definition_finite_A}
\\
I(A:B)_\psi
&\coloneqq
H(A)_\psi
-
H(A|B)_\psi
\label{eq:I_AB_definition_finite_A}
\end{align}
are finite.
These definitions remain meaningful even when the alternative
expressions
\begin{align}
H(A)_\psi
+
H(R)_\psi
-
H(AR)_\psi
\end{align}
and
\begin{align}
H(A)_\psi
+
H(B)_\psi
-
H(AB)_\psi
\end{align}
are not well defined because they may contain an indeterminate
combination of the form
\(+\infty-\infty\).

More generally, let
\begin{align}
\rho^{ABR}
\in
\operatorname{D}
\left(
\mathcal H^A
\otimes
\mathcal H^B
\otimes
\mathcal H^R
\right)
\end{align}
satisfy
\begin{align}
H(A)_\rho
<
+\infty.
\end{align}
The quantum conditional mutual information is defined by
\begin{align}
I(A:R|B)_\rho
&\coloneqq
H(A|B)_\rho
-
H(A|BR)_\rho
\label{eq:conditional_mutual_information_definition_finite_A}
\\
&=
I(A:BR)_\rho
-
I(A:B)_\rho.
\label{eq:conditional_mutual_information_mutual_information_difference}
\end{align}
Both conditional entropies in
\eqref{eq:conditional_mutual_information_definition_finite_A}
are finite under the assumption
\(H(A)_\rho<+\infty\).
Thus, neither expression contains an indeterminate subtraction of
infinite quantities.
By strong subadditivity and the finite conditional-entropy
bounds, the conditional mutual information satisfies
\begin{align}
0
\le
I(A:R|B)_\rho
\le
2H(A)_\rho
<
+\infty.
\label{eq:conditional_mutual_information_bounds_finite_A}
\end{align}

If, in addition,
\begin{align}
H(B)_\rho
<
+\infty,
\end{align}
then
\(H(AB)_\rho<+\infty\),
and the chain rules take the form
\begin{align}
\label{eq:mutual_information_chain_rule}
I(AB:R)_\rho
&=
I(B:R)_\rho
+
I(A:R|B)_\rho,
&
I(AB:R)_\rho
&=
I(A:R)_\rho
+
I(B:R|A)_\rho.
\end{align}

\subsection{Sandwiched R\'enyi relative entropy}
\label{subsec:sandwiched_renyi_relative_entropy}

We define the sandwiched R\'enyi relative entropy for
positive-semidefinite operators on separable Hilbert spaces.
Because negative powers of a positive-semidefinite operator may be
unbounded in infinite dimensions, we first state the domain condition
under which the sandwiched product is well defined.
We then verify this condition for the reference operator
\(I^A\otimes\rho^R\)
appearing in the conditional sandwiched R\'enyi entropy.

Let
\(\mathcal H\)
be a separable Hilbert space, and let
\(\alpha>1\).
For states
\(X,Y\)
on a finite-dimensional Hilbert space, the sandwiched R\'enyi
relative entropy is commonly written as~\cite{muller2013quantum,wilde2014strong}
\begin{align}
\widetilde D_\alpha(X\|Y)
\coloneqq
\frac{1}{\alpha-1}
\log
\operatorname{Tr}
\left[
\left(
Y^{\frac{1-\alpha}{2\alpha}}
X
Y^{\frac{1-\alpha}{2\alpha}}
\right)^\alpha
\right],
\label{eq:sandwiched_renyi_divergence_finite_dimensional_form}
\end{align}
with the convention that
\begin{align}
\widetilde D_\alpha(X\|Y)
=
+\infty
\end{align}
unless
\begin{align}
\operatorname{supp}(X)
\subseteq
\operatorname{supp}(Y).
\label{eq:finite_dimensional_sandwiched_support_condition}
\end{align}
The trace expression in
\eqref{eq:sandwiched_renyi_divergence_finite_dimensional_form}
is, up to the logarithmic prefactor and normalization convention, the
case \(p=2\alpha\) of the two-parameter family discussed in
Ref.~\cite[Sec.~4.4.3]{10.1093/acprof:oso/9780199652495.003.0004}.
For nonnormalized positive-semidefinite operators, different
normalization conventions are used in the literature.
Reference~\cite[Definition~2]{muller2013quantum}
includes a factor
\((\operatorname{Tr}[X])^{-1}\)
inside the logarithm, whereas
Ref.~\cite[Definition~4]{wilde2014strong}
uses the expression in
\eqref{eq:sandwiched_renyi_divergence_finite_dimensional_form}
without this factor.
The two conventions coincide when \(X\) is a state.
In this work, we adopt the latter, unnormalized convention.

In finite dimensions, the support condition
\eqref{eq:finite_dimensional_sandwiched_support_condition}
is sufficient to make
\eqref{eq:sandwiched_renyi_divergence_finite_dimensional_form}
well defined, because the negative power of \(Y\) is bounded on
\(\operatorname{supp}(Y)\).
In infinite dimensions, this is no longer automatic.
Even for
\begin{align}
X,Y
\in
\operatorname{T}_+(\mathcal H),
\end{align}
the generalized power
\begin{align}
Y^{\frac{1-\alpha}{2\alpha}}
\end{align}
may be unbounded because the nonzero spectrum of \(Y\) may accumulate
at \(0\).
Consequently, the formal product
\begin{align}
Y^{\frac{1-\alpha}{2\alpha}}
X
Y^{\frac{1-\alpha}{2\alpha}}
\label{eq:formal_sandwiched_product}
\end{align}
need not be a bounded or closed operator.
The product in
\eqref{eq:formal_sandwiched_product}
is initially defined on the natural domain
\begin{align}
\Biggl\{
\ket{\psi}
\in
\operatorname{dom}
\left(
Y^{\frac{1-\alpha}{2\alpha}}
\right):
X
Y^{\frac{1-\alpha}{2\alpha}}
\ket{\psi}
\in
\operatorname{dom}
\left(
Y^{\frac{1-\alpha}{2\alpha}}
\right)
\Biggr\}.
\label{eq:natural_domain_formal_sandwiched_product}
\end{align}
This domain may not be all of
\(\mathcal H\),
and the resulting operator may not be closed or self-adjoint.
Its \(\alpha\)-th power therefore cannot in general be defined
directly by functional calculus.
The extended trace of a positive-semidefinite self-adjoint operator
is well defined as an element of
\([0,+\infty]\)
by
\eqref{eq:extended_trace_unbounded_positive_operator}; however, this fact applies only after the relevant positive-semidefinite self-adjoint operator has been constructed.
It does not by itself establish that
\eqref{eq:formal_sandwiched_product}
has such an interpretation.
Infinite-dimensional formulations addressing this issue were
developed in, e.g.,
Refs.~\cite{berta2018renyi,jencova2018renyi,jenvcova2021renyi,
gu2019interpolationquasinoncommutativelpspaces,mosonyi2023strong}.

The following lemma is the specialization of Ref.~\cite[Lemma~3.1]{mosonyi2023strong} to the case of the sandwiched R\'enyi relative entropy \(z=\alpha\).

\begin{lemma}[Domain condition for the sandwiched product]
\label{lem:mosonyi_domain_condition_sandwiched}
Let
\(\mathcal H\)
be a separable Hilbert space, let
\(\alpha>1\), and let
\begin{align}
0
\ne
X,Y
\in
\operatorname{B}_+(\mathcal H).
\end{align}
The following conditions are equivalent:
\begin{enumerate}
\item
There exists an operator
\(K\in\operatorname{B}(\mathcal H)\)
such that
\begin{align}
X
=
Y^{\frac{\alpha-1}{2\alpha}}
K
Y^{\frac{\alpha-1}{2\alpha}}.
\label{eq:sandwiched_domain_factorization_condition}
\end{align}

\item
The range inclusion
\begin{align}
\operatorname{ran}
\left(
X^{1/2}
\right)
\subseteq
\operatorname{ran}
\left(
Y^{\frac{\alpha-1}{2\alpha}}
\right)
\label{eq:sandwiched_domain_range_condition}
\end{align}
holds.

\item
The composition
\begin{align}
Y^{\frac{1-\alpha}{2\alpha}}
X^{1/2},
\label{eq:sandwiched_domain_bounded_product_condition}
\end{align}
initially defined on its natural domain, is defined on all of
\(\mathcal H\)
and belongs to
\(\operatorname{B}(\mathcal H)\).

\item
There exists a constant
\(\lambda\ge0\)
such that
\begin{align}
X
\le
\lambda
Y^{\frac{\alpha-1}{\alpha}}.
\label{eq:sandwiched_domain_domination_condition}
\end{align}
\end{enumerate}
All powers of \(Y\), including its negative powers, are understood in
the sense of
\eqref{eq:generalized_real_power}.
\end{lemma}

Suppose that the equivalent conditions in
Lemma~\ref{lem:mosonyi_domain_condition_sandwiched}
hold.
Among all operators \(K\) satisfying
\eqref{eq:sandwiched_domain_factorization_condition},
there exists a unique positive-semidefinite operator
\begin{align}
X_{Y,\alpha}
\in
\operatorname{B}_+(\mathcal H)
\end{align}
satisfying
\begin{align}
\operatorname{supp}(X_{Y,\alpha})
\subseteq
\operatorname{supp}(Y).
\end{align}
It is given by
\begin{align}
X_{Y,\alpha}
\coloneqq
\left(
Y^{\frac{1-\alpha}{2\alpha}}
X^{1/2}
\right)
\left(
Y^{\frac{1-\alpha}{2\alpha}}
X^{1/2}
\right)^\dagger
\label{eq:sandwiched_operator_definition}
\end{align}
and satisfies
\begin{align}
X
=
Y^{\frac{\alpha-1}{2\alpha}}
X_{Y,\alpha}
Y^{\frac{\alpha-1}{2\alpha}}.
\label{eq:sandwiched_operator_factorization}
\end{align}
Moreover,
\begin{align}
X_{Y,\alpha}
=
\overline{
Y^{\frac{1-\alpha}{2\alpha}}
X
Y^{\frac{1-\alpha}{2\alpha}}
},
\label{eq:sandwiched_operator_closure}
\end{align}
where the product on the right-hand side is initially defined on
\eqref{eq:natural_domain_formal_sandwiched_product}, and under the equivalent conditions of Lemma~\ref{lem:mosonyi_domain_condition_sandwiched}, this product is closable and its closure is the bounded operator \(X_{Y,\alpha}\).
We use
\eqref{eq:sandwiched_operator_definition},
rather than the formal product
\eqref{eq:formal_sandwiched_product},
as the primary definition.

Following
Ref.~\cite[Definition~3.9]{mosonyi2023strong},
for
\begin{align}
0
\ne
X,Y
\in
\operatorname{B}_+(\mathcal H),
\end{align}
we define
\begin{align}
\widetilde Q_\alpha(X\|Y)
\coloneqq
\begin{cases}
\displaystyle
\operatorname{Tr}
\left[
\left(
X_{Y,\alpha}
\right)^\alpha
\right],
&
\begin{array}{l}
\text{if the equivalent conditions in}\\[-0.5ex]
\text{Lemma~\ref{lem:mosonyi_domain_condition_sandwiched}
hold},
\end{array}
\\[3ex]
+\infty,
&
\text{otherwise},
\end{cases}
\label{eq:sandwiched_Q_alpha_infinite_dimensional}
\end{align}
and
\begin{align}
\widetilde D_\alpha(X\|Y)
\coloneqq
\frac{1}{\alpha-1}
\log
\widetilde Q_\alpha(X\|Y),
\label{eq:sandwiched_renyi_divergence_infinite_dimensional}
\end{align}
where
\begin{align}
\log(+\infty)
\coloneqq
+\infty.
\end{align}
In the notation of
Ref.~\cite[Definition~3.9]{mosonyi2023strong},
the quantities
\(\widetilde Q_\alpha\)
and
\(\widetilde D_\alpha\)
defined here correspond to
\(Q_\alpha^*\)
and
\(D_\alpha^*\),
respectively.
Note that we do not include the additional normalization term for trace-class operators introduced in Ref.~\cite[Definition~3.10]{mosonyi2023strong}; thus, the tilde in our notation indicates the sandwiched R\'enyi quantity, rather than the trace-normalized quantity used in Ref.~\cite{mosonyi2023strong}.

If the equivalent domain conditions in
Lemma~\ref{lem:mosonyi_domain_condition_sandwiched} hold, then
\begin{align}
X_{Y,\alpha}
\in
\operatorname{B}_+(\mathcal H)
\end{align}
and hence
\begin{align}
\left(
X_{Y,\alpha}
\right)^\alpha
\in
\operatorname{B}_+(\mathcal H).
\end{align}
Its trace is therefore well defined in the extended sense:
\begin{align}
\operatorname{Tr}
\left[
\left(
X_{Y,\alpha}
\right)^\alpha
\right]
\in
[0,+\infty].
\end{align}
Since \(X\ne0\), one has
\begin{align}
\widetilde Q_\alpha(X\|Y)
>
0,
\end{align}
and consequently
\begin{align}
\widetilde D_\alpha(X\|Y)
\in
(-\infty,+\infty].
\end{align}

We next verify that the domain condition is automatically satisfied
for the reference operator appearing in conditional sandwiched
R\'enyi entropies.

\begin{proposition}[Domain condition for conditional sandwiched
R\'enyi entropies]
\label{prop:conditional_sandwiched_domain_condition}
Let
\(\mathcal H^A\)
be finite-dimensional, let
\(\mathcal H^R\)
be separable, let
\(\alpha>1\), and let
\begin{align}
0
\ne
 Z^{AR}
\in
\operatorname{T}_+
\left(
\mathcal H^A
\otimes
\mathcal H^R
\right).
\end{align}
Define
\begin{align}
 Z^R
\coloneqq
\operatorname{Tr}_A[ Z^{AR}].
\end{align}
Then
\begin{align}
0
\ne
I^A\otimes Z^R
\in
\operatorname{T}_+
\left(
\mathcal H^A
\otimes
\mathcal H^R
\right),
\end{align}
and the states
\(
 Z^{AR}
\)
and
\(
I^A\otimes Z^R
\)
satisfy the equivalent conditions in
Lemma~\ref{lem:mosonyi_domain_condition_sandwiched}.
\end{proposition}

\begin{proof}
Since \(A\) is finite-dimensional,
\(I^A\otimes Z^R\)
is trace class and
\begin{align}
\operatorname{Tr}
\left[
I^A\otimes Z^R
\right]
&=
|A|
\operatorname{Tr}[ Z^R]
\nonumber\\
&=
|A|
\operatorname{Tr}[ Z^{AR}]
<
+\infty.
\label{eq:IA_tensor_rhoR_trace_class}
\end{align}
Moreover,
\( Z^{AR}\ne0\)
implies
\( Z^R\ne0\),
and hence
\(I^A\otimes Z^R\ne0\).

We first show that every
\begin{align}
Z^{AR}
\in
\operatorname{T}_+
\left(
\mathcal H^A
\otimes
\mathcal H^R
\right)
\end{align}
satisfies
\begin{align}
Z^{AR}
\le
|A|
I^A
\otimes
Z^R.
\label{eq:general_domination_by_marginal}
\end{align}
Consider first
\begin{align}
Z^{AR}
=
\ket{\psi}\bra{\psi}^{AR}.
\end{align}
Let
\begin{align}
\ket{\psi}^{AR}
=
\sum_{j=0}^{r-1}
\sqrt{p_j}
\ket{a_j}^{A}
\otimes
\ket{r_j}^{R}
\end{align}
be a Schmidt decomposition, where
\(r\le|A|\).
Then
\begin{align}
Z^R
=
\sum_{j=0}^{r-1}
p_j
\ket{r_j}\bra{r_j}^{R}.
\end{align}
For every
\(\ket{\varphi}^{AR}\),
the Cauchy--Schwarz inequality gives
\begin{align}
\left|
\bra{\psi}\ket{\varphi}
\right|^2
&=
\left|
\sum_{j=0}^{r-1}
\sqrt{p_j}
\left(\bra{a_j}^{A}\otimes \bra{r_j}^{R}\right)
\ket{\varphi}^{AR}
\right|^2
\nonumber\\
&\le
r
\sum_{j=0}^{r-1}
p_j
\left|
\left(\bra{a_j}^{A}\otimes \bra{r_j}^{R}\right)
\ket{\varphi}^{AR}
\right|^2
\nonumber\\
&\le
r
\bra{\varphi}
\left(
I^A\otimes Z^R
\right)
\ket{\varphi}
\nonumber\\
&\le
|A|
\bra{\varphi}
\left(
I^A\otimes Z^R
\right)
\ket{\varphi}.
\end{align}
Therefore,
\begin{align}
\ket{\psi}\bra{\psi}^{AR}
\le
|A|
I^A\otimes Z^R.
\label{eq:rank_one_domination_by_marginal}
\end{align}
Applying
\eqref{eq:rank_one_domination_by_marginal}
termwise to the spectral decomposition of \(Z^{AR}\) and summing
gives
\eqref{eq:general_domination_by_marginal}; specifically, the positive partial sums converge to \(Z^{AR}\) in trace norm, their partial traces converge to \(Z^R\)
in trace norm, and the positive cone is closed in the operator norm.

For every
\begin{align}
t
\in
[0,\| Z^R\|],
\end{align}
one has
\begin{align}
t
\le
\| Z^R\|^{1/\alpha}
t^{\frac{\alpha-1}{\alpha}}.
\end{align}
Functional calculus therefore gives
\begin{align}
I^A\otimes Z^R
\le
\| Z^R\|^{1/\alpha}
\left(
I^A\otimes Z^R
\right)^{\frac{\alpha-1}{\alpha}}.
\label{eq:marginal_power_domination}
\end{align}
Combining
\eqref{eq:general_domination_by_marginal}
and
\eqref{eq:marginal_power_domination},
we obtain
\begin{align}
 Z^{AR}
\le
|A|
\| Z^R\|^{1/\alpha}
\left(
I^A\otimes Z^R
\right)^{\frac{\alpha-1}{\alpha}}.
\label{eq:conditional_renyi_mosonyi_domination}
\end{align}
Thus,
\eqref{eq:sandwiched_domain_domination_condition}
holds with
\begin{align}
X
&=
 Z^{AR},
&
Y
&=
I^A\otimes Z^R,
&
\lambda
&=
|A|
\| Z^R\|^{1/\alpha}.
\end{align}
The conclusion follows from
Lemma~\ref{lem:mosonyi_domain_condition_sandwiched}.
\end{proof}

When
\(\mathcal H^A\)
is finite-dimensional and
\(\mathcal H^R\)
is separable, for every
\(\alpha>1\)
and every nonzero
\begin{align}
\rho^{AR}
\in
\operatorname{T}_+
\left(
\mathcal H^A
\otimes
\mathcal H^R
\right),
\end{align}
we define the sandwiched conditional R\'enyi entropy of \(A\)
conditioned on \(R\) by
\begin{align}
\widetilde H_\alpha(A|R)_\rho
\coloneqq
-
\widetilde D_\alpha
\left(
\rho^{AR}
\middle\|
I^A\otimes\rho^R
\right),
\label{eq:sandwiched_conditional_entropy_downarrow}
\end{align}
where
\begin{align}
\rho^R
\coloneqq
\operatorname{Tr}_A[\rho^{AR}].
\end{align}
Proposition~\ref{prop:conditional_sandwiched_domain_condition}
shows that the domain condition in the definition of the sandwiched
R\'enyi relative entropy is automatically satisfied.
Thus,
\eqref{eq:sandwiched_conditional_entropy_downarrow}
is well defined as an extended real number even when
\(\mathcal H^R\)
is infinite-dimensional.

For normalized states, the following lemma establishes the
order-one limit.
The right-sided limit under bounded domination is given in
Ref.~\cite[Theorem~13]{berta2018renyi}, and the identification with the sandwiched R\'enyi relative entropy used here is given in Ref.~\cite[Theorem~3.3]{jencova2018renyi}.

\begin{lemma}[Order-one limit of conditional sandwiched R\'enyi
entropy]
\label{lem:conditional_sandwiched_order_one_limit}
Let
\(\mathcal H^A\)
be finite-dimensional, let
\(\mathcal H^R\)
be separable, and let
\begin{align}
\rho^{AR}
\in
\operatorname{D}
\left(
\mathcal H^A
\otimes
\mathcal H^R
\right).
\end{align}
Then
\begin{align}
\lim_{\alpha\downarrow1}
\widetilde H_\alpha(A|R)_\rho
=
H(A|R)_\rho.
\label{eq:conditional_sandwiched_entropy_order_one_limit}
\end{align}
\end{lemma}

\begin{proof}
By
\eqref{eq:general_domination_by_marginal},
\begin{align}
\rho^{AR}
\le
|A|
I^A\otimes\rho^R.
\label{eq:conditional_sandwiched_order_one_domination}
\end{align}
Thus,
\(\rho^{AR}\)
is boundedly dominated by
\(I^A\otimes\rho^R\).
The right-sided order-one limit under bounded domination in
Ref.~\cite[Theorem~13, Eq.~(66)]{berta2018renyi},
together with the identification in
Ref.~\cite[Theorem~3.3]{jencova2018renyi}
and the homogeneity relation in
Ref.~\cite[Eq.~(14)]{jencova2018renyi},
therefore yields
\begin{align}
\lim_{\alpha\downarrow1}
\widetilde D_\alpha
\left(
\rho^{AR}
\middle\|
I^A\otimes\rho^R
\right)
=
D
\left(
\rho^{AR}
\middle\|
I^A\otimes\rho^R
\right),
\label{eq:conditional_sandwiched_divergence_order_one_limit}
\end{align}
where the homogeneity relation accounts for the nonnormalized second argument
\(I^A\otimes\rho^R\);
see also the discussion preceding
Ref.~\cite[Definition~3.9]{mosonyi2023strong}
for the agreement with the definition used here.

Since \(A\) is finite-dimensional,
\begin{align}
H(A)_\rho
\le
\log|A|
<
+\infty.
\end{align}
The relative-entropy chain rule
\begin{align}
D
\left(
\rho^{AR}
\middle\|
\rho^A\otimes\rho^R
\right)
=
D
\left(
\rho^{AR}
\middle\|
I^A\otimes\rho^R
\right)
+
H(A)_\rho
\label{eq:relative_entropy_chain_rule_finite_A}
\end{align}
therefore implies
\begin{align}
D
\left(
\rho^{AR}
\middle\|
I^A\otimes\rho^R
\right)
&=
D
\left(
\rho^{AR}
\middle\|
\rho^A\otimes\rho^R
\right)
-
H(A)_\rho
\nonumber\\
&=
-
H(A|R)_\rho.
\label{eq:conditional_entropy_as_relative_entropy_to_marginal}
\end{align}
Combining
\eqref{eq:sandwiched_conditional_entropy_downarrow},
\eqref{eq:conditional_sandwiched_divergence_order_one_limit}, and
\eqref{eq:conditional_entropy_as_relative_entropy_to_marginal}
proves
\eqref{eq:conditional_sandwiched_entropy_order_one_limit}.
\end{proof}

\section{One-shot decoupling for finite-dimensional input and
separable reference and output systems}
\label{sec:finite_input_infinite_reference_one_shot_decoupling}

In this section, we formulate and analyze one-shot decoupling for a
general completely positive map with finite-dimensional input and
separable, possibly infinite-dimensional reference and output
systems.
In Subsection~\ref{subsec:finite_input_decoupling_formulation},
we formulate the one-shot decoupling task for a nonzero bounded
completely positive map, using the Haar-averaged
relative-entropy error.
In Subsection~\ref{subsec:optimal_trace_inequality_separable}, we improve the key trace inequality for positive-semidefinite trace-class operators on separable Hilbert spaces introduced in Ref.~\cite{cheng2025sharpestimatesquantumcovering}, by extending the optimal bound in Ref.~\cite{gour2026optimaltraceinequalitiessingleshot} to infinite dimensions.
In Subsection~\ref{subsec:haar_averaged_Jensen_trace_inequality}, we present a Haar-averaged Jensen's trace inequality to avoid a problematic operator-concavity step in the existing analysis of the error exponent of one-shot decoupling in Ref.~\cite{berta2026tightanyshotquantumdecoupling} (see also Remark~\ref{rem:tight_any_shot_invalid_log_concavity_step} for the details of this point).
In Subsection~\ref{subsec:haar_averaged_relative_entropy_decoupling_bound}, we combine these inequalities with the Haar-random-unitary second-moment identity to bound the decoupling error.
Finally, in Subsection~\ref{subsec:finite_input_one_shot_decoupling_results}, we derive the one-shot decoupling bounds and their asymptotic error exponent.

\subsection{Formulation of one-shot decoupling}
\label{subsec:finite_input_decoupling_formulation}

We formulate the one-shot decoupling task analyzed in this section.

Let
\(\mathcal H^A\)
be finite-dimensional, and let
\(\mathcal H^R\) and \(\mathcal H^E\)
be separable and possibly infinite-dimensional.
Let
\begin{align}
\rho^{AR}
\in
\operatorname{D}
\left(
\mathcal H^A
\otimes
\mathcal H^R
\right),
\end{align}
and let
\begin{align}
\mathcal E^{A\to E}
:
\operatorname{T}
\left(
\mathcal H^A
\right)
\longrightarrow
\operatorname{T}
\left(
\mathcal H^E
\right)
\end{align}
be a nonzero bounded completely positive map.
Define
\begin{align}
|A|
&\coloneqq
\dim\mathcal H^A,
&
\pi^A
&\coloneqq
\frac{
I^A
}{
|A|
}.
\end{align}

The normalization factor and the corresponding normalized completely
positive map are
\begin{align}
\kappa_{\mathcal E}
&\coloneqq
\operatorname{Tr}
\left[
\mathcal E^{A\to E}
\left(
\pi^A
\right)
\right]
>
0,
\\
\widehat{\mathcal E}^{A\to E}
&\coloneqq
\frac{1}{
\kappa_{\mathcal E}
}
\mathcal E^{A\to E}.
\label{eq:finite_input_normalized_map}
\end{align}
By construction,
\begin{align}
\operatorname{Tr}
\left[
\widehat{\mathcal E}^{A\to E}
\left(
\pi^A
\right)
\right]
=
1.
\end{align}
This normalization fixes only the overall scale of
\(\mathcal E^{A\to E}\).
In particular,
\(\widehat{\mathcal E}^{A\to E}\)
need not be trace-preserving.

Let
\begin{align}
\mathcal H^{A'}
\simeq
\mathcal H^A,
\end{align}
and let
\(\Phi^{AA'}\)
denote the normalized maximally entangled state.
Define the normalized Choi state and its \(E\)-marginal by
\begin{align}
\omega_{\mathcal E}^{A'E}
&\coloneqq
\left(
\widehat{\mathcal E}^{A\to E}
\otimes
\id^{A'}
\right)
\left(
\Phi^{AA'}
\right),
&
\omega_{\mathcal E}^{E}
&\coloneqq
\operatorname{Tr}_{A'}
\left[
\omega_{\mathcal E}^{A'E}
\right]
=
\widehat{\mathcal E}^{A\to E}
\left(
\pi^A
\right).
\label{eq:finite_input_normalized_choi_state_and_marginal}
\end{align}
Then
\begin{align}
\operatorname{Tr}
\left[
\omega_{\mathcal E}^{A'E}
\right]
=
\operatorname{Tr}
\left[
\omega_{\mathcal E}^{E}
\right]
=
1.
\end{align}

For every
\begin{align}
U
\in
\operatorname{U}
\left(
\mathcal H^A
\right),
\end{align}
define the possibly unnormalized output operator by
\begin{align}
\widetilde\tau_{\mathcal E,U}^{ER}
\coloneqq
\left(
\widehat{\mathcal E}^{A\to E}
\otimes
\id^R
\right)
\left(
\left(
U^A
\otimes
I^R
\right)
\rho^{AR}
\left(
U^{A\dagger}
\otimes
I^R
\right)
\right).
\label{eq:finite_input_randomized_output}
\end{align}
Haar-random-unitary twirling gives
\begin{align}
&\mathbb E_{
U\sim
\mathrm{Haar}
\left(
\operatorname{U}
\left(
\mathcal H^A
\right)
\right)
}
\left[
\left(
U^A
\otimes
I^R
\right)
\rho^{AR}
\left(
U^{A\dagger}
\otimes
I^R
\right)
\right]
=
\pi^A
\otimes
\rho^R,
\end{align}
and hence
\begin{align}
\mathbb E_{
U\sim
\mathrm{Haar}
\left(
\operatorname{U}
\left(
\mathcal H^A
\right)
\right)
}
\left[
\widetilde\tau_{\mathcal E,U}^{ER}
\right]
=
\omega_{\mathcal E}^{E}
\otimes
\rho^R.
\label{eq:finite_input_average_output}
\end{align}
In particular,
\begin{align}
\mathbb E_{
U\sim
\mathrm{Haar}
\left(
\operatorname{U}
\left(
\mathcal H^A
\right)
\right)
}
\left[
\operatorname{Tr}
\left[
\widetilde\tau_{\mathcal E,U}^{ER}
\right]
\right]
=
1.
\end{align}

Following the formulation in
Refs.~\cite{he2026quantum,berta2026tightanyshotquantumdecoupling},
we quantify the one-shot decoupling error by
\begin{align}
\mathbb E_{
U\sim
\mathrm{Haar}
\left(
\operatorname{U}
\left(
\mathcal H^A
\right)
\right)
}
\left[
D
\left(
\widetilde\tau_{\mathcal E,U}^{ER}
\middle\|
\omega_{\mathcal E}^{E}
\otimes
\rho^R
\right)
\right].
\label{eq:finite_input_average_relative_entropy_error}
\end{align}
For a general completely positive map, the two arguments of the
relative entropy in
\eqref{eq:finite_input_average_relative_entropy_error}
need not have the same trace for a fixed \(U\).
Consequently, the fixed-unitary quantity may be negative.
Nevertheless, the Haar-averaged quantity is nonnegative; indeed, joint convexity of the relative entropy and
\eqref{eq:finite_input_average_output} give
\begin{align}
&\mathbb E_{
U\sim
\mathrm{Haar}
\left(
\operatorname{U}
\left(
\mathcal H^A
\right)
\right)
}
\left[
D
\left(
\widetilde\tau_{\mathcal E,U}^{ER}
\middle\|
\omega_{\mathcal E}^{E}
\otimes
\rho^R
\right)
\right]
\nonumber\\
&\ge
D
\left(
\mathbb E_{
U\sim
\mathrm{Haar}
\left(
\operatorname{U}
\left(
\mathcal H^A
\right)
\right)
}
\left[
\widetilde\tau_{\mathcal E,U}^{ER}
\right]
\middle\|
\omega_{\mathcal E}^{E}
\otimes
\rho^R
\right)
\nonumber\\
&=
D
\left(
\omega_{\mathcal E}^{E}
\otimes
\rho^R
\middle\|
\omega_{\mathcal E}^{E}
\otimes
\rho^R
\right)
\nonumber\\
&=
0.
\label{eq:finite_input_average_relative_entropy_nonnegative}
\end{align}
We identify tensor products that differ only by the canonical ordering
of their tensor factors throughout this section.

\subsection{Optimal trace inequality for separable Hilbert spaces}
\label{subsec:optimal_trace_inequality_separable}

We present the key trace inequality used below to obtain a bound
in terms of the sandwiched R\'enyi relative entropy in the decoupling
analysis.

For nonzero
\(X,Y\in\operatorname{T}_+(\mathcal H)\),
define
\begin{align}
\Xi(X\|Y)
\coloneqq
D(X+Y\|Y)
+
D(Y\|X+Y)
\in
\mathbb R
\cup
\{+\infty\},
\label{eq:decoupling_logarithmic_increment_definition}
\end{align}
where the second term is finite due to
\(Y\le X+Y\).
Thus,
\(\Xi(X\|Y)\)
is well defined as an extended-real-valued quantity.
It can be regarded as the quantum Jeffreys divergence between
\(X+Y\) and \(Y\)~\cite{audenaert2013asymmetry}.
The definition in
\eqref{eq:decoupling_logarithmic_increment_definition}
provides a rigorous interpretation of the formal expression
\begin{align}
\operatorname{Tr}
\left[
X
\left(
\log[X+Y]
-
\log[Y]
\right)
\right],
\label{eq:decoupling_formal_logarithmic_increment}
\end{align}
which was used in
Ref.~\cite{cheng2025sharpestimatesquantumcovering}.

Reference~\cite[Theorem~1]{cheng2025sharpestimatesquantumcovering}
established an upper bound on this quantity for positive-semidefinite
trace-class operators on a separable Hilbert space, which was subsequently used as a key trace inequality in the existing decoupling analysis in Ref.~\cite{berta2026tightanyshotquantumdecoupling}.
In finite dimensions,
Ref.~\cite[Theorem~1 and Eq.~(29)]{
gour2026optimaltraceinequalitiessingleshot} improved the prefactor of this trace inequality to the optimal constant, denoted below by \(G_s\).
We show that the finite-rank approximation argument of
Ref.~\cite[Theorem~1]{cheng2025sharpestimatesquantumcovering}
extends this sharper bound to separable Hilbert spaces.

\begin{proposition}[Optimal trace inequality for separable Hilbert
spaces]
\label{prop:infinite_dimensional_trace_inequality}
Let
\(\mathcal H\)
be a separable Hilbert space, and let
\(X,Y\in\operatorname{T}_+(\mathcal H)\)
be nonzero.
Suppose that
\begin{align}
\operatorname{supp}(X)
\subseteq
\operatorname{supp}(Y)
\end{align}
and
\begin{align}
\Xi(X\|Y)
<
\infty,
\end{align}
where \(\Xi(X\|Y)\) is defined in~\eqref{eq:decoupling_logarithmic_increment_definition}.
Then, for every
\(s\in(0,1]\),
\begin{align}
\Xi(X\|Y)
&\le
G_s
\exp
\left[
s
\widetilde D_{1+s}(X\|Y)
\right],
\label{eq:infinite_dimensional_trace_inequality}
\\
G_s
&\coloneqq
\sup_{r>0}
\frac{
\log[1+r]
}{
r^s
},
\label{eq:decoupling_Gs_definition}
\end{align}
where the sandwiched R\'enyi relative entropy
\(\widetilde D_{1+s}(X\|Y)\)
is defined in
\eqref{eq:sandwiched_renyi_divergence_infinite_dimensional}.
Moreover,
\(G_s\)
is the smallest constant for which
\eqref{eq:infinite_dimensional_trace_inequality}
holds uniformly over all nonzero positive-semidefinite trace-class
operators \(X\) and \(Y\) satisfying the stated assumptions on
arbitrary separable Hilbert spaces.
\end{proposition}

\begin{proof}
It suffices to consider the case in which
\(\operatorname{supp}(Y)\)
is infinite-dimensional, since otherwise both \(X\) and \(Y\) are
supported on a finite-dimensional Hilbert space and the assertion
follows from
Ref.~\cite[Theorem~1 and Eq.~(29)]{
gour2026optimaltraceinequalitiessingleshot}.
We restrict the underlying Hilbert space to
\(\operatorname{supp}(Y)\), so that
\begin{align}
\operatorname{supp}(Y)
=
\mathcal H.
\label{eq:trace_inequality_Y_faithful_restriction}
\end{align}

\medskip
\noindent
\textbf{Finite-rank approximation.}
Let
\(\{\Pi_n\}_{n\in\{1,2,\ldots\}}\)
be an increasing sequence of finite-rank projections satisfying
\begin{align}
\operatorname{rank}\Pi_n
&=
n,
&
\Pi_n
&\le
\Pi_{n+1},
&
\Pi_n
&\longrightarrow
I
\end{align}
in the strong operator topology.
By
\eqref{eq:trace_inequality_Y_faithful_restriction},
for every nonzero
\(\ket{\varphi}\in\Pi_n\mathcal H\),
\begin{align}
\bra{\varphi}
\Pi_nY\Pi_n
\ket{\varphi}
=
\bra{\varphi}
Y
\ket{\varphi}
>
0.
\end{align}
Thus,
\(\Pi_nY\Pi_n\)
is strictly positive on
\(\Pi_n\mathcal H\), and hence
\begin{align}
\operatorname{supp}
\left(
\Pi_nX\Pi_n
\right)
\subseteq
\operatorname{supp}
\left(
\Pi_nY\Pi_n
\right)
=
\Pi_n\mathcal H.
\label{eq:trace_inequality_compressed_support_condition}
\end{align}

The finite-rank approximation results used below are collected in
Ref.~\cite[Lemma~A.1]{cheng2025sharpestimatesquantumcovering}.
More precisely, for nonzero
\(S,T\in\operatorname{T}_+(\mathcal H)\),
one has
\begin{align}
&\lim_{n\to\infty}
\Biggl\{
D
\left(
\Pi_nS\Pi_n
\middle\|
\Pi_nT\Pi_n
\right)
+
\operatorname{Tr}
\left[
\Pi_n(T-S)\Pi_n
\right]
\Biggr\}
\nonumber\\
&=
D(S\|T)
+
\operatorname{Tr}[T-S],
\label{eq:corrected_relative_entropy_finite_rank_convergence}
\end{align}
and, for every
\(s\in(0,1]\),
\begin{align}
\lim_{n\to\infty}
\widetilde Q_{1+s}
\left(
\Pi_nX\Pi_n
\middle\|
\Pi_nY\Pi_n
\right)
=
\widetilde Q_{1+s}(X\|Y).
\label{eq:sandwiched_quasi_divergence_finite_rank_convergence}
\end{align}
The relative-entropy approximation goes back to
Ref.~\cite[Lemma~4]{lindblad1974expectations}, whereas the
sandwiched quasi-divergence approximation for an arbitrary increasing
sequence of finite-rank projections converging strongly to the
identity is given in
Ref.~\cite[Lemma~3.39]{mosonyi2023strong}.

Since \(S\) is positive semidefinite and trace class,
\begin{align}
\operatorname{Tr}
\left[
\Pi_nS\Pi_n
\right]
=
\operatorname{Tr}
\left[
S^{1/2}
\Pi_n
S^{1/2}
\right].
\end{align}
Moreover,
\(\{S^{1/2}\Pi_nS^{1/2}\}_{n\in\{1,2,\ldots\}}\)
is increasing and converges strongly to \(S\).
The monotone convergence property of the trace therefore gives
\begin{align}
\lim_{n\to\infty}
\operatorname{Tr}
\left[
\Pi_nS\Pi_n
\right]
=
\operatorname{Tr}[S].
\label{eq:trace_finite_rank_convergence_S}
\end{align}
Similarly,
\begin{align}
\lim_{n\to\infty}
\operatorname{Tr}
\left[
\Pi_nT\Pi_n
\right]
=
\operatorname{Tr}[T].
\label{eq:trace_finite_rank_convergence_T}
\end{align}
Consequently,
\begin{align}
\lim_{n\to\infty}
\operatorname{Tr}
\left[
\Pi_n(T-S)\Pi_n
\right]
=
\operatorname{Tr}[T-S].
\label{eq:trace_correction_finite_rank_convergence}
\end{align}
Subtracting
\eqref{eq:trace_correction_finite_rank_convergence}
from
\eqref{eq:corrected_relative_entropy_finite_rank_convergence}
yields
\begin{align}
\lim_{n\to\infty}
D
\left(
\Pi_nS\Pi_n
\middle\|
\Pi_nT\Pi_n
\right)
=
D(S\|T),
\label{eq:homogeneous_relative_entropy_finite_rank_convergence}
\end{align}
where convergence to \(+\infty\) is understood in
terms of extended real numbers.

\medskip
\noindent
\textbf{Proof of the inequality.}
Fix
\(s\in(0,1)\).
Equation~\eqref{eq:trace_finite_rank_convergence_S},
applied with \(S=X\), gives
\begin{align}
\lim_{n\to\infty}
\operatorname{Tr}
\left[
\Pi_nX\Pi_n
\right]
=
\operatorname{Tr}[X]
>
0,
\end{align}
and hence
\(\Pi_nX\Pi_n\)
is nonzero for all sufficiently large \(n\).
Moreover,
\(\Pi_nY\Pi_n\)
is nonzero for every \(n\) by
\eqref{eq:trace_inequality_Y_faithful_restriction}.
Thus, for every sufficiently large \(n\), the finite-dimensional
result of
Ref.~\cite[Theorem~1 and Eq.~(29)]{
gour2026optimaltraceinequalitiessingleshot},
expressed using the natural-logarithm convention adopted in this
work, gives
\begin{align}
\Xi
\left(
\Pi_nX\Pi_n
\middle\|
\Pi_nY\Pi_n
\right)
\le
G_s
\widetilde Q_{1+s}
\left(
\Pi_nX\Pi_n
\middle\|
\Pi_nY\Pi_n
\right).
\label{eq:finite_rank_optimal_trace_inequality}
\end{align}
For \(s=1\), the same inequality holds with
\begin{align}
G_1
=
\sup_{r>0}
\frac{
\log[1+r]
}{
r
}
=
1
\end{align}
by the \(s=1\) case of
Ref.~\cite[Theorem~1]{cheng2025sharpestimatesquantumcovering}.

Applying
\eqref{eq:homogeneous_relative_entropy_finite_rank_convergence}
first with
\begin{align}
S
&=
X+Y,
&
T
&=
Y,
\end{align}
and then with
\begin{align}
S
&=
Y,
&
T
&=
X+Y,
\end{align}
gives
\begin{align}
&\lim_{n\to\infty}
D
\left(
\Pi_nX\Pi_n
+
\Pi_nY\Pi_n
\middle\|
\Pi_nY\Pi_n
\right)
=
D(X+Y\|Y),
\label{eq:first_relative_entropy_Xi_finite_rank_convergence}
\\
&\lim_{n\to\infty}
D
\left(
\Pi_nY\Pi_n
\middle\|
\Pi_nX\Pi_n
+
\Pi_nY\Pi_n
\right)
=
D(Y\|X+Y).
\label{eq:second_relative_entropy_Xi_finite_rank_convergence}
\end{align}
Since
\(\Xi(X\|Y)<\infty\),
both quantities on the right-hand sides are finite.
Therefore,
\begin{align}
\lim_{n\to\infty}
\Xi
\left(
\Pi_nX\Pi_n
\middle\|
\Pi_nY\Pi_n
\right)
=
\Xi(X\|Y).
\label{eq:Xi_finite_rank_convergence}
\end{align}
On the other hand,
\eqref{eq:sandwiched_quasi_divergence_finite_rank_convergence}
together with
\eqref{eq:sandwiched_renyi_divergence_infinite_dimensional}
gives
\begin{align}
&\lim_{n\to\infty}
\widetilde Q_{1+s}
\left(
\Pi_nX\Pi_n
\middle\|
\Pi_nY\Pi_n
\right)
\nonumber\\
&=
\widetilde Q_{1+s}(X\|Y)
=
\exp
\left[
s
\widetilde D_{1+s}(X\|Y)
\right].
\label{eq:sandwiched_quasi_divergence_finite_rank_convergence_D}
\end{align}
Taking the limit in
\eqref{eq:finite_rank_optimal_trace_inequality}
and using
\eqref{eq:Xi_finite_rank_convergence}
and
\eqref{eq:sandwiched_quasi_divergence_finite_rank_convergence_D}
proves
\eqref{eq:infinite_dimensional_trace_inequality}.

\medskip
\noindent
\textbf{Optimality.}
For
\(s\in(0,1)\),
the optimality of \(G_s\) for the corresponding inequality over
finite-dimensional positive-semidefinite operators is proved in
Ref.~\cite[Theorem~1 and Remark~1]{
gour2026optimaltraceinequalitiessingleshot}.
Since finite-dimensional Hilbert spaces form a subclass of separable
Hilbert spaces, no smaller constant can satisfy
\eqref{eq:infinite_dimensional_trace_inequality}
uniformly over the class considered here.

For \(s=1\), consider the one-dimensional Hilbert space with
\begin{align}
Y
&=
y,
&
X
&=
ry,
\end{align}
where
\(y>0\)
and
\(r>0\).
Then
\begin{align}
\frac{
\Xi(X\|Y)
}{
\widetilde Q_2(X\|Y)
}
=
\frac{
\log[1+r]
}{
r
}.
\end{align}
Letting
\(r\downarrow0\)
shows that any uniform constant must be at least
\begin{align}
\sup_{r>0}
\frac{
\log[1+r]
}{
r
}
=
1
=
G_1.
\end{align}
\end{proof}

\subsection{Haar-averaged Jensen's trace inequality}
\label{subsec:haar_averaged_Jensen_trace_inequality}

We prove the Jensen's trace inequalities needed to control the logarithmic terms in the Haar-averaged relative entropy to complete the proof of bounds on the achievable relative-entropy error exponent of one-shot decoupling.

We first show that the relative entropy can be recovered by adding
small identity-operator perturbations to the two logarithms.

\begin{proposition}[Perturbation formula for relative entropy]
\label{prop:relative_entropy_strictly_positive_perturbation}
Let
\(\mathcal H\)
be separable, and let
\(X,Y\in\operatorname{T}_+(\mathcal H)\)
be positive-semidefinite trace-class operators satisfying
\begin{align}
D(X\|Y)
<
+\infty.
\end{align}
Then
\begin{align}
D(X\|Y)
=
\lim_{\delta\downarrow0}
\operatorname{Tr}
\left[
X
\left(
\log[X+\delta I]
-
\log[Y+\delta I]
\right)
\right].
\label{eq:relative_entropy_strictly_positive_perturbation}
\end{align}
\end{proposition}

\begin{proof}
The statement is immediate if \(X=0\), so suppose that \(X\ne0\).
Since
\begin{align}
D(X\|Y)
<
+\infty,
\end{align}
the support condition
\begin{align}
\operatorname{supp}(X)
\subseteq
\operatorname{supp}(Y)
\end{align}
holds.

Let
\begin{align}
X
=
\sum_{i\in\mathcal I_X}
\lambda_i
\ket{\psi_i}\bra{\psi_i},
\qquad
\lambda_i>0,
\end{align}
be the spectral decomposition in
\eqref{eq:spectral_decomposition_trace_class_positive},
and let
\(\mu_i^Y\)
be the probability measure defined in
\eqref{eq:umegaki_Y_spectral_measure}.
The support condition implies
\begin{align}
\mu_i^Y
\left(
(0,\|Y\|]
\right)
=
1
\end{align}
for every
\(i\in\mathcal I_X\).

For
\(x,\lambda,\delta>0\),
define
\begin{align}
f_\delta(x,\lambda)
&\coloneqq
\log
\left(
\frac{x+\delta}{\lambda+\delta}
\right),
&
f(x,\lambda)
&\coloneqq
\log
\left(
\frac{x}{\lambda}
\right).
\label{eq:relative_entropy_regularization_integrands}
\end{align}
For every
\(\delta>0\),
the functions
\begin{align}
x
&\longmapsto
\log(x+\delta),
&
\lambda
&\longmapsto
\log(\lambda+\delta)
\end{align}
are bounded and continuous on
\([0,\|X\|]\)
and
\([0,\|Y\|]\),
respectively.
Therefore, the spectral theorem and the fact that
\(\mu_i^Y((0,\|Y\|])=1\)
give
\begin{align}
&\operatorname{Tr}
\left[
X
\left(
\log[X+\delta I]
-
\log[Y+\delta I]
\right)
\right]
\nonumber\\
&=
\sum_{i\in\mathcal I_X}
\lambda_i
\left[
\log(\lambda_i+\delta)
-
\int_{(0,\|Y\|]}
\log(\lambda+\delta)
\,
\mathrm{d}\mu_i^Y(\lambda)
\right]
\nonumber\\
&=
\sum_{i\in\mathcal I_X}
\lambda_i
\int_{(0,\|Y\|]}
f_\delta(\lambda_i,\lambda)
\,
\mathrm{d}\mu_i^Y(\lambda).
\label{eq:regularized_relative_entropy_spectral_expression}
\end{align}

For every
\(x,\lambda>0\),
\begin{align}
\lim_{\delta\downarrow0}
f_\delta(x,\lambda)
=
f(x,\lambda).
\label{eq:regularized_logarithm_pointwise_limit}
\end{align}
Moreover,
\begin{align}
\frac{x+\delta}{\lambda+\delta}
=
\frac{\lambda}{\lambda+\delta}
\frac{x}{\lambda}
+
\frac{\delta}{\lambda+\delta}.
\end{align}
Hence,
\((x+\delta)/(\lambda+\delta)\)
is a convex combination of
\(x/\lambda\)
and \(1\).
Since the logarithm is increasing,
\begin{align}
\min
\left\{
0,
f(x,\lambda)
\right\}
\le
f_\delta(x,\lambda)
\le
\max
\left\{
0,
f(x,\lambda)
\right\},
\end{align}
and consequently
\begin{align}
\left|
f_\delta(x,\lambda)
\right|
\le
\left|
f(x,\lambda)
\right|.
\label{eq:regularized_logarithm_domination}
\end{align}

We next verify the integrability of the right-hand side of
\eqref{eq:regularized_logarithm_domination}
with respect to the weighted family of measures
\(\{\lambda_i\mu_i^Y\}_{i\in\mathcal I_X}\).
Its negative part satisfies
\begin{align}
&
\sum_{i\in\mathcal I_X}
\lambda_i
\int_{(0,\|Y\|]}
\left[
f(\lambda_i,\lambda)
\right]_-
\,
\mathrm{d}\mu_i^Y(\lambda)
\nonumber\\
&=
\sum_{i\in\mathcal I_X}
\int_{(\lambda_i,\|Y\|]}
\lambda_i
\log
\left(
\frac{\lambda}{\lambda_i}
\right)
\,
\mathrm{d}\mu_i^Y(\lambda)
\nonumber\\
&\le
\sum_{i\in\mathcal I_X}
\int_{(0,\|Y\|]}
\lambda
\,
\mathrm{d}\mu_i^Y(\lambda)
\nonumber\\
&=
\sum_{i\in\mathcal I_X}
\bra{\psi_i}
Y
\ket{\psi_i}
\nonumber\\
&=
\operatorname{Tr}
\left[
\Pi_XY\Pi_X
\right]
\nonumber\\
&\le
\operatorname{Tr}[Y]
<
+\infty,
\label{eq:relative_entropy_regularization_negative_part_bound}
\end{align}
where
\begin{align}
[a]_-
\coloneqq
\max\{-a,0\}
\end{align}
and we used
\begin{align}
\lambda_i
\log
\left(
\frac{\lambda}{\lambda_i}
\right)
\le
\lambda-\lambda_i
\le
\lambda
\qquad
\text{for \(\lambda>\lambda_i\)}.
\end{align}

By
\eqref{eq:umegaki_relative_entropy_positive_operators},
\begin{align}
D(X\|Y)
=
\sum_{i\in\mathcal I_X}
\lambda_i
\int_{(0,\|Y\|]}
f(\lambda_i,\lambda)
\,
\mathrm{d}\mu_i^Y(\lambda).
\end{align}
Since
\(D(X\|Y)<+\infty\)
and the total negative part is finite by
\eqref{eq:relative_entropy_regularization_negative_part_bound},
the total positive part is also finite.
It follows that
\begin{align}
&
\sum_{i\in\mathcal I_X}
\lambda_i
\int_{(0,\|Y\|]}
\left|
f(\lambda_i,\lambda)
\right|
\,
\mathrm{d}\mu_i^Y(\lambda)
\nonumber\\
&<
+\infty.
\label{eq:relative_entropy_regularization_dominating_function}
\end{align}

To apply the dominated convergence theorem, consider the disjoint
union
\begin{align}
\Omega
\coloneqq
\bigsqcup_{i\in\mathcal I_X}
\left(
\{i\}
\times
(0,\|Y\|]
\right)
\end{align}
and define a finite positive measure \(\nu\) on \(\Omega\) by
\begin{align}
\nu
\left(
\{i\}\times\Delta
\right)
\coloneqq
\lambda_i
\mu_i^Y(\Delta).
\end{align}
Indeed,
\begin{align}
\nu(\Omega)
=
\sum_{i\in\mathcal I_X}
\lambda_i
=
\operatorname{Tr}[X]
<
+\infty.
\end{align}
Define
\begin{align}
g_\delta(i,\lambda)
&\coloneqq
f_\delta(\lambda_i,\lambda),
&
g(i,\lambda)
&\coloneqq
f(\lambda_i,\lambda).
\end{align}
Equations
\eqref{eq:regularized_logarithm_pointwise_limit}
and
\eqref{eq:regularized_logarithm_domination}
give
\begin{align}
g_\delta(i,\lambda)
&\longrightarrow
g(i,\lambda),
\\
|g_\delta(i,\lambda)|
&\le
|g(i,\lambda)|
\end{align}
for every
\((i,\lambda)\in\Omega\),
while
\eqref{eq:relative_entropy_regularization_dominating_function}
shows that
\(|g|\)
is \(\nu\)-integrable.
The dominated convergence theorem
\cite[Theorem~8.3.4]{tao2023analysis}
therefore yields
\begin{align}
&\lim_{\delta\downarrow0}
\sum_{i\in\mathcal I_X}
\lambda_i
\int_{(0,\|Y\|]}
f_\delta(\lambda_i,\lambda)
\,
\mathrm{d}\mu_i^Y(\lambda)
\nonumber\\
&=\sum_{i\in\mathcal I_X}
\lambda_i
\int_{(0,\|Y\|]}
\lim_{\delta\downarrow0}
f_\delta(\lambda_i,\lambda)
\,
\mathrm{d}\mu_i^Y(\lambda)
\nonumber\\
&=
\sum_{i\in\mathcal I_X}
\lambda_i
\int_{(0,\|Y\|]}
f(\lambda_i,\lambda)
\,
\mathrm{d}\mu_i^Y(\lambda)
\nonumber\\
&=
D(X\|Y).
\end{align}
Combining this identity with
\eqref{eq:regularized_relative_entropy_spectral_expression}
proves
\eqref{eq:relative_entropy_strictly_positive_perturbation}.
\end{proof}

Using the perturbation formula, we now derive the Haar-averaged trace
inequality used in the decoupling argument.

\begin{proposition}[Haar-averaged Jensen's trace inequality]
\label{prop:decoupling_operator_Jensen_reduction}
Let
\(\mathcal H^A\)
be finite-dimensional, and let
\(\mathcal H^R\) and \(\mathcal H^E\)
be separable.
Let
\begin{align}
\rho^{AR}
\in
\operatorname{D}
\left(
\mathcal H^A
\otimes
\mathcal H^R
\right),
\end{align}
and let
\begin{align}
\mathcal E^{A\to E}
:
\operatorname{T}
\left(
\mathcal H^A
\right)
\longrightarrow
\operatorname{T}
\left(
\mathcal H^E
\right)
\end{align}
be a nonzero bounded completely positive map.
Let
\(\widehat{\mathcal E}^{A\to E}\),
\(\omega_{\mathcal E}^{A'E}\),
\(\omega_{\mathcal E}^{E}\), and
\(\widetilde\tau_{\mathcal E,U}^{ER}\)
be defined in
\eqref{eq:finite_input_normalized_map},
\eqref{eq:finite_input_normalized_choi_state_and_marginal}, and
\eqref{eq:finite_input_randomized_output},
respectively.
Then, for every
\(\delta>0\),
\begin{align}
&\mathbb E_{
U\sim
\mathrm{Haar}
\left(
\operatorname{U}
\left(
\mathcal H^A
\right)
\right)
}
\left[
\operatorname{Tr}
\left[
\widetilde\tau_{\mathcal E,U}^{ER}
\log
\left[
\widetilde\tau_{\mathcal E,U}^{ER}
+
\delta I^{ER}
\right]
\right]
\right]
\nonumber\\
&\le
\operatorname{Tr}
\Biggl[
\left(
\rho^{AR}
\otimes
\omega_{\mathcal E}^{A'E}
\right)
\log
\Biggl[
|A|^2
\mathbb E_{
U\sim
\mathrm{Haar}
\left(
\operatorname{U}
\left(
\mathcal H^A
\right)
\right)
}
\Biggl[
\left(
U^{A\dagger}
\otimes
I^{A'}
\right)
\Phi^{AA'}
\left(
U^A
\otimes
I^{A'}
\right)
\otimes
\widetilde\tau_{\mathcal E,U}^{ER}
\Biggr]
+
\delta I^{AA'ER}
\Biggr]
\Biggr].
\label{eq:decoupling_trace_Jensen_first_term_bound}
\end{align}
\end{proposition}

\begin{proof}
Let
\begin{align}
\mathcal K
\coloneqq
L^2
\left(
\operatorname{U}
\left(
\mathcal H^A
\right),
\mathrm{Haar};
\mathcal H^{AA'ER}
\right)
\end{align}
be the Hilbert space of equivalence classes of strongly measurable
functions
\begin{align}
\varphi
:
\operatorname{U}
\left(
\mathcal H^A
\right)
\longrightarrow
\mathcal H^{AA'ER}
\end{align}
satisfying
\begin{align}
\int_{
\operatorname{U}
\left(
\mathcal H^A
\right)
}
\|
\varphi[U]
\|^2
\,\mathrm dU
<
\infty,
\end{align}
where
\(\mathrm dU\)
denotes the normalized Haar measure.
Its inner product is
\begin{align}
\langle
\varphi_1,
\varphi_2
\rangle_{\mathcal K}
\coloneqq
\int_{
\operatorname{U}
\left(
\mathcal H^A
\right)
}
\bra{\varphi_1[U]}
\ket{\varphi_2[U]}
\,\mathrm dU.
\end{align}

Define a linear map
\begin{align}
V
:
\mathcal H^{AA'ER}
\longrightarrow
\mathcal K
\end{align}
by
\begin{align}
\left(
V
\ket{\psi}
\right)[U]
\coloneqq
|A|
\Biggl(
\left(
U^{A\dagger}
\otimes
I^{A'}
\right)
\Phi^{AA'}
\left(
U^A
\otimes
I^{A'}
\right)
\otimes
I^{ER}
\Biggr)
\ket{\psi}.
\label{eq:decoupling_trace_Jensen_isometry}
\end{align}
Haar-random-unitary twirling gives
\begin{align}
&\mathbb E_{
U\sim
\mathrm{Haar}
\left(
\operatorname{U}
\left(
\mathcal H^A
\right)
\right)
}
\left[
\left(
U^{A\dagger}
\otimes
I^{A'}
\right)
\Phi^{AA'}
\left(
U^A
\otimes
I^{A'}
\right)
\right]
=
\frac{
I^{AA'}
}{
|A|^2
}.
\label{eq:decoupling_trace_Jensen_first_moment}
\end{align}
Since
\(\Phi^{AA'}\)
is a projection, for every
\(\ket{\psi}\in\mathcal H^{AA'ER}\),
\begin{align}
\|
V
\ket{\psi}
\|_{\mathcal K}^2
&=
|A|^2
\mathbb E_{
U\sim
\mathrm{Haar}
\left(
\operatorname{U}
\left(
\mathcal H^A
\right)
\right)
}
\Biggl[
\bra{\psi}
\Biggl(
\left(
U^{A\dagger}
\otimes
I^{A'}
\right)
\Phi^{AA'}
\left(
U^A
\otimes
I^{A'}
\right)
\otimes
I^{ER}
\Biggr)
\ket{\psi}
\Biggr]
\nonumber\\
&=
\|
\ket{\psi}
\|^2,
\end{align}
where
\eqref{eq:decoupling_trace_Jensen_first_moment}
was used in the last equality.
Therefore,
\(V\)
is an isometry.

For every
\(\delta>0\),
define the multiplication operator
\(M_\delta\)
on
\(\mathcal K\)
by
\begin{align}
\ket{\left(
M_\delta
\varphi
\right)[U]}
\coloneqq
\left(
I^{AA'}
\otimes
\left[
\widetilde\tau_{\mathcal E,U}^{ER}
+
\delta I^{ER}
\right]
\right)
\ket{\varphi[U]}.
\label{eq:decoupling_trace_Jensen_multiplication_operator}
\end{align}
The map
\begin{align}
U
\longmapsto
\widetilde\tau_{\mathcal E,U}^{ER}
\end{align}
is continuous in the trace norm and hence bounded in the operator
norm, since
\(\operatorname{U}(\mathcal H^A)\)
is compact.
Therefore,
\(M_\delta\)
is bounded and satisfies
\begin{align}
M_\delta
\ge
\delta I_{\mathcal K}.
\label{eq:decoupling_trace_Jensen_multiplication_lower_bound}
\end{align}
The continuous functional calculus for multiplication operators acts
pointwise:
\begin{align}
\ket{\left(
\log[M_\delta]
\varphi
\right)[U]}
=
\left(
I^{AA'}
\otimes
\log
\left[
\widetilde\tau_{\mathcal E,U}^{ER}
+
\delta I^{ER}
\right]
\right)
\ket{\varphi[U]}.
\label{eq:decoupling_trace_Jensen_log_multiplication_operator}
\end{align}

Since \(V\) is an isometry, Jensen's operator
inequality~\cite[Theorem~2.1]{hansen2003jensen}
for the operator-concave function \(\log\) on the spectrum in
\eqref{eq:decoupling_trace_Jensen_multiplication_lower_bound}
gives
\begin{align}
V^\dagger
\log[M_\delta]
V
\le
\log
\left[
V^\dagger
M_\delta
V
\right].
\label{eq:decoupling_trace_Jensen_compression}
\end{align}
Due to
\eqref{eq:decoupling_trace_Jensen_isometry}
and
\eqref{eq:decoupling_trace_Jensen_log_multiplication_operator},
for arbitrary
\(\ket{\psi_1},\ket{\psi_2}\in\mathcal H^{AA'ER}\),
the definitions of \(V\), \(M_\delta\), and the inner product on
\(\mathcal K\) give
\begin{align}
& \left\langle
\left(
V\ket{\psi_1}
\right),
\log[M_\delta]
\left(
V
\ket{\psi_2}
\right)
\right\rangle_\mathcal{K}
\nonumber\\
&=
|A|^2
\mathbb E_{
U\sim
\mathrm{Haar}
\left(
\operatorname{U}
\left(
\mathcal H^A
\right)
\right)
}
\Biggl[
\bra{\psi_1}
\Biggl(
\left(
U^{A\dagger}
\otimes
I^{A'}
\right)
\Phi^{AA'}
\left(
U^A
\otimes
I^{A'}
\right)
\otimes
\log
\left[
\widetilde\tau_{\mathcal E,U}^{ER}
+
\delta I^{ER}
\right]
\Biggr)
\ket{\psi_2}
\Biggr].
\end{align}
Thus, the left-hand side of
\eqref{eq:decoupling_trace_Jensen_compression}
is
\begin{align}
&V^\dagger
\log[M_\delta]
V
\nonumber\\
&=
|A|^2
\mathbb E_{
U\sim
\mathrm{Haar}
\left(
\operatorname{U}
\left(
\mathcal H^A
\right)
\right)
}
\Biggl[
\left(
U^{A\dagger}
\otimes
I^{A'}
\right)
\Phi^{AA'}
\left(
U^A
\otimes
I^{A'}
\right)
\otimes
\log
\left[
\widetilde\tau_{\mathcal E,U}^{ER}
+
\delta I^{ER}
\right]
\Biggr];
\label{eq:decoupling_trace_Jensen_compressed_log}
\end{align}
Similarly,
\begin{align}
&V^\dagger
M_\delta
V
\nonumber\\
&=
|A|^2
\mathbb E_{
U\sim
\mathrm{Haar}
\left(
\operatorname{U}
\left(
\mathcal H^A
\right)
\right)
}
\Biggl[
\left(
U^{A\dagger}
\otimes
I^{A'}
\right)
\Phi^{AA'}
\left(
U^A
\otimes
I^{A'}
\right)
\otimes
\left(
\widetilde\tau_{\mathcal E,U}^{ER}
+
\delta I^{ER}
\right)
\Biggr]
\nonumber\\
&=
|A|^2
\mathbb E_{
U\sim
\mathrm{Haar}
\left(
\operatorname{U}
\left(
\mathcal H^A
\right)
\right)
}
\Biggl[
\left(
U^{A\dagger}
\otimes
I^{A'}
\right)
\Phi^{AA'}
\left(
U^A
\otimes
I^{A'}
\right)
\otimes
\widetilde\tau_{\mathcal E,U}^{ER}
\Biggr]
+
\delta I^{AA'ER},
\label{eq:decoupling_trace_Jensen_compressed_operator}
\end{align}
where
\eqref{eq:decoupling_trace_Jensen_first_moment}
was used in the last equality.
Substituting
\eqref{eq:decoupling_trace_Jensen_compressed_log}
and
\eqref{eq:decoupling_trace_Jensen_compressed_operator}
into
\eqref{eq:decoupling_trace_Jensen_compression}
yields
\begin{align}
&|A|^2
\mathbb E_{
U\sim
\mathrm{Haar}
\left(
\operatorname{U}
\left(
\mathcal H^A
\right)
\right)
}
\Biggl[
\left(
U^{A\dagger}
\otimes
I^{A'}
\right)
\Phi^{AA'}
\left(
U^A
\otimes
I^{A'}
\right)
\otimes
\log
\left[
\widetilde\tau_{\mathcal E,U}^{ER}
+
\delta I^{ER}
\right]
\Biggr]
\nonumber\\
&\le
\log
\Biggl[
|A|^2
\mathbb E_{
U\sim
\mathrm{Haar}
\left(
\operatorname{U}
\left(
\mathcal H^A
\right)
\right)
}
\Biggl[
\left(
U^{A\dagger}
\otimes
I^{A'}
\right)
\Phi^{AA'}
\left(
U^A
\otimes
I^{A'}
\right)
\otimes
\widetilde\tau_{\mathcal E,U}^{ER}
\Biggr]
+
\delta I^{AA'ER}
\Biggr].
\label{eq:decoupling_trace_Jensen_regularized_operator_inequality}
\end{align}
The normalized Choi representation gives
\begin{align}
\widetilde\tau_{\mathcal E,U}^{ER}
=
|A|^2
\operatorname{Tr}_{AA'}
\Biggl[
\Biggl(
\left(
U^{A\dagger}
\otimes
I^{A'}
\right)
\Phi^{AA'}
\left(
U^A
\otimes
I^{A'}
\right)
\otimes
I^{ER}
\Biggr)
\left(
\rho^{AR}
\otimes
\omega_{\mathcal E}^{A'E}
\right)
\Biggr].
\label{eq:decoupling_trace_Jensen_choi_representation}
\end{align}
Multiplying
\eqref{eq:decoupling_trace_Jensen_regularized_operator_inequality}
by
\(\rho^{AR}\otimes\omega_{\mathcal E}^{A'E}\)
and taking the trace, the left-hand side becomes
\begin{align}
&|A|^2
\mathbb E_{
U\sim
\mathrm{Haar}
\left(
\operatorname{U}
\left(
\mathcal H^A
\right)
\right)
}
\Biggl[
\operatorname{Tr}
\Biggl[
\left(
\rho^{AR}
\otimes
\omega_{\mathcal E}^{A'E}
\right)
\Biggl(
\left(
U^{A\dagger}
\otimes
I^{A'}
\right)
\Phi^{AA'}
\left(
U^A
\otimes
I^{A'}
\right)
\otimes
\log
\left[
\widetilde\tau_{\mathcal E,U}^{ER}
+
\delta I^{ER}
\right]
\Biggr)
\Biggr]
\Biggr]
\nonumber\\
&=
\mathbb E_{
U\sim
\mathrm{Haar}
\left(
\operatorname{U}
\left(
\mathcal H^A
\right)
\right)
}
\left[
\operatorname{Tr}
\left[
\widetilde\tau_{\mathcal E,U}^{ER}
\log
\left[
\widetilde\tau_{\mathcal E,U}^{ER}
+
\delta I^{ER}
\right]
\right]
\right],
\end{align}
where
\eqref{eq:decoupling_trace_Jensen_choi_representation}
was used.
Therefore,
\eqref{eq:decoupling_trace_Jensen_regularized_operator_inequality}
implies
\eqref{eq:decoupling_trace_Jensen_first_term_bound}, i.e.,
\begin{align}
&\mathbb E_{
U\sim
\mathrm{Haar}
\left(
\operatorname{U}
\left(
\mathcal H^A
\right)
\right)
}
\left[
\operatorname{Tr}
\left[
\widetilde\tau_{\mathcal E,U}^{ER}
\log
\left[
\widetilde\tau_{\mathcal E,U}^{ER}
+
\delta I^{ER}
\right]
\right]
\right]
\nonumber\\
&\le
\operatorname{Tr}
\Biggl[
\left(
\rho^{AR}
\otimes
\omega_{\mathcal E}^{A'E}
\right)
\log
\Biggl[
|A|^2
\mathbb E_{
U\sim
\mathrm{Haar}
\left(
\operatorname{U}
\left(
\mathcal H^A
\right)
\right)
}
\Biggl[
\left(
U^{A\dagger}
\otimes
I^{A'}
\right)
\Phi^{AA'}
\left(
U^A
\otimes
I^{A'}
\right)
\otimes
\widetilde\tau_{\mathcal E,U}^{ER}
\Biggr]
+
\delta I^{AA'ER}
\Biggr]
\Biggr].
\label{eq:decoupling_trace_Jensen_first_term_bound_2}
\end{align}
\end{proof}

\begin{remark}[Comparison with the existing analysis of tight
any-shot decoupling]
\label{rem:tight_any_shot_invalid_log_concavity_step}
In the analysis of
Ref.~\cite[Theorem~3 and Eqs.~(29), (30), and~(31)]{
berta2026tightanyshotquantumdecoupling},
the following inequality is asserted, in the notation of this work:
\begin{align}
&\mathbb E_{
U\sim
\mathrm{Haar}
\left(
\operatorname{U}
\left(
\mathcal H^A
\right)
\right)
}
\left[
\operatorname{Tr}
\left[
\widetilde\tau_{\mathcal E,U}^{ER}
\log
\widetilde\tau_{\mathcal E,U}^{ER}
\right]
\right]
\nonumber\\
&\le
|A|^2
\operatorname{Tr}
\Biggl[
\left(
\rho^{AR}
\otimes
\omega_{\mathcal E}^{A'E}
\right)
\log
\mathbb E_{
U\sim
\mathrm{Haar}
\left(
\operatorname{U}
\left(
\mathcal H^A
\right)
\right)
}
\Biggl[
\left(
U^{A\dagger}
\otimes
I^{A'}
\right)
\Phi^{AA'}
\left(
U^A
\otimes
I^{A'}
\right)
\otimes
\widetilde\tau_{\mathcal E,U}^{ER}
\Biggr]
\Biggr].
\label{eq:tight_any_shot_claimed_log_concavity_step}
\end{align}
In contrast to
\eqref{eq:decoupling_trace_Jensen_regularized_operator_inequality},
the factor
\(|A|^2\)
in
\eqref{eq:tight_any_shot_claimed_log_concavity_step}
is outside the trace rather than inside the logarithm.

The inequality
\eqref{eq:tight_any_shot_claimed_log_concavity_step}
does not hold in general, even when
\(\mathcal H^R\) and \(\mathcal H^E\)
are finite-dimensional and
\(\mathcal E^{A\to E}\)
is trace-preserving.
To see this, assume that
\begin{align}
|A|
\ge
2,
\end{align}
let
\begin{align}
\mathcal H^R
&=
\mathbb C,
&
\mathcal H^E
&=
\mathbb C
=
\operatorname{span}
\left\{
\ket{0}
\right\},
\end{align}
and define
\begin{align}
\mathcal E^{A\to E}
\left(
X^A
\right)
\coloneqq
\operatorname{Tr}
\left[
X^A
\right]
\ket{0}\bra{0}^{E}.
\label{eq:tight_any_shot_counterexample_channel}
\end{align}
This map is completely positive and trace-preserving.
Moreover,
\begin{align}
\kappa_{\mathcal E}
&=
1,
&
\widehat{\mathcal E}^{A\to E}
&=
\mathcal E^{A\to E},
\\
\omega_{\mathcal E}^{A'E}
&=
\pi^{A'}
\otimes
\ket{0}\bra{0}^{E},
&
\omega_{\mathcal E}^{E}
&=
\ket{0}\bra{0}^{E}.
\end{align}
For every state
\(\rho^A\)
and every
\(U\in\operatorname{U}(\mathcal H^A)\),
\begin{align}
\widetilde\tau_{\mathcal E,U}^{E}
=
\ket{0}\bra{0}^{E}.
\end{align}
Hence, the left-hand side of
\eqref{eq:tight_any_shot_claimed_log_concavity_step}
is zero.
On the other hand,
\eqref{eq:decoupling_trace_Jensen_first_moment}
gives
\begin{align}
&\mathbb E_{
U\sim
\mathrm{Haar}
\left(
\operatorname{U}
\left(
\mathcal H^A
\right)
\right)
}
\Biggl[
\left(
U^{A\dagger}
\otimes
I^{A'}
\right)
\Phi^{AA'}
\left(
U^A
\otimes
I^{A'}
\right)
\otimes
\widetilde\tau_{\mathcal E,U}^{E}
\Biggr]
\nonumber\\
&=
\frac{
I^{AA'}
}{
|A|^2
}
\otimes
\ket{0}\bra{0}^{E}.
\end{align}
Consequently, the right-hand side of
\eqref{eq:tight_any_shot_claimed_log_concavity_step}
is
\begin{align}
&|A|^2
\operatorname{Tr}
\Biggl[
\left(
\rho^A
\otimes
\pi^{A'}
\otimes
\ket{0}\bra{0}^{E}
\right)
\log
\Biggl[
\frac{
I^{AA'}
}{
|A|^2
}
\otimes
\ket{0}\bra{0}^{E}
\Biggr]
\Biggr]
\nonumber\\
&=
-|A|^2
\log
\left[
|A|^2
\right]
<
0,
\end{align}
contradicting
\eqref{eq:tight_any_shot_claimed_log_concavity_step}.

The issue is that the derivation in
Ref.~\cite[Eqs.~(29), (30), and~(31)]{
berta2026tightanyshotquantumdecoupling}
uses the identity
\begin{align}
\phi
\otimes
\log T
=
\log
\left[
\phi
\otimes
T
\right]
\end{align}
for a rank-one projection \(\phi\).
With the conventional logarithmic functional calculus, the
right-hand side is not a bounded operator on the full Hilbert space
because
\(\phi\otimes T\)
has a nontrivial kernel.
Under the generalized logarithm adopted in
\eqref{eq:generalized_operator_logarithm},
the identity holds algebraically.
However, the scalar function obtained by extending
\(\log t\)
with the value zero at \(t=0\) is not concave; indeed, for
\begin{align}
\ell(0)
&\coloneqq
0,
&
\ell(t)
&\coloneqq
\log[t],
\qquad
t>0,
\end{align}
one has
\begin{align}
\ell
\left(
\frac{
0+1
}{
2
}
\right)
=
-\log2
<
0
=
\frac{
\ell(0)
+
\ell(1)
}{
2
},
\end{align}
which is not concave.
Consequently, this generalized logarithm cannot be used in the
operator-concavity step.

Proposition~\ref{prop:decoupling_operator_Jensen_reduction}
avoids this issue by applying Jensen's operator inequality to the
isometry \(V\) and the strictly positive multiplication operator
\(M_\delta\).
The normalization required to make \(V\) isometric introduces the
factor
\(|A|^2\)
inside the logarithm.
Consequently,
Theorem~\ref{thm:finite_input_infinite_reference_decoupling}
below yields the corresponding one-shot bound with the extra
dimension-dependent factor
\(C_{|A|}\)
compared to
Ref.~\cite[Theorem~3]{
berta2026tightanyshotquantumdecoupling}.
However, since
\begin{align}
\lim_{|A|\to\infty}C_{|A|}
=
2,
\end{align}
the factor
\(C_{|A|}\)
remains bounded under tensor powers and hence does not affect the
leading asymptotic error exponent.
\end{remark}

\subsection{Decoupling error from the trace inequality}
\label{subsec:haar_averaged_relative_entropy_decoupling_bound}

We now combine the Haar-averaged Jensen inequality with the Haar-random-unitary second-moment identity to bound the averaged
relative-entropy decoupling error in terms of \(\Xi\).
The following Haar-random-unitary second-moment identity was used in the original analysis of decoupling in
Ref.~\cite{horodecki2007quantum}.
A derivation can be found, for example, in
Ref.~\cite[Eq.~(124)]{Mele2024introductiontohaar}.
The dimension-dependent prefactors in
\eqref{eq:decoupling_haar_second_moment}
reflect the convention that the maximally entangled states are
normalized.

\begin{lemma}[Haar-random-unitary second moment]
\label{lem:decoupling_haar_second_moment}
Let
\(\mathcal H^A\)
be finite-dimensional, with
\(|A|\)
defined in
\eqref{eq:dimension_shorthand_A}, and suppose that
\begin{align}
|A|
\ge
2.
\end{align}
Let
\(\mathcal H^{A'}\),
\(\mathcal H^{\widetilde A}\), and
\(\mathcal H^{\widetilde A'}\)
be fixed copies of
\(\mathcal H^A\).
Let
\(\Phi^{AA'}\)
and
\(\Phi^{\widetilde A\widetilde A'}\)
be the normalized maximally entangled states defined as in
\eqref{eq:normalized_maximally_entangled_state}
on
\(\mathcal H^A\otimes\mathcal H^{A'}\)
and
\(\mathcal H^{\widetilde A}\otimes
\mathcal H^{\widetilde A'}\),
respectively.
Let
\(F^{A\widetilde A}\)
and
\(F^{A'\widetilde A'}\)
be the corresponding swap operators defined as in
\eqref{eq:swap_operator_definition}.
For every
\(U\in\operatorname{U}(\mathcal H^A)\),
let
\(U^{\widetilde A}\)
denote the corresponding copy of \(U\) acting on
\(\mathcal H^{\widetilde A}\).
Then,
\begin{align}
&\mathbb E_{
U\sim
\mathrm{Haar}
\left(
\operatorname{U}
\left(
\mathcal H^A
\right)
\right)
}
\Biggl[
\left(
U^{A\dagger}
\otimes
I^{A'}
\right)
\Phi^{AA'}
\left(
U^A
\otimes
I^{A'}
\right)
\nonumber\\
&\hspace{28mm}\otimes
\left(
U^{\widetilde A\dagger}
\otimes
I^{\widetilde A'}
\right)
\Phi^{\widetilde A\widetilde A'}
\left(
U^{\widetilde A}
\otimes
I^{\widetilde A'}
\right)
\Biggr]
\nonumber\\
&=
\frac{
I^{A\widetilde A}
\otimes
I^{A'\widetilde A'}
}{
|A|^2
\left(
|A|^2-1
\right)
}
-
\frac{
F^{A\widetilde A}
\otimes
I^{A'\widetilde A'}
}{
|A|^3
\left(
|A|^2-1
\right)
}
\nonumber\\
&\quad-
\frac{
I^{A\widetilde A}
\otimes
F^{A'\widetilde A'}
}{
|A|^3
\left(
|A|^2-1
\right)
}
+
\frac{
F^{A\widetilde A}
\otimes
F^{A'\widetilde A'}
}{
|A|^2
\left(
|A|^2-1
\right)
}.
\label{eq:decoupling_haar_second_moment}
\end{align}
\end{lemma}

We then derive the following bound on Haar-averaged relative-entropy decoupling error via \(\Xi\)

\begin{proposition}[Bound on Haar-averaged relative-entropy
decoupling error via \(\Xi\)]
\label{prop:decoupling_Xi_reduction}
Let
\(\mathcal H^A\)
be finite-dimensional, and let
\(\mathcal H^R\) and \(\mathcal H^E\)
be separable.
Let
\begin{align}
\rho^{AR}
\in
\operatorname{D}
\left(
\mathcal H^A
\otimes
\mathcal H^R
\right),
\end{align}
and let
\begin{align}
\mathcal E^{A\to E}
:
\operatorname{T}
\left(
\mathcal H^A
\right)
\longrightarrow
\operatorname{T}
\left(
\mathcal H^E
\right)
\end{align}
be a nonzero bounded completely positive map.
Let
\(\widehat{\mathcal E}^{A\to E}\),
\(\omega_{\mathcal E}^{A'E}\),
\(\omega_{\mathcal E}^{E}\), and
\(\widetilde\tau_{\mathcal E,U}^{ER}\)
be defined in
\eqref{eq:finite_input_normalized_map},
\eqref{eq:finite_input_normalized_choi_state_and_marginal}, and
\eqref{eq:finite_input_randomized_output},
respectively.
Then
\begin{align}
&\Xi
\Biggl(
\rho^{AR}
\otimes
\omega_{\mathcal E}^{A'E}
\Biggm\|
\left(
I^A
\otimes
\rho^R
\right)
\otimes
\left(
I^{A'}
\otimes
\omega_{\mathcal E}^{E}
\right)
\Biggr)
<
\infty,
\label{eq:decoupling_Xi_reduction_finiteness}
\end{align}
and
\begin{align}
&\mathbb E_{
U\sim
\mathrm{Haar}
\left(
\operatorname{U}
\left(
\mathcal H^A
\right)
\right)
}
\left[
D
\left(
\widetilde\tau_{\mathcal E,U}^{ER}
\middle\|
\omega_{\mathcal E}^{E}
\otimes
\rho^R
\right)
\right]
\nonumber\\
&\le
C_{|A|}
\Xi
\Biggl(
\rho^{AR}
\otimes
\omega_{\mathcal E}^{A'E}
\Biggm\|
\left(
I^A
\otimes
\rho^R
\right)
\otimes
\left(
I^{A'}
\otimes
\omega_{\mathcal E}^{E}
\right)
\Biggr),
\label{eq:decoupling_Xi_reduction}
\end{align}
where \(\Xi\) is defined in
\eqref{eq:decoupling_logarithmic_increment_definition},
and
\begin{align}
C_{|A|}
\coloneqq
\begin{cases}
0,
&
\text{if \(|A|=1\),}
\\[1mm]
1
+
\displaystyle
\frac{
\log
\left[
\frac{
|A|^2
}{
|A|^2-1
}
\right]
}{
\log
\left[
1+\frac{1}{|A|^2}
\right]
},
&
\text{if \(|A|\ge2\).}
\end{cases}
\label{eq:decoupling_dimension_dependent_constant}
\end{align}
\end{proposition}

\begin{proof}
\textbf{The case \(\lvert A\rvert=1\).}
If
\(|A|=1\),
then, for every \(U\),
\begin{align}
\widetilde\tau_{\mathcal E,U}^{ER}
=
\omega_{\mathcal E}^{E}
\otimes
\rho^R.
\end{align}
Hence, the Haar-averaged relative entropy vanishes.
Moreover, the two arguments of \(\Xi\) coincide, so that
\begin{align}
\Xi
\left(
\omega_{\mathcal E}^{E}
\otimes
\rho^R
\middle\|
\omega_{\mathcal E}^{E}
\otimes
\rho^R
\right)
=
\log2
<
\infty.
\end{align}
Thus,
\begin{align}
&\mathbb E_{
U\sim
\mathrm{Haar}
\left(
\operatorname{U}
\left(
\mathcal H^A
\right)
\right)
}
\left[
D
\left(
\widetilde\tau_{\mathcal E,U}^{ER}
\middle\|
\omega_{\mathcal E}^{E}
\otimes
\rho^R
\right)
\right]
\nonumber\\
&=
0
=
C_1
\Xi
\left(
\omega_{\mathcal E}^{E}
\otimes
\rho^R
\middle\|
\omega_{\mathcal E}^{E}
\otimes
\rho^R
\right).
\label{eq:decoupling_Xi_reduction_dimension_one}
\end{align}

In the following, suppose that
\(|A|\ge2\).

\textbf{Finiteness of \(\Xi\).}
We first show the finiteness of \(\Xi\) in
\eqref{eq:decoupling_Xi_reduction_finiteness}.
It holds that
\begin{align}
\rho^{AR}
&\le
|A|
\left(
I^A
\otimes
\rho^R
\right),
\\
\omega_{\mathcal E}^{A'E}
&\le
|A|
\left(
I^{A'}
\otimes
\omega_{\mathcal E}^{E}
\right).
\end{align}
Taking the tensor product gives
\begin{align}
\rho^{AR}
\otimes
\omega_{\mathcal E}^{A'E}
\le
|A|^2
\left(
I^A
\otimes
\rho^R
\right)
\otimes
\left(
I^{A'}
\otimes
\omega_{\mathcal E}^{E}
\right).
\label{eq:decoupling_Xi_reduction_first_domination}
\end{align}
Consequently,
\begin{align}
&\left(
I^A
\otimes
\rho^R
\right)
\otimes
\left(
I^{A'}
\otimes
\omega_{\mathcal E}^{E}
\right)
+
\rho^{AR}
\otimes
\omega_{\mathcal E}^{A'E}
\nonumber\\
&\le
\left(
1+|A|^2
\right)
\left(
I^A
\otimes
\rho^R
\right)
\otimes
\left(
I^{A'}
\otimes
\omega_{\mathcal E}^{E}
\right).
\end{align}
This operator inequality implies
\begin{align}
&D
\Biggl(
\left(
I^A
\otimes
\rho^R
\right)
\otimes
\left(
I^{A'}
\otimes
\omega_{\mathcal E}^{E}
\right)
+
\rho^{AR}
\otimes
\omega_{\mathcal E}^{A'E}
\Biggm\|
\left(
I^A
\otimes
\rho^R
\right)
\otimes
\left(
I^{A'}
\otimes
\omega_{\mathcal E}^{E}
\right)
\Biggr)
\nonumber\\
&\le
\operatorname{Tr}
\Biggl[
\left(
I^A
\otimes
\rho^R
\right)
\otimes
\left(
I^{A'}
\otimes
\omega_{\mathcal E}^{E}
\right)
+
\rho^{AR}
\otimes
\omega_{\mathcal E}^{A'E}
\Biggr]
\log
\left[
1+|A|^2
\right]
<
\infty.
\end{align}
Similarly,
\begin{align}
&\left(
I^A
\otimes
\rho^R
\right)
\otimes
\left(
I^{A'}
\otimes
\omega_{\mathcal E}^{E}
\right)
\nonumber\\
&\le
\left(
I^A
\otimes
\rho^R
\right)
\otimes
\left(
I^{A'}
\otimes
\omega_{\mathcal E}^{E}
\right)
+
\rho^{AR}
\otimes
\omega_{\mathcal E}^{A'E}
\end{align}
implies
\begin{align}
&D
\Biggl(
\left(
I^A
\otimes
\rho^R
\right)
\otimes
\left(
I^{A'}
\otimes
\omega_{\mathcal E}^{E}
\right)
\Biggm\|
\left(
I^A
\otimes
\rho^R
\right)
\otimes
\left(
I^{A'}
\otimes
\omega_{\mathcal E}^{E}
\right)
+
\rho^{AR}
\otimes
\omega_{\mathcal E}^{A'E}
\Biggr)
\le
0
<
\infty.
\end{align}
Consequently,
\begin{align}
&\Xi
\Biggl(
\rho^{AR}
\otimes
\omega_{\mathcal E}^{A'E}
\Biggm\|
\left(
I^A
\otimes
\rho^R
\right)
\otimes
\left(
I^{A'}
\otimes
\omega_{\mathcal E}^{E}
\right)
\Biggr)
<
\infty.
\label{eq:decoupling_Xi_reduction_finiteness_proof}
\end{align}
This proves
\eqref{eq:decoupling_Xi_reduction_finiteness}.

\medskip
\noindent
\textbf{Perturbation and Fatou's lemma.}
We first represent the Haar-averaged relative entropy using the
perturbations in
Proposition~\ref{prop:relative_entropy_strictly_positive_perturbation}.
It holds that
\begin{align}
\rho^{AR}
\le
|A|
\left(
I^A
\otimes
\rho^R
\right).
\label{eq:finite_input_rho_dominated_by_marginal}
\end{align}
Consequently,
\begin{align}
\left(
U^A
\otimes
I^R
\right)
\rho^{AR}
\left(
U^{A\dagger}
\otimes
I^R
\right)
\le
|A|
I^A
\otimes
\rho^R.
\end{align}
Applying the positive map
\(\widehat{\mathcal E}^{A\to E}\otimes\id^R\)
and using
\begin{align}
\widehat{\mathcal E}^{A\to E}
\left(
I^A
\right)
=
|A|
\omega_{\mathcal E}^{E}
\end{align}
gives
\begin{align}
\widetilde\tau_{\mathcal E,U}^{ER}
\le
|A|^2
\omega_{\mathcal E}^{E}
\otimes
\rho^R.
\label{eq:decoupling_Xi_reduction_output_domination}
\end{align}
This operator inequality implies, for every \(U\),
\begin{align}
D
\left(
\widetilde\tau_{\mathcal E,U}^{ER}
\middle\|
\omega_{\mathcal E}^{E}
\otimes
\rho^R
\right)
\le
\operatorname{Tr}
\left[
\widetilde\tau_{\mathcal E,U}^{ER}
\right]
\log
\left[
|A|^2
\right]
<
\infty.
\end{align}
Proposition~\ref{prop:relative_entropy_strictly_positive_perturbation}
thus gives
\begin{align}
&D
\left(
\widetilde\tau_{\mathcal E,U}^{ER}
\middle\|
\omega_{\mathcal E}^{E}
\otimes
\rho^R
\right)
\nonumber\\
&=
\lim_{\delta\downarrow0}
\operatorname{Tr}
\Biggl[
\widetilde\tau_{\mathcal E,U}^{ER}
\Biggl(
\log
\left[
\widetilde\tau_{\mathcal E,U}^{ER}
+
\delta I^{ER}
\right]
-
\log
\left[
\omega_{\mathcal E}^{E}
\otimes
\rho^R
+
\delta I^{ER}
\right]
\Biggr)
\Biggr].
\label{eq:decoupling_Xi_reduction_pointwise_limit}
\end{align}

For
\(\delta>0\),
define
\begin{align}
f_\delta(U)
\coloneqq
\operatorname{Tr}
\Biggl[
\widetilde\tau_{\mathcal E,U}^{ER}
\Biggl(
\log
\left[
\widetilde\tau_{\mathcal E,U}^{ER}
+
\delta I^{ER}
\right]
-
\log
\left[
\omega_{\mathcal E}^{E}
\otimes
\rho^R
+
\delta I^{ER}
\right]
\Biggr)
\Biggr].
\end{align}
Equations
\eqref{eq:regularized_logarithm_domination}
and
\eqref{eq:relative_entropy_regularization_negative_part_bound},
together with
\begin{align}
\operatorname{Tr}
\left[
\omega_{\mathcal E}^{E}
\otimes
\rho^R
\right]
=
1,
\end{align}
give
\begin{align}
f_\delta(U)
\ge
-1
\label{eq:decoupling_Xi_reduction_uniform_lower_bound}
\end{align}
uniformly in \(U\) and \(\delta\).

By the sequential characterization of the lower limit, there exists
a sequence
\(\{\delta_n\}_{n=1}^{\infty}\)
satisfying
\(\delta_n\downarrow0\)
such that
\begin{align}
&\lim_{n\to\infty}
\mathbb E_{
U\sim
\mathrm{Haar}
\left(
\operatorname{U}
\left(
\mathcal H^A
\right)
\right)
}
\left[
f_{\delta_n}(U)
\right]
\nonumber\\
&=
\liminf_{\delta\downarrow0}
\mathbb E_{
U\sim
\mathrm{Haar}
\left(
\operatorname{U}
\left(
\mathcal H^A
\right)
\right)
}
\left[
f_\delta(U)
\right].
\end{align}
Fatou's lemma~\cite[Lemma~8.2.13]{tao2023analysis},
applied to the nonnegative functions
\(f_{\delta_n}+1\),
together with
\eqref{eq:decoupling_Xi_reduction_pointwise_limit},
gives
\begin{align}
&\mathbb E_{
U\sim
\mathrm{Haar}
\left(
\operatorname{U}
\left(
\mathcal H^A
\right)
\right)
}
\left[
D
\left(
\widetilde\tau_{\mathcal E,U}^{ER}
\middle\|
\omega_{\mathcal E}^{E}
\otimes
\rho^R
\right)
\right]
\nonumber\\
&=
\mathbb E_{
U\sim
\mathrm{Haar}
\left(
\operatorname{U}
\left(
\mathcal H^A
\right)
\right)
}
\left[
\lim_{n\to\infty}
f_{\delta_n}(U)
\right]
\nonumber\\
&\le
\liminf_{n\to\infty}
\mathbb E_{
U\sim
\mathrm{Haar}
\left(
\operatorname{U}
\left(
\mathcal H^A
\right)
\right)
}
\left[
f_{\delta_n}(U)
\right]
\nonumber\\
&=
\liminf_{\delta\downarrow0}
\mathbb E_{
U\sim
\mathrm{Haar}
\left(
\operatorname{U}
\left(
\mathcal H^A
\right)
\right)
}
\left[
f_\delta(U)
\right],
\label{eq:decoupling_Xi_reduction_Fatou}
\end{align}
where the additive constant arising from applying Fatou's lemma to
\(f_{\delta_n}+1\)
cancels from both sides.

\medskip
\noindent
\textbf{Application of Jensen's operator inequality.}
For every
\(\delta>0\),
Proposition~\ref{prop:decoupling_operator_Jensen_reduction}
gives
\begin{align}
&\mathbb E_{
U\sim
\mathrm{Haar}
\left(
\operatorname{U}
\left(
\mathcal H^A
\right)
\right)
}
\left[
\operatorname{Tr}
\left[
\widetilde\tau_{\mathcal E,U}^{ER}
\log
\left[
\widetilde\tau_{\mathcal E,U}^{ER}
+
\delta I^{ER}
\right]
\right]
\right]
\nonumber\\
&\le
\operatorname{Tr}
\Biggl[
\left(
\rho^{AR}
\otimes
\omega_{\mathcal E}^{A'E}
\right)
\log
\Biggl[
|A|^2
\mathbb E_{
U\sim
\mathrm{Haar}
\left(
\operatorname{U}
\left(
\mathcal H^A
\right)
\right)
}
\Biggl[
\left(
U^{A\dagger}
\otimes
I^{A'}
\right)
\Phi^{AA'}
\left(
U^A
\otimes
I^{A'}
\right)
\otimes
\widetilde\tau_{\mathcal E,U}^{ER}
\Biggr]
+
\delta I^{AA'ER}
\Biggr]
\Biggr].
\label{eq:decoupling_Xi_reduction_first_term_bound}
\end{align}
Moreover,
\begin{align}
&\left(
I^A
\otimes
\rho^R
\right)
\otimes
\left(
I^{A'}
\otimes
\omega_{\mathcal E}^{E}
\right)
=
I^{AA'}
\otimes
\left(
\omega_{\mathcal E}^{E}
\otimes
\rho^R
\right),
\end{align}
and hence
\begin{align}
&\log
\Biggl[
\left(
I^A
\otimes
\rho^R
\right)
\otimes
\left(
I^{A'}
\otimes
\omega_{\mathcal E}^{E}
\right)
+
\delta I^{AA'ER}
\Biggr]
\nonumber\\
&=
I^{AA'}
\otimes
\log
\left[
\omega_{\mathcal E}^{E}
\otimes
\rho^R
+
\delta I^{ER}
\right].
\label{eq:decoupling_Xi_reduction_reference_logarithm}
\end{align}
Since
\begin{align}
\operatorname{Tr}_{AA'}
\left[
\rho^{AR}
\otimes
\omega_{\mathcal E}^{A'E}
\right]
=
\omega_{\mathcal E}^{E}
\otimes
\rho^R,
\end{align}
it follows from
\eqref{eq:finite_input_average_output}
that
\begin{align}
&\operatorname{Tr}
\Biggl[
\left(
\rho^{AR}
\otimes
\omega_{\mathcal E}^{A'E}
\right)
\log
\Biggl[
\left(
I^A
\otimes
\rho^R
\right)
\otimes
\left(
I^{A'}
\otimes
\omega_{\mathcal E}^{E}
\right)
+
\delta I^{AA'ER}
\Biggr]
\Biggr]
\nonumber\\
&=
\mathbb E_{
U\sim
\mathrm{Haar}
\left(
\operatorname{U}
\left(
\mathcal H^A
\right)
\right)
}
\left[
\operatorname{Tr}
\left[
\widetilde\tau_{\mathcal E,U}^{ER}
\log
\left[
\omega_{\mathcal E}^{E}
\otimes
\rho^R
+
\delta I^{ER}
\right]
\right]
\right].
\label{eq:decoupling_Xi_reduction_reference_term}
\end{align}
Subtracting
\eqref{eq:decoupling_Xi_reduction_reference_term}
from
\eqref{eq:decoupling_Xi_reduction_first_term_bound}
gives
\begin{align}
&\mathbb E_{
U\sim
\mathrm{Haar}
\left(
\operatorname{U}
\left(
\mathcal H^A
\right)
\right)
}
\Biggl[
\operatorname{Tr}
\Biggl[
\widetilde\tau_{\mathcal E,U}^{ER}
\Biggl(
\log
\left[
\widetilde\tau_{\mathcal E,U}^{ER}
+
\delta I^{ER}
\right]
-
\log
\left[
\omega_{\mathcal E}^{E}
\otimes
\rho^R
+
\delta I^{ER}
\right]
\Biggr)
\Biggr]
\Biggr]
\nonumber\\
&\le
\operatorname{Tr}
\Biggl[
\left(
\rho^{AR}
\otimes
\omega_{\mathcal E}^{A'E}
\right)
\Biggl(
\log
\Biggl[
|A|^2
\mathbb E_{
U\sim
\mathrm{Haar}
\left(
\operatorname{U}
\left(
\mathcal H^A
\right)
\right)
}
\Biggl[
\left(
U^{A\dagger}
\otimes
I^{A'}
\right)
\Phi^{AA'}
\left(
U^A
\otimes
I^{A'}
\right)
\otimes
\widetilde\tau_{\mathcal E,U}^{ER}
\Biggr]
+
\delta I^{AA'ER}
\Biggr]
\nonumber\\
&\hspace{34mm}
-
\log
\Biggl[
\left(
I^A
\otimes
\rho^R
\right)
\otimes
\left(
I^{A'}
\otimes
\omega_{\mathcal E}^{E}
\right)
+
\delta I^{AA'ER}
\Biggr]
\Biggr)
\Biggr].
\label{eq:decoupling_Xi_reduction_regularized_bound}
\end{align}

\medskip
\noindent
\textbf{Second moment of Haar-random unitaries.}
Then, using the Haar second-moment identity
\eqref{eq:decoupling_haar_second_moment} in Lemma~\ref{lem:decoupling_haar_second_moment}, together with the Choi representation
\eqref{eq:decoupling_trace_Jensen_choi_representation},
one obtains
\begin{align}
&|A|^2
\mathbb E_{
U\sim
\mathrm{Haar}
\left(
\operatorname{U}
\left(
\mathcal H^A
\right)
\right)
}
\Biggl[
\left(
U^{A\dagger}
\otimes
I^{A'}
\right)
\Phi^{AA'}
\left(
U^A
\otimes
I^{A'}
\right)
\otimes
\widetilde\tau_{\mathcal E,U}^{ER}
\Biggr]
\nonumber\\
&=
\frac{
|A|^2
}{
|A|^2-1
}
\Biggl[
\left(
I^A
\otimes
\rho^R
\right)
\otimes
\left(
I^{A'}
\otimes
\omega_{\mathcal E}^{E}
\right)
+
\rho^{AR}
\otimes
\omega_{\mathcal E}^{A'E}
\Biggr]
\nonumber\\
&\quad-
\frac{
|A|
}{
|A|^2-1
}
\Biggl[
\rho^{AR}
\otimes
\left(
I^{A'}
\otimes
\omega_{\mathcal E}^{E}
\right)
+
\left(
I^A
\otimes
\rho^R
\right)
\otimes
\omega_{\mathcal E}^{A'E}
\Biggr].
\label{eq:decoupling_relative_entropy_Jensen_exact_second_moment}
\end{align}
Both terms in the last bracket are positive semidefinite.
Thus, dropping the corresponding subtracted term gives
\begin{align}
&|A|^2
\mathbb E_{
U\sim
\mathrm{Haar}
\left(
\operatorname{U}
\left(
\mathcal H^A
\right)
\right)
}
\Biggl[
\left(
U^{A\dagger}
\otimes
I^{A'}
\right)
\Phi^{AA'}
\left(
U^A
\otimes
I^{A'}
\right)
\otimes
\widetilde\tau_{\mathcal E,U}^{ER}
\Biggr]
\nonumber\\
&\le
\frac{
|A|^2
}{
|A|^2-1
}
\Biggl[
\left(
I^A
\otimes
\rho^R
\right)
\otimes
\left(
I^{A'}
\otimes
\omega_{\mathcal E}^{E}
\right)
+
\rho^{AR}
\otimes
\omega_{\mathcal E}^{A'E}
\Biggr].
\label{eq:decoupling_Xi_reduction_second_moment_upper_bound}
\end{align}
Since
\(|A|^2/(|A|^2-1)>1\),
it follows that, for every
\(\delta>0\),
\begin{align}
&|A|^2
\mathbb E_{
U\sim
\mathrm{Haar}
\left(
\operatorname{U}
\left(
\mathcal H^A
\right)
\right)
}
\Biggl[
\left(
U^{A\dagger}
\otimes
I^{A'}
\right)
\Phi^{AA'}
\left(
U^A
\otimes
I^{A'}
\right)
\otimes
\widetilde\tau_{\mathcal E,U}^{ER}
\Biggr]
+
\delta I^{AA'ER}
\nonumber\\
&\le
\frac{
|A|^2
}{
|A|^2-1
}
\Biggl[
\left(
I^A
\otimes
\rho^R
\right)
\otimes
\left(
I^{A'}
\otimes
\omega_{\mathcal E}^{E}
\right)
+
\rho^{AR}
\otimes
\omega_{\mathcal E}^{A'E}
+
\delta I^{AA'ER}
\Biggr].
\end{align}
Operator monotonicity of the logarithm therefore gives
\begin{align}
&\log
\Biggl[
|A|^2
\mathbb E_{
U\sim
\mathrm{Haar}
\left(
\operatorname{U}
\left(
\mathcal H^A
\right)
\right)
}
\Biggl[
\left(
U^{A\dagger}
\otimes
I^{A'}
\right)
\Phi^{AA'}
\left(
U^A
\otimes
I^{A'}
\right)
\otimes
\widetilde\tau_{\mathcal E,U}^{ER}
\Biggr]
+
\delta I^{AA'ER}
\Biggr]
\nonumber\\
&\le
\log
\left[
\frac{
|A|^2
}{
|A|^2-1
}
\right]
I^{AA'ER}
\nonumber\\
&\quad+
\log
\Biggl[
\left(
I^A
\otimes
\rho^R
\right)
\otimes
\left(
I^{A'}
\otimes
\omega_{\mathcal E}^{E}
\right)
+
\rho^{AR}
\otimes
\omega_{\mathcal E}^{A'E}
+
\delta I^{AA'ER}
\Biggr].
\end{align}
Therefore, the right-hand side of
\eqref{eq:decoupling_Xi_reduction_regularized_bound}
is bounded by
\begin{align}
&\log
\left[
\frac{
|A|^2
}{
|A|^2-1
}
\right]
\nonumber\\
&\quad+
\operatorname{Tr}
\Biggl[
\left(
\rho^{AR}
\otimes
\omega_{\mathcal E}^{A'E}
\right)
\Biggl(
\log
\Biggl[
\left(
I^A
\otimes
\rho^R
\right)
\otimes
\left(
I^{A'}
\otimes
\omega_{\mathcal E}^{E}
\right)
+
\rho^{AR}
\otimes
\omega_{\mathcal E}^{A'E}
+
\delta I^{AA'ER}
\Biggr]
\nonumber\\
&\hspace{34mm}
-
\log
\Biggl[
\left(
I^A
\otimes
\rho^R
\right)
\otimes
\left(
I^{A'}
\otimes
\omega_{\mathcal E}^{E}
\right)
+
\delta I^{AA'ER}
\Biggr]
\Biggr)
\Biggr].
\label{eq:decoupling_Xi_reduction_log_upper_bound}
\end{align}

\medskip
\noindent
\textbf{Bound in \(\Xi\) with an additive constant.}
Combining
\eqref{eq:decoupling_Xi_reduction_Fatou}
and
\eqref{eq:decoupling_Xi_reduction_log_upper_bound}
gives
\begin{align}
&\mathbb E_{
U\sim
\mathrm{Haar}
\left(
\operatorname{U}
\left(
\mathcal H^A
\right)
\right)
}
\left[
D
\left(
\widetilde\tau_{\mathcal E,U}^{ER}
\middle\|
\omega_{\mathcal E}^{E}
\otimes
\rho^R
\right)
\right]
\nonumber\\
&\le
\log
\left[
\frac{
|A|^2
}{
|A|^2-1
}
\right]
\nonumber\\
&\quad+
\liminf_{\delta\downarrow0}
\operatorname{Tr}
\Biggl[
\left(
\rho^{AR}
\otimes
\omega_{\mathcal E}^{A'E}
\right)
\Biggl(
\log
\Biggl[
\left(
I^A
\otimes
\rho^R
\right)
\otimes
\left(
I^{A'}
\otimes
\omega_{\mathcal E}^{E}
\right)
+
\rho^{AR}
\otimes
\omega_{\mathcal E}^{A'E}
+
\delta I^{AA'ER}
\Biggr]
\nonumber\\
&\hspace{34mm}
-
\log
\Biggl[
\left(
I^A
\otimes
\rho^R
\right)
\otimes
\left(
I^{A'}
\otimes
\omega_{\mathcal E}^{E}
\right)
+
\delta I^{AA'ER}
\Biggr]
\Biggr)
\Biggr].
\label{eq:decoupling_Xi_reduction_additive_bound_liminf}
\end{align}

The finiteness established in
\eqref{eq:decoupling_Xi_reduction_finiteness_proof}
allows us to identify this lower limit.
Proposition~\ref{prop:relative_entropy_strictly_positive_perturbation},
applied to the two relative entropies in the definition of \(\Xi\),
gives
\begin{align}
&\lim_{\delta\downarrow0}
\operatorname{Tr}
\Biggl[
\left(
\rho^{AR}
\otimes
\omega_{\mathcal E}^{A'E}
\right)
\Biggl(
\log
\Biggl[
\left(
I^A
\otimes
\rho^R
\right)
\otimes
\left(
I^{A'}
\otimes
\omega_{\mathcal E}^{E}
\right)
+
\rho^{AR}
\otimes
\omega_{\mathcal E}^{A'E}
+
\delta I^{AA'ER}
\Biggr]
\nonumber\\
&\hspace{34mm}
-
\log
\Biggl[
\left(
I^A
\otimes
\rho^R
\right)
\otimes
\left(
I^{A'}
\otimes
\omega_{\mathcal E}^{E}
\right)
+
\delta I^{AA'ER}
\Biggr]
\Biggr)
\Biggr]
\nonumber\\
&=
\Xi
\Biggl(
\rho^{AR}
\otimes
\omega_{\mathcal E}^{A'E}
\Biggm\|
\left(
I^A
\otimes
\rho^R
\right)
\otimes
\left(
I^{A'}
\otimes
\omega_{\mathcal E}^{E}
\right)
\Biggr).
\label{eq:decoupling_Xi_reduction_regularized_limit}
\end{align}
Therefore,
\begin{align}
&\mathbb E_{
U\sim
\mathrm{Haar}
\left(
\operatorname{U}
\left(
\mathcal H^A
\right)
\right)
}
\left[
D
\left(
\widetilde\tau_{\mathcal E,U}^{ER}
\middle\|
\omega_{\mathcal E}^{E}
\otimes
\rho^R
\right)
\right]
\nonumber\\
&\le
\log
\left[
\frac{
|A|^2
}{
|A|^2-1
}
\right]
+
\Xi
\Biggl(
\rho^{AR}
\otimes
\omega_{\mathcal E}^{A'E}
\Biggm\|
\left(
I^A
\otimes
\rho^R
\right)
\otimes
\left(
I^{A'}
\otimes
\omega_{\mathcal E}^{E}
\right)
\Biggr).
\label{eq:decoupling_Xi_reduction_additive_bound}
\end{align}

\medskip
\noindent
\textbf{Multiplicative bound in \(\Xi\).}
The data processing inequality under the trace map gives
\begin{align}
&\Xi
\Biggl(
\rho^{AR}
\otimes
\omega_{\mathcal E}^{A'E}
\Biggm\|
\left(
I^A
\otimes
\rho^R
\right)
\otimes
\left(
I^{A'}
\otimes
\omega_{\mathcal E}^{E}
\right)
\Biggr)
\nonumber\\
&\ge
D
\left(
|A|^2+1
\middle\|
|A|^2
\right)
+
D
\left(
|A|^2
\middle\|
|A|^2+1
\right)
\nonumber\\
&=
\log
\left[
1+\frac{1}{|A|^2}
\right],
\label{eq:decoupling_Xi_reduction_universal_lower_bound}
\end{align}
where
\begin{align}
\operatorname{Tr}
\left[
\rho^{AR}
\otimes
\omega_{\mathcal E}^{A'E}
\right]
&=
1,
\\
\operatorname{Tr}
\Biggl[
\left(
I^A
\otimes
\rho^R
\right)
\otimes
\left(
I^{A'}
\otimes
\omega_{\mathcal E}^{E}
\right)
\Biggr]
&=
|A|^2.
\end{align}
Therefore,
\begin{align}
&\log
\left[
\frac{
|A|^2
}{
|A|^2-1
}
\right]
\nonumber\\
&\le
\frac{
\log
\left[
\frac{
|A|^2
}{
|A|^2-1
}
\right]
}{
\log
\left[
1+\frac{1}{|A|^2}
\right]
}
\Xi
\Biggl(
\rho^{AR}
\otimes
\omega_{\mathcal E}^{A'E}
\Biggm\|
\left(
I^A
\otimes
\rho^R
\right)
\otimes
\left(
I^{A'}
\otimes
\omega_{\mathcal E}^{E}
\right)
\Biggr).
\label{eq:decoupling_Xi_reduction_additive_constant_absorption}
\end{align}
Combining
\eqref{eq:decoupling_Xi_reduction_additive_constant_absorption}
with
\eqref{eq:decoupling_Xi_reduction_additive_bound}
proves
\eqref{eq:decoupling_Xi_reduction}.
\end{proof}

\subsection{Bounds on one-shot decoupling}
\label{subsec:finite_input_one_shot_decoupling_results}

We conclude by combining the Haar-averaged relative-entropy decoupling bound in terms of \(\Xi\) with the optimal
trace inequality for \(\Xi\).

\begin{theorem}[One-shot decoupling for finite-dimensional input and
separable reference and output systems]
\label{thm:finite_input_infinite_reference_decoupling}
Let
\(\mathcal H^A\)
be finite-dimensional, and let
\(\mathcal H^R\) and \(\mathcal H^E\)
be separable and possibly infinite-dimensional.
Let
\begin{align}
\rho^{AR}
\in
\operatorname{D}
\left(
\mathcal H^A
\otimes
\mathcal H^R
\right),
\end{align}
and let
\begin{align}
\mathcal E^{A\to E}
:
\operatorname{T}
\left(
\mathcal H^A
\right)
\longrightarrow
\operatorname{T}
\left(
\mathcal H^E
\right)
\end{align}
be a nonzero bounded completely positive map.
Let
\(\widehat{\mathcal E}^{A\to E}\),
\(\omega_{\mathcal E}^{A'E}\),
\(\omega_{\mathcal E}^{E}\), and
\(\widetilde\tau_{\mathcal E,U}^{ER}\)
be defined in
\eqref{eq:finite_input_normalized_map},
\eqref{eq:finite_input_normalized_choi_state_and_marginal}, and
\eqref{eq:finite_input_randomized_output},
respectively.
Then
\begin{align}
&\mathbb E_{
U\sim
\mathrm{Haar}
\left(
\operatorname{U}
\left(
\mathcal H^A
\right)
\right)
}
\left[
D
\left(
\widetilde\tau_{\mathcal E,U}^{ER}
\middle\|
\omega_{\mathcal E}^{E}
\otimes
\rho^R
\right)
\right]
\nonumber\\
&\le
\inf_{0<s\le1}
C_{|A|}
G_s
\exp
\Biggl[
-s
\Biggl(
\widetilde H_{1+s}(A|R)_\rho
+
\widetilde H_{1+s}(A'|E)_{\omega_{\mathcal E}}
\Biggr)
\Biggr],
\label{eq:finite_input_infinite_reference_decoupling_bound}
\end{align}
where
\(C_{|A|}\)
and
\(G_s\)
are defined in
\eqref{eq:decoupling_dimension_dependent_constant}
and
\eqref{eq:decoupling_Gs_definition},
respectively, and the sandwiched conditional R\'enyi entropies are
defined in
\eqref{eq:sandwiched_conditional_entropy_downarrow}.
In particular, there exists
\begin{align}
U
\in
\operatorname{U}
\left(
\mathcal H^A
\right)
\end{align}
such that
\begin{align}
&D
\left(
\widetilde\tau_{\mathcal E,U}^{ER}
\middle\|
\omega_{\mathcal E}^{E}
\otimes
\rho^R
\right)
\nonumber\\
&\le
\inf_{0<s\le1}
C_{|A|}
G_s
\exp
\Biggl[
-s
\Biggl(
\widetilde H_{1+s}(A|R)_\rho
+
\widetilde H_{1+s}(A'|E)_{\omega_{\mathcal E}}
\Biggr)
\Biggr].
\label{eq:finite_input_infinite_reference_decoupling_existential_bound}
\end{align}
\end{theorem}

\begin{proof}
\textbf{The case \(\lvert A\rvert=1\).}
Suppose first that
\(|A|=1\).
Then
\begin{align}
\rho^{AR}
=
\ket{0}\bra{0}^{A}
\otimes
\rho^R,
\end{align}
and, for every \(U\),
\begin{align}
\widetilde\tau_{\mathcal E,U}^{ER}
=
\omega_{\mathcal E}^{E}
\otimes
\rho^R.
\end{align}
Consequently,
\begin{align}
\mathbb E_{
U\sim
\mathrm{Haar}
\left(
\operatorname{U}
\left(
\mathcal H^A
\right)
\right)
}
\left[
D
\left(
\widetilde\tau_{\mathcal E,U}^{ER}
\middle\|
\omega_{\mathcal E}^{E}
\otimes
\rho^R
\right)
\right]
=
0.
\end{align}
Moreover, for every
\(s\in(0,1]\),
\begin{align}
\widetilde Q_{1+s}
\left(
\rho^{AR}
\middle\|
I^A
\otimes
\rho^R
\right)
&=
1,
\\
\widetilde Q_{1+s}
\left(
\omega_{\mathcal E}^{A'E}
\middle\|
I^{A'}
\otimes
\omega_{\mathcal E}^{E}
\right)
&=
1,
\end{align}
and hence
\begin{align}
\widetilde H_{1+s}(A|R)_\rho
&=
0,
\\
\widetilde H_{1+s}(A'|E)_{\omega_{\mathcal E}}
&=
0.
\end{align}
Since
\(C_1=0\),
the expression inside the infimum on the right-hand side of
\eqref{eq:finite_input_infinite_reference_decoupling_bound}
is equal to zero for every
\(s\in(0,1]\).
Thus, both sides of
\eqref{eq:finite_input_infinite_reference_decoupling_bound}
are equal to zero.

\medskip
\noindent
\textbf{The case \(\lvert A\rvert\ge2\).}
Suppose henceforth that
\(|A|\ge2\),
and fix
\(s\in(0,1]\).
The operator inequality
\eqref{eq:decoupling_Xi_reduction_first_domination}
implies
\begin{align}
\operatorname{supp}
\left(
\rho^{AR}
\otimes
\omega_{\mathcal E}^{A'E}
\right)
\subseteq
\operatorname{supp}
\Biggl[
\left(
I^A
\otimes
\rho^R
\right)
\otimes
\left(
I^{A'}
\otimes
\omega_{\mathcal E}^{E}
\right)
\Biggr],
\end{align}
and
Proposition~\ref{prop:decoupling_Xi_reduction}
gives
\eqref{eq:decoupling_Xi_reduction_finiteness}.
Thus, the assumptions of
Proposition~\ref{prop:infinite_dimensional_trace_inequality}
are satisfied.
Combining
Propositions~\ref{prop:decoupling_Xi_reduction}
and
\ref{prop:infinite_dimensional_trace_inequality}
gives
\begin{align}
&\mathbb E_{
U\sim
\mathrm{Haar}
\left(
\operatorname{U}
\left(
\mathcal H^A
\right)
\right)
}
\left[
D
\left(
\widetilde\tau_{\mathcal E,U}^{ER}
\middle\|
\omega_{\mathcal E}^{E}
\otimes
\rho^R
\right)
\right]
\nonumber\\
&\le
C_{|A|}
G_s
\widetilde Q_{1+s}
\Biggl(
\rho^{AR}
\otimes
\omega_{\mathcal E}^{A'E}
\Biggm\|
\left(
I^A
\otimes
\rho^R
\right)
\otimes
\left(
I^{A'}
\otimes
\omega_{\mathcal E}^{E}
\right)
\Biggr).
\label{eq:finite_input_decoupling_before_product_factorization}
\end{align}
Multiplicativity of the sandwiched quasi-divergence gives
\begin{align}
&\widetilde Q_{1+s}
\Biggl(
\rho^{AR}
\otimes
\omega_{\mathcal E}^{A'E}
\Biggm\|
\left(
I^A
\otimes
\rho^R
\right)
\otimes
\left(
I^{A'}
\otimes
\omega_{\mathcal E}^{E}
\right)
\Biggr)
\nonumber\\
&=
\widetilde Q_{1+s}
\left(
\rho^{AR}
\middle\|
I^A
\otimes
\rho^R
\right)
\widetilde Q_{1+s}
\left(
\omega_{\mathcal E}^{A'E}
\middle\|
I^{A'}
\otimes
\omega_{\mathcal E}^{E}
\right).
\label{eq:finite_input_decoupling_sandwiched_product}
\end{align}
By the definition in
\eqref{eq:sandwiched_conditional_entropy_downarrow},
\begin{align}
\widetilde Q_{1+s}
\left(
\rho^{AR}
\middle\|
I^A
\otimes
\rho^R
\right)
&=
\exp
\left[
-s
\widetilde H_{1+s}(A|R)_\rho
\right],
\\
\widetilde Q_{1+s}
\left(
\omega_{\mathcal E}^{A'E}
\middle\|
I^{A'}
\otimes
\omega_{\mathcal E}^{E}
\right)
&=
\exp
\left[
-s
\widetilde H_{1+s}(A'|E)_{\omega_{\mathcal E}}
\right].
\end{align}
Substituting these identities and
\eqref{eq:finite_input_decoupling_sandwiched_product}
into
\eqref{eq:finite_input_decoupling_before_product_factorization}
gives
\begin{align}
&\mathbb E_{
U\sim
\mathrm{Haar}
\left(
\operatorname{U}
\left(
\mathcal H^A
\right)
\right)
}
\left[
D
\left(
\widetilde\tau_{\mathcal E,U}^{ER}
\middle\|
\omega_{\mathcal E}^{E}
\otimes
\rho^R
\right)
\right]
\nonumber\\
&\le
C_{|A|}
G_s
\exp
\Biggl[
-s
\Biggl(
\widetilde H_{1+s}(A|R)_\rho
+
\widetilde H_{1+s}(A'|E)_{\omega_{\mathcal E}}
\Biggr)
\Biggr].
\end{align}
Since this inequality holds for every
\(s\in(0,1]\),
taking the infimum over \(s\) proves
\eqref{eq:finite_input_infinite_reference_decoupling_bound}.

The existence statement
\eqref{eq:finite_input_infinite_reference_decoupling_existential_bound}
follows because there exists a unitary whose fixed-unitary relative
entropy is no larger than its Haar average.
\end{proof}

\begin{remark}[One-shot decoupling in purified distance and
trace distance]
Suppose, in addition to the assumptions of
Theorem~\ref{thm:finite_input_infinite_reference_decoupling},
that the normalized completely positive map
\(\widehat{\mathcal E}^{A\to E}\)
is trace-preserving.
Then
\(\widetilde\tau_{\mathcal E,U}^{ER}\)
is a normalized state for every
\(U\in\operatorname{U}(\mathcal H^A)\).
Therefore,
\eqref{eq:decoupling_relative_entropy_controls_PD}
and
\eqref{eq:decoupling_relative_entropy_controls_trace_distance},
together with
Theorem~\ref{thm:finite_input_infinite_reference_decoupling},
imply
\begin{align}
&\mathbb E_{
U\sim
\mathrm{Haar}
\left(
\operatorname{U}
\left(
\mathcal H^A
\right)
\right)
}
\left[
P
\left(
\widetilde\tau_{\mathcal E,U}^{ER},
\omega_{\mathcal E}^{E}
\otimes
\rho^R
\right)^2
\right]
\nonumber\\
&\le
\inf_{0<s\le1}
C_{|A|}
G_s
\exp
\Biggl[
-s
\Biggl(
\widetilde H_{1+s}(A|R)_\rho
+
\widetilde H_{1+s}(A'|E)_{\omega_{\mathcal E}}
\Biggr)
\Biggr],
\label{eq:finite_input_decoupling_purified_distance_bound}
\\
&\mathbb E_{
U\sim
\mathrm{Haar}
\left(
\operatorname{U}
\left(
\mathcal H^A
\right)
\right)
}
\left[
\frac{1}{4}
\left\|
\widetilde\tau_{\mathcal E,U}^{ER}
-
\omega_{\mathcal E}^{E}
\otimes
\rho^R
\right\|_1^2
\right]
\nonumber\\
&\le
\inf_{0<s\le1}
C_{|A|}
G_s
\exp
\Biggl[
-s
\Biggl(
\widetilde H_{1+s}(A|R)_\rho
+
\widetilde H_{1+s}(A'|E)_{\omega_{\mathcal E}}
\Biggr)
\Biggr].
\label{eq:finite_input_decoupling_trace_distance_bound}
\end{align}
\end{remark}

As a direct consequence of
Theorem~\ref{thm:finite_input_infinite_reference_decoupling},
we obtain the following asymptotic error-exponent bound.

\begin{corollary}[Asymptotic error exponent of one-shot decoupling for
finite-dimensional input and separable reference and output systems]
\label{cor:finite_input_infinite_reference_decoupling_exponent}
Let
\(\mathcal H^A\)
be finite-dimensional, and let
\(\mathcal H^R\) and \(\mathcal H^E\)
be separable and possibly infinite-dimensional.
Let
\begin{align}
\rho^{AR}
\in
\operatorname{D}
\left(
\mathcal H^A
\otimes
\mathcal H^R
\right),
\end{align}
and let
\begin{align}
\mathcal E^{A\to E}
:
\operatorname{T}
\left(
\mathcal H^A
\right)
\longrightarrow
\operatorname{T}
\left(
\mathcal H^E
\right)
\end{align}
be a nonzero bounded completely positive map.
Let
\(\widehat{\mathcal E}^{A\to E}\),
\(\omega_{\mathcal E}^{A'E}\),
\(\omega_{\mathcal E}^{E}\), and
\(\widetilde\tau_{\mathcal E,U}^{ER}\)
be defined in
\eqref{eq:finite_input_normalized_map},
\eqref{eq:finite_input_normalized_choi_state_and_marginal}, and
\eqref{eq:finite_input_randomized_output},
respectively.
Then
\begin{align}
&\liminf_{n\to\infty}
-\frac{1}{n}
\log
\mathbb E_{
U_n\sim
\mathrm{Haar}
\left(
\operatorname{U}
\left(
(\mathcal H^A)^{\otimes n}
\right)
\right)
}
\left[
D
\left(
\widetilde\tau_{
\mathcal E^{\otimes n},U_n
}^{E^nR^n}
\middle\|
\left(
\omega_{\mathcal E}^{E}
\right)^{\otimes n}
\otimes
\left(
\rho^R
\right)^{\otimes n}
\right)
\right]
\nonumber\\
&\ge
\sup_{0<s\le1}
s
\Biggl(
\widetilde H_{1+s}(A|R)_\rho
+
\widetilde H_{1+s}(A'|E)_{\omega_{\mathcal E}}
\Biggr).
\label{eq:finite_input_infinite_reference_decoupling_exponent}
\end{align}
\end{corollary}

\begin{proof}
Fix
\(s\in(0,1]\).
The normalization in
\eqref{eq:finite_input_normalized_map}
is multiplicative under tensor products.
Hence, for every
\(n\in\{1,2,\ldots\}\),
the normalized Choi state and its marginal associated with
\(\mathcal E^{\otimes n}\)
are
\begin{align}
\omega_{\mathcal E^{\otimes n}}^{(A')^nE^n}
&=
\left(
\omega_{\mathcal E}^{A'E}
\right)^{\otimes n},
\\
\omega_{\mathcal E^{\otimes n}}^{E^n}
&=
\left(
\omega_{\mathcal E}^{E}
\right)^{\otimes n}.
\end{align}
Applying
Theorem~\ref{thm:finite_input_infinite_reference_decoupling}
to
\(\rho^{\otimes n}\)
and
\(\mathcal E^{\otimes n}\)
gives
\begin{align}
&\mathbb E_{
U_n\sim
\mathrm{Haar}
\left(
\operatorname{U}
\left(
(\mathcal H^A)^{\otimes n}
\right)
\right)
}
\left[
D
\left(
\widetilde\tau_{
\mathcal E^{\otimes n},U_n
}^{E^nR^n}
\middle\|
\left(
\omega_{\mathcal E}^{E}
\right)^{\otimes n}
\otimes
\left(
\rho^R
\right)^{\otimes n}
\right)
\right]
\nonumber\\
&\le
\inf_{0<t\le1}
C_{|A|^n}
G_t
\exp
\Biggl[
-t
\Biggl(
\widetilde H_{1+t}
\left(
A^n
\middle|
R^n
\right)_{\rho^{\otimes n}}
+
\widetilde H_{1+t}
\left(
(A')^n
\middle|
E^n
\right)_{\omega_{\mathcal E}^{\otimes n}}
\Biggr)
\Biggr]
\nonumber\\
&\le
C_{|A|^n}
G_s
\exp
\Biggl[
-s
\Biggl(
\widetilde H_{1+s}
\left(
A^n
\middle|
R^n
\right)_{\rho^{\otimes n}}
+
\widetilde H_{1+s}
\left(
(A')^n
\middle|
E^n
\right)_{\omega_{\mathcal E}^{\otimes n}}
\Biggr)
\Biggr].
\label{eq:finite_input_infinite_reference_decoupling_iid_bound}
\end{align}
By multiplicativity of the sandwiched quasi-divergence, the
sandwiched conditional R\'enyi entropy is additive under tensor
products.
Therefore,
\begin{align}
\widetilde H_{1+s}
\left(
A^n
\middle|
R^n
\right)_{\rho^{\otimes n}}
&=
n
\widetilde H_{1+s}(A|R)_\rho,
\\
\widetilde H_{1+s}
\left(
(A')^n
\middle|
E^n
\right)_{\omega_{\mathcal E}^{\otimes n}}
&=
n
\widetilde H_{1+s}(A'|E)_{\omega_{\mathcal E}}.
\end{align}
Substituting these identities into
\eqref{eq:finite_input_infinite_reference_decoupling_iid_bound}
gives
\begin{align}
&\mathbb E_{
U_n\sim
\mathrm{Haar}
\left(
\operatorname{U}
\left(
(\mathcal H^A)^{\otimes n}
\right)
\right)
}
\left[
D
\left(
\widetilde\tau_{
\mathcal E^{\otimes n},U_n
}^{E^nR^n}
\middle\|
\left(
\omega_{\mathcal E}^{E}
\right)^{\otimes n}
\otimes
\left(
\rho^R
\right)^{\otimes n}
\right)
\right]
\nonumber\\
&\le
C_{|A|^n}
G_s
\exp
\Biggl[
-ns
\Biggl(
\widetilde H_{1+s}(A|R)_\rho
+
\widetilde H_{1+s}(A'|E)_{\omega_{\mathcal E}}
\Biggr)
\Biggr].
\label{eq:finite_input_infinite_reference_decoupling_iid_additive_bound}
\end{align}

Using the conventions
\begin{align}
\log0
&=
-\infty,
&
-\log0
&=
+\infty,
\end{align}
we may apply the extended-real logarithm to
\eqref{eq:finite_input_infinite_reference_decoupling_iid_additive_bound}.
If
\(C_{|A|^n}=0\),
then
\eqref{eq:finite_input_average_relative_entropy_nonnegative}
and
\eqref{eq:finite_input_infinite_reference_decoupling_iid_additive_bound}
imply that the Haar-averaged relative entropy is zero.
Consequently, the left-hand side of
\eqref{eq:finite_input_infinite_reference_decoupling_exponent}
is \(+\infty\), and the claim follows immediately.

It remains to consider the case in which
\begin{align}
C_{|A|^n}
>
0
\qquad
\text{for every
\(n\in\{1,2,\ldots\}\)}.
\end{align}
Taking the logarithm of
\eqref{eq:finite_input_infinite_reference_decoupling_iid_additive_bound}
and dividing by \(-n\) gives
\begin{align}
&-\frac{1}{n}
\log
\mathbb E_{
U_n\sim
\mathrm{Haar}
\left(
\operatorname{U}
\left(
(\mathcal H^A)^{\otimes n}
\right)
\right)
}
\left[
D
\left(
\widetilde\tau_{
\mathcal E^{\otimes n},U_n
}^{E^nR^n}
\middle\|
\left(
\omega_{\mathcal E}^{E}
\right)^{\otimes n}
\otimes
\left(
\rho^R
\right)^{\otimes n}
\right)
\right]
\nonumber\\
&\ge
s
\Biggl(
\widetilde H_{1+s}(A|R)_\rho
+
\widetilde H_{1+s}(A'|E)_{\omega_{\mathcal E}}
\Biggr)
-
\frac{1}{n}
\log
\left[
C_{|A|^n}
G_s
\right].
\label{eq:finite_input_infinite_reference_decoupling_exponent_fixed_s}
\end{align}
By
\eqref{eq:decoupling_dimension_dependent_constant},
\begin{align}
\lim_{n\to\infty}
C_{|A|^n}
=
2.
\end{align}
Since
\(G_s\)
is positive, finite, and independent of \(n\),
\begin{align}
\lim_{n\to\infty}
\frac{1}{n}
\log
\left[
C_{|A|^n}
G_s
\right]
=
0.
\label{eq:finite_input_infinite_reference_decoupling_prefactor_rate}
\end{align}
Taking the lower limit in
\eqref{eq:finite_input_infinite_reference_decoupling_exponent_fixed_s}
and using
\eqref{eq:finite_input_infinite_reference_decoupling_prefactor_rate}
gives
\begin{align}
&\liminf_{n\to\infty}
-\frac{1}{n}
\log
\mathbb E_{
U_n\sim
\mathrm{Haar}
\left(
\operatorname{U}
\left(
(\mathcal H^A)^{\otimes n}
\right)
\right)
}
\left[
D
\left(
\widetilde\tau_{
\mathcal E^{\otimes n},U_n
}^{E^nR^n}
\middle\|
\left(
\omega_{\mathcal E}^{E}
\right)^{\otimes n}
\otimes
\left(
\rho^R
\right)^{\otimes n}
\right)
\right]
\nonumber\\
&\ge
s
\Biggl(
\widetilde H_{1+s}(A|R)_\rho
+
\widetilde H_{1+s}(A'|E)_{\omega_{\mathcal E}}
\Biggr).
\end{align}
Since this inequality holds for every
\(s\in(0,1]\),
taking the supremum over \(s\) proves
\eqref{eq:finite_input_infinite_reference_decoupling_exponent}.
\end{proof}

\begin{remark}[Converse for one-shot decoupling]
\label{remark:converse_one_shot_decoupling}
Due to the converse bounds shown in Ref.~\cite[Theorem~5, Corollary~6, and Theorem~7]{berta2026tightanyshotquantumdecoupling}, the asymptotic decoupling-exponent bound achieved in
Corollary~\ref{cor:finite_input_infinite_reference_decoupling_exponent}
is ensemble-tight up to a critical rate when the reference system \(\mathcal H^R\)
is finite-dimensional and
\(\mathcal E^{A\to E}\)
is a partial-trace map.
The converse bound on the decoupling exponent in
Ref.~\cite[Theorem~5]{
berta2026tightanyshotquantumdecoupling}
relies on a pinching inequality
\cite[Lemma~9 and its proof]{hayashi2002optimal}
with respect to the spectral decomposition of the reference marginal.
Extending this converse to an infinite-dimensional reference system may therefore require different techniques or additional assumptions, and we leave the analysis of such extensions for future work.
Nevertheless, the contribution of the present work is to provide the decoupling protocol that applies regardless of whether the reference system is finite- or infinite-dimensional and provably achieves the ensemble-tight exponent when the reference system is finite-dimensional.
\end{remark}

\section{Infinite-dimensional IID decoupling}
\label{sec:infinite_dimensional_IID_decoupling}

In this section, we formulate and analyze IID decoupling by partial
trace for separable and possibly infinite-dimensional input and
reference systems.
In
Subsection~\ref{subsec:formulation_infinite_dimensional_IID_decoupling},
we define the infinite-dimensional IID decoupling task, its error
criteria, and the asymptotic dimension rates considered below.
In
Subsection~\ref{subsec:finite_rank_spectral_cutoff},
for each fixed block length, we construct finite-rank spectral cutoffs
for the corresponding block of IID states.
In
Subsection~\ref{subsec:high_probability_finite_rank_IID_decoupling},
we use these blockwise cutoffs to construct high-probability
finite-rank projections on the entire \(n\)-copy input space.
In
Subsection~\ref{subsec:infinite_dimensional_IID_decoupling_achievability},
we combine these approximation tools with random-unitary decoupling
arguments for finite-dimensional input systems and separable,
possibly infinite-dimensional reference systems to prove the
achievability bounds on infinite-dimensional IID decoupling.
Finally, in
Subsection~\ref{subsec:converse_infinite_dimensional_IID_decoupling},
we prove the corresponding converse bounds and thereby identify the
optimal asymptotic rates.

\subsection{Formulation of infinite-dimensional IID decoupling}
\label{subsec:formulation_infinite_dimensional_IID_decoupling}

We formulate the infinite-dimensional IID decoupling task via partial
trace.

Let
\(\mathcal H^A\)
and
\(\mathcal H^R\)
be separable and possibly infinite-dimensional, and let
\begin{align}
\rho^{AR}
\in
\operatorname{D}
\left(
\mathcal H^A
\otimes
\mathcal H^R
\right).
\end{align}
Throughout this section, we assume that
\begin{align}
H(A)_\rho
<
\infty,
\label{eq:IID_decoupling_finite_entropy_assumption}
\end{align}
which may be viewed as an energy constraint, as discussed in
\eqref{eq:finite_entropy_finite_energy_condition}.

For every
\(n\in\{1,2,\ldots\}\),
let
\(\Pi_n^{A^n}\)
be a nonzero finite-rank projection on
\(\mathcal H^{A^n}\), and define
\begin{align}
\Pi_n^{\perp A^n}
&\coloneqq
I^{A^n}
-
\Pi_n^{A^n},
\label{eq:IID_decoupling_projection_complement}
\\
\delta_{\Pi_n}
&\coloneqq
\operatorname{Tr}
\left[
\Pi_n^{\perp A^n}
\left(
\rho^A
\right)^{\otimes n}
\right].
\label{eq:IID_decoupling_projection_error}
\end{align}
Thus,
\(\delta_{\Pi_n}\)
is the probability of obtaining the failure outcome when performing
the binary projective measurement
\(\{\Pi_n^{A^n},\Pi_n^{\perp A^n}\}\)
on
\((\rho^A)^{\otimes n}\).
We assume that
\begin{align}
\delta_{\Pi_n}
<
1.
\end{align}

Define the projected finite-dimensional input space and its dimension
by
\begin{align}
\mathcal H^{A_{\Pi_n}^n}
&\coloneqq
\Pi_n^{A^n}
\mathcal H^{A^n},
\label{eq:IID_decoupling_projected_input_space}
\\
\left|
A_{\Pi_n}^n
\right|
&\coloneqq
\operatorname{rank}
\Pi_n^{A^n}.
\label{eq:IID_decoupling_projected_input_dimension}
\end{align}
Conditioned on successful projection, the normalized projected state
and its input and reference marginals are
\begin{align}
\rho_{\Pi_n}^{A_{\Pi_n}^nR^n}
&\coloneqq
\frac{
\left(
\Pi_n^{A^n}
\otimes
I^{R^n}
\right)
\left(
\rho^{AR}
\right)^{\otimes n}
\left(
\Pi_n^{A^n}
\otimes
I^{R^n}
\right)
}{
1-\delta_{\Pi_n}
},
\label{eq:IID_decoupling_projected_state}
\\
\rho_{\Pi_n}^{A_{\Pi_n}^n}
&\coloneqq
\operatorname{Tr}_{R^n}
\left[
\rho_{\Pi_n}^{A_{\Pi_n}^nR^n}
\right],
\label{eq:IID_decoupling_projected_input_marginal}
\\
\rho_{\Pi_n}^{R^n}
&\coloneqq
\operatorname{Tr}_{A_{\Pi_n}^n}
\left[
\rho_{\Pi_n}^{A_{\Pi_n}^nR^n}
\right].
\label{eq:IID_decoupling_projected_reference}
\end{align}

Define the complementary input space by
\begin{align}
\mathcal H^{A_{\Pi_n^\perp}^n}
\coloneqq
\Pi_n^{\perp A^n}
\mathcal H^{A^n}.
\label{eq:IID_decoupling_complementary_input_space}
\end{align}
Whenever
\(\delta_{\Pi_n}>0\),
the normalized state conditioned on the complementary outcome and its
input and reference marginals are defined by
\begin{align}
\rho_{\Pi_n^\perp}^{A_{\Pi_n^\perp}^nR^n}
&\coloneqq
\frac{
\left(
\Pi_n^{\perp A^n}
\otimes
I^{R^n}
\right)
\left(
\rho^{AR}
\right)^{\otimes n}
\left(
\Pi_n^{\perp A^n}
\otimes
I^{R^n}
\right)
}{
\delta_{\Pi_n}
},
\label{eq:IID_decoupling_complementary_state}
\\
\rho_{\Pi_n^\perp}^{A_{\Pi_n^\perp}^n}
&\coloneqq
\operatorname{Tr}_{R^n}
\left[
\rho_{\Pi_n^\perp}^{A_{\Pi_n^\perp}^nR^n}
\right],
\label{eq:IID_decoupling_complementary_input_marginal}
\\
\rho_{\Pi_n^\perp}^{R^n}
&\coloneqq
\operatorname{Tr}_{A_{\Pi_n^\perp}^n}
\left[
\rho_{\Pi_n^\perp}^{A_{\Pi_n^\perp}^nR^n}
\right].
\label{eq:IID_decoupling_complementary_reference_marginal}
\end{align}
When
\(\delta_{\Pi_n}=0\),
the complementary branch is absent, and every expression of the form
\begin{align}
\delta_{\Pi_n}
f
\left(
\rho_{\Pi_n^\perp}
\right)
\label{eq:IID_decoupling_zero_complementary_convention}
\end{align}
appearing below is defined to be zero.

Let
\(\mathcal H^{E_n}\)
and
\(\mathcal H^{M_n}\)
be finite-dimensional, and suppose that there exists a unitary
isomorphism
\begin{align}
V_n^{A_{\Pi_n}^n\to E_nM_n}
:
\mathcal H^{A_{\Pi_n}^n}
\longrightarrow
\mathcal H^{E_n}
\otimes
\mathcal H^{M_n}.
\label{eq:IID_decoupling_partial_trace_factorization}
\end{align}
The corresponding partial-trace channel is defined by
\begin{align}
\mathcal T_n^{A_{\Pi_n}^n\to E_n}
\left(
X
\right)
\coloneqq
\operatorname{Tr}_{M_n}
\left[
V_n
X
V_n^\dagger
\right].
\label{eq:IID_decoupling_partial_trace_channel}
\end{align}
The unitary isomorphism in
\eqref{eq:IID_decoupling_partial_trace_factorization}
implies
\begin{align}
\left|
A_{\Pi_n}^n
\right|
=
\left|
E_n
\right|
\left|
M_n
\right|.
\label{eq:IID_decoupling_dimension_identity}
\end{align}

Let
\begin{align}
\mathcal H^{A_{\Pi_n}^{n\prime}}
\simeq
\mathcal H^{A_{\Pi_n}^n}
\end{align}
be a copy of the projected input space.
The normalized Choi state of
\(\mathcal T_n^{A_{\Pi_n}^n\to E_n}\)
and its
\(E_n\)-marginal are defined by
\begin{align}
\omega_n^{A_{\Pi_n}^{n\prime}E_n}
&\coloneqq
\left(
\mathcal T_n^{A_{\Pi_n}^n\to E_n}
\otimes
\id^{A_{\Pi_n}^{n\prime}}
\right)
\left(
\Phi^{A_{\Pi_n}^nA_{\Pi_n}^{n\prime}}
\right),
\label{eq:IID_decoupling_partial_trace_Choi_state}
\\
\omega_n^{E_n}
&\coloneqq
\operatorname{Tr}_{A_{\Pi_n}^{n\prime}}
\left[
\omega_n^{A_{\Pi_n}^{n\prime}E_n}
\right].
\label{eq:IID_decoupling_partial_trace_Choi_marginal}
\end{align}
Under the identification
\begin{align}
\mathcal H^{A_{\Pi_n}^{n\prime}}
\simeq
\mathcal H^{E_n'}
\otimes
\mathcal H^{M_n'},
\end{align}
where
\(\mathcal H^{E_n'}\simeq\mathcal H^{E_n}\)
and
\(\mathcal H^{M_n'}\simeq\mathcal H^{M_n}\),
the normalized maximally entangled state factorizes, up to the
canonical ordering of tensor factors, as
\begin{align}
\Phi^{A_{\Pi_n}^nA_{\Pi_n}^{n\prime}}
=
\Phi^{E_nE_n'}
\otimes
\Phi^{M_nM_n'}.
\end{align}
Consequently, up to the canonical ordering of tensor factors,
\begin{align}
\omega_n^{A_{\Pi_n}^{n\prime}E_n}
&=
\pi^{M_n'}
\otimes
\Phi^{E_n'E_n},
\\
\omega_n^{E_n}
&=
\pi^{E_n}.
\label{eq:IID_decoupling_partial_trace_output_state}
\end{align}

Let
\begin{align}
U_n
\in
\operatorname{U}
\left(
\mathcal H^{A_{\Pi_n}^n}
\right).
\end{align}
Conditioned on successful projection, the corresponding output state
is
\begin{align}
&\tau_{n,U_n}^{E_nR^n}
\nonumber\\
&\coloneqq
\left(
\mathcal T_n^{A_{\Pi_n}^n\to E_n}
\otimes
\id^{R^n}
\right)
\left[
\left(
U_n
\otimes
I^{R^n}
\right)
\rho_{\Pi_n}^{A_{\Pi_n}^nR^n}
\left(
U_n^\dagger
\otimes
I^{R^n}
\right)
\right].
\label{eq:IID_decoupling_output_state}
\end{align}
Since
\(\mathcal T_n^{A_{\Pi_n}^n\to E_n}\)
is trace-preserving,
\(\tau_{n,U_n}^{E_nR^n}\)
is normalized, and
\begin{align}
\tau_{n,U_n}^{R^n}
=
\rho_{\Pi_n}^{R^n}.
\label{eq:IID_decoupling_output_reference}
\end{align}

The operational goal of the decoupling task is to decouple the output
system
\(E_n\)
from the reference while preserving the original reference marginal
\((\rho^R)^{\otimes n}\).
Accordingly, the decoupling error is quantified by
\begin{align}
D
\left(
\tau_{n,U_n}^{E_nR^n}
\middle\|
\pi^{E_n}
\otimes
\left(
\rho^R
\right)^{\otimes n}
\right).
\label{eq:IID_decoupling_original_reference_error}
\end{align}

A sequence of protocols specified by
\begin{align}
\left\{
\left(
\Pi_n^{A^n},
V_n^{A_{\Pi_n}^n\to E_nM_n},
U_n
\right)
\right\}_{n=1}^{\infty}
\end{align}
achieves IID decoupling if
\begin{align}
\lim_{n\to\infty}
\delta_{\Pi_n}
&=
0,
\label{eq:IID_decoupling_definition_projection_limit}
\\
\lim_{n\to\infty}
D
\left(
\tau_{n,U_n}^{E_nR^n}
\middle\|
\pi^{E_n}
\otimes
\left(
\rho^R
\right)^{\otimes n}
\right)
&=
0.
\label{eq:IID_decoupling_definition_original_error_limit}
\end{align}
The first condition ensures that the projection succeeds with
asymptotically unit probability, whereas the second condition is the
operational decoupling criterion.

For the mathematical analysis of the random-unitary decoupling step,
we also consider the auxiliary projected-reference quantity
\begin{align}
D
\left(
\tau_{n,U_n}^{E_nR^n}
\middle\|
\pi^{E_n}
\otimes
\rho_{\Pi_n}^{R^n}
\right).
\label{eq:IID_decoupling_projected_reference_error}
\end{align}
This quantity measures decoupling from the reference marginal of the
normalized projected state and is used only as an intermediate
quantity in the analysis.

Indeed, using
\eqref{eq:IID_decoupling_output_reference},
the relative-entropy chain-rule identity gives
\begin{align}
&D
\left(
\tau_{n,U_n}^{E_nR^n}
\middle\|
\pi^{E_n}
\otimes
\left(
\rho^R
\right)^{\otimes n}
\right)
\nonumber\\
&=
D
\left(
\tau_{n,U_n}^{E_nR^n}
\middle\|
\pi^{E_n}
\otimes
\rho_{\Pi_n}^{R^n}
\right)
+
D
\left(
\rho_{\Pi_n}^{R^n}
\middle\|
\left(
\rho^R
\right)^{\otimes n}
\right).
\label{eq:IID_decoupling_original_projected_error_decomposition}
\end{align}

Moreover, expanding
\begin{align}
I^{A^n}
=
\Pi_n^{A^n}
+
\Pi_n^{\perp A^n}
\end{align}
on both sides of
\((\rho^{AR})^{\otimes n}\),
the two cross terms vanish under the partial trace over
\(A^n\).
Indeed, cyclicity of the partial trace with respect to bounded
operators acting on
\(\mathcal H^{A^n}\)
gives, for example,
\begin{align}
&\operatorname{Tr}_{A^n}
\left[
\left(
\Pi_n^{A^n}
\otimes
I^{R^n}
\right)
\left(
\rho^{AR}
\right)^{\otimes n}
\left(
\Pi_n^{\perp A^n}
\otimes
I^{R^n}
\right)
\right]
\nonumber\\
&=
\operatorname{Tr}_{A^n}
\left[
\left(
\Pi_n^{\perp A^n}
\Pi_n^{A^n}
\otimes
I^{R^n}
\right)
\left(
\rho^{AR}
\right)^{\otimes n}
\right]
\nonumber\\
&=
0.
\end{align}
Consequently,
\begin{align}
\left(
\rho^R
\right)^{\otimes n}
&=
\operatorname{Tr}_{A^n}
\left[
\left(
\Pi_n^{A^n}
\otimes
I^{R^n}
\right)
\left(
\rho^{AR}
\right)^{\otimes n}
\left(
\Pi_n^{A^n}
\otimes
I^{R^n}
\right)
\right]
\nonumber\\
&\quad+
\operatorname{Tr}_{A^n}
\left[
\left(
\Pi_n^{\perp A^n}
\otimes
I^{R^n}
\right)
\left(
\rho^{AR}
\right)^{\otimes n}
\left(
\Pi_n^{\perp A^n}
\otimes
I^{R^n}
\right)
\right]
\nonumber\\
&=
\left(
1-\delta_{\Pi_n}
\right)
\rho_{\Pi_n}^{R^n}
+
\delta_{\Pi_n}
\rho_{\Pi_n^\perp}^{R^n},
\label{eq:IID_decoupling_reference_branch_decomposition}
\end{align}
where the final term is understood to be zero when
\(\delta_{\Pi_n}=0\)
according to
\eqref{eq:IID_decoupling_zero_complementary_convention}.
In particular,
\begin{align}
\left(
\rho^R
\right)^{\otimes n}
\ge
\left(
1-\delta_{\Pi_n}
\right)
\rho_{\Pi_n}^{R^n}.
\label{eq:IID_decoupling_projected_reference_domination}
\end{align}
By
\eqref{eq:IID_decoupling_projected_reference_domination},
operator monotonicity of the logarithm, and the scaling relation
\eqref{eq:homogeneous_relative_entropy_scaling},
\begin{align}
D
\left(
\rho_{\Pi_n}^{R^n}
\middle\|
\left(
\rho^R
\right)^{\otimes n}
\right)
&\le
D
\left(
\rho_{\Pi_n}^{R^n}
\middle\|
\left(
1-\delta_{\Pi_n}
\right)
\rho_{\Pi_n}^{R^n}
\right)
\nonumber\\
&=
-\log
\left(
1-\delta_{\Pi_n}
\right).
\label{eq:IID_decoupling_projected_reference_relative_entropy_bound}
\end{align}
Combining
\eqref{eq:IID_decoupling_original_projected_error_decomposition}
and
\eqref{eq:IID_decoupling_projected_reference_relative_entropy_bound}
gives
\begin{align}
&D
\left(
\tau_{n,U_n}^{E_nR^n}
\middle\|
\pi^{E_n}
\otimes
\left(
\rho^R
\right)^{\otimes n}
\right)
\nonumber\\
&\le
D
\left(
\tau_{n,U_n}^{E_nR^n}
\middle\|
\pi^{E_n}
\otimes
\rho_{\Pi_n}^{R^n}
\right)
-
\log
\left(
1-\delta_{\Pi_n}
\right).
\label{eq:IID_decoupling_original_error_from_projected_error}
\end{align}
Hence, if
\begin{align}
\lim_{n\to\infty}
\delta_{\Pi_n}
&=
0,
\\
\lim_{n\to\infty}
D
\left(
\tau_{n,U_n}^{E_nR^n}
\middle\|
\pi^{E_n}
\otimes
\rho_{\Pi_n}^{R^n}
\right)
&=
0,
\end{align}
then the operational decoupling conditions
\eqref{eq:IID_decoupling_definition_projection_limit}
and
\eqref{eq:IID_decoupling_definition_original_error_limit}
hold.

The asymptotic rates analyzed below are the projected-input dimension
rate
\begin{align}
\limsup_{n\to\infty}
\frac{1}{n}
\log
\left|
A_{\Pi_n}^n
\right|
\label{eq:IID_decoupling_projected_input_dimension_rate}
\end{align}
and the discarded-system dimension rate
\begin{align}
\limsup_{n\to\infty}
\frac{1}{n}
\log
\left|
M_n
\right|.
\label{eq:IID_decoupling_discarded_dimension_rate}
\end{align}

\subsection{Finite-rank spectral cutoffs}
\label{subsec:finite_rank_spectral_cutoff}

For each block length \(m\), we construct an auxiliary finite-rank cutoff projection \(\widehat{\Pi}_m\) on the \(m\)-copy input space from an eigenbasis of \((\rho^A)^{\otimes m}\). The logarithms of the cutoff ranks have the finite asymptotic rate \(H(A)_\rho<\infty\), while the projected states preserve the entropy, conditional entropy, and mutual information to first order. In Subsection~\ref{subsec:high_probability_finite_rank_IID_decoupling}, for each fixed block length \(m\), we use the blockwise cutoff projection \(\widehat{\Pi}_m\) to construct a high-probability finite-rank projection \(\Pi_n\) on the entire \(n\)-copy input space.

For every \(m\in\{1,2,\ldots\}\), let
\(\mathcal I_m=\{0,\ldots,|A|^m-1\}\) if
\(\dim\mathcal H^{A}<\infty\), and let
\(\mathcal I_m=\{0,1,\ldots\}\) otherwise.
Write
\begin{align}
(\rho^A)^{\otimes m}
=
\sum_{i\in\mathcal I_m}
\lambda_i^{A^m}
\ket{\psi_i}\bra{\psi_i}^{A^m},
\qquad
\lambda_0^{A^m}
\ge
\lambda_1^{A^m}
\ge
\cdots
\ge
0,
\label{eq:decoupling_IID_m_copy_spectral_decomposition}
\end{align}
where
\(\{\ket{\psi_i}^{A^m}\}_{i\in\mathcal I_m}\)
is an orthonormal basis of
\(\mathcal H^{A^m}\).
If
\((\rho^A)^{\otimes m}\)
has finite rank, the positive-eigenvalue eigenvectors are completed
by an orthonormal basis of its kernel.
When the eigenvalue sequence is used, we may formally extend it by zeros.

Let
\begin{align}
\{N_m\}_{m=1}^{\infty}
\label{eq:decoupling_IID_cutoff_rank_sequence}
\end{align}
be a sequence of positive integers.
If
\(\mathcal H^{A^m}\)
is finite-dimensional, assume additionally that
\begin{align}
N_m
\le
|A|^m.
\label{eq:decoupling_IID_cutoff_rank_range}
\end{align}
Define the rank-\(N_m\) eigenbasis cutoff projection by
\begin{align}
\widehat{\Pi}_m^{A^m}
\coloneqq
\sum_{i=0}^{N_m-1}
\ket{\psi_i}\bra{\psi_i}^{A^m}.
\label{eq:decoupling_IID_spectral_cutoff_projection}
\end{align}

All projection-conditioned spaces, errors, states, marginals, and
complementary-branch quantities introduced below are the
specializations of
\eqref{eq:IID_decoupling_projection_complement}--%
\eqref{eq:IID_decoupling_zero_complementary_convention}
to the block projection
\(\widehat{\Pi}_m^{A^m}\), with \(n\) replaced by \(m\).
The separate equation labels below are retained for convenient
reference to the blockwise spectral-cutoff construction.

Define the corresponding cutoff error by
\begin{align}
\delta_{\widehat{\Pi}_m}
\coloneqq
\operatorname{Tr}
\left[
\left(
I^{A^m}
-
\widehat{\Pi}_m^{A^m}
\right)
(\rho^A)^{\otimes m}
\right].
\label{eq:decoupling_IID_spectral_cutoff_error}
\end{align}
Since
\(N_m\ge1\)
and
\(\lambda_0^{A^m}>0\),
one has
\begin{align}
\delta_{\widehat{\Pi}_m}
<
1.
\label{eq:decoupling_IID_spectral_cutoff_error_strict_bound}
\end{align}

Define
\begin{align}
\mathcal H^{A_{\widehat{\Pi}_m}^m}
&\coloneqq
\widehat{\Pi}_m^{A^m}
\mathcal H^{A^m},
\label{eq:decoupling_IID_projected_input_space}
\\
\left|
A_{\widehat{\Pi}_m}^m
\right|
&\coloneqq
\operatorname{rank}
\widehat{\Pi}_m^{A^m}
=
N_m.
\label{eq:decoupling_IID_projected_input_dimension}
\end{align}
Define the normalized projected block state and its input and
reference marginals by
\begin{align}
\rho_{\widehat{\Pi}_m}^{A_{\widehat{\Pi}_m}^mR^m}
&\coloneqq
\frac{
\left(
\widehat{\Pi}_m^{A^m}
\otimes
I^{R^m}
\right)
(\rho^{AR})^{\otimes m}
\left(
\widehat{\Pi}_m^{A^m}
\otimes
I^{R^m}
\right)
}{
1-\delta_{\widehat{\Pi}_m}
},
\label{eq:decoupling_IID_normalized_projected_block_state}
\\
\rho_{\widehat{\Pi}_m}^{A_{\widehat{\Pi}_m}^m}
&\coloneqq
\operatorname{Tr}_{R^m}
\left[
\rho_{\widehat{\Pi}_m}^{A_{\widehat{\Pi}_m}^mR^m}
\right],
\label{eq:decoupling_IID_projected_block_input_marginal}
\\
\rho_{\widehat{\Pi}_m}^{R^m}
&\coloneqq
\operatorname{Tr}_{A_{\widehat{\Pi}_m}^m}
\left[
\rho_{\widehat{\Pi}_m}^{A_{\widehat{\Pi}_m}^mR^m}
\right].
\label{eq:decoupling_IID_projected_block_reference_marginal}
\end{align}

Define the complementary projection and its range by
\begin{align}
\widehat{\Pi}_m^{\perp A^m}
&\coloneqq
I^{A^m}
-
\widehat{\Pi}_m^{A^m},
\label{eq:decoupling_IID_complementary_projection}
\\
\mathcal H^{A_{\widehat{\Pi}_m^\perp}^m}
&\coloneqq
\widehat{\Pi}_m^{\perp A^m}
\mathcal H^{A^m}.
\label{eq:decoupling_IID_complementary_input_space}
\end{align}
Whenever
\(\delta_{\widehat{\Pi}_m}>0\),
define the normalized complementary block state and its input marginal
by
\begin{align}
\rho_{\widehat{\Pi}_m^\perp}^{A_{\widehat{\Pi}_m^\perp}^mR^m}
&\coloneqq
\frac{
\left(
\widehat{\Pi}_m^{\perp A^m}
\otimes
I^{R^m}
\right)
(\rho^{AR})^{\otimes m}
\left(
\widehat{\Pi}_m^{\perp A^m}
\otimes
I^{R^m}
\right)
}{
\delta_{\widehat{\Pi}_m}
},
\label{eq:decoupling_IID_normalized_complementary_block_state}
\\
\rho_{\widehat{\Pi}_m^\perp}^{A_{\widehat{\Pi}_m^\perp}^m}
&\coloneqq
\operatorname{Tr}_{R^m}
\left[
\rho_{\widehat{\Pi}_m^\perp}^{A_{\widehat{\Pi}_m^\perp}^mR^m}
\right].
\label{eq:decoupling_IID_complementary_block_marginal}
\end{align}
When
\(\delta_{\widehat{\Pi}_m}=0\),
the complementary branch is absent, and every weighted expression of
the form
\begin{align}
\delta_{\widehat{\Pi}_m}
f
\left(
\rho_{\widehat{\Pi}_m^\perp}
\right)
\label{eq:decoupling_IID_zero_complementary_convention}
\end{align}
appearing below is defined to be zero, consistently with
\eqref{eq:IID_decoupling_zero_complementary_convention}.

The sequence of cutoff ranks
\(\{N_m\}_{m=1}^{\infty}\)
is specified by the following proposition.
It shows that, for every prescribed partial-trace rate, the cutoff
ranks can be chosen so that the corresponding projected states have
the required first-order entropic properties and the projected input
spaces admit exact finite-dimensional factorizations.

\begin{proposition}[Finite-entropy spectral cutoffs]
\label{prop:decoupling_IID_finite_entropy_spectral_cutoffs}
Let
\(\mathcal H^A\)
and
\(\mathcal H^R\)
be separable and possibly infinite-dimensional, and let
\begin{align}
\rho^{AR}
\in
\operatorname{D}
\left(
\mathcal H^A
\otimes
\mathcal H^R
\right)
\end{align}
satisfy
\begin{align}
H(A)_\rho
<
\infty.
\end{align}
For every
\(m\in\{1,2,\ldots\}\),
let
\((\rho^A)^{\otimes m}\)
have the spectral decomposition in
\eqref{eq:decoupling_IID_m_copy_spectral_decomposition}.
Then, for every
\begin{align}
q
\in
\left[
0,
H(A)_\rho
\right],
\end{align}
there exist a sequence of positive integers
\(\{N_m\}_{m=1}^{\infty}\)
satisfying
\eqref{eq:decoupling_IID_cutoff_rank_range}
whenever
\(\mathcal H^{A^m}\)
is finite-dimensional, sequences of finite-dimensional systems
\(\{E_m\}_{m=1}^{\infty}\)
and
\(\{M_m\}_{m=1}^{\infty}\),
and unitary isomorphisms
\begin{align}
V_m^{A_{\widehat{\Pi}_m}^m\to E_mM_m}
:
\mathcal H^{A_{\widehat{\Pi}_m}^m}
\longrightarrow
\mathcal H^{E_m}
\otimes
\mathcal H^{M_m}
\label{eq:IID_decoupling_prescribed_common_factorization}
\end{align}
such that the projections
\(\widehat{\Pi}_m^{A^m}\),
the cutoff errors
\(\delta_{\widehat{\Pi}_m}\),
the projected spaces
\(\mathcal H^{A_{\widehat{\Pi}_m}^m}\),
the projected dimensions
\(\lvert A_{\widehat{\Pi}_m}^m\rvert\),
the projected states
\(\rho_{\widehat{\Pi}_m}^{A_{\widehat{\Pi}_m}^mR^m}\),
the projected input marginals
\(\rho_{\widehat{\Pi}_m}^{A_{\widehat{\Pi}_m}^m}\),
the projected reference marginals
\(\rho_{\widehat{\Pi}_m}^{R^m}\),
the complementary projections
\(\widehat{\Pi}_m^{\perp A^m}\),
and the complementary spaces
\(\mathcal H^{A_{\widehat{\Pi}_m^\perp}^m}\)
defined in
\eqref{eq:decoupling_IID_spectral_cutoff_projection},
\eqref{eq:decoupling_IID_spectral_cutoff_error},
\eqref{eq:decoupling_IID_projected_input_space},
\eqref{eq:decoupling_IID_projected_input_dimension},
\eqref{eq:decoupling_IID_normalized_projected_block_state},
\eqref{eq:decoupling_IID_projected_block_input_marginal},
\eqref{eq:decoupling_IID_projected_block_reference_marginal},
\eqref{eq:decoupling_IID_complementary_projection}, and
\eqref{eq:decoupling_IID_complementary_input_space},
respectively, and, whenever
\(\delta_{\widehat{\Pi}_m}>0\),
the complementary states
\(\rho_{\widehat{\Pi}_m^\perp}^{A_{\widehat{\Pi}_m^\perp}^mR^m}\)
and their input marginals
\(\rho_{\widehat{\Pi}_m^\perp}^{A_{\widehat{\Pi}_m^\perp}^m}\)
defined in
\eqref{eq:decoupling_IID_normalized_complementary_block_state}
and
\eqref{eq:decoupling_IID_complementary_block_marginal},
respectively, satisfy
\begin{align}
\lim_{m\to\infty}
\delta_{\widehat{\Pi}_m}
&=
0,
\label{eq:decoupling_IID_spectral_cutoff_error_limit}
\\
\lim_{m\to\infty}
\frac{1}{m}
H
\left(
A_{\widehat{\Pi}_m}^m
\right)_{\rho_{\widehat{\Pi}_m}}
&=
H(A)_\rho,
\label{eq:decoupling_IID_projected_entropy_limit}
\\
\lim_{m\to\infty}
\frac{
\delta_{\widehat{\Pi}_m}
}{
m
}
H
\left(
A_{\widehat{\Pi}_m^\perp}^m
\right)_{\rho_{\widehat{\Pi}_m^\perp}}
&=
0,
\label{eq:decoupling_IID_weighted_complementary_entropy_limit}
\\
\lim_{m\to\infty}
\frac{1}{m}
H
\left(
A_{\widehat{\Pi}_m}^m
\middle|
R^m
\right)_{\rho_{\widehat{\Pi}_m}}
&=
H(A|R)_\rho,
\label{eq:decoupling_IID_projected_conditional_entropy_limit}
\\
\lim_{m\to\infty}
\frac{1}{m}
I
\left(
A_{\widehat{\Pi}_m}^m:R^m
\right)_{\rho_{\widehat{\Pi}_m}}
&=
I(A:R)_\rho,
\label{eq:decoupling_IID_projected_mutual_information_limit}
\\
\lim_{m\to\infty}
\frac{1}{m}
\log
\left|
A_{\widehat{\Pi}_m}^m
\right|
&=
H(A)_\rho,
\label{eq:decoupling_IID_spectral_cutoff_rank_rate}
\\
\lim_{m\to\infty}
\frac{1}{m}
\log
\left|
M_m
\right|
&=
q.
\label{eq:IID_decoupling_prescribed_common_discarded_rate}
\end{align}
\end{proposition}

\begin{proof}
Fix
\begin{align}
q
\in
\left[
0,
H(A)_\rho
\right].
\end{align}

\medskip
\noindent
\textbf{The finite-dimensional maximally mixed case.}
We first treat the case in which
\(\mathcal H^A\)
is finite-dimensional and
\begin{align}
H(A)_\rho
=
\log
\left|
A
\right|.
\label{eq:decoupling_IID_maximal_entropy_case}
\end{align}
Condition
\eqref{eq:decoupling_IID_maximal_entropy_case}
implies
\begin{align}
\rho^A
=
\pi^A.
\end{align}
For every \(m\), set
\begin{align}
\widehat{\Pi}_m^{A^m}
&\coloneqq
I^{A^m},
&
N_m
&\coloneqq
\left|
A
\right|^m.
\end{align}
Then
\begin{align}
\delta_{\widehat{\Pi}_m}
&=
0,
&
\rho_{\widehat{\Pi}_m}^{A_{\widehat{\Pi}_m}^mR^m}
&=
\left(
\rho^{AR}
\right)^{\otimes m}.
\end{align}
Consequently, additivity gives
\begin{align}
H
\left(
A_{\widehat{\Pi}_m}^m
\right)_{\rho_{\widehat{\Pi}_m}}
&=
mH(A)_\rho,
\\
H
\left(
A_{\widehat{\Pi}_m}^m
\middle|
R^m
\right)_{\rho_{\widehat{\Pi}_m}}
&=
mH(A|R)_\rho,
\\
I
\left(
A_{\widehat{\Pi}_m}^m:R^m
\right)_{\rho_{\widehat{\Pi}_m}}
&=
mI(A:R)_\rho.
\end{align}
The weighted complementary entropy is defined to be zero, and
\begin{align}
\frac{1}{m}
\log
\left|
A_{\widehat{\Pi}_m}^m
\right|
=
\log
\left|
A
\right|
=
H(A)_\rho.
\end{align}

If
\begin{align}
\left|
A
\right|
=
1,
\end{align}
then \(q=0\).
For every \(m\), take both
\(E_m\)
and
\(M_m\)
to be one-dimensional and choose the canonical unitary isomorphism
\begin{align}
V_m^{A_{\widehat{\Pi}_m}^m\to E_mM_m}
:
\mathcal H^{A_{\widehat{\Pi}_m}^m}
\longrightarrow
\mathcal H^{E_m}
\otimes
\mathcal H^{M_m}.
\end{align}
Then
\begin{align}
\frac{1}{m}
\log
\left|
M_m
\right|
=
0
=
q.
\end{align}

Suppose now that
\begin{align}
\left|
A
\right|
\ge
2.
\end{align}
For every \(m\), define
\begin{align}
a_m
\coloneqq
\left\lfloor
\frac{
mq
}{
\log
\left|
A
\right|
}
\right\rfloor.
\end{align}
Since
\begin{align}
0
\le
q
\le
\log
\left|
A
\right|,
\end{align}
one has
\begin{align}
a_m
\in
\{0,\ldots,m\},
\end{align}
and
\begin{align}
\lim_{m\to\infty}
\frac{a_m}{m}
=
\frac{
q
}{
\log
\left|
A
\right|
}.
\end{align}
Choose finite-dimensional systems \(M_m\) and \(E_m\) satisfying
\begin{align}
\left|
M_m
\right|
&\coloneqq
\left|
A
\right|^{a_m},
&
\left|
E_m
\right|
&\coloneqq
\left|
A
\right|^{m-a_m}.
\end{align}
Since
\begin{align}
\left|
A_{\widehat{\Pi}_m}^m
\right|
=
\left|
A
\right|^m
=
\left|
E_m
\right|
\left|
M_m
\right|,
\end{align}
there exists a unitary isomorphism
\begin{align}
V_m^{A_{\widehat{\Pi}_m}^m\to E_mM_m}
:
\mathcal H^{A_{\widehat{\Pi}_m}^m}
\longrightarrow
\mathcal H^{E_m}
\otimes
\mathcal H^{M_m}.
\end{align}
Moreover,
\begin{align}
\lim_{m\to\infty}
\frac{1}{m}
\log
\left|
M_m
\right|
&=
\lim_{m\to\infty}
\frac{a_m}{m}
\log
\left|
A
\right|
\nonumber\\
&=
q.
\end{align}
Thus, all the claimed limits hold in the finite-dimensional maximally
mixed case.

\medskip
\noindent
\textbf{Choice of the cutoff ranks.}
Henceforth, suppose that either
\(\mathcal H^A\)
is infinite-dimensional or
\(\mathcal H^A\)
is finite-dimensional and
\begin{align}
H(A)_\rho
<
\log
\left|
A
\right|.
\label{eq:decoupling_IID_nonmaximal_entropy_case}
\end{align}

Let
\begin{align}
\rho^A
=
\sum_{j=0}^{\infty}
\lambda_j^A
\ket{\psi_j}\bra{\psi_j}^A
\end{align}
be the spectral decomposition obtained from
\eqref{eq:decoupling_IID_m_copy_spectral_decomposition}
with \(m=1\).
Let \(J\) be a random variable taking values in
\(\{0,1,\ldots\}\)
according to
\begin{align}
\Pr
\left\{
J=j
\right\}
=
\lambda_j^A,
\end{align}
and define
\begin{align}
\imath_A(J)
\coloneqq
-\log
\lambda_J^A.
\end{align}
The indices \(j\) for which
\(\lambda_j^A=0\)
occur with probability zero.
By
\eqref{eq:IID_decoupling_finite_entropy_assumption}
and
\eqref{eq:entropy_eigenvalue_formula},
\begin{align}
\mathbb E
\left[
\imath_A(J)
\right]
&=
-\sum_{j=0}^{\infty}
\lambda_j^A
\log
\lambda_j^A
\nonumber\\
&=
H(A)_\rho
<
\infty.
\label{eq:decoupling_IID_information_density_expectation}
\end{align}
Thus,
\(\imath_A(J)\)
is integrable.

Let
\(J_1,J_2,\ldots\)
be IID copies of \(J\), write
\begin{align}
J^m
\coloneqq
\left(
J_1,\ldots,J_m
\right),
\end{align}
and define
\begin{align}
\overline{\imath}_{A,m}(J^m)
\coloneqq
\frac{1}{m}
\sum_{\ell=1}^{m}
\imath_A(J_\ell).
\label{eq:decoupling_IID_average_information_density}
\end{align}
By the law of large numbers in \(L^1\),
\(\overline{\imath}_{A,m}(J^m)\)
converges in mean to
\(H(A)_\rho\)~\cite[Example~6.2.3]{durrett2019probability}.
Thus, defining
\begin{align}
\Delta_m
\coloneqq
\mathbb E
\left[
\left|
\overline{\imath}_{A,m}(J^m)
-
H(A)_\rho
\right|
\right],
\label{eq:decoupling_IID_information_density_mean_deviation}
\end{align}
one has
\begin{align}
\lim_{m\to\infty}
\Delta_m
=
0.
\label{eq:decoupling_IID_information_density_mean_convergence}
\end{align}

Choose positive parameters
\(\{\zeta_m\}_{m=1}^{\infty}\)
such that
\begin{align}
\lim_{m\to\infty}
\zeta_m
&=
0,
&
\lim_{m\to\infty}
\frac{\Delta_m}{\zeta_m}
&=
0.
\label{eq:decoupling_IID_information_density_window_properties}
\end{align}
For example, one may take
\begin{align}
\zeta_m
\coloneqq
\sqrt{\Delta_m}
+
\frac{1}{m},
\label{eq:decoupling_IID_information_density_window_choice}
\end{align}
which satisfies
\begin{align}
0
\le
\frac{\Delta_m}{\zeta_m}
=
\frac{
\Delta_m
}{
\sqrt{\Delta_m}+1/m
}
\le
\sqrt{\Delta_m}.
\label{eq:decoupling_IID_information_density_window_ratio_bound}
\end{align}

Define
\begin{align}
L_m
&\coloneqq
\left\lceil
\exp
\left[
m
\left(
H(A)_\rho+\zeta_m
\right)
\right]
\right\rceil,
\label{eq:decoupling_IID_base_cutoff_rank}
\\
K_m
&\coloneqq
\left\lceil
\exp
\left[
mq
\right]
\right\rceil,
\label{eq:decoupling_IID_prescribed_discarded_dimension}
\\
\widehat N_m
&\coloneqq
K_m
\left\lceil
\frac{L_m}{K_m}
\right\rceil.
\label{eq:decoupling_IID_adjusted_cutoff_rank_candidate}
\end{align}
Since
\(q\le H(A)_\rho\)
and
\(\zeta_m>0\),
one has
\begin{align}
K_m
\le
L_m.
\label{eq:decoupling_IID_prescribed_dimension_below_base_rank}
\end{align}
Consequently,
\begin{align}
L_m
\le
\widehat N_m
<
L_m+K_m
\le
2L_m.
\label{eq:decoupling_IID_adjusted_cutoff_rank_bounds}
\end{align}

By
\eqref{eq:decoupling_IID_base_cutoff_rank},
\begin{align}
L_m
&=
\left\lceil
\exp
\left[
m
\left(
H(A)_\rho+\zeta_m
\right)
\right]
\right\rceil
\nonumber\\
&\le
2
\exp
\left[
m
\left(
H(A)_\rho+\zeta_m
\right)
\right],
\label{eq:decoupling_IID_base_cutoff_rank_upper_bound}
\end{align}
where we used
\(\exp[m(H(A)_\rho+\zeta_m)]\ge1\).
Combining
\eqref{eq:decoupling_IID_adjusted_cutoff_rank_bounds}
and
\eqref{eq:decoupling_IID_base_cutoff_rank_upper_bound}
gives
\begin{align}
\widehat N_m
<
4
\exp
\left[
m
\left(
H(A)_\rho+\zeta_m
\right)
\right].
\label{eq:decoupling_IID_adjusted_cutoff_rank_exponential_upper_bound}
\end{align}

If
\(\mathcal H^A\)
is infinite-dimensional, set
\begin{align}
m_*
\coloneqq
1.
\end{align}
If
\(\mathcal H^A\)
is finite-dimensional, then
\eqref{eq:decoupling_IID_nonmaximal_entropy_case},
\eqref{eq:decoupling_IID_adjusted_cutoff_rank_exponential_upper_bound},
and
\eqref{eq:decoupling_IID_information_density_window_properties}
give
\begin{align}
\limsup_{m\to\infty}
\frac{1}{m}
\log
\widehat N_m
\le
H(A)_\rho
<
\log
\left|
A
\right|.
\label{eq:decoupling_IID_adjusted_cutoff_rank_rate_strict_bound}
\end{align}
Therefore, there exists
\(m_*\in\{1,2,\ldots\}\)
such that
\begin{align}
\widehat N_m
\le
\left|
A
\right|^m
=
\left|A^m\right|
\label{eq:decoupling_IID_adjusted_cutoff_rank_finite_dimensional_range}
\end{align}
for every
\(m\ge m_*\).

Define the final cutoff ranks by
\begin{align}
N_m
\coloneqq
\begin{cases}
1,
&
m<m_*,
\\[1mm]
\widehat N_m,
&
m\ge m_*,
\end{cases}
\label{eq:decoupling_IID_cutoff_rank_choice}
\end{align}
where the first case can occur only when
\(\mathcal H^A\)
is finite-dimensional.
Equations
\eqref{eq:decoupling_IID_adjusted_cutoff_rank_finite_dimensional_range}
and
\eqref{eq:decoupling_IID_cutoff_rank_choice}
show that this definition satisfies
\eqref{eq:decoupling_IID_cutoff_rank_range}.

For every \(m\ge m_*\),
\eqref{eq:decoupling_IID_cutoff_rank_choice} and
\eqref{eq:decoupling_IID_adjusted_cutoff_rank_exponential_upper_bound}
give
\begin{align}
N_m
<
4
\exp
\left[
m
\left(
H(A)_\rho+\zeta_m
\right)
\right].
\label{eq:decoupling_IID_cutoff_rank_exponential_upper_bound}
\end{align}

Choose finite-dimensional systems
\(E_m\)
and
\(M_m\)
such that
\begin{align}
\left(
\left|
E_m
\right|,
\left|
M_m
\right|
\right)
\coloneqq
\begin{cases}
(1,1),
&
m<m_*,
\\[1mm]
\left(
N_m/K_m,
K_m
\right),
&
m\ge m_*.
\end{cases}
\label{eq:decoupling_IID_factor_dimensions}
\end{align}
By
\eqref{eq:decoupling_IID_adjusted_cutoff_rank_candidate}
and
\eqref{eq:decoupling_IID_cutoff_rank_choice},
\(K_m\) divides \(N_m\) whenever \(m\ge m_*\).
Consequently, by
\eqref{eq:decoupling_IID_projected_input_dimension},
\begin{align}
\left|
A_{\widehat{\Pi}_m}^m
\right|
=
N_m
=
\left|
E_m
\right|
\left|
M_m
\right|
\label{eq:decoupling_IID_cutoff_rank_factorization_identity}
\end{align}
for every \(m\).
Therefore, there exists a unitary isomorphism
\begin{align}
V_m^{A_{\widehat{\Pi}_m}^m\to E_mM_m}
:
\mathcal H^{A_{\widehat{\Pi}_m}^m}
\longrightarrow
\mathcal H^{E_m}
\otimes
\mathcal H^{M_m}.
\label{eq:decoupling_IID_cutoff_rank_unitary_factorization}
\end{align}
The finitely many terms with
\(m<m_*\)
do not affect any of the limits below.

\medskip
\noindent
\textbf{Proof of
\eqref{eq:decoupling_IID_spectral_cutoff_error_limit}.}
For a multi-index
\(j^m=(j_1,\ldots,j_m)\), write
\begin{align}
\lambda_{j^m}^{A^m}
\coloneqq
\prod_{\ell=1}^{m}
\lambda_{j_\ell}^A.
\end{align}
These product eigenvalues, counted with multiplicity and rearranged in
non-increasing order, are precisely the eigenvalues in
\eqref{eq:decoupling_IID_m_copy_spectral_decomposition}.
Define
\begin{align}
\mathcal T_m
\coloneqq
\left\{
j^m:
\lambda_{j^m}^{A^m}>0,\
-\frac{1}{m}
\log
\lambda_{j^m}^{A^m}
\le
H(A)_\rho+\zeta_m
\right\}.
\end{align}
Every
\(j^m\in\mathcal T_m\)
satisfies
\begin{align}
\lambda_{j^m}^{A^m}
\ge
\exp
\left[
-m
\left(
H(A)_\rho+\zeta_m
\right)
\right].
\end{align}
Since the sum of all product eigenvalues is one,
\begin{align}
\left|
\mathcal T_m
\right|
\le
\exp
\left[
m
\left(
H(A)_\rho+\zeta_m
\right)
\right]
\le
L_m.
\end{align}
For every \(m\ge m_*\),
\begin{align}
L_m
\le
N_m.
\end{align}
Thus, the top-\(N_m\) spectral projection retains every product
eigenvalue indexed by
\(\mathcal T_m\).
Therefore, for every \(m\ge m_*\),
\begin{align}
\delta_{\widehat{\Pi}_m}
&\le
\Pr
\left\{
\overline{\imath}_{A,m}(J^m)
>
H(A)_\rho+\zeta_m
\right\}
\nonumber\\
&\le
\Pr
\left\{
\left|
\overline{\imath}_{A,m}(J^m)
-
H(A)_\rho
\right|
>
\zeta_m
\right\}
\nonumber\\
&\le
\frac{
\mathbb E
\left[
\left|
\overline{\imath}_{A,m}(J^m)
-
H(A)_\rho
\right|
\right]
}{
\zeta_m
}
\nonumber\\
&=
\frac{\Delta_m}{\zeta_m},
\label{eq:decoupling_IID_cutoff_error_information_density_bound}
\end{align}
where the third inequality is Markov's inequality.
Equations
\eqref{eq:decoupling_IID_information_density_window_properties}
and
\eqref{eq:decoupling_IID_cutoff_error_information_density_bound}
prove
\eqref{eq:decoupling_IID_spectral_cutoff_error_limit}.

\medskip
\noindent
\textbf{Proof of
\eqref{eq:decoupling_IID_projected_entropy_limit}.}
Define the per-copy entropy contribution of the omitted spectral tail
by
\begin{align}
t_m
\coloneqq
\frac{1}{m}
\sum_{i=N_m}^{\infty}
\left(
-\lambda_i^{A^m}
\log
\lambda_i^{A^m}
\right).
\label{eq:decoupling_IID_weighted_spectral_tail}
\end{align}
Since every product eigenvalue indexed by
\(\mathcal T_m\)
is retained for every \(m\ge m_*\),
\begin{align}
t_m
&\le
\mathbb E
\left[
\overline{\imath}_{A,m}(J^m)
\mathbf 1
\left\{
\overline{\imath}_{A,m}(J^m)
>
H(A)_\rho+\zeta_m
\right\}
\right]
\nonumber\\
&=
H(A)_\rho
\Pr
\left\{
\overline{\imath}_{A,m}(J^m)
>
H(A)_\rho+\zeta_m
\right\}
\nonumber\\
&\quad+
\mathbb E
\left[
\left(
\overline{\imath}_{A,m}(J^m)
-
H(A)_\rho
\right)
\mathbf 1
\left\{
\overline{\imath}_{A,m}(J^m)
>
H(A)_\rho+\zeta_m
\right\}
\right]
\nonumber\\
&\le
H(A)_\rho
\Pr
\left\{
\left|
\overline{\imath}_{A,m}(J^m)
-
H(A)_\rho
\right|
>
\zeta_m
\right\}
+
\Delta_m
\nonumber\\
&\le
H(A)_\rho
\frac{\Delta_m}{\zeta_m}
+
\Delta_m.
\label{eq:decoupling_IID_weighted_spectral_tail_bound}
\end{align}
Hence,
\begin{align}
\lim_{m\to\infty}
t_m
=
0.
\label{eq:decoupling_IID_weighted_spectral_tail_limit}
\end{align}

The eigenvalues of
\(\rho_{\widehat{\Pi}_m}^{A_{\widehat{\Pi}_m}^m}\)
are
\begin{align}
\frac{
\lambda_i^{A^m}
}{
1-\delta_{\widehat{\Pi}_m}
},
\qquad
i\in\{0,\ldots,N_m-1\}.
\end{align}
Therefore,
\begin{align}
H
\left(
A_{\widehat{\Pi}_m}^m
\right)_{\rho_{\widehat{\Pi}_m}}
&=
-\sum_{i=0}^{N_m-1}
\frac{
\lambda_i^{A^m}
}{
1-\delta_{\widehat{\Pi}_m}
}
\log
\frac{
\lambda_i^{A^m}
}{
1-\delta_{\widehat{\Pi}_m}
}
\nonumber\\
&=
\frac{
mH(A)_\rho-mt_m
}{
1-\delta_{\widehat{\Pi}_m}
}
+
\log
\left(
1-\delta_{\widehat{\Pi}_m}
\right).
\label{eq:decoupling_IID_projected_marginal_entropy_formula}
\end{align}
Dividing
\eqref{eq:decoupling_IID_projected_marginal_entropy_formula}
by \(m\), and using
\eqref{eq:decoupling_IID_spectral_cutoff_error_limit}
and
\eqref{eq:decoupling_IID_weighted_spectral_tail_limit},
proves
\eqref{eq:decoupling_IID_projected_entropy_limit}.

\medskip
\noindent
\textbf{Proof of
\eqref{eq:decoupling_IID_weighted_complementary_entropy_limit}.}
Since
\(\widehat{\Pi}_m^{A^m}\)
commutes with
\((\rho^A)^{\otimes m}\),
the block-diagonal entropy decomposition gives
\begin{align}
mH(A)_\rho
&=
h_2
\left(
\delta_{\widehat{\Pi}_m}
\right)
+
\left(
1-\delta_{\widehat{\Pi}_m}
\right)
H
\left(
A_{\widehat{\Pi}_m}^m
\right)_{\rho_{\widehat{\Pi}_m}}
\nonumber\\
&\quad+
\delta_{\widehat{\Pi}_m}
H
\left(
A_{\widehat{\Pi}_m^\perp}^m
\right)_{\rho_{\widehat{\Pi}_m^\perp}},
\label{eq:decoupling_IID_marginal_entropy_branch_decomposition}
\end{align}
where the final term is defined to be zero when
\(\delta_{\widehat{\Pi}_m}=0\).
In particular, when
\(\delta_{\widehat{\Pi}_m}>0\),
the entropy
\begin{align}
H
\left(
A_{\widehat{\Pi}_m^\perp}^m
\right)_{\rho_{\widehat{\Pi}_m^\perp}}
\end{align}
is finite.
Consequently,
\begin{align}
&\frac{
\delta_{\widehat{\Pi}_m}
}{
m
}
H
\left(
A_{\widehat{\Pi}_m^\perp}^m
\right)_{\rho_{\widehat{\Pi}_m^\perp}}
\nonumber\\
&=
H(A)_\rho
-
\frac{1}{m}
h_2
\left(
\delta_{\widehat{\Pi}_m}
\right)
-
\left(
1-\delta_{\widehat{\Pi}_m}
\right)
\frac{1}{m}
H
\left(
A_{\widehat{\Pi}_m}^m
\right)_{\rho_{\widehat{\Pi}_m}}.
\end{align}
Using
\eqref{eq:decoupling_IID_spectral_cutoff_error_limit},
\eqref{eq:decoupling_IID_projected_entropy_limit}, and
\begin{align}
0
\le
h_2
\left(
\delta_{\widehat{\Pi}_m}
\right)
\le
\log2,
\end{align}
we obtain
\eqref{eq:decoupling_IID_weighted_complementary_entropy_limit}.

\medskip
\noindent
\textbf{Proof of
\eqref{eq:decoupling_IID_projected_conditional_entropy_limit}.}
Since
\begin{align}
H(A^m)_{\rho^{\otimes m}}
=
mH(A)_\rho
<
\infty,
\end{align}
one has
\begin{align}
I(A^m:R^m)_{\rho^{\otimes m}}
\le
2H(A^m)_{\rho^{\otimes m}}
=
2mH(A)_\rho
<
\infty.
\label{eq:decoupling_IID_tensor_power_mutual_information_finite}
\end{align}

Define the pinched state
\begin{align}
\overline\rho_m^{A^mR^m}
&\coloneqq
\left(
\widehat{\Pi}_m^{A^m}
\otimes
I^{R^m}
\right)
\left(
\rho^{AR}
\right)^{\otimes m}
\left(
\widehat{\Pi}_m^{A^m}
\otimes
I^{R^m}
\right)
\nonumber\\
&\quad+
\left(
\widehat{\Pi}_m^{\perp A^m}
\otimes
I^{R^m}
\right)
\left(
\rho^{AR}
\right)^{\otimes m}
\left(
\widehat{\Pi}_m^{\perp A^m}
\otimes
I^{R^m}
\right).
\end{align}
Set
\begin{align}
U_m^{A^m,\mathrm{pin}}
\coloneqq
\widehat{\Pi}_m^{A^m}
-
\widehat{\Pi}_m^{\perp A^m}.
\end{align}
Then
\(U_m^{A^m,\mathrm{pin}}\)
is unitary and
\begin{align}
\overline\rho_m^{A^mR^m}
&=
\frac{1}{2}
\left(
\rho^{AR}
\right)^{\otimes m}
\nonumber\\
&\quad+
\frac{1}{2}
\left(
U_m^{A^m,\mathrm{pin}}
\otimes
I^{R^m}
\right)
\left(
\rho^{AR}
\right)^{\otimes m}
\left(
U_m^{A^m,\mathrm{pin}}
\otimes
I^{R^m}
\right).
\label{eq:decoupling_IID_pinching_random_unitary_representation}
\end{align}
The \(A^m\)-marginal is unchanged by the pinching, and hence
\begin{align}
H(A^m)_{\overline\rho_m}
=
mH(A)_\rho
<
\infty.
\end{align}
In particular,
\begin{align}
I(A^m:R^m)_{\overline\rho_m}
<
\infty.
\end{align}

Introduce a classical register \(X\), and define
\begin{align}
\overline\rho_m^{XA^mR^m}
&\coloneqq
\frac{1}{2}
\ket{0}\bra{0}^X
\otimes
\left(
\rho^{AR}
\right)^{\otimes m}
\nonumber\\
&\quad+
\frac{1}{2}
\ket{1}\bra{1}^X
\otimes
\left(
U_m^{A^m,\mathrm{pin}}
\otimes
I^{R^m}
\right)
\left(
\rho^{AR}
\right)^{\otimes m}
\left(
U_m^{A^m,\mathrm{pin}}
\otimes
I^{R^m}
\right).
\end{align}
The chain rule for mutual information in
\eqref{eq:mutual_information_chain_rule}
gives
\begin{align}
I(XA^m:R^m)_{\overline\rho_m}
=
I(X:R^m)_{\overline\rho_m}
+
I(A^m:R^m|X)_{\overline\rho_m}.
\label{eq:decoupling_IID_pinching_mutual_information_chain_rule_X_first}
\end{align}
Since the two conditional branches have the same \(R^m\)-marginal,
one has
\begin{align}
I(X:R^m)_{\overline\rho_m}
=
0.
\end{align}
Moreover, conditional mutual information for a classical register is
the average of the mutual informations of the conditional states.
Hence,
\begin{align}
I(A^m:R^m|X)_{\overline\rho_m}
&=
\frac{1}{2}
I(A^m:R^m)_{\rho^{\otimes m}}
\nonumber\\
&\quad+
\frac{1}{2}
I(A^m:R^m)_{
\left(
U_m^{A^m,\mathrm{pin}}
\otimes
I^{R^m}
\right)
\left(
\rho^{AR}
\right)^{\otimes m}
\left(
U_m^{A^m,\mathrm{pin}}
\otimes
I^{R^m}
\right)
}
\nonumber\\
&=
I(A^m:R^m)_{\rho^{\otimes m}},
\label{eq:decoupling_IID_pinching_conditional_mutual_information}
\end{align}
where the second equality follows from invariance of mutual
information under local unitaries.
Substituting these identities into
\eqref{eq:decoupling_IID_pinching_mutual_information_chain_rule_X_first}
gives
\begin{align}
I(XA^m:R^m)_{\overline\rho_m}
=
I(A^m:R^m)_{\rho^{\otimes m}}.
\label{eq:decoupling_IID_pinching_mutual_information_first_expansion}
\end{align}

Applying the chain rule
\eqref{eq:mutual_information_chain_rule}
in the opposite order gives
\begin{align}
I(XA^m:R^m)_{\overline\rho_m}
=
I(A^m:R^m)_{\overline\rho_m}
+
I(X:R^m|A^m)_{\overline\rho_m}.
\label{eq:decoupling_IID_pinching_mutual_information_chain_rule_A_first}
\end{align}
Combining
\eqref{eq:decoupling_IID_pinching_mutual_information_first_expansion}
and
\eqref{eq:decoupling_IID_pinching_mutual_information_chain_rule_A_first}
yields
\begin{align}
I(A^m:R^m)_{\rho^{\otimes m}}
-
I(A^m:R^m)_{\overline\rho_m}
=
I(X:R^m|A^m)_{\overline\rho_m}.
\label{eq:decoupling_IID_pinching_mutual_information_difference}
\end{align}

Since \(X\) is classical,
\begin{align}
0
&\le
I(X:R^m|A^m)_{\overline\rho_m}
\nonumber\\
&\le
H(X|A^m)_{\overline\rho_m}
\nonumber\\
&\le
H(X)_{\overline\rho_m}
=
\log2.
\label{eq:decoupling_IID_pinching_conditional_mutual_information_bound}
\end{align}
Since the \(A^m\)-marginal is unchanged by the pinching,
\eqref{eq:I_AR_definition_finite_A_general} gives
\begin{align}
&H(A^m|R^m)_{\overline\rho_m}
-
H(A^m|R^m)_{\rho^{\otimes m}}
\nonumber\\
&=
I(A^m:R^m)_{\rho^{\otimes m}}
-
I(A^m:R^m)_{\overline\rho_m}.
\label{eq:decoupling_IID_pinching_conditional_entropy_difference}
\end{align}
Moreover, additivity of the conditional entropy under tensor products
gives
\begin{align}
H(A^m|R^m)_{\rho^{\otimes m}}
=
mH(A|R)_\rho.
\label{eq:decoupling_IID_conditional_entropy_tensor_power_additivity}
\end{align}
Combining
\eqref{eq:decoupling_IID_pinching_mutual_information_difference},
\eqref{eq:decoupling_IID_pinching_conditional_mutual_information_bound},
\eqref{eq:decoupling_IID_pinching_conditional_entropy_difference}, and
\eqref{eq:decoupling_IID_conditional_entropy_tensor_power_additivity},
we obtain
\begin{align}
0
\le
H(A^m|R^m)_{\overline\rho_m}
-
mH(A|R)_\rho
\le
\log2.
\label{eq:decoupling_IID_pinching_conditional_entropy_bound}
\end{align}

Introduce a binary classical register \(Y\).
When
\(\delta_{\widehat{\Pi}_m}>0\), define
\begin{align}
\widetilde\rho_m^{YA^mR^m}
&\coloneqq
\left(
1-\delta_{\widehat{\Pi}_m}
\right)
\ket{0}\bra{0}^Y
\otimes
\rho_{\widehat{\Pi}_m}^{A_{\widehat{\Pi}_m}^mR^m}
\nonumber\\
&\quad+
\delta_{\widehat{\Pi}_m}
\ket{1}\bra{1}^Y
\otimes
\rho_{\widehat{\Pi}_m^\perp}^{A_{\widehat{\Pi}_m^\perp}^mR^m},
\end{align}
where the conditional states are canonically embedded into the
orthogonal subspaces
\(\mathcal H^{A_{\widehat{\Pi}_m}^m}\)
and
\(\mathcal H^{A_{\widehat{\Pi}_m^\perp}^m}\)
of
\(\mathcal H^{A^m}\).
When
\(\delta_{\widehat{\Pi}_m}=0\), define
\begin{align}
\widetilde\rho_m^{YA^mR^m}
\coloneqq
\ket{0}\bra{0}^Y
\otimes
\left(
\rho^{AR}
\right)^{\otimes m}.
\end{align}
In either case, its \(A^mR^m\)-marginal is
\(\overline\rho_m^{A^mR^m}\).
Since \(Y\) is determined by the orthogonal support sectors of the
\(A^m\)-system,
\begin{align}
H(Y|A^mR^m)_{\widetilde\rho_m}
=
0.
\end{align}
Consequently,
\begin{align}
H(A^m|R^m)_{\overline\rho_m}
&=
H(Y|R^m)_{\widetilde\rho_m}
+
\left(
1-\delta_{\widehat{\Pi}_m}
\right)
H
\left(
A_{\widehat{\Pi}_m}^m
\middle|
R^m
\right)_{\rho_{\widehat{\Pi}_m}}
\nonumber\\
&\quad+
\delta_{\widehat{\Pi}_m}
H
\left(
A_{\widehat{\Pi}_m^\perp}^m
\middle|
R^m
\right)_{\rho_{\widehat{\Pi}_m^\perp}}.
\label{eq:decoupling_IID_conditional_entropy_branch_decomposition}
\end{align}
Because \(Y\) is classical,
\begin{align}
0
\le
H(Y|R^m)_{\widetilde\rho_m}
\le
H(Y)_{\widetilde\rho_m}
=
h_2
\left(
\delta_{\widehat{\Pi}_m}
\right)
\le
\log2.
\end{align}
Combining this bound with
\eqref{eq:decoupling_IID_pinching_conditional_entropy_bound}
and
\eqref{eq:decoupling_IID_conditional_entropy_branch_decomposition}
gives
\begin{align}
\Biggl|
&
\left(
1-\delta_{\widehat{\Pi}_m}
\right)
H
\left(
A_{\widehat{\Pi}_m}^m
\middle|
R^m
\right)_{\rho_{\widehat{\Pi}_m}}
\nonumber\\
&\quad+
\delta_{\widehat{\Pi}_m}
H
\left(
A_{\widehat{\Pi}_m^\perp}^m
\middle|
R^m
\right)_{\rho_{\widehat{\Pi}_m^\perp}}
-
mH(A|R)_\rho
\Biggr|
\le
\log2.
\label{eq:decoupling_IID_weighted_conditional_entropy_bound}
\end{align}

Whenever
\(\delta_{\widehat{\Pi}_m}>0\),
the marginal entropy
\begin{align}
H
\left(
A_{\widehat{\Pi}_m^\perp}^m
\right)_{\rho_{\widehat{\Pi}_m^\perp}}
\end{align}
is finite by
\eqref{eq:decoupling_IID_marginal_entropy_branch_decomposition}.
Therefore,
\eqref{eq:conditional_entropy_bounds_finite_marginal}
gives
\begin{align}
\left|
H
\left(
A_{\widehat{\Pi}_m^\perp}^m
\middle|
R^m
\right)_{\rho_{\widehat{\Pi}_m^\perp}}
\right|
\le
H
\left(
A_{\widehat{\Pi}_m^\perp}^m
\right)_{\rho_{\widehat{\Pi}_m^\perp}}.
\end{align}
Consequently,
\begin{align}
&\left|
\frac{
\delta_{\widehat{\Pi}_m}
}{
m
}
H
\left(
A_{\widehat{\Pi}_m^\perp}^m
\middle|
R^m
\right)_{\rho_{\widehat{\Pi}_m^\perp}}
\right|
\nonumber\\
&\le
\frac{
\delta_{\widehat{\Pi}_m}
}{
m
}
H
\left(
A_{\widehat{\Pi}_m^\perp}^m
\right)_{\rho_{\widehat{\Pi}_m^\perp}}.
\end{align}
Equation
\eqref{eq:decoupling_IID_weighted_complementary_entropy_limit}
therefore implies
\begin{align}
\lim_{m\to\infty}
\frac{
\delta_{\widehat{\Pi}_m}
}{
m
}
H
\left(
A_{\widehat{\Pi}_m^\perp}^m
\middle|
R^m
\right)_{\rho_{\widehat{\Pi}_m^\perp}}
=
0.
\label{eq:decoupling_IID_weighted_complementary_conditional_entropy_limit}
\end{align}
Dividing
\eqref{eq:decoupling_IID_weighted_conditional_entropy_bound}
by \(m\), and using
\eqref{eq:decoupling_IID_spectral_cutoff_error_limit}
and
\eqref{eq:decoupling_IID_weighted_complementary_conditional_entropy_limit},
proves
\eqref{eq:decoupling_IID_projected_conditional_entropy_limit}.

\medskip
\noindent
\textbf{Proof of
\eqref{eq:decoupling_IID_projected_mutual_information_limit}.}
By
\eqref{eq:I_AR_definition_finite_A_general},
\begin{align}
I
\left(
A_{\widehat{\Pi}_m}^m:R^m
\right)_{\rho_{\widehat{\Pi}_m}}
&=
H
\left(
A_{\widehat{\Pi}_m}^m
\right)_{\rho_{\widehat{\Pi}_m}}
\nonumber\\
&\quad-
H
\left(
A_{\widehat{\Pi}_m}^m
\middle|
R^m
\right)_{\rho_{\widehat{\Pi}_m}}.
\end{align}
Equations
\eqref{eq:decoupling_IID_projected_entropy_limit}
and
\eqref{eq:decoupling_IID_projected_conditional_entropy_limit}
therefore prove
\eqref{eq:decoupling_IID_projected_mutual_information_limit}.

\medskip
\noindent
\textbf{Proof of
\eqref{eq:IID_decoupling_prescribed_common_discarded_rate}.}
For every \(m\ge m_*\),
\eqref{eq:decoupling_IID_prescribed_discarded_dimension}
and
\eqref{eq:decoupling_IID_factor_dimensions}
give
\begin{align}
\left|
M_m
\right|
=
\left\lceil
\exp
\left[
mq
\right]
\right\rceil.
\end{align}
Since
\begin{align}
\exp
\left[
mq
\right]
\le
\left\lceil
\exp
\left[
mq
\right]
\right\rceil
\le
2
\exp
\left[
mq
\right],
\end{align}
we obtain
\begin{align}
q
\le
\frac{1}{m}
\log
\left|
M_m
\right|
\le
q
+
\frac{\log2}{m}.
\end{align}
This proves
\eqref{eq:IID_decoupling_prescribed_common_discarded_rate}.

\medskip
\noindent
\textbf{Proof of
\eqref{eq:decoupling_IID_spectral_cutoff_rank_rate}.}
By
\eqref{eq:decoupling_IID_projected_input_dimension}
and
\eqref{eq:decoupling_IID_cutoff_rank_exponential_upper_bound},
for every \(m\ge m_*\),
\begin{align}
\left|
A_{\widehat{\Pi}_m}^m
\right|
=
N_m
<
4
\exp
\left[
m
\left(
H(A)_\rho+\zeta_m
\right)
\right].
\label{eq:decoupling_IID_projected_dimension_exponential_upper_bound}
\end{align}
Taking logarithms in
\eqref{eq:decoupling_IID_projected_dimension_exponential_upper_bound}
and dividing by \(m\) gives
\begin{align}
\frac{1}{m}
\log
\left|
A_{\widehat{\Pi}_m}^m
\right|
<
H(A)_\rho
+
\zeta_m
+
\frac{\log4}{m}.
\label{eq:decoupling_IID_spectral_cutoff_rank_rate_upper_estimate}
\end{align}
Using
\eqref{eq:decoupling_IID_information_density_window_properties}
in
\eqref{eq:decoupling_IID_spectral_cutoff_rank_rate_upper_estimate},
we obtain
\begin{align}
\limsup_{m\to\infty}
\frac{1}{m}
\log
\left|
A_{\widehat{\Pi}_m}^m
\right|
\le
H(A)_\rho.
\label{eq:decoupling_IID_spectral_cutoff_rank_rate_upper_bound}
\end{align}

On the other hand,
\(\rho_{\widehat{\Pi}_m}^{A_{\widehat{\Pi}_m}^m}\)
is supported on the finite-dimensional space
\(\mathcal H^{A_{\widehat{\Pi}_m}^m}\).
The entropy--dimension bound therefore gives
\begin{align}
H
\left(
A_{\widehat{\Pi}_m}^m
\right)_{\rho_{\widehat{\Pi}_m}}
\le
\log
\left|
A_{\widehat{\Pi}_m}^m
\right|.
\label{eq:decoupling_IID_projected_entropy_dimension_bound}
\end{align}
Dividing
\eqref{eq:decoupling_IID_projected_entropy_dimension_bound}
by \(m\) gives
\begin{align}
\frac{1}{m}
H
\left(
A_{\widehat{\Pi}_m}^m
\right)_{\rho_{\widehat{\Pi}_m}}
\le
\frac{1}{m}
\log
\left|
A_{\widehat{\Pi}_m}^m
\right|.
\label{eq:decoupling_IID_projected_entropy_dimension_rate_bound}
\end{align}
Taking the limit inferior in
\eqref{eq:decoupling_IID_projected_entropy_dimension_rate_bound}
and using
\eqref{eq:decoupling_IID_projected_entropy_limit},
we obtain
\begin{align}
\liminf_{m\to\infty}
\frac{1}{m}
\log
\left|
A_{\widehat{\Pi}_m}^m
\right|
\ge
H(A)_\rho.
\label{eq:decoupling_IID_spectral_cutoff_rank_rate_lower_bound}
\end{align}

Combining
\eqref{eq:decoupling_IID_spectral_cutoff_rank_rate_upper_bound}
and
\eqref{eq:decoupling_IID_spectral_cutoff_rank_rate_lower_bound}
proves
\eqref{eq:decoupling_IID_spectral_cutoff_rank_rate}.
\end{proof}

\subsection{High-probability finite-rank projections for IID states}
\label{subsec:high_probability_finite_rank_IID_decoupling}

Using the blockwise spectral cutoffs constructed in
Subsection~\ref{subsec:finite_rank_spectral_cutoff},
we now construct a single finite-rank projection on the entire
\(n\)-copy input space.

Fix
\begin{align}
q
\in
\left[
0,
H(A)_\rho
\right],
\end{align}
and choose sequences
\begin{align}
\left\{
\widehat{\Pi}_m^{A^m},
E_m,
M_m,
V_m^{A_{\widehat{\Pi}_m}^m\to E_mM_m}
\right\}_{m=1}^{\infty}
\end{align}
satisfying
Proposition~\ref{prop:decoupling_IID_finite_entropy_spectral_cutoffs}
for this value of \(q\).
Fix
\begin{align}
m
\in
\{1,2,\ldots\}.
\end{align}
The block length \(m\) remains fixed throughout the following
construction, while the total number \(n\) of copies tends to
infinity.

A measurement of
\begin{align}
\left\{
\widehat{\Pi}_m^{A^m},
\widehat{\Pi}_m^{\perp A^m}
\right\}
\end{align}
on
\((\rho^A)^{\otimes m}\)
produces the common outcome with probability
\(1-\delta_{\widehat{\Pi}_m}\)
and the rare outcome with probability
\(\delta_{\widehat{\Pi}_m}\),
where the corresponding projections and spaces are defined in
\eqref{eq:decoupling_IID_spectral_cutoff_projection},
\eqref{eq:decoupling_IID_projected_input_space},
\eqref{eq:decoupling_IID_complementary_projection}, and
\eqref{eq:decoupling_IID_complementary_input_space}.
We refer to
\(\mathcal H^{A_{\widehat{\Pi}_m}^m}\)
as the ``common'' block subspace and to
\(\mathcal H^{A_{\widehat{\Pi}_m^\perp}^m}\)
as the ``rare'' block subspace.
The $m$-copy space $\mathcal{H}^{A^m}$ is decomposed into
\begin{align}
\label{eq:decoupling_IID_Pi_perp_decomposition}
    \mathcal{H}^{A^m}=\mathcal H^{A_{\widehat{\Pi}_m}^m}\oplus\mathcal H^{A_{\widehat{\Pi}_m^\perp}^m}.
\end{align}

A direct repetition of the common-block projection does not yield a
high-probability projection on a growing number of blocks.
Indeed, the tensor-product projection
\begin{align}
\left(
\widehat{\Pi}_m^{A^m}
\right)^{\otimes b}
\end{align}
accepts
\((\rho^A)^{\otimes mb}\)
with probability
\begin{align}
\operatorname{Tr}
\left[
\left(
\widehat{\Pi}_m^{A^m}
\right)^{\otimes b}
(\rho^A)^{\otimes mb}
\right]
=
\left(
1-\delta_{\widehat{\Pi}_m}
\right)^b,
\label{eq:IID_decoupling_naive_repeated_cutoff_probability}
\end{align}
which converges to zero as
\(b\to\infty\)
whenever
\(\delta_{\widehat{\Pi}_m}>0\).
Thus, even when
\(\delta_{\widehat{\Pi}_m}\)
is small, requiring every block to lie in the common subspace is
asymptotically too restrictive.

The construction below instead retains all common--rare patterns
whose empirical rare-block frequency is close to
\(\delta_{\widehat{\Pi}_m}\).
Within each such pattern, we group a fixed number of common blocks into a subsystem to be decoupled by Haar-random unitary in the subsequent decoupling protocol, while keeping the remaining common blocks and the rare blocks.
The rare blocks, which may be infinite-dimensional, are projected into a finite-rank typical subspace.
The resulting projection on \(n\) copies has asymptotically unit
success probability.
Moreover, with the above grouping, 
the range of the projection admits a pattern-independent subsystem consisting of the fixed number of identical copies of \(\mathcal H^{A_{\widehat{\Pi}_m}^m}\).
This provides the finite-dimensional tensor-product structure needed for the subsequent decoupling protocol, while the asymptotic choice of the blockwise cutoffs is governed by
Proposition~\ref{prop:decoupling_IID_finite_entropy_spectral_cutoffs}.

For the fixed block length \(m\) and every \(n\ge m\), let
\(b_{n,m}\in\{1,2,\ldots\}\)
and
\(c_{n,m}\in\{0,\ldots,m-1\}\)
be uniquely determined by
\begin{align}
n
=
mb_{n,m}
+
c_{n,m},
\label{eq:decoupling_IID_number_of_blocks}
\end{align}
where
\begin{align}
b_{n,m}
&=
\left\lfloor
\frac{n}{m}
\right\rfloor,
&
\lim_{n\to\infty}
\frac{
b_{n,m}
}{
n
}
&=
\frac{1}{m},
&
0
\le
c_{n,m}
<
m.
\label{eq:decoupling_IID_block_quotient_properties}
\end{align}
Here,
\(b_{n,m}\)
is the number of complete \(m\)-copy blocks, while
\(c_{n,m}\)
is the number of residual copies not contained in a complete block.
The entire $n$-copy input space then decomposes into
\begin{align}
    \left(\mathcal H^A\right)^{\otimes n}=\left(\mathcal H^{A^m}\right)^{\otimes b_{n,m}}\otimes\mathcal H^{A^{c_{n,m}}},
\end{align}
the first factor consists of \(b\) complete \(m\)-copy blocks, while the second is the residual subsystem.

For a binary pattern
\begin{align}
x^{b_{n,m}}
=
\left(
x_1,\ldots,x_{b_{n,m}}
\right)
\in
\{0,1\}^{b_{n,m}},
\end{align}
the value \(x_j=0\) indicates that the \(j\)-th complete block lies in
the common subspace
\(\mathcal H^{A_{\widehat{\Pi}_m}^m}\),
whereas \(x_j=1\) indicates that it lies in the rare subspace
\(\mathcal H^{A_{\widehat{\Pi}_m^\perp}^m}\).
Due to~\eqref{eq:decoupling_IID_Pi_perp_decomposition}, the space of the complete blocks decomposes into
\begin{align}
    \left(\mathcal H^{A^m}\right)^{\otimes b_{n,m}}
    &=\bigoplus_{x^{b_{n,m}}\in \{0,1\}^{b_{n,m}}}\mathcal{H}_{x_1}\otimes\cdots\otimes\mathcal{H}_{x_{b_{n,m}}}, \\
    \mathcal H_x
&\coloneqq
\begin{cases}
\mathcal H^{A_{\widehat{\Pi}_m}^m},
&
x=0,
\\
\mathcal H^{A_{\widehat{\Pi}_m^\perp}^m},
&
x=1.
\end{cases}
\end{align}

By
\eqref{eq:decoupling_IID_spectral_cutoff_error_strict_bound},
one always has
\begin{align}
\delta_{\widehat{\Pi}_m}
<
1.
\end{align}
We now treat separately the case in which the rare branch has positive
probability and the case in which it is absent.

\medskip
\noindent
\textbf{Typical common--rare patterns (\(0<\delta_{\widehat{\Pi}_m}<1\)).}
Assume first that
\begin{align}
0
<
\delta_{\widehat{\Pi}_m}
<
1.
\label{eq:IID_decoupling_global_cutoff_positive_rare_case}
\end{align}
In this case, the rare branch occurs with positive probability.

Define the number of rare blocks by
\begin{align}
k
\left(
x^{b_{n,m}}
\right)
\coloneqq
\sum_{j=1}^{b_{n,m}}
x_j.
\label{eq:IID_decoupling_global_cutoff_rare_block_number}
\end{align}
Thus,
\(k(x^{b_{n,m}})\)
is determined by the pattern and is not an independently chosen
parameter.
Fix a pattern-typicality width
\begin{align}
0
<
\eta
<
\min
\left\{
\delta_{\widehat{\Pi}_m},
1-\delta_{\widehat{\Pi}_m}
\right\}.
\label{eq:IID_decoupling_global_cutoff_pattern_tolerance}
\end{align}
Define the weakly typical set for the common--rare patterns by
\begin{align}
\mathcal T_{n,m,\eta}
\coloneqq
\left\{
x^{b_{n,m}}
\in
\{0,1\}^{b_{n,m}}:
\left|
\frac{
k(x^{b_{n,m}})
}{
b_{n,m}
}
-
\delta_{\widehat{\Pi}_m}
\right|
\le
\eta
\right\}.
\label{eq:IID_decoupling_global_cutoff_pattern_typical_set}
\end{align}
Because the complete blocks are in a product state, the common--rare
outcomes are IID Bernoulli random variables with rare-outcome
probability
\(\delta_{\widehat{\Pi}_m}\).
The law of large numbers therefore gives
\begin{align}
\lim_{n\to\infty}
\sum_{
x^{b_{n,m}}
\in
\mathcal T_{n,m,\eta}
}
\left(
1-\delta_{\widehat{\Pi}_m}
\right)^{
b_{n,m}-k(x^{b_{n,m}})
}
\delta_{\widehat{\Pi}_m}^{k(x^{b_{n,m}})}
=
1.
\label{eq:IID_decoupling_global_cutoff_pattern_probability_limit}
\end{align}
The method of types~\cite[Eq.~(12.1), Theorem~12.1.3]{Cover2006} gives
\begin{align}
\left|
\mathcal T_{n,m,\eta}
\right|
\le
\left(
b_{n,m}+1
\right)
\exp
\Biggl[
b_{n,m}
\max
\left\{
h_2(t):
t\in[0,1],\
\left|
t-\delta_{\widehat{\Pi}_m}
\right|
\le
\eta
\right\}
\Biggr].
\label{eq:IID_decoupling_global_cutoff_pattern_cardinality}
\end{align}

Every pattern in
\(\mathcal T_{n,m,\eta}\)
contains at least
\begin{align}
b_{n,m}
\left(
1-\delta_{\widehat{\Pi}_m}-\eta
\right)
\end{align}
common blocks.
Define
\begin{align}
\ell_{n,m,\eta}
\coloneqq
\left\lfloor
b_{n,m}
\left(
1-\delta_{\widehat{\Pi}_m}-\eta
\right)
\right\rfloor.
\label{eq:IID_decoupling_global_cutoff_common_block_number}
\end{align}
The integer
\(\ell_{n,m,\eta}\)
is the fixed number of common blocks isolated from every pattern in
\(\mathcal T_{n,m,\eta}\).
It depends on \(n\), the fixed block length \(m\), and the
pattern-typicality width \(\eta\), but not on the individual pattern.

For every
\(x^{b_{n,m}}\in\mathcal T_{n,m,\eta}\),
write the common-block indices in increasing order as
\begin{align}
\left\{
j:
x_j=0
\right\}
=
\left\{
j_1
\left(
x^{b_{n,m}}
\right)
<
\cdots
<
j_{
b_{n,m}-k(x^{b_{n,m}})
}
\left(
x^{b_{n,m}}
\right)
\right\}.
\end{align}
Define three disjoint subsets of
\(\{1,\ldots,b_{n,m}\}\)
by
\begin{align}
\mathcal I_{\mathrm{fix}}
\left(
x^{b_{n,m}}
\right)
&\coloneqq
\left\{
j_1
\left(
x^{b_{n,m}}
\right),
\ldots,
j_{\ell_{n,m,\eta}}
\left(
x^{b_{n,m}}
\right)
\right\},
\nonumber\\
\mathcal I_{\mathrm{rem}}
\left(
x^{b_{n,m}}
\right)
&\coloneqq
\left\{
j:
x_j=0
\right\}
\setminus
\mathcal I_{\mathrm{fix}}
\left(
x^{b_{n,m}}
\right),
\nonumber\\
\mathcal I_{\mathrm{rare}}
\left(
x^{b_{n,m}}
\right)
&\coloneqq
\left\{
j:
x_j=1
\right\}.
\label{eq:IID_decoupling_global_cutoff_pattern_index_sets}
\end{align}
The first set is understood to be empty when
\(\ell_{n,m,\eta}=0\).

The set
\(\mathcal I_{\mathrm{fix}}\)
labels the fixed-count common blocks.
Here, ``fixed'' refers to their number:
\begin{align}
\left|
\mathcal I_{\mathrm{fix}}
\left(
x^{b_{n,m}}
\right)
\right|
=
\ell_{n,m,\eta}
\label{eq:ell}
\end{align}
for every typical pattern, although their original positions may
depend on the pattern.
After a pattern-dependent permutation, these blocks form the
pattern-independent subsystem
\begin{align}
\left(
\mathcal H^{A_{\widehat{\Pi}_m}^m}
\right)^{\otimes\ell_{n,m,\eta}}
\end{align}
used in the subsequent random-unitary partial-trace protocol.

The set
\(\mathcal I_{\mathrm{rem}}\)
labels the remaining common blocks after the fixed-count common blocks
have been isolated.
Its cardinality is
\begin{align}
\left|
\mathcal I_{\mathrm{rem}}
\left(
x^{b_{n,m}}
\right)
\right|
=
b_{n,m}
-
k
\left(
x^{b_{n,m}}
\right)
-
\ell_{n,m,\eta}.
\end{align}
The convention of taking the first
\(\ell_{n,m,\eta}\)
common-block indices in increasing order is used only to specify this
decomposition canonically.
By
\eqref{eq:IID_decoupling_global_cutoff_pattern_typical_set},
\begin{align}
&b_{n,m}
\left(
\delta_{\widehat{\Pi}_m}-\eta
\right)
\le
k
\left(
x^{b_{n,m}}
\right)
\le
b_{n,m}
\left(
\delta_{\widehat{\Pi}_m}+\eta
\right),
\label{eq:IID_decoupling_global_cutoff_rare_count}
\end{align}
and hence,
\begin{align}
&0
\le
\left|
\mathcal I_{\mathrm{rem}}
\left(
x^{b_{n,m}}
\right)
\right|
\le
2\eta b_{n,m}+1.
\label{eq:IID_decoupling_global_cutoff_remaining_common_count}
\end{align}

The set
\(\mathcal I_{\mathrm{rare}}\)
labels the blocks lying in
\(\mathcal H^{A_{\widehat{\Pi}_m^\perp}^m}\), and
\begin{align}
\left|
\mathcal I_{\mathrm{rare}}
\left(
x^{b_{n,m}}
\right)
\right|
=
k
\left(
x^{b_{n,m}}
\right),
\end{align}
which is in the range~\eqref{eq:IID_decoupling_global_cutoff_rare_count}.
In particular, since
\(\eta<\delta_{\widehat{\Pi}_m}\),
\eqref{eq:IID_decoupling_global_cutoff_rare_count} implies
\begin{align}
\lim_{n\to\infty}
\min_{
x^{b_{n,m}}
\in
\mathcal T_{n,m,\eta}
}
k
\left(
x^{b_{n,m}}
\right)
=
\infty.
\label{eq:IID_decoupling_global_cutoff_uniform_rare_count_divergence}
\end{align}

Let
\begin{align}
U_{x^{b_{n,m}}}^{A,\mathrm{perm}}
\in
\operatorname{U}
\left(
\mathcal H^{A^{mb_{n,m}}}
\right)
\end{align}
denote the unitary that permutes the
\(b_{n,m}\)
complete \(m\)-copy blocks so that they appear in the order
\begin{align}
\mathcal I_{\mathrm{fix}}
\left(
x^{b_{n,m}}
\right),
\qquad
\mathcal I_{\mathrm{rem}}
\left(
x^{b_{n,m}}
\right),
\qquad
\mathcal I_{\mathrm{rare}}
\left(
x^{b_{n,m}}
\right),
\end{align}
with the indices within each set kept in increasing order.
Its restriction to the pattern sector
\begin{align}
\mathcal H_{x_1}
\otimes\cdots\otimes
\mathcal H_{x_{b_{n,m}}}
\end{align}
is a unitary isomorphism onto
\begin{align}
&
\left(
\mathcal H^{A_{\widehat{\Pi}_m}^m}
\right)^{\otimes\ell_{n,m,\eta}}
\otimes
\left(
\mathcal H^{A_{\widehat{\Pi}_m}^m}
\right)^{
\otimes
\left(
b_{n,m}
-
k(x^{b_{n,m}})
-
\ell_{n,m,\eta}
\right)
}
\otimes
\left(
\mathcal H^{A_{\widehat{\Pi}_m^\perp}^m}
\right)^{
\otimes k(x^{b_{n,m}})
}.
\end{align}

\medskip
\noindent
\textbf{Finite-rank cutoff for the rare blocks ($0<\delta_{\widehat{\Pi}_m}<1$).}
Whereas $\mathcal{H}^{A_{\widehat{\Pi}_m}^m}$ for the common blocks is finite-dimensional, $\mathcal{H}^{A_{\widehat{\Pi}_m^\perp}^m}$ for the rare blocks may be infinite-dimensional; thus, we further use weak typicality to introduce a finite-dimensional cutoff.
The normalized rare-branch state
\(\rho_{\widehat{\Pi}_m^\perp}^{A_{\widehat{\Pi}_m^\perp}^mR^m}\)
and its input marginal
\(\rho_{\widehat{\Pi}_m^\perp}^{A_{\widehat{\Pi}_m^\perp}^m}\)
are those defined in
\eqref{eq:decoupling_IID_normalized_complementary_block_state}
and
\eqref{eq:decoupling_IID_complementary_block_marginal},
respectively.
Equation
\eqref{eq:decoupling_IID_marginal_entropy_branch_decomposition}
gives
\begin{align}
&\delta_{\widehat{\Pi}_m}
H
\left(
A_{\widehat{\Pi}_m^\perp}^m
\right)_{\rho_{\widehat{\Pi}_m^\perp}}
\nonumber\\
&=
mH(A)_\rho
-
h_2
\left(
\delta_{\widehat{\Pi}_m}
\right)
-
\left(
1-\delta_{\widehat{\Pi}_m}
\right)
H
\left(
A_{\widehat{\Pi}_m}^m
\right)_{\rho_{\widehat{\Pi}_m}}
<
\infty.
\label{eq:IID_decoupling_global_cutoff_rare_entropy_finite}
\end{align}
Since
\(\delta_{\widehat{\Pi}_m}>0\),
it follows that
\begin{align}
H
\left(
A_{\widehat{\Pi}_m^\perp}^m
\right)_{\rho_{\widehat{\Pi}_m^\perp}}
<
\infty,
\label{eq:IID_decoupling_global_cutoff_normalized_rare_entropy_finite}
\end{align}
which allows us to apply weak typicality.

Fix a weak-typicality width
\begin{align}
\xi
>
0.
\end{align}
The nonzero eigenvalues of
\(\rho_{\widehat{\Pi}_m^\perp}^{A_{\widehat{\Pi}_m^\perp}^m}\)
are
\begin{align}
\frac{
\lambda_i^{A^m}
}{
\delta_{\widehat{\Pi}_m}
},
\qquad
i
\in \mathcal J_m^{\mathrm{rare}},
\end{align}
where \(\mathcal J_m^{\mathrm{rare}}\) is the index set
\begin{align}
\mathcal J_m^{\mathrm{rare}}\coloneqq\left\{
i\in\{N_m,N_m+1,\ldots\}:
\lambda_i^{A^m}>0
\right\},
\end{align}
and the index set terminates when
\(\mathcal H^{A^m}\)
is finite-dimensional.
For every
\(k\in\{1,2,\ldots\}\),
where \(k\) denotes the number of rare \(m\)-copy blocks, define
\begin{align}
\mathcal T_{k,m,\xi}^{\mathrm{rare}}
\coloneqq
\Biggl\{
i^k
\in
\left(
\mathcal J_m^{\mathrm{rare}}
\right)^k:
\Biggl|
-\frac{1}{k}
\log
\prod_{j=1}^{k}
\frac{
\lambda_{i_j}^{A^m}
}{
\delta_{\widehat{\Pi}_m}
}
-
H
\left(
A_{\widehat{\Pi}_m^\perp}^m
\right)_{\rho_{\widehat{\Pi}_m^\perp}}
\Biggr|
\le
\xi
\Biggr\}.
\label{eq:IID_decoupling_global_cutoff_rare_typical_set}
\end{align}
The corresponding finite-rank projection on
\begin{align}
\left(
\mathcal H^{A_{\widehat{\Pi}_m^\perp}^m}
\right)^{\otimes k}
\end{align}
is
\begin{align}
\Pi_{k,m,\xi}^{\mathrm{rare}}
\coloneqq
\sum_{
i^k
\in
\mathcal T_{k,m,\xi}^{\mathrm{rare}}
}
\bigotimes_{j=1}^{k}
\ket{\psi_{i_j}}\bra{\psi_{i_j}},
\label{eq:IID_decoupling_global_cutoff_rare_typical_projection}
\end{align}
and the corresponding subspace is denoted by
\begin{align}
    \mathcal{H}^{A_{\Pi_{k,m,\xi}^{\mathrm{rare}}}^{mk}}
    &\coloneqq
    \Pi_{k,m,\xi}^{\mathrm{rare}}\left(
\mathcal H^{A_{\widehat{\Pi}_m^\perp}^m}
\right)^{\otimes k}.
\end{align}
Whenever
\(\Pi_{k,m,\xi}^{\mathrm{rare}}\)
is regarded as an operator on
\(\left(\mathcal H^{A^m}\right)^{\otimes k}\),
we identify it with its extension by zero on the orthogonal complement
of
\(\left(\mathcal H^{A_{\widehat{\Pi}_m^\perp}^m}\right)^{\otimes k}\).
For \(k=0\), set
\begin{align}
\mathcal T_{0,m,\xi}^{\mathrm{rare}}
&\coloneqq
\{\emptyset\},
&
\Pi_{0,m,\xi}^{\mathrm{rare}}
&\coloneqq
1,
&\mathcal{H}^{A_{\Pi_{0,m,\xi}^{\mathrm{rare}}}^{0}}
&\coloneqq\mathbb{C}.
\label{eq:IID_decoupling_global_cutoff_rare_zero_copy_definition}
\end{align}

The law of large numbers and
\eqref{eq:IID_decoupling_global_cutoff_normalized_rare_entropy_finite}
give
\begin{align}
\lim_{k\to\infty}
\operatorname{Tr}
\left[
\Pi_{k,m,\xi}^{\mathrm{rare}}
\left(
\rho_{\widehat{\Pi}_m^\perp}^{A_{\widehat{\Pi}_m^\perp}^m}
\right)^{\otimes k}
\right]
=
1.
\label{eq:IID_decoupling_global_cutoff_rare_probability_limit}
\end{align}
Moreover, every sequence in
\(\mathcal T_{k,m,\xi}^{\mathrm{rare}}\)
has probability at least
\begin{align}
\exp
\left[
-k
\left(
H
\left(
A_{\widehat{\Pi}_m^\perp}^m
\right)_{\rho_{\widehat{\Pi}_m^\perp}}
+
\xi
\right)
\right].
\end{align}
Since the sum of the probabilities of all sequences is at most one,
it follows that
\begin{align}
\operatorname{rank}
\Pi_{k,m,\xi}^{\mathrm{rare}}
\le
\exp
\left[
k
\left(
H
\left(
A_{\widehat{\Pi}_m^\perp}^m
\right)_{\rho_{\widehat{\Pi}_m^\perp}}
+
\xi
\right)
\right].
\label{eq:IID_decoupling_global_cutoff_rare_rank_bound}
\end{align}

For fixed \(m\) and \(\xi\),  define the asymptotic rank exponent of the rare typical projections by
\begin{align}
\gamma_{m,\xi}^{\mathrm{rare}}
\coloneqq
\lim_{k\to\infty}
\frac{1}{k}
\log
\operatorname{rank}
\Pi_{k,m,\xi}^{\mathrm{rare}},
\label{eq:IID_decoupling_global_cutoff_rare_rank_exponent}
\end{align}
where we adopt the convention
\begin{align}
\log 0
\coloneqq
-\infty.
\label{eq:IID_decoupling_global_cutoff_log_zero_convention}
\end{align}
This limit exists.
To see its existence, if
\begin{align}
i^k
\in
\mathcal T_{k,m,\xi}^{\mathrm{rare}}
\qquad
\text{and}
\qquad
j^{k'}
\in
\mathcal T_{k',m,\xi}^{\mathrm{rare}},
\end{align}
then their concatenation belongs to
\(\mathcal T_{k+k',m,\xi}^{\mathrm{rare}}\).
Consequently,
\begin{align}
\operatorname{rank}
\Pi_{k+k',m,\xi}^{\mathrm{rare}}
\ge
\operatorname{rank}
\Pi_{k,m,\xi}^{\mathrm{rare}}
\operatorname{rank}
\Pi_{k',m,\xi}^{\mathrm{rare}}.
\label{eq:IID_decoupling_global_cutoff_rare_rank_supermultiplicativity}
\end{align}
Therefore, the extended-real-valued sequence
\begin{align}
\left\{
\log
\operatorname{rank}
\Pi_{k,m,\xi}^{\mathrm{rare}}
\right\}_{k\in\{1,2,\ldots\}}
\end{align}
is superadditive.
Moreover,
\eqref{eq:IID_decoupling_global_cutoff_rare_probability_limit}
implies that
\begin{align}
\operatorname{rank}
\Pi_{k,m,\xi}^{\mathrm{rare}}
>
0
\end{align}
for all sufficiently large \(k\).
Thus, the sequence is finite from some point onward, and Fekete's
lemma~\cite{fekete1923verteilung} proves the existence of the limit of the extended-real-valued sequence in
\eqref{eq:IID_decoupling_global_cutoff_rare_rank_exponent}.

\medskip
\noindent
\textbf{Finite-rank cutoff for the residual subsystem.}
It remains to control the separable, possibly infinite-dimensional residual subsystem, which is treated identically
in the cases
\(\delta_{\widehat{\Pi}_m}>0\)
and
\(\delta_{\widehat{\Pi}_m}=0\).
For \(c_{n,m}>0\), define
\begin{align}
r_{n,m}^{\mathrm{res}}
\coloneqq
\begin{cases}
\displaystyle
\min\left\{
\left\lceil \exp(\sqrt n)\right\rceil,
|A|^{c_{n,m}}
\right\},
& |A|<\infty,\\[1.2ex]
\displaystyle
\left\lceil \exp(\sqrt n)\right\rceil,
& |A|=\infty.
\end{cases}
\label{eq:IID_decoupling_global_cutoff_residual_rank}
\end{align}
Let
\(\Pi_{n,m}^{\mathrm{res},A^{c_{n,m}}}\)
be the top-\(r_{n,m}^{\mathrm{res}}\) spectral projection of
\((\rho^A)^{\otimes c_{n,m}}\), with eigenvalues counted with
multiplicity.
The corresponding subspace is denoted by
\begin{align}
    \mathcal{H}^{A_{\Pi_{n,m}^{\mathrm{res},A^{c_{n,m}}}}^{c_{n,m}}}\coloneqq\Pi_{n,m}^{\mathrm{res},A^{c_{n,m}}}\mathcal{H}^{A^{c_{n,m}}}.
\end{align}
Define
\begin{align}
\delta_{n,m}^{\mathrm{res}}
\coloneqq
\operatorname{Tr}
\left[
\left(
I^{A^{c_{n,m}}}
-
\Pi_{n,m}^{\mathrm{res},A^{c_{n,m}}}
\right)
(\rho^A)^{\otimes c_{n,m}}
\right].
\label{eq:IID_decoupling_global_cutoff_residual_error}
\end{align}
When
\(c_{n,m}=0\), set
\begin{align}
\mathcal{H}^{A_{\Pi_{n,m}^{\mathrm{res},A^{c_{n,m}}}}^{c_{n,m}}}
&\coloneqq
\mathbb C,
&
\Pi_{n,m}^{\mathrm{res},A^{c_{n,m}}}
&\coloneqq
1,
&
\delta_{n,m}^{\mathrm{res}}
&\coloneqq
0.
\label{eq:IID_decoupling_zero_residual_definition}
\end{align}

For each fixed
\(c\in\{1,\ldots,m-1\}\),
the retained rank tends to infinity when
\(\mathcal H^{A^c}\)
is infinite-dimensional and eventually equals
\(\left|A\right|^c\)
when
\(\mathcal H^{A^c}\)
is finite-dimensional.
Consequently, along every subsequence on which
\(c_{n,m}=c\),
monotone convergence of the corresponding spectral sums gives
\begin{align}
\delta_{n,m}^{\mathrm{res}}
\longrightarrow
0.
\end{align}
Since only finitely many residual copy numbers can occur,
\begin{align}
\lim_{n\to\infty}
\delta_{n,m}^{\mathrm{res}}
=
0.
\label{eq:IID_decoupling_global_cutoff_residual_error_limit}
\end{align}
Moreover,
\begin{align}
\operatorname{rank}
\Pi_{n,m}^{\mathrm{res},A^{c_{n,m}}}
\le
\left\lceil
\exp
\left[
\sqrt n
\right]
\right\rceil
\le
2
\exp
\left[
\sqrt n
\right]
\end{align}
whenever \(c_{n,m}>0\), while the rank equals one when
\(c_{n,m}=0\).
Therefore,
\begin{align}
\lim_{n\to\infty}
\frac{1}{n}
\log
\operatorname{rank}
\Pi_{n,m}^{\mathrm{res},A^{c_{n,m}}}
=
0.
\label{eq:IID_decoupling_global_cutoff_residual_rank_rate}
\end{align}

\medskip
\noindent
\textbf{The high-probability finite-rank projections in the $n$-copy space ($0<\delta_{\widehat{\Pi}_m}<1$).}
Under $0<\delta_{\widehat{\Pi}_m}<1$, for every sufficiently large \(n\), define
\begin{align}
&\Pi_{n,m,\eta,\xi}^{A^n}
\nonumber\\
&\coloneqq
\sum_{
x^{b_{n,m}}
\in
\mathcal T_{n,m,\eta}
}
\Biggl\{
\left(
U_{x^{b_{n,m}}}^{A,\mathrm{perm}}
\right)^\dagger
\Biggl[
\left(
\widehat{\Pi}_m^{A^m}
\right)^{\otimes\ell_{n,m,\eta}}
\otimes
\left(
\widehat{\Pi}_m^{A^m}
\right)^{
\otimes
\left(b_{n,m} - k \left( x^{b_{n,m}} \right) - \ell_{n,m,\eta}\right)
}
\otimes
\Pi_{
k(x^{b_{n,m}}),m,\xi
}^{\mathrm{rare}}
\Biggr]
U_{x^{b_{n,m}}}^{A,\mathrm{perm}}
\Biggr\}
\nonumber\\
&\quad\otimes
\Pi_{n,m}^{\mathrm{res},A^{c_{n,m}}}.
\label{eq:IID_decoupling_global_cutoff_projection_fixed_block}
\end{align}
Distinct common--rare patterns correspond to mutually orthogonal
tensor-product subspaces.
Hence, the summands in
\eqref{eq:IID_decoupling_global_cutoff_projection_fixed_block}
have mutually orthogonal ranges, and their sum is a finite-rank
projection.

Equations
\eqref{eq:IID_decoupling_global_cutoff_pattern_probability_limit},
\eqref{eq:IID_decoupling_global_cutoff_rare_probability_limit}, and
\eqref{eq:IID_decoupling_global_cutoff_uniform_rare_count_divergence}
imply that the right-hand side of
\eqref{eq:IID_decoupling_global_cutoff_projection_fixed_block}
is nonzero for all sufficiently large \(n\).
For the remaining finitely many \(n\ge m\), define
\(\Pi_{n,m,\eta,\xi}^{A^n}\)
to be the rank-one projection onto an eigenvector associated with the
largest eigenvalue of
\((\rho^A)^{\otimes n}\).
For the finitely many values
\(n<m\),
we use the same exceptional convention: we choose
\(\Pi_{n,m,\eta,\xi}^{A^n}\)
to be a rank-one spectral projection corresponding to a positive
eigenvalue of
\((\rho^A)^{\otimes n}\),
take
\(M_{n,m,\eta,\xi}\)
to be one-dimensional, and identify
\(E_{n,m,\eta,\xi}\)
with the range of the projection.
These choices do not affect any asymptotic statement.

For all sufficiently large \(n\), define the finite-dimensional
remainder space by the external Hilbert-space direct sum
\begin{align}
&\mathcal H^{B_{n,m,\eta,\xi}}
\nonumber\\
&\coloneqq
\Biggl[
\bigoplus_{
x^{b_{n,m}}
\in
\mathcal T_{n,m,\eta}
}
\Biggl(
\left(
\mathcal H^{A_{\widehat{\Pi}_m}^m}
\right)^{
\otimes
\left(b_{n,m} - k \left( x^{b_{n,m}} \right) - \ell_{n,m,\eta}\right)
}
\otimes
\mathcal{H}^{A_{\Pi_{k(x^{b_{n,m}}),m,\xi}^{\mathrm{rare}}}^{mk(x^{b_{n,m}})}}
\Biggr)
\Biggr]
\otimes
\mathcal{H}^{A_{\Pi_{n,m}^{\mathrm{res},A^{c_{n,m}}}}^{c_{n,m}}}.
\label{eq:IID_decoupling_global_cutoff_remainder_space}
\end{align}
The external direct sum retains the mutually orthogonal
typical-pattern sector label.
Thus,
\(B_{n,m,\eta,\xi}\)
contains the remaining common blocks, the rare blocks, the pattern
sector, and the residual copies.

\medskip
\noindent
\textbf{The high-probability finite-rank projections in the $n$-copy space 
(\(\delta_{\widehat{\Pi}_m}=0\)).}
Assume now that
\begin{align}
\delta_{\widehat{\Pi}_m}
=
0.
\label{eq:IID_decoupling_global_cutoff_common_only_case}
\end{align}
In this case, the rare branch is absent.

Set
\begin{align}
\eta
=
\xi
=
0,
\label{eq:eta_xi_zero}
\end{align}
and adopt the convention
\begin{align}
\gamma_{m,0}^{\mathrm{rare}}
\coloneqq
0.
\label{eq:IID_decoupling_global_cutoff_common_only_rare_rank_exponent}
\end{align}
Define
\begin{align}
\mathcal T_{n,m,0}
&\coloneqq
\left\{
0^{b_{n,m}}
\right\},
\label{eq:IID_decoupling_common_only_pattern_set}
\\
k
\left(
0^{b_{n,m}}
\right)
&\coloneqq
0,
\label{eq:IID_decoupling_common_only_rare_count}
\\
\ell_{n,m,0}
&\coloneqq
b_{n,m},
\label{eq:IID_decoupling_common_only_fixed_count}
\\
\mathcal I_{\mathrm{fix}}
\left(
0^{b_{n,m}}
\right)
&\coloneqq
\left\{
1,\ldots,b_{n,m}
\right\},
\label{eq:IID_decoupling_common_only_fixed_set}
\\
\mathcal I_{\mathrm{rem}}
\left(
0^{b_{n,m}}
\right)
&\coloneqq
\emptyset,
\label{eq:IID_decoupling_common_only_remaining_set}
\\
\mathcal I_{\mathrm{rare}}
\left(
0^{b_{n,m}}
\right)
&\coloneqq
\emptyset,
\label{eq:IID_decoupling_common_only_rare_set}
\\
U_{0^{b_{n,m}}}^{A,\mathrm{perm}}
&\coloneqq
I^{A^{mb_{n,m}}}.
\label{eq:IID_decoupling_common_only_permutation}
\end{align}
Define the global projection and the remainder space by
\begin{align}
\Pi_{n,m,0,0}^{A^n}
&\coloneqq
\left(
\widehat{\Pi}_m^{A^m}
\right)^{\otimes b_{n,m}}
\otimes
\Pi_{n,m}^{\mathrm{res},A^{c_{n,m}}},
\label{eq:IID_decoupling_global_cutoff_projection_common_only}
\\
\mathcal H^{B_{n,m,0,0}}
&\coloneqq
\mathcal{H}^{A_{\Pi_{n,m}^{\mathrm{res},A^{c_{n,m}}}}^{c_{n,m}}}.
\label{eq:IID_decoupling_global_cutoff_common_only_remainder_space}
\end{align}
Thus, every complete block is a fixed-count common block, and there
are no remaining common blocks or rare blocks.

\medskip
\noindent
\textbf{Unified definitions.}
The following projection-conditioned spaces, errors, states, and
marginals are the specializations of
\eqref{eq:IID_decoupling_projection_complement}--%
\eqref{eq:IID_decoupling_zero_complementary_convention}
to the projection
\(\Pi_{n,m,\eta,\xi}^{A^n}\).
The separate equation labels are retained for convenient reference to
the global-cutoff construction.

In either case, define the complementary projection, cutoff error,
projected input spaces, and projected dimension by
\begin{align}
\Pi_{n,m,\eta,\xi}^{\perp A^n}
&\coloneqq
I^{A^n}
-
\Pi_{n,m,\eta,\xi}^{A^n},
\label{eq:IID_decoupling_global_cutoff_complement}
\\
\delta_{\Pi_{n,m,\eta,\xi}}
&\coloneqq
\operatorname{Tr}
\left[
\Pi_{n,m,\eta,\xi}^{\perp A^n}
(\rho^A)^{\otimes n}
\right],
\label{eq:IID_decoupling_global_cutoff_error}
\\
\mathcal H^{A_{\Pi_{n,m,\eta,\xi}}^n}
&\coloneqq
\Pi_{n,m,\eta,\xi}^{A^n}
\mathcal H^{A^n},
\label{eq:IID_decoupling_global_cutoff_input_space}
\\
\mathcal H^{A_{\Pi_{n,m,\eta,\xi}^{\perp}}^n}
&\coloneqq
\Pi_{n,m,\eta,\xi}^{\perp A^n}
\mathcal H^{A^n},
\label{eq:IID_decoupling_global_cutoff_complementary_input_space}
\\
\left|
A_{\Pi_{n,m,\eta,\xi}}^n
\right|
&\coloneqq
\operatorname{rank}
\Pi_{n,m,\eta,\xi}^{A^n}.
\label{eq:IID_decoupling_global_cutoff_input_dimension}
\end{align}

Each cutoff projection is chosen to be diagonal in an eigenbasis of
the corresponding marginal and therefore commutes with that marginal.
Moreover, the permutations of the complete \(m\)-copy blocks commute
with
\((\rho^A)^{\otimes mb_{n,m}}\).
The exceptional rank-one projections are chosen onto eigenvectors of
\((\rho^A)^{\otimes n}\).
Therefore, one has the commutation relation
\begin{align}
\left[
\Pi_{n,m,\eta,\xi}^{A^n},
(\rho^A)^{\otimes n}
\right]
=
0.
\label{eq:IID_decoupling_global_cutoff_commutation}
\end{align}
By construction,
\(\Pi_{n,m,\eta,\xi}^{A^n}\)
is nonzero and its range contains an eigenvector corresponding to a
positive eigenvalue of
\((\rho^A)^{\otimes n}\).
Consequently,
\begin{align}
\delta_{\Pi_{n,m,\eta,\xi}}
<
1.
\label{eq:IID_decoupling_global_cutoff_error_strict_bound}
\end{align}

The normalized state conditioned on successful global projection and
its input and reference marginals are
\begin{align}
\rho_{\Pi_{n,m,\eta,\xi}}^{
A_{\Pi_{n,m,\eta,\xi}}^nR^n
}
&\coloneqq
\frac{
\left(
\Pi_{n,m,\eta,\xi}^{A^n}
\otimes
I^{R^n}
\right)
(\rho^{AR})^{\otimes n}
\left(
\Pi_{n,m,\eta,\xi}^{A^n}
\otimes
I^{R^n}
\right)
}{
1-\delta_{\Pi_{n,m,\eta,\xi}}
},
\label{eq:IID_decoupling_global_cutoff_normalized_state}
\\
\rho_{\Pi_{n,m,\eta,\xi}}^{
A_{\Pi_{n,m,\eta,\xi}}^n
}
&\coloneqq
\operatorname{Tr}_{R^n}
\left[
\rho_{\Pi_{n,m,\eta,\xi}}^{
A_{\Pi_{n,m,\eta,\xi}}^nR^n
}
\right],
\label{eq:IID_decoupling_global_cutoff_input_marginal}
\\
\rho_{\Pi_{n,m,\eta,\xi}}^{R^n}
&\coloneqq
\operatorname{Tr}_{A_{\Pi_{n,m,\eta,\xi}}^n}
\left[
\rho_{\Pi_{n,m,\eta,\xi}}^{
A_{\Pi_{n,m,\eta,\xi}}^nR^n
}
\right].
\label{eq:IID_decoupling_global_cutoff_reference_marginal}
\end{align}

Whenever
\begin{align}
\delta_{\Pi_{n,m,\eta,\xi}}
>
0,
\end{align}
the normalized state conditioned on the complementary outcome and
its input and reference marginals are defined by
\begin{align}
\rho_{\Pi_{n,m,\eta,\xi}^{\perp}}^{
A_{\Pi_{n,m,\eta,\xi}^{\perp}}^nR^n
}
&\coloneqq
\frac{
\left(
\Pi_{n,m,\eta,\xi}^{\perp A^n}
\otimes
I^{R^n}
\right)
(\rho^{AR})^{\otimes n}
\left(
\Pi_{n,m,\eta,\xi}^{\perp A^n}
\otimes
I^{R^n}
\right)
}{
\delta_{\Pi_{n,m,\eta,\xi}}
},
\label{eq:IID_decoupling_global_cutoff_normalized_complementary_state}
\\
\rho_{\Pi_{n,m,\eta,\xi}^{\perp}}^{
A_{\Pi_{n,m,\eta,\xi}^{\perp}}^n
}
&\coloneqq
\operatorname{Tr}_{R^n}
\left[
\rho_{\Pi_{n,m,\eta,\xi}^{\perp}}^{
A_{\Pi_{n,m,\eta,\xi}^{\perp}}^nR^n
}
\right],
\label{eq:IID_decoupling_global_cutoff_complementary_input_marginal}
\\
\rho_{\Pi_{n,m,\eta,\xi}^{\perp}}^{R^n}
&\coloneqq
\operatorname{Tr}_{A_{\Pi_{n,m,\eta,\xi}^{\perp}}^n}
\left[
\rho_{\Pi_{n,m,\eta,\xi}^{\perp}}^{
A_{\Pi_{n,m,\eta,\xi}^{\perp}}^nR^n
}
\right].
\label{eq:IID_decoupling_global_cutoff_complementary_reference_marginal}
\end{align}
When
\(\delta_{\Pi_{n,m,\eta,\xi}}=0\),
the complementary branch is absent, and every weighted expression
involving
\(\rho_{\Pi_{n,m,\eta,\xi}^{\perp}}\)
is defined to be zero according to
\eqref{eq:IID_decoupling_zero_complementary_convention}.

For all sufficiently large \(n\) in the case
\(0<\delta_{\widehat{\Pi}_m}<1\), and for every \(n\ge m\) in the case
\(\delta_{\widehat{\Pi}_m}=0\), the restrictions of the sectorwise
permutations
\(U_{x^{b_{n,m}}}^{A,\mathrm{perm}}\)
to the corresponding mutually orthogonal pattern sectors, together
with the identity on the residual subsystem, define a unitary
isomorphism
\begin{align}
&V_{n,m,\eta,\xi}^{\mathrm{fix}}
:
\mathcal H^{A_{\Pi_{n,m,\eta,\xi}}^n}
\nonumber\\
&\hspace{20mm}\longrightarrow
\left(
\mathcal H^{A_{\widehat{\Pi}_m}^m}
\right)^{\otimes\ell_{n,m,\eta}}
\otimes
\mathcal H^{B_{n,m,\eta,\xi}},
\label{eq:IID_decoupling_global_cutoff_fixed_factorization}
\end{align}
where for each typical pattern
\(x^{b_{n,m}}\),
the permutation unitary
\(U_{x^{b_{n,m}}}^{A,\mathrm{perm}}\)
identifies the corresponding sector with
\begin{align}
\left(
\mathcal H^{A_{\widehat{\Pi}_m}^m}
\right)^{\otimes\ell_{n,m,\eta}}
\otimes
\left(
\mathcal H^{A_{\widehat{\Pi}_m}^m}
\right)^{
\otimes
\left(b_{n,m} - k \left( x^{b_{n,m}} \right) - \ell_{n,m,\eta}\right)
}
\otimes
\mathcal{H}^{A_{\Pi_{k(x^{b_{n,m}}),m,\xi}^{\mathrm{rare}}}^{mk(x^{b_{n,m}})}}
\otimes
\mathcal{H}^{A_{\Pi_{n,m}^{\mathrm{res},A^{c_{n,m}}}}^{c_{n,m}}}.
\end{align}
Taking the external direct sum over the pattern sectors and using the
canonical distributive unitary isomorphism
\begin{align}
\bigoplus_x
\left(
\mathcal H
\otimes
\mathcal H_x
\right)
\simeq
\mathcal H
\otimes
\left(
\bigoplus_x
\mathcal H_x
\right)
\end{align}
gives
\eqref{eq:IID_decoupling_global_cutoff_fixed_factorization}.

For these values of \(n\), define finite-dimensional systems
\(E_{n,m,\eta,\xi}\)
and
\(M_{n,m,\eta,\xi}\)
by
\begin{align}
\mathcal H^{E_{n,m,\eta,\xi}}
&\coloneqq
\left(
\mathcal H^{E_m}
\right)^{\otimes\ell_{n,m,\eta}},
\label{eq:IID_decoupling_global_cutoff_output_space}
\\
\mathcal H^{M_{n,m,\eta,\xi}}
&\coloneqq
\left(
\mathcal H^{M_m}
\right)^{\otimes\ell_{n,m,\eta}}
\otimes
\mathcal H^{B_{n,m,\eta,\xi}}.
\label{eq:IID_decoupling_global_cutoff_discarded_space}
\end{align}
Let
\begin{align}
V_{n,m,\eta,\xi}^{\mathrm{reord}}
:
\left(
\mathcal H^{E_m}
\otimes
\mathcal H^{M_m}
\right)^{\otimes\ell_{n,m,\eta}}
\otimes
\mathcal H^{B_{n,m,\eta,\xi}}
\longrightarrow
\mathcal H^{E_{n,m,\eta,\xi}}
\otimes
\mathcal H^{M_{n,m,\eta,\xi}}
\end{align}
be the canonical unitary isomorphism that reorders and regroups the
tensor factors.
Define the unitary isomorphism
\begin{align}
&V_{n,m,\eta,\xi}^{
A_{\Pi_{n,m,\eta,\xi}}^n
\to
E_{n,m,\eta,\xi}M_{n,m,\eta,\xi}
}
\nonumber\\
&\coloneqq
V_{n,m,\eta,\xi}^{\mathrm{reord}}
\left[
\left(
V_m^{A_{\widehat{\Pi}_m}^m\to E_mM_m}
\right)^{\otimes\ell_{n,m,\eta}}
\otimes
I^{B_{n,m,\eta,\xi}}
\right]
V_{n,m,\eta,\xi}^{\mathrm{fix}}.
\label{eq:IID_decoupling_global_cutoff_partial_trace_factorization}
\end{align}
Thus,
\begin{align}
\mathcal H^{A_{\Pi_{n,m,\eta,\xi}}^n}
\simeq
\mathcal H^{E_{n,m,\eta,\xi}}
\otimes
\mathcal H^{M_{n,m,\eta,\xi}}.
\label{eq:IID_decoupling_global_cutoff_output_discarded_factorization}
\end{align}
The corresponding partial-trace channel is
\begin{align}
&\mathcal T_{n,m,\eta,\xi}^{
A_{\Pi_{n,m,\eta,\xi}}^n
\to
E_{n,m,\eta,\xi}
}
\left(
X
\right)
\nonumber\\
&\coloneqq
\operatorname{Tr}_{M_{n,m,\eta,\xi}}
\left[
V_{n,m,\eta,\xi}
X
V_{n,m,\eta,\xi}^{\dagger}
\right].
\label{eq:IID_decoupling_global_cutoff_partial_trace_channel}
\end{align}

For the finitely many exceptional \(n\) in the case
\(0<\delta_{\widehat{\Pi}_m}<1\)
for which the rank-one definition of
\(\Pi_{n,m,\eta,\xi}^{A^n}\)
is used, choose
\(M_{n,m,\eta,\xi}\)
to be one-dimensional and set
\begin{align}
\mathcal H^{E_{n,m,\eta,\xi}}
\coloneqq
\mathcal H^{A_{\Pi_{n,m,\eta,\xi}}^n}.
\end{align}
Choose the canonical unitary isomorphism
\begin{align}
&V_{n,m,\eta,\xi}^{
A_{\Pi_{n,m,\eta,\xi}}^n
\to
E_{n,m,\eta,\xi}M_{n,m,\eta,\xi}
}
:
\mathcal H^{A_{\Pi_{n,m,\eta,\xi}}^n}
\nonumber\\
&\hspace{25mm}\longrightarrow
\mathcal H^{E_{n,m,\eta,\xi}}
\otimes
\mathcal H^{M_{n,m,\eta,\xi}},
\end{align}
and define
\begin{align}
&\mathcal T_{n,m,\eta,\xi}^{
A_{\Pi_{n,m,\eta,\xi}}^n
\to
E_{n,m,\eta,\xi}
}
\left(
X
\right)
\nonumber\\
&\coloneqq
\operatorname{Tr}_{M_{n,m,\eta,\xi}}
\left[
V_{n,m,\eta,\xi}
X
V_{n,m,\eta,\xi}^{\dagger}
\right].
\end{align}
These finitely many choices do not affect any asymptotic statement.

The preceding construction yields high-probability finite-rank projections for IID states such that the corresponding projected spaces admit a pattern-independent subsystem consisting of identical finite-dimensional blocks and, consequently, an exact factorization of the partial-trace map. The relevant asymptotic properties are summarized in the following theorem.

\begin{theorem}[High-probability finite-rank projections for IID states]
\label{thm:IID_decoupling_high_probability_global_cutoffs}
Let
\(\mathcal H^A\)
and
\(\mathcal H^R\)
be separable and possibly infinite-dimensional, and let
\begin{align}
\rho^{AR}
\in
\operatorname{D}
\left(
\mathcal H^A
\otimes
\mathcal H^R
\right)
\end{align}
satisfy
\begin{align}
H(A)_\rho
<
\infty.
\end{align}
Fix
\begin{align}
q
\in
\left[
0,
H(A)_\rho
\right],
\end{align}
and let
\(
\left\{
\widehat{\Pi}_m^{A^m},
E_m,
M_m,
V_m^{A_{\widehat{\Pi}_m}^m\to E_mM_m}
\right\}_{m=1}^{\infty}
\)
be a choice satisfying
Proposition~\ref{prop:decoupling_IID_finite_entropy_spectral_cutoffs}
for this value of \(q\).

Fix
\begin{align}
m
\in
\{1,2,\ldots\}
\end{align}
such that
\(
0
<
\delta_{\widehat{\Pi}_m}
<
1
\),
and fix parameters
\begin{align}
0
<
\eta
<
\min
\left\{
\delta_{\widehat{\Pi}_m},
1-\delta_{\widehat{\Pi}_m}
\right\},
\qquad
\xi
>
0.
\end{align}
Let
\(\ell_{n,m,\eta}\),
\(\gamma_{m,\xi}^{\mathrm{rare}}\),
\(\Pi_{n,m,\eta,\xi}^{A^n}\),
\(\delta_{\Pi_{n,m,\eta,\xi}}\),
\(\mathcal H^{A_{\Pi_{n,m,\eta,\xi}}^n}\),
\(\lvert A_{\Pi_{n,m,\eta,\xi}}^n\rvert\),
\(\rho_{\Pi_{n,m,\eta,\xi}}^{
A_{\Pi_{n,m,\eta,\xi}}^nR^n
}\),
\(\rho_{\Pi_{n,m,\eta,\xi}}^{
A_{\Pi_{n,m,\eta,\xi}}^n
}\),
and
\(\rho_{\Pi_{n,m,\eta,\xi}}^{R^n}\)
be those defined in
\eqref{eq:IID_decoupling_global_cutoff_common_block_number},
\eqref{eq:IID_decoupling_global_cutoff_rare_rank_exponent},
\eqref{eq:IID_decoupling_global_cutoff_projection_fixed_block},
\eqref{eq:IID_decoupling_global_cutoff_error},
\eqref{eq:IID_decoupling_global_cutoff_input_space},
\eqref{eq:IID_decoupling_global_cutoff_input_dimension},
\eqref{eq:IID_decoupling_global_cutoff_normalized_state},
\eqref{eq:IID_decoupling_global_cutoff_input_marginal}, and
\eqref{eq:IID_decoupling_global_cutoff_reference_marginal},
respectively.
For all sufficiently large \(n\), let
\(B_{n,m,\eta,\xi}\),
\(E_{n,m,\eta,\xi}\),
\(M_{n,m,\eta,\xi}\),
\(V_{n,m,\eta,\xi}^{
A_{\Pi_{n,m,\eta,\xi}}^n
\to
E_{n,m,\eta,\xi}M_{n,m,\eta,\xi}
}\),
and
\(\mathcal T_{n,m,\eta,\xi}\)
be those defined in
\eqref{eq:IID_decoupling_global_cutoff_remainder_space},
\eqref{eq:IID_decoupling_global_cutoff_output_space},
\eqref{eq:IID_decoupling_global_cutoff_discarded_space},
\eqref{eq:IID_decoupling_global_cutoff_partial_trace_factorization}, and
\eqref{eq:IID_decoupling_global_cutoff_partial_trace_channel},
respectively.
For the remaining finitely many exceptional values of \(n\), use the
prescribed definitions following
\eqref{eq:IID_decoupling_global_cutoff_partial_trace_channel}.
Then the following limits hold:
\begin{align}
\lim_{n\to\infty}
\delta_{\Pi_{n,m,\eta,\xi}}
&=
0,
\label{eq:IID_decoupling_global_cutoff_fixed_block_error_limit}
\\
\lim_{n\to\infty}
\frac{1}{n}
H
\left(
A_{\Pi_{n,m,\eta,\xi}}^n
\right)_{\rho_{\Pi_{n,m,\eta,\xi}}}
&=
H(A)_\rho,
\label{eq:IID_decoupling_global_cutoff_fixed_block_entropy_limit}
\\
\lim_{n\to\infty}
\frac{1}{n}
H
\left(
A_{\Pi_{n,m,\eta,\xi}}^n
\middle|
R^n
\right)_{\rho_{\Pi_{n,m,\eta,\xi}}}
&=
H(A|R)_\rho,
\label{eq:IID_decoupling_global_cutoff_fixed_block_conditional_limit}
\\
\lim_{n\to\infty}
\frac{1}{n}
I
\left(
A_{\Pi_{n,m,\eta,\xi}}^n:R^n
\right)_{\rho_{\Pi_{n,m,\eta,\xi}}}
&=
I(A:R)_\rho,
\label{eq:IID_decoupling_global_cutoff_fixed_block_mutual_information_limit}
\\
\lim_{n\to\infty}
\frac{1}{n}
\log
\left|
A_{\Pi_{n,m,\eta,\xi}}^n
\right|
&=
\frac{1}{m}
\max_{
\substack{
t\in[0,1]\\
\left|
t-\delta_{\widehat{\Pi}_m}
\right|
\le
\eta
}
}
\Biggl\{
h_2(t)
+
(1-t)
\log
\left|
A_{\widehat{\Pi}_m}^m
\right|
+
t
\gamma_{m,\xi}^{\mathrm{rare}}
\Biggr\},
\label{eq:IID_decoupling_global_cutoff_fixed_block_rank_rate}
\\
\lim_{n\to\infty}
\frac{1}{n}
\log
\left|
M_{n,m,\eta,\xi}
\right|
&=
\frac{1}{m}
\max_{
\substack{
t\in[0,1]\\
\left|
t-\delta_{\widehat{\Pi}_m}
\right|
\le
\eta
}
}
\Biggl\{
h_2(t)
+
(1-t)
\log
\left|
A_{\widehat{\Pi}_m}^m
\right|
+
t
\gamma_{m,\xi}^{\mathrm{rare}}
\Biggr\}
\nonumber\\
&\quad
-
\frac{
1-\delta_{\widehat{\Pi}_m}-\eta
}{
m
}
\Biggl[
\log
\left|
A_{\widehat{\Pi}_m}^m
\right|
-
\log
\left|
M_m
\right|
\Biggr].
\label{eq:IID_decoupling_global_cutoff_fixed_partial_trace_rate}
\end{align}

For every fixed
\(m\in\{1,2,\ldots\}\)
such that
\(
\delta_{\widehat{\Pi}_m}
=
0
\),
set
\(
\eta
=
\xi
=
0
\)
as in
\eqref{eq:eta_xi_zero},
and adopt the convention
\(
\gamma_{m,0}^{\mathrm{rare}}
=
0
\)
in
\eqref{eq:IID_decoupling_global_cutoff_common_only_rare_rank_exponent}.
Let all corresponding objects be defined using the common-only
definitions in
\eqref{eq:IID_decoupling_common_only_pattern_set}--%
\eqref{eq:IID_decoupling_global_cutoff_common_only_remainder_space}
and the unified definitions in
\eqref{eq:IID_decoupling_global_cutoff_complement}--%
\eqref{eq:IID_decoupling_global_cutoff_partial_trace_channel}.
Then
\eqref{eq:IID_decoupling_global_cutoff_fixed_block_error_limit}--%
\eqref{eq:IID_decoupling_global_cutoff_fixed_block_mutual_information_limit}
hold with
\(\eta=\xi=0\), and
\begin{align}
\lim_{n\to\infty}
\frac{1}{n}
\log
\left|
A_{\Pi_{n,m,0,0}}^n
\right|
&=
\frac{1}{m}
\log
\left|
A_{\widehat{\Pi}_m}^m
\right|,
\label{eq:IID_decoupling_global_cutoff_common_only_rank_rate}
\\
\lim_{n\to\infty}
\frac{1}{n}
\log
\left|
M_{n,m,0,0}
\right|
&=
\frac{1}{m}
\log
\left|
M_m
\right|.
\label{eq:IID_decoupling_global_cutoff_common_only_partial_trace_rate}
\end{align}

Finally, let
\(\{\eta_m\}_{m=1}^{\infty}\)
and
\(\{\xi_m\}_{m=1}^{\infty}\)
be any sequences satisfying
\begin{align}
0
<
\eta_m
&<
\min
\left\{
\delta_{\widehat{\Pi}_m},
1-\delta_{\widehat{\Pi}_m}
\right\},
\qquad
\xi_m>0,
&&
\text{if }
0<\delta_{\widehat{\Pi}_m}<1,
\nonumber\\
\eta_m
&=
\xi_m
=
0,
&&
\text{if }
\delta_{\widehat{\Pi}_m}=0,
\label{eq:IID_decoupling_global_cutoff_admissible_parameter_sequences}
\end{align}
and
\begin{align}
\lim_{m\to\infty}
\eta_m
&=
0,
&
\lim_{m\to\infty}
\xi_m
&=
0.
\label{eq:IID_decoupling_global_cutoff_parameter_sequence_limits}
\end{align}
Then
\begin{align}
\lim_{m\to\infty}
\lim_{n\to\infty}
\frac{1}{n}
\log
\left|
A_{\Pi_{n,m,\eta_m,\xi_m}}^n
\right|
&=
H(A)_\rho,
\label{eq:IID_decoupling_global_cutoff_projected_dimension_iterated_limit}
\\
\lim_{m\to\infty}
\lim_{n\to\infty}
\frac{1}{n}
\log
\left|
M_{n,m,\eta_m,\xi_m}
\right|
&=
q.
\label{eq:IID_decoupling_global_cutoff_prescribed_partial_trace_rate}
\end{align}
\end{theorem}

\begin{proof}
By
Proposition~\ref{prop:decoupling_IID_finite_entropy_spectral_cutoffs},
one has
\begin{align}
0
\le
\delta_{\widehat{\Pi}_m}
<
1
\end{align}
for every \(m\), so the two cases considered in the theorem are
exhaustive.

\medskip
\noindent
\textbf{Proof of
\eqref{eq:IID_decoupling_global_cutoff_fixed_block_error_limit}.}
Suppose first that
\(0<\delta_{\widehat{\Pi}_m}<1\).
For every sufficiently large \(n\), the product structure of
\((\rho^A)^{\otimes n}\), together with
\eqref{eq:IID_decoupling_global_cutoff_projection_fixed_block},
gives
\begin{align}
&1-\delta_{\Pi_{n,m,\eta,\xi}}
\nonumber\\
&=
\left(
1-\delta_{n,m}^{\mathrm{res}}
\right)
\sum_{
x^{b_{n,m}}
\in
\mathcal T_{n,m,\eta}
}
\left(
1-\delta_{\widehat{\Pi}_m}
\right)^{
b_{n,m}-k(x^{b_{n,m}})
}
\delta_{\widehat{\Pi}_m}^{k(x^{b_{n,m}})}
\nonumber\\
&\qquad\qquad\times
\operatorname{Tr}
\left[
\Pi_{k(x^{b_{n,m}}),m,\xi}^{\mathrm{rare}}
\left(
\rho_{\widehat{\Pi}_m^\perp}^{A_{\widehat{\Pi}_m^\perp}^m}
\right)^{
\otimes k(x^{b_{n,m}})
}
\right].
\label{eq:IID_decoupling_global_cutoff_success_probability_proof}
\end{align}
Consequently,
\begin{align}
1-\delta_{\Pi_{n,m,\eta,\xi}}
&\ge
\left(
1-\delta_{n,m}^{\mathrm{res}}
\right)
\nonumber\\
&\quad\times
\sum_{
x^{b_{n,m}}
\in
\mathcal T_{n,m,\eta}
}
\left(
1-\delta_{\widehat{\Pi}_m}
\right)^{
b_{n,m}-k(x^{b_{n,m}})
}
\delta_{\widehat{\Pi}_m}^{k(x^{b_{n,m}})}
\nonumber\\
&\quad\times
\inf_{
x^{b_{n,m}}
\in
\mathcal T_{n,m,\eta}
}
\operatorname{Tr}
\left[
\Pi_{k(x^{b_{n,m}}),m,\xi}^{\mathrm{rare}}
\left(
\rho_{\widehat{\Pi}_m^\perp}^{A_{\widehat{\Pi}_m^\perp}^m}
\right)^{
\otimes k(x^{b_{n,m}})
}
\right].
\end{align}
The first factor has limit one by
\eqref{eq:IID_decoupling_global_cutoff_residual_error_limit},
and the second factor has limit one by
\eqref{eq:IID_decoupling_global_cutoff_pattern_probability_limit}.
By
\eqref{eq:IID_decoupling_global_cutoff_uniform_rare_count_divergence},
the number of rare blocks tends to infinity uniformly over all
patterns in
\(\mathcal T_{n,m,\eta}\).
Hence,
\eqref{eq:IID_decoupling_global_cutoff_rare_probability_limit}
implies
\begin{align}
&\lim_{n\to\infty}
\inf_{
x^{b_{n,m}}
\in
\mathcal T_{n,m,\eta}
}
\operatorname{Tr}
\left[
\Pi_{k(x^{b_{n,m}}),m,\xi}^{\mathrm{rare}}
\left(
\rho_{\widehat{\Pi}_m^\perp}^{A_{\widehat{\Pi}_m^\perp}^m}
\right)^{
\otimes k(x^{b_{n,m}})
}
\right]
=
1.
\end{align}
Since
\(1-\delta_{\Pi_{n,m,\eta,\xi}}\le1\),
we obtain
\begin{align}
\lim_{n\to\infty}
\delta_{\Pi_{n,m,\eta,\xi}}
=
0.
\end{align}

If
\(\delta_{\widehat{\Pi}_m}=0\),
then
\eqref{eq:IID_decoupling_global_cutoff_projection_common_only}
gives
\begin{align}
1-\delta_{\Pi_{n,m,0,0}}
=
1-\delta_{n,m}^{\mathrm{res}}.
\end{align}
Thus,
\eqref{eq:IID_decoupling_global_cutoff_residual_error_limit}
proves
\eqref{eq:IID_decoupling_global_cutoff_fixed_block_error_limit}
also in this case.

\medskip
\noindent
\textbf{Proof of
\eqref{eq:IID_decoupling_global_cutoff_fixed_block_entropy_limit}.}
By
\eqref{eq:IID_decoupling_global_cutoff_commutation},
the projection
\(\Pi_{n,m,\eta,\xi}^{A^n}\)
commutes with
\((\rho^A)^{\otimes n}\).
Hence, these two operators admit a common orthonormal eigenbasis.
Each positive eigenvalue of
\((\rho^A)^{\otimes n}\)
is of the form
\begin{align}
\prod_{j=1}^{n}
\lambda_{i_j}^A
\end{align}
for a sequence
\begin{align}
i^n
=
(i_1,\ldots,i_n),
\qquad
i_j\in\{0,1,\ldots\},
\qquad
\lambda_{i_j}^A>0.
\end{align}
Within each degenerate eigenspace, the common eigenbasis may be
labelled by such sequences while preserving their multiplicities.

For every \(L>0\), split the rejected positive-eigenvalue
eigenvectors according to whether their normalized self-information
is at most \(L\) or greater than \(L\).
After division by \(n\), the contribution of the former is at most
\begin{align}
L
\delta_{\Pi_{n,m,\eta,\xi}}.
\end{align}
To bound the contribution of the latter, consider any sequence
\(i^n=(i_1,\ldots,i_n)\)
satisfying
\begin{align}
\frac{1}{n}
\sum_{j=1}^{n}
\left(
-\log\lambda_{i_j}^A
\right)
>
L.
\label{eq:IID_decoupling_global_cutoff_large_self_information_sequence}
\end{align}
Separating the indices for which
\(-\log\lambda_{i_j}^A\le L/2\)
from those for which
\(-\log\lambda_{i_j}^A>L/2\), we obtain
\begin{align}
\sum_{j=1}^{n}
\left(
-\log\lambda_{i_j}^A
\right)
&\le
\frac{nL}{2}
+
\sum_{
\substack{
1\le j\le n\\
-\log\lambda_{i_j}^A>L/2
}
}
\left(
-\log\lambda_{i_j}^A
\right).
\end{align}
By
\eqref{eq:IID_decoupling_global_cutoff_large_self_information_sequence},
the sum on the left-hand side is greater than \(nL\).
Together with the preceding upper bound, this implies
\begin{align}
\sum_{
\substack{
1\le j\le n\\
-\log\lambda_{i_j}^A>L/2
}
}
\left(
-\log\lambda_{i_j}^A
\right)
>
\frac{nL}{2}.
\end{align}
Consequently,
\begin{align}
\sum_{j=1}^{n}
\left(
-\log\lambda_{i_j}^A
\right)
&\le
\frac{nL}{2}
+
\sum_{
\substack{
1\le j\le n\\
-\log\lambda_{i_j}^A>L/2
}
}
\left(
-\log\lambda_{i_j}^A
\right)
\nonumber\\
&<
2
\sum_{
\substack{
1\le j\le n\\
-\log\lambda_{i_j}^A>L/2
}
}
\left(
-\log\lambda_{i_j}^A
\right).
\label{eq:IID_decoupling_global_cutoff_self_information_tail_bound}
\end{align}

Averaging
\eqref{eq:IID_decoupling_global_cutoff_self_information_tail_bound}
over the product eigenvalue distribution and enlarging the summation
from the rejected eigenvectors to all product eigenvectors gives
\begin{align}
&0
\le
\frac{1}{n}
\operatorname{Tr}
\left[
\Pi_{n,m,\eta,\xi}^{\perp A^n}
(\rho^A)^{\otimes n}
\left(
-\log
(\rho^A)^{\otimes n}
\right)
\right]
\nonumber\\
&\le
L
\delta_{\Pi_{n,m,\eta,\xi}}
+
\frac{2}{n}
\sum_{j=1}^{n}
\sum_{
\substack{
i\in\{0,1,\ldots\}\\
\lambda_i^A>0\\
-\log\lambda_i^A>L/2
}
}
\lambda_i^A
\left(
-\log\lambda_i^A
\right)
\nonumber\\
&=
L
\delta_{\Pi_{n,m,\eta,\xi}}
+
2
\sum_{
\substack{
i\in\{0,1,\ldots\}\\
\lambda_i^A>0\\
-\log\lambda_i^A>L/2
}
}
\lambda_i^A
\left(
-\log\lambda_i^A
\right).
\label{eq:IID_decoupling_global_cutoff_information_tail_bound}
\end{align}

By
\eqref{eq:IID_decoupling_global_cutoff_fixed_block_error_limit},
the first term on the right-hand side has limit zero for every fixed
\(L\).
Moreover,
\begin{align}
H(A)_\rho
=
\sum_{
\substack{
i\in\{0,1,\ldots\}\\
\lambda_i^A>0
}
}
\lambda_i^A
\left(
-\log\lambda_i^A
\right)
<
\infty,
\end{align}
and therefore
\begin{align}
\lim_{L\to\infty}
\sum_{
\substack{
i\in\{0,1,\ldots\}\\
\lambda_i^A>0\\
-\log\lambda_i^A>L/2
}
}
\lambda_i^A
\left(
-\log\lambda_i^A
\right)
=
0.
\end{align}
It follows that
\begin{align}
\lim_{n\to\infty}
\frac{1}{n}
\operatorname{Tr}
\left[
\Pi_{n,m,\eta,\xi}^{\perp A^n}
(\rho^A)^{\otimes n}
\left(
-\log
(\rho^A)^{\otimes n}
\right)
\right]
=
0.
\label{eq:IID_decoupling_global_cutoff_information_tail_limit}
\end{align}

Taking the \(A^n\)-marginal in
\eqref{eq:IID_decoupling_global_cutoff_normalized_state}
and using
\eqref{eq:IID_decoupling_global_cutoff_commutation},
we obtain
\begin{align}
&\left(
1-\delta_{\Pi_{n,m,\eta,\xi}}
\right)
H
\left(
A_{\Pi_{n,m,\eta,\xi}}^n
\right)_{\rho_{\Pi_{n,m,\eta,\xi}}}
\nonumber\\
&=
\operatorname{Tr}
\left[
\Pi_{n,m,\eta,\xi}^{A^n}
(\rho^A)^{\otimes n}
\left(
-\log
(\rho^A)^{\otimes n}
\right)
\right]
\nonumber\\
&\quad+
\left(
1-\delta_{\Pi_{n,m,\eta,\xi}}
\right)
\log
\left(
1-\delta_{\Pi_{n,m,\eta,\xi}}
\right).
\label{eq:IID_decoupling_global_cutoff_projected_entropy_identity}
\end{align}
Since
\begin{align}
\operatorname{Tr}
\left[
(\rho^A)^{\otimes n}
\left(
-\log
(\rho^A)^{\otimes n}
\right)
\right]
=
nH(A)_\rho,
\end{align}
equations
\eqref{eq:IID_decoupling_global_cutoff_fixed_block_error_limit},
\eqref{eq:IID_decoupling_global_cutoff_information_tail_limit}, and
\eqref{eq:IID_decoupling_global_cutoff_projected_entropy_identity}
give
\begin{align}
\lim_{n\to\infty}
\frac{
1-\delta_{\Pi_{n,m,\eta,\xi}}
}{
n
}
H
\left(
A_{\Pi_{n,m,\eta,\xi}}^n
\right)_{\rho_{\Pi_{n,m,\eta,\xi}}}
=
H(A)_\rho.
\end{align}
Since
\begin{align}
\lim_{n\to\infty}
\left(
1-\delta_{\Pi_{n,m,\eta,\xi}}
\right)
=
1,
\end{align}
we conclude that
\begin{align}
\lim_{n\to\infty}
\frac{1}{n}
H
\left(
A_{\Pi_{n,m,\eta,\xi}}^n
\right)_{\rho_{\Pi_{n,m,\eta,\xi}}}
=
H(A)_\rho.
\end{align}

The argument uses only
\eqref{eq:IID_decoupling_global_cutoff_commutation}
and
\eqref{eq:IID_decoupling_global_cutoff_fixed_block_error_limit}.
It therefore applies without change when
\(\delta_{\widehat{\Pi}_m}=0\).

\medskip
\noindent
\textbf{Proof of
\eqref{eq:IID_decoupling_global_cutoff_fixed_block_conditional_limit}.}
Whenever
\(\delta_{\Pi_{n,m,\eta,\xi}}>0\),
the normalized complementary branch and its input marginal are those
defined in
\eqref{eq:IID_decoupling_global_cutoff_normalized_complementary_state}
and
\eqref{eq:IID_decoupling_global_cutoff_complementary_input_marginal}.
By
\eqref{eq:IID_decoupling_global_cutoff_commutation},
one has
\begin{align}
&\delta_{\Pi_{n,m,\eta,\xi}}
H
\left(
A_{\Pi_{n,m,\eta,\xi}^{\perp}}^n
\right)_{\rho_{\Pi_{n,m,\eta,\xi}^{\perp}}}
\nonumber\\
&=
\operatorname{Tr}
\left[
\Pi_{n,m,\eta,\xi}^{\perp A^n}
(\rho^A)^{\otimes n}
\left(
-\log
(\rho^A)^{\otimes n}
\right)
\right]
\nonumber\\
&\quad+
\delta_{\Pi_{n,m,\eta,\xi}}
\log
\delta_{\Pi_{n,m,\eta,\xi}}.
\label{eq:IID_decoupling_global_cutoff_complementary_entropy_identity}
\end{align}
The right-hand side is finite because
\begin{align}
0
&\le
\operatorname{Tr}
\left[
\Pi_{n,m,\eta,\xi}^{\perp A^n}
(\rho^A)^{\otimes n}
\left(
-\log
(\rho^A)^{\otimes n}
\right)
\right]
\nonumber\\
&\le
\operatorname{Tr}
\left[
(\rho^A)^{\otimes n}
\left(
-\log
(\rho^A)^{\otimes n}
\right)
\right]
\nonumber\\
&=
nH(A)_\rho
<
\infty.
\end{align}
Thus,
\eqref{eq:conditional_entropy_bounds_finite_marginal}
gives
\begin{align}
&\left|
\delta_{\Pi_{n,m,\eta,\xi}}
H
\left(
A_{\Pi_{n,m,\eta,\xi}^{\perp}}^n
\middle|
R^n
\right)_{\rho_{\Pi_{n,m,\eta,\xi}^{\perp}}}
\right|
\nonumber\\
&\le
\delta_{\Pi_{n,m,\eta,\xi}}
H
\left(
A_{\Pi_{n,m,\eta,\xi}^{\perp}}^n
\right)_{\rho_{\Pi_{n,m,\eta,\xi}^{\perp}}}
\nonumber\\
&\le
\operatorname{Tr}
\left[
\Pi_{n,m,\eta,\xi}^{\perp A^n}
(\rho^A)^{\otimes n}
\left(
-\log
(\rho^A)^{\otimes n}
\right)
\right],
\label{eq:IID_decoupling_global_cutoff_complementary_conditional_entropy_bound}
\end{align}
where the final inequality follows from
\begin{align}
\delta_{\Pi_{n,m,\eta,\xi}}
\log
\delta_{\Pi_{n,m,\eta,\xi}}
\le
0.
\end{align}
When
\(\delta_{\Pi_{n,m,\eta,\xi}}=0\),
the weighted complementary contribution is defined to be zero, and
\eqref{eq:IID_decoupling_global_cutoff_complementary_conditional_entropy_bound}
remains valid.

Let \(X\) be a binary register, and define the isometry
\begin{align}
V_{\Pi_{n,m,\eta,\xi}}^{A^n\to XA^n}
\coloneqq
\ket{0}^X
\otimes
\Pi_{n,m,\eta,\xi}^{A^n}
+
\ket{1}^X
\otimes
\Pi_{n,m,\eta,\xi}^{\perp A^n}.
\end{align}
Define
\begin{align}
&\omega_{n,m,\eta,\xi}^{XA^nR^n}
\nonumber\\
&\coloneqq
\left(
V_{\Pi_{n,m,\eta,\xi}}^{A^n\to XA^n}
\otimes
I^{R^n}
\right)
(\rho^{AR})^{\otimes n}
\left(
\left(
V_{\Pi_{n,m,\eta,\xi}}^{A^n\to XA^n}
\right)^\dagger
\otimes
I^{R^n}
\right),
\label{eq:IID_decoupling_global_cutoff_coherent_measurement_state}
\end{align}
and
\begin{align}
&\overline{\omega}_{n,m,\eta,\xi}^{XA^nR^n}
\nonumber\\
&\coloneqq
\left(
\Delta^X
\otimes
\id^{A^nR^n}
\right)
\left(
\omega_{n,m,\eta,\xi}^{XA^nR^n}
\right),
\label{eq:IID_decoupling_global_cutoff_dephased_measurement_state}
\end{align}
where
\(\Delta^X\)
denotes the dephasing channel in the basis
\(\{\ket{0},\ket{1}\}\).

Invariance of conditional entropy under a local isometry gives
\begin{align}
H
\left(
XA^n
\middle|
R^n
\right)_{\omega_{n,m,\eta,\xi}}
=
nH(A|R)_\rho.
\label{eq:IID_decoupling_global_cutoff_coherent_measurement_conditional_entropy}
\end{align}
The off-diagonal blocks of the
\(XA^n\)-marginal of
\(\omega_{n,m,\eta,\xi}\)
are
\begin{align}
\ket{0}\bra{1}^X
\otimes
\Pi_{n,m,\eta,\xi}^{A^n}
(\rho^A)^{\otimes n}
\Pi_{n,m,\eta,\xi}^{\perp A^n}
\end{align}
and its adjoint.
They vanish by
\eqref{eq:IID_decoupling_global_cutoff_commutation}.
It follows that
\begin{align}
\omega_{n,m,\eta,\xi}^{XA^n}
=
\overline{\omega}_{n,m,\eta,\xi}^{XA^n}.
\label{eq:IID_decoupling_global_cutoff_dephasing_marginal_invariance}
\end{align}
Moreover, since the dephasing acts trivially on \(R^n\),
\begin{align}
\omega_{n,m,\eta,\xi}^{R^n}
=
\overline{\omega}_{n,m,\eta,\xi}^{R^n}.
\end{align}

Since
\(V_{\Pi_{n,m,\eta,\xi}}^{A^n\to XA^n}\)
is an isometry,
\begin{align}
H
\left(
XA^n
\right)_{\omega_{n,m,\eta,\xi}}
=
H(A^n)_{\rho^{\otimes n}}
=
nH(A)_\rho
<
\infty.
\end{align}
Equation
\eqref{eq:IID_decoupling_global_cutoff_dephasing_marginal_invariance}
therefore gives
\begin{align}
H
\left(
XA^n
\right)_{\overline{\omega}_{n,m,\eta,\xi}}
=
nH(A)_\rho
<
\infty.
\end{align}
In particular,
\begin{align}
I
\left(
XA^n:R^n
\right)_{\omega_{n,m,\eta,\xi}}
&\le
2nH(A)_\rho
<
\infty,
\\
I
\left(
XA^n:R^n
\right)_{\overline{\omega}_{n,m,\eta,\xi}}
&\le
2nH(A)_\rho
<
\infty.
\end{align}

Define the unitary
\begin{align}
U^{X,\mathrm{deph}}
\coloneqq
\ket{0}\bra{0}^X
-
\ket{1}\bra{1}^X.
\end{align}
Then the dephasing channel has the random-unitary representation
\begin{align}
&\overline{\omega}_{n,m,\eta,\xi}^{XA^nR^n}
\nonumber\\
&=
\frac{1}{2}
\omega_{n,m,\eta,\xi}^{XA^nR^n}
\nonumber\\
&\quad+
\frac{1}{2}
\left(
U^{X,\mathrm{deph}}
\otimes
I^{A^nR^n}
\right)
\omega_{n,m,\eta,\xi}^{XA^nR^n}
\left(
U^{X,\mathrm{deph}}
\otimes
I^{A^nR^n}
\right).
\label{eq:IID_decoupling_global_cutoff_dephasing_random_unitary_representation}
\end{align}

Introduce a binary classical register \(Y\), and define
\begin{align}
&\widetilde{\omega}_{n,m,\eta,\xi}^{YXA^nR^n}
\nonumber\\
&\coloneqq
\frac{1}{2}
\ket{0}\bra{0}^Y
\otimes
\omega_{n,m,\eta,\xi}^{XA^nR^n}
\nonumber\\
&\quad+
\frac{1}{2}
\ket{1}\bra{1}^Y
\otimes
\left(
U^{X,\mathrm{deph}}
\otimes
I^{A^nR^n}
\right)
\omega_{n,m,\eta,\xi}^{XA^nR^n}
\left(
U^{X,\mathrm{deph}}
\otimes
I^{A^nR^n}
\right).
\end{align}
Its \(XA^nR^n\)-marginal is
\(\overline{\omega}_{n,m,\eta,\xi}^{XA^nR^n}\).

The chain rule for mutual information in
\eqref{eq:mutual_information_chain_rule}
gives
\begin{align}
&I
\left(
YXA^n:R^n
\right)_{\widetilde{\omega}_{n,m,\eta,\xi}}
\nonumber\\
&=
I
\left(
Y:R^n
\right)_{\widetilde{\omega}_{n,m,\eta,\xi}}
+
I
\left(
XA^n:R^n
\middle|
Y
\right)_{\widetilde{\omega}_{n,m,\eta,\xi}}.
\label{eq:IID_decoupling_global_cutoff_dephasing_mutual_information_chain_rule_Y_first}
\end{align}
Since the two conditional branches have the same \(R^n\)-marginal,
one has
\begin{align}
I
\left(
Y:R^n
\right)_{\widetilde{\omega}_{n,m,\eta,\xi}}
=
0.
\end{align}
Moreover, conditional mutual information for a classical register is
the average of the mutual informations of the conditional states.
Hence,
\begin{align}
&I
\left(
XA^n:R^n
\middle|
Y
\right)_{\widetilde{\omega}_{n,m,\eta,\xi}}
\nonumber\\
&=
\frac{1}{2}
I
\left(
XA^n:R^n
\right)_{\omega_{n,m,\eta,\xi}}
\nonumber\\
&\quad+
\frac{1}{2}
I
\left(
XA^n:R^n
\right)_{
\left(
U^{X,\mathrm{deph}}
\otimes
I^{A^nR^n}
\right)
\omega_{n,m,\eta,\xi}
\left(
U^{X,\mathrm{deph}}
\otimes
I^{A^nR^n}
\right)
}
\nonumber\\
&=
I
\left(
XA^n:R^n
\right)_{\omega_{n,m,\eta,\xi}},
\label{eq:IID_decoupling_global_cutoff_dephasing_conditional_mutual_information}
\end{align}
where the final equality follows from invariance of mutual information
under local unitaries.
Substituting these identities into
\eqref{eq:IID_decoupling_global_cutoff_dephasing_mutual_information_chain_rule_Y_first}
gives
\begin{align}
I
\left(
YXA^n:R^n
\right)_{\widetilde{\omega}_{n,m,\eta,\xi}}
=
I
\left(
XA^n:R^n
\right)_{\omega_{n,m,\eta,\xi}}.
\label{eq:IID_decoupling_global_cutoff_dephasing_mutual_information_first_expansion}
\end{align}

Applying the chain rule
\eqref{eq:mutual_information_chain_rule}
in the opposite order gives
\begin{align}
&I
\left(
YXA^n:R^n
\right)_{\widetilde{\omega}_{n,m,\eta,\xi}}
\nonumber\\
&=
I
\left(
XA^n:R^n
\right)_{\overline{\omega}_{n,m,\eta,\xi}}
+
I
\left(
Y:R^n
\middle|
XA^n
\right)_{\widetilde{\omega}_{n,m,\eta,\xi}}.
\label{eq:IID_decoupling_global_cutoff_dephasing_mutual_information_chain_rule_XA_first}
\end{align}
Combining
\eqref{eq:IID_decoupling_global_cutoff_dephasing_mutual_information_first_expansion}
and
\eqref{eq:IID_decoupling_global_cutoff_dephasing_mutual_information_chain_rule_XA_first}
yields
\begin{align}
&I
\left(
XA^n:R^n
\right)_{\omega_{n,m,\eta,\xi}}
-
I
\left(
XA^n:R^n
\right)_{\overline{\omega}_{n,m,\eta,\xi}}
\nonumber\\
&=
I
\left(
Y:R^n
\middle|
XA^n
\right)_{\widetilde{\omega}_{n,m,\eta,\xi}}.
\label{eq:IID_decoupling_global_cutoff_mutual_information_pinching_decomposition}
\end{align}
Since \(Y\) is classical,
\begin{align}
0
&\le
I
\left(
Y:R^n
\middle|
XA^n
\right)_{\widetilde{\omega}_{n,m,\eta,\xi}}
\nonumber\\
&\le
H
\left(
Y
\middle|
XA^n
\right)_{\widetilde{\omega}_{n,m,\eta,\xi}}
\nonumber\\
&\le
H(Y)_{\widetilde{\omega}_{n,m,\eta,\xi}}
=
\log 2.
\label{eq:IID_decoupling_global_cutoff_dephasing_conditional_mutual_information_bound}
\end{align}
Since the \(XA^n\)-marginals of
\(\omega_{n,m,\eta,\xi}\)
and
\(\overline{\omega}_{n,m,\eta,\xi}\)
coincide,
\eqref{eq:I_AR_definition_finite_A_general} gives
\begin{align}
&H
\left(
XA^n
\middle|
R^n
\right)_{\overline{\omega}_{n,m,\eta,\xi}}
-
H
\left(
XA^n
\middle|
R^n
\right)_{\omega_{n,m,\eta,\xi}}
\nonumber\\
&=
I
\left(
XA^n:R^n
\right)_{\omega_{n,m,\eta,\xi}}
-
I
\left(
XA^n:R^n
\right)_{\overline{\omega}_{n,m,\eta,\xi}}.
\label{eq:IID_decoupling_global_cutoff_dephasing_conditional_entropy_difference}
\end{align}
Combining
\eqref{eq:IID_decoupling_global_cutoff_mutual_information_pinching_decomposition},
\eqref{eq:IID_decoupling_global_cutoff_dephasing_conditional_mutual_information_bound},
and
\eqref{eq:IID_decoupling_global_cutoff_dephasing_conditional_entropy_difference},
we obtain
\begin{align}
0
&\le
H
\left(
XA^n
\middle|
R^n
\right)_{\overline{\omega}_{n,m,\eta,\xi}}
-
H
\left(
XA^n
\middle|
R^n
\right)_{\omega_{n,m,\eta,\xi}}
\le
\log 2.
\label{eq:IID_decoupling_global_cutoff_dephased_conditional_entropy_comparison}
\end{align}

The state
\(\overline{\omega}_{n,m,\eta,\xi}^{XA^nR^n}\)
has the form
\begin{align}
&\overline{\omega}_{n,m,\eta,\xi}^{XA^nR^n}
\nonumber\\
&=
\left(
1-\delta_{\Pi_{n,m,\eta,\xi}}
\right)
\ket{0}\bra{0}^X
\otimes
\rho_{\Pi_{n,m,\eta,\xi}}^{
A_{\Pi_{n,m,\eta,\xi}}^nR^n
}
\nonumber\\
&\quad+
\delta_{\Pi_{n,m,\eta,\xi}}
\ket{1}\bra{1}^X
\otimes
\rho_{\Pi_{n,m,\eta,\xi}^{\perp}}^{
A_{\Pi_{n,m,\eta,\xi}^{\perp}}^nR^n
},
\label{eq:IID_decoupling_global_cutoff_dephased_state_branch_decomposition}
\end{align}
where the conditional states are canonically embedded into the
orthogonal subspaces
\(\mathcal H^{A_{\Pi_{n,m,\eta,\xi}}^n}\)
and
\(\mathcal H^{A_{\Pi_{n,m,\eta,\xi}^{\perp}}^n}\)
of
\(\mathcal H^{A^n}\).
When
\(\delta_{\Pi_{n,m,\eta,\xi}}=0\),
the second term is absent.

Since \(X\) is classical, the chain rule gives
\begin{align}
&H
\left(
XA^n
\middle|
R^n
\right)_{\overline{\omega}_{n,m,\eta,\xi}}
\nonumber\\
&=
H
\left(
X
\middle|
R^n
\right)_{\overline{\omega}_{n,m,\eta,\xi}}
\nonumber\\
&\quad+
\left(
1-\delta_{\Pi_{n,m,\eta,\xi}}
\right)
H
\left(
A_{\Pi_{n,m,\eta,\xi}}^n
\middle|
R^n
\right)_{\rho_{\Pi_{n,m,\eta,\xi}}}
\nonumber\\
&\quad+
\delta_{\Pi_{n,m,\eta,\xi}}
H
\left(
A_{\Pi_{n,m,\eta,\xi}^{\perp}}^n
\middle|
R^n
\right)_{\rho_{\Pi_{n,m,\eta,\xi}^{\perp}}}.
\label{eq:IID_decoupling_global_cutoff_dephased_conditional_entropy_chain_rule}
\end{align}
Moreover,
\begin{align}
0
&\le
H
\left(
X
\middle|
R^n
\right)_{\overline{\omega}_{n,m,\eta,\xi}}
\nonumber\\
&\le
H(X)_{\overline{\omega}_{n,m,\eta,\xi}}
\nonumber\\
&=
h_2
\left(
\delta_{\Pi_{n,m,\eta,\xi}}
\right)
\le
\log 2.
\end{align}
It follows from
\eqref{eq:IID_decoupling_global_cutoff_coherent_measurement_conditional_entropy},
\eqref{eq:IID_decoupling_global_cutoff_dephased_conditional_entropy_comparison},
and
\eqref{eq:IID_decoupling_global_cutoff_dephased_conditional_entropy_chain_rule}
that
\begin{align}
-\log 2
&\le
\left(
1-\delta_{\Pi_{n,m,\eta,\xi}}
\right)
H
\left(
A_{\Pi_{n,m,\eta,\xi}}^n
\middle|
R^n
\right)_{\rho_{\Pi_{n,m,\eta,\xi}}}
\nonumber\\
&\quad+
\delta_{\Pi_{n,m,\eta,\xi}}
H
\left(
A_{\Pi_{n,m,\eta,\xi}^{\perp}}^n
\middle|
R^n
\right)_{\rho_{\Pi_{n,m,\eta,\xi}^{\perp}}}
-
nH(A|R)_\rho
\nonumber\\
&\le
\log 2.
\end{align}
Equivalently,
\begin{align}
&\Biggl|
\left(
1-\delta_{\Pi_{n,m,\eta,\xi}}
\right)
H
\left(
A_{\Pi_{n,m,\eta,\xi}}^n
\middle|
R^n
\right)_{\rho_{\Pi_{n,m,\eta,\xi}}}
\nonumber\\
&\quad+
\delta_{\Pi_{n,m,\eta,\xi}}
H
\left(
A_{\Pi_{n,m,\eta,\xi}^{\perp}}^n
\middle|
R^n
\right)_{\rho_{\Pi_{n,m,\eta,\xi}^{\perp}}}
-
nH(A|R)_\rho
\Biggr|
\le
\log 2.
\label{eq:IID_decoupling_global_cutoff_conditional_entropy_branch_comparison}
\end{align}
When
\(\delta_{\Pi_{n,m,\eta,\xi}}=0\),
the weighted complementary contribution is defined to be zero.

Equations
\eqref{eq:IID_decoupling_global_cutoff_conditional_entropy_branch_comparison}
and
\eqref{eq:IID_decoupling_global_cutoff_complementary_conditional_entropy_bound}
give
\begin{align}
&\Biggl|
\frac{
1-\delta_{\Pi_{n,m,\eta,\xi}}
}{
n
}
H
\left(
A_{\Pi_{n,m,\eta,\xi}}^n
\middle|
R^n
\right)_{\rho_{\Pi_{n,m,\eta,\xi}}}
-
H(A|R)_\rho
\Biggr|
\nonumber\\
&\le
\frac{\log 2}{n}
\nonumber\\
&\quad+
\frac{1}{n}
\operatorname{Tr}
\left[
\Pi_{n,m,\eta,\xi}^{\perp A^n}
(\rho^A)^{\otimes n}
\left(
-\log
(\rho^A)^{\otimes n}
\right)
\right].
\label{eq:IID_decoupling_global_cutoff_weighted_conditional_entropy_bound}
\end{align}
The first term on the right-hand side has limit zero, while the second
has limit zero by
\eqref{eq:IID_decoupling_global_cutoff_information_tail_limit}.
Consequently,
\begin{align}
\lim_{n\to\infty}
\frac{
1-\delta_{\Pi_{n,m,\eta,\xi}}
}{
n
}
H
\left(
A_{\Pi_{n,m,\eta,\xi}}^n
\middle|
R^n
\right)_{\rho_{\Pi_{n,m,\eta,\xi}}}
=
H(A|R)_\rho.
\end{align}
By
\eqref{eq:IID_decoupling_global_cutoff_fixed_block_error_limit},
\begin{align}
\lim_{n\to\infty}
\left(
1-\delta_{\Pi_{n,m,\eta,\xi}}
\right)
=
1.
\end{align}
For all sufficiently large \(n\), this factor is strictly positive,
and therefore
\begin{align}
&\frac{1}{n}
H
\left(
A_{\Pi_{n,m,\eta,\xi}}^n
\middle|
R^n
\right)_{\rho_{\Pi_{n,m,\eta,\xi}}}
\nonumber\\
&=
\frac{1}{
1-\delta_{\Pi_{n,m,\eta,\xi}}
}
\frac{
1-\delta_{\Pi_{n,m,\eta,\xi}}
}{
n
}
H
\left(
A_{\Pi_{n,m,\eta,\xi}}^n
\middle|
R^n
\right)_{\rho_{\Pi_{n,m,\eta,\xi}}}.
\end{align}
Taking the limit gives
\begin{align}
\lim_{n\to\infty}
\frac{1}{n}
H
\left(
A_{\Pi_{n,m,\eta,\xi}}^n
\middle|
R^n
\right)_{\rho_{\Pi_{n,m,\eta,\xi}}}
=
H(A|R)_\rho.
\end{align}

The proof uses only the unified definitions and estimates and thus
also applies when
\(\delta_{\widehat{\Pi}_m}=0\).

\medskip
\noindent
\textbf{Proof of
\eqref{eq:IID_decoupling_global_cutoff_fixed_block_mutual_information_limit}.}
Since
\(H(A)_\rho<\infty\),
the mutual information is finite and
\begin{align}
I(A:R)_\rho
=
H(A)_\rho
-
H(A|R)_\rho.
\end{align}
Moreover, since
\(A_{\Pi_{n,m,\eta,\xi}}^n\)
is finite-dimensional,
\begin{align}
&I
\left(
A_{\Pi_{n,m,\eta,\xi}}^n:R^n
\right)_{\rho_{\Pi_{n,m,\eta,\xi}}}
\nonumber\\
&=
H
\left(
A_{\Pi_{n,m,\eta,\xi}}^n
\right)_{\rho_{\Pi_{n,m,\eta,\xi}}}
-
H
\left(
A_{\Pi_{n,m,\eta,\xi}}^n
\middle|
R^n
\right)_{\rho_{\Pi_{n,m,\eta,\xi}}}.
\end{align}
Consequently,
\eqref{eq:IID_decoupling_global_cutoff_fixed_block_entropy_limit}
and
\eqref{eq:IID_decoupling_global_cutoff_fixed_block_conditional_limit}
give
\begin{align}
\lim_{n\to\infty}
\frac{1}{n}
I
\left(
A_{\Pi_{n,m,\eta,\xi}}^n:R^n
\right)_{\rho_{\Pi_{n,m,\eta,\xi}}}
=
I(A:R)_\rho.
\end{align}
The same argument applies when
\(\delta_{\widehat{\Pi}_m}=0\).

\medskip
\noindent
\textbf{Proof of
\eqref{eq:IID_decoupling_global_cutoff_fixed_block_rank_rate}.}
Suppose that
\(0<\delta_{\widehat{\Pi}_m}<1\).
For all sufficiently large \(n\), distinct patterns in
\(\mathcal T_{n,m,\eta}\)
give mutually orthogonal sectors in
\eqref{eq:IID_decoupling_global_cutoff_projection_fixed_block}.
For each integer \(k\) satisfying
\begin{align}
0
\le
k
\le
b_{n,m},
\qquad
\left|
\frac{k}{b_{n,m}}
-
\delta_{\widehat{\Pi}_m}
\right|
\le
\eta,
\end{align}
there are
\(\binom{b_{n,m}}{k}\)
patterns containing exactly \(k\) rare blocks.
The sector corresponding to each such pattern has dimension
\begin{align}
\left|
A_{\widehat{\Pi}_m}^m
\right|^{b_{n,m}-k}
\operatorname{rank}
\Pi_{k,m,\xi}^{\mathrm{rare}},
\end{align}
and the residual subsystem contributes the common multiplicative
factor
\begin{align}
\operatorname{rank}
\Pi_{n,m}^{\mathrm{res},A^{c_{n,m}}}.
\end{align}
Consequently,
\begin{align}
&\left|
A_{\Pi_{n,m,\eta,\xi}}^n
\right|
\nonumber\\
&=
\operatorname{rank}
\Pi_{n,m}^{\mathrm{res},A^{c_{n,m}}}
\sum_{
\substack{
0\le k\le b_{n,m}\\
\left|
k/b_{n,m}-\delta_{\widehat{\Pi}_m}
\right|
\le
\eta
}
}
\binom{b_{n,m}}{k}
\left|
A_{\widehat{\Pi}_m}^m
\right|^{b_{n,m}-k}
\operatorname{rank}
\Pi_{k,m,\xi}^{\mathrm{rare}}.
\label{eq:IID_decoupling_global_cutoff_rank_decomposition_proof}
\end{align}

For every
\(0\le k\le b_{n,m}\),
the binary type-class bound
\cite[Example~12.1.3]{Cover2006}
gives
\begin{align}
\frac{1}{
b_{n,m}+1
}
\exp
\left[
b_{n,m}
h_2
\left(
\frac{k}{b_{n,m}}
\right)
\right]
&\le
\binom{b_{n,m}}{k}
\nonumber\\
&\le
\exp
\left[
b_{n,m}
h_2
\left(
\frac{k}{b_{n,m}}
\right)
\right].
\label{eq:IID_decoupling_global_cutoff_binomial_type_bounds}
\end{align}

By the assumed choice of \(\eta\) and
\eqref{eq:decoupling_IID_block_quotient_properties},
the summation range in
\eqref{eq:IID_decoupling_global_cutoff_rank_decomposition_proof}
is nonempty for all sufficiently large \(n\).
Since
\(\eta<\delta_{\widehat{\Pi}_m}\),
every \(k\) in this range satisfies
\begin{align}
k
\ge
b_{n,m}
\left(
\delta_{\widehat{\Pi}_m}-\eta
\right).
\end{align}
Together with
\eqref{eq:decoupling_IID_block_quotient_properties},
this yields
\begin{align}
\lim_{n\to\infty}
\min_{
\substack{
0\le k\le b_{n,m}\\
\left|
k/b_{n,m}-\delta_{\widehat{\Pi}_m}
\right|
\le
\eta
}
}
k
=
\infty.
\label{eq:IID_decoupling_global_cutoff_minimum_rare_count_limit}
\end{align}
By
\eqref{eq:IID_decoupling_global_cutoff_rare_probability_limit},
the projections
\(\Pi_{k,m,\xi}^{\mathrm{rare}}\)
have positive rank for all sufficiently large \(k\).
Therefore,
\eqref{eq:IID_decoupling_global_cutoff_minimum_rare_count_limit}
and
\eqref{eq:IID_decoupling_global_cutoff_rare_rank_exponent}
give
\begin{align}
&\lim_{n\to\infty}
\max_{
\substack{
0\le k\le b_{n,m}\\
\left|
k/b_{n,m}-\delta_{\widehat{\Pi}_m}
\right|
\le
\eta
}
}
\left|
\frac{1}{k}
\log
\operatorname{rank}
\Pi_{k,m,\xi}^{\mathrm{rare}}
-
\gamma_{m,\xi}^{\mathrm{rare}}
\right|
=
0.
\label{eq:IID_decoupling_global_cutoff_rare_rank_uniform_limit}
\end{align}

For every sufficiently large \(n\) and every \(k\) in the summation
range, one has
\begin{align}
&\frac{1}{b_{n,m}}
\log
\Biggl[
\binom{b_{n,m}}{k}
\left|
A_{\widehat{\Pi}_m}^m
\right|^{b_{n,m}-k}
\operatorname{rank}
\Pi_{k,m,\xi}^{\mathrm{rare}}
\Biggr]
\nonumber\\
&\quad-
\Biggl[
h_2
\left(
\frac{k}{b_{n,m}}
\right)
+
\left(
1-\frac{k}{b_{n,m}}
\right)
\log
\left|
A_{\widehat{\Pi}_m}^m
\right|
+
\frac{k}{b_{n,m}}
\gamma_{m,\xi}^{\mathrm{rare}}
\Biggr]
\nonumber\\
&=
\frac{1}{b_{n,m}}
\log
\binom{b_{n,m}}{k}
-
h_2
\left(
\frac{k}{b_{n,m}}
\right)
\nonumber\\
&\quad+
\frac{k}{b_{n,m}}
\left[
\frac{1}{k}
\log
\operatorname{rank}
\Pi_{k,m,\xi}^{\mathrm{rare}}
-
\gamma_{m,\xi}^{\mathrm{rare}}
\right].
\label{eq:IID_decoupling_global_cutoff_rank_summand_error_decomposition}
\end{align}

By
\eqref{eq:IID_decoupling_global_cutoff_binomial_type_bounds},
\begin{align}
\left|
\frac{1}{b_{n,m}}
\log
\binom{b_{n,m}}{k}
-
h_2
\left(
\frac{k}{b_{n,m}}
\right)
\right|
\le
\frac{
\log
\left(
b_{n,m}+1
\right)
}{
b_{n,m}
}.
\label{eq:IID_decoupling_global_cutoff_binomial_type_uniform_error}
\end{align}
Since
\begin{align}
0
\le
\frac{k}{b_{n,m}}
\le
1,
\end{align}
it follows that
\begin{align}
&\max_{
\substack{
0\le k\le b_{n,m}\\
\left|
k/b_{n,m}-\delta_{\widehat{\Pi}_m}
\right|
\le
\eta
}
}
\Biggl|
\frac{1}{b_{n,m}}
\log
\Biggl[
\binom{b_{n,m}}{k}
\left|
A_{\widehat{\Pi}_m}^m
\right|^{b_{n,m}-k}
\operatorname{rank}
\Pi_{k,m,\xi}^{\mathrm{rare}}
\Biggr]
\nonumber\\
&\hspace{14mm}-
\Biggl[
h_2
\left(
\frac{k}{b_{n,m}}
\right)
+
\left(
1-\frac{k}{b_{n,m}}
\right)
\log
\left|
A_{\widehat{\Pi}_m}^m
\right|
+
\frac{k}{b_{n,m}}
\gamma_{m,\xi}^{\mathrm{rare}}
\Biggr]
\Biggr|
\nonumber\\
&\le
\frac{
\log
\left(
b_{n,m}+1
\right)
}{
b_{n,m}
}
\nonumber\\
&\quad+
\max_{
\substack{
0\le k\le b_{n,m}\\
\left|
k/b_{n,m}-\delta_{\widehat{\Pi}_m}
\right|
\le
\eta
}
}
\left|
\frac{1}{k}
\log
\operatorname{rank}
\Pi_{k,m,\xi}^{\mathrm{rare}}
-
\gamma_{m,\xi}^{\mathrm{rare}}
\right|.
\label{eq:IID_decoupling_global_cutoff_rank_summand_uniform_bound}
\end{align}
By
\eqref{eq:decoupling_IID_block_quotient_properties},
\begin{align}
\lim_{n\to\infty}
\frac{
\log
\left(
b_{n,m}+1
\right)
}{
b_{n,m}
}
=
0.
\label{eq:IID_decoupling_global_cutoff_number_of_rank_terms_rate}
\end{align}
Together with
\eqref{eq:IID_decoupling_global_cutoff_rare_rank_uniform_limit},
this gives
\begin{align}
&\lim_{n\to\infty}
\max_{
\substack{
0\le k\le b_{n,m}\\
\left|
k/b_{n,m}-\delta_{\widehat{\Pi}_m}
\right|
\le
\eta
}
}
\Biggl|
\frac{1}{b_{n,m}}
\log
\Biggl[
\binom{b_{n,m}}{k}
\left|
A_{\widehat{\Pi}_m}^m
\right|^{b_{n,m}-k}
\operatorname{rank}
\Pi_{k,m,\xi}^{\mathrm{rare}}
\Biggr]
\nonumber\\
&\hspace{14mm}-
\Biggl[
h_2
\left(
\frac{k}{b_{n,m}}
\right)
+
\left(
1-\frac{k}{b_{n,m}}
\right)
\log
\left|
A_{\widehat{\Pi}_m}^m
\right|
+
\frac{k}{b_{n,m}}
\gamma_{m,\xi}^{\mathrm{rare}}
\Biggr]
\Biggr|
=
0.
\label{eq:IID_decoupling_global_cutoff_rank_summand_uniform_limit}
\end{align}

For all sufficiently large \(n\), every summand in
\eqref{eq:IID_decoupling_global_cutoff_rank_decomposition_proof}
is positive.
Since the sum contains at most
\(b_{n,m}+1\)
terms, one has
\begin{align}
&\max_{
\substack{
0\le k\le b_{n,m}\\
\left|
k/b_{n,m}-\delta_{\widehat{\Pi}_m}
\right|
\le
\eta
}
}
\Biggl\{
\binom{b_{n,m}}{k}
\left|
A_{\widehat{\Pi}_m}^m
\right|^{b_{n,m}-k}
\operatorname{rank}
\Pi_{k,m,\xi}^{\mathrm{rare}}
\Biggr\}
\nonumber\\
&\le
\sum_{
\substack{
0\le k\le b_{n,m}\\
\left|
k/b_{n,m}-\delta_{\widehat{\Pi}_m}
\right|
\le
\eta
}
}
\binom{b_{n,m}}{k}
\left|
A_{\widehat{\Pi}_m}^m
\right|^{b_{n,m}-k}
\operatorname{rank}
\Pi_{k,m,\xi}^{\mathrm{rare}}
\nonumber\\
&\le
\left(
b_{n,m}+1
\right)
\max_{
\substack{
0\le k\le b_{n,m}\\
\left|
k/b_{n,m}-\delta_{\widehat{\Pi}_m}
\right|
\le
\eta
}
}
\Biggl\{
\binom{b_{n,m}}{k}
\left|
A_{\widehat{\Pi}_m}^m
\right|^{b_{n,m}-k}
\operatorname{rank}
\Pi_{k,m,\xi}^{\mathrm{rare}}
\Biggr\}.
\label{eq:IID_decoupling_global_cutoff_rank_sum_maximum_comparison}
\end{align}
Taking logarithms and dividing by \(b_{n,m}\) gives
\begin{align}
0
&\le
\frac{1}{b_{n,m}}
\log
\sum_{
\substack{
0\le k\le b_{n,m}\\
\left|
k/b_{n,m}-\delta_{\widehat{\Pi}_m}
\right|
\le
\eta
}
}
\binom{b_{n,m}}{k}
\left|
A_{\widehat{\Pi}_m}^m
\right|^{b_{n,m}-k}
\operatorname{rank}
\Pi_{k,m,\xi}^{\mathrm{rare}}
\nonumber\\
&\quad-
\max_{
\substack{
0\le k\le b_{n,m}\\
\left|
k/b_{n,m}-\delta_{\widehat{\Pi}_m}
\right|
\le
\eta
}
}
\frac{1}{b_{n,m}}
\log
\Biggl[
\binom{b_{n,m}}{k}
\left|
A_{\widehat{\Pi}_m}^m
\right|^{b_{n,m}-k}
\operatorname{rank}
\Pi_{k,m,\xi}^{\mathrm{rare}}
\Biggr]
\nonumber\\
&\le
\frac{
\log
\left(
b_{n,m}+1
\right)
}{
b_{n,m}
}.
\label{eq:IID_decoupling_global_cutoff_rank_sum_maximum_log_comparison}
\end{align}

The inequality
\begin{align}
\left|
\max_k x_k-\max_k y_k
\right|
\le
\max_k
\left|
x_k-y_k
\right|
\end{align}
and
\eqref{eq:IID_decoupling_global_cutoff_rank_summand_uniform_limit}
give
\begin{align}
&\lim_{n\to\infty}
\Biggl|
\max_{
\substack{
0\le k\le b_{n,m}\\
\left|
k/b_{n,m}-\delta_{\widehat{\Pi}_m}
\right|
\le
\eta
}
}
\frac{1}{b_{n,m}}
\log
\Biggl[
\binom{b_{n,m}}{k}
\left|
A_{\widehat{\Pi}_m}^m
\right|^{b_{n,m}-k}
\operatorname{rank}
\Pi_{k,m,\xi}^{\mathrm{rare}}
\Biggr]
\nonumber\\
&\quad-
\max_{
\substack{
0\le k\le b_{n,m}\\
\left|
k/b_{n,m}-\delta_{\widehat{\Pi}_m}
\right|
\le
\eta
}
}
\Biggl[
h_2
\left(
\frac{k}{b_{n,m}}
\right)
+
\left(
1-\frac{k}{b_{n,m}}
\right)
\log
\left|
A_{\widehat{\Pi}_m}^m
\right|
+
\frac{k}{b_{n,m}}
\gamma_{m,\xi}^{\mathrm{rare}}
\Biggr]
\Biggr|
=
0.
\label{eq:IID_decoupling_global_cutoff_largest_rank_summand_rate}
\end{align}

The function
\begin{align}
t
\longmapsto
h_2(t)
+
(1-t)
\log
\left|
A_{\widehat{\Pi}_m}^m
\right|
+
t
\gamma_{m,\xi}^{\mathrm{rare}}
\end{align}
is continuous on the nonempty compact interval
\begin{align}
\left\{
t\in[0,1]:
\left|
t-\delta_{\widehat{\Pi}_m}
\right|
\le
\eta
\right\}.
\end{align}
For all sufficiently large \(n\), every point in this interval is at
distance at most
\(1/b_{n,m}\)
from an admissible point of the form
\(k/b_{n,m}\).
Since
\begin{align}
\lim_{n\to\infty}
\frac{1}{b_{n,m}}
=
0,
\end{align}
uniform continuity gives
\begin{align}
&\lim_{n\to\infty}
\max_{
\substack{
0\le k\le b_{n,m}\\
\left|
k/b_{n,m}-\delta_{\widehat{\Pi}_m}
\right|
\le
\eta
}
}
\Biggl[
h_2
\left(
\frac{k}{b_{n,m}}
\right)
+
\left(
1-\frac{k}{b_{n,m}}
\right)
\log
\left|
A_{\widehat{\Pi}_m}^m
\right|
+
\frac{k}{b_{n,m}}
\gamma_{m,\xi}^{\mathrm{rare}}
\Biggr]
\nonumber\\
&=
\max_{
\substack{
t\in[0,1]\\
\left|
t-\delta_{\widehat{\Pi}_m}
\right|
\le
\eta
}
}
\Biggl\{
h_2(t)
+
(1-t)
\log
\left|
A_{\widehat{\Pi}_m}^m
\right|
+
t
\gamma_{m,\xi}^{\mathrm{rare}}
\Biggr\}.
\label{eq:IID_decoupling_global_cutoff_continuous_rank_maximum_limit}
\end{align}
Combining
\eqref{eq:IID_decoupling_global_cutoff_number_of_rank_terms_rate},
\eqref{eq:IID_decoupling_global_cutoff_rank_sum_maximum_log_comparison},
\eqref{eq:IID_decoupling_global_cutoff_largest_rank_summand_rate}, and
\eqref{eq:IID_decoupling_global_cutoff_continuous_rank_maximum_limit},
we obtain
\begin{align}
&\lim_{n\to\infty}
\frac{1}{b_{n,m}}
\log
\sum_{
\substack{
0\le k\le b_{n,m}\\
\left|
k/b_{n,m}-\delta_{\widehat{\Pi}_m}
\right|
\le
\eta
}
}
\binom{b_{n,m}}{k}
\left|
A_{\widehat{\Pi}_m}^m
\right|^{b_{n,m}-k}
\operatorname{rank}
\Pi_{k,m,\xi}^{\mathrm{rare}}
\nonumber\\
&=
\max_{
\substack{
t\in[0,1]\\
\left|
t-\delta_{\widehat{\Pi}_m}
\right|
\le
\eta
}
}
\Biggl\{
h_2(t)
+
(1-t)
\log
\left|
A_{\widehat{\Pi}_m}^m
\right|
+
t
\gamma_{m,\xi}^{\mathrm{rare}}
\Biggr\}.
\label{eq:IID_decoupling_global_cutoff_rank_sum_rate}
\end{align}

Taking the logarithm of
\eqref{eq:IID_decoupling_global_cutoff_rank_decomposition_proof}
and dividing by \(n\) gives
\begin{align}
&\frac{1}{n}
\log
\left|
A_{\Pi_{n,m,\eta,\xi}}^n
\right|
\nonumber\\
&=
\frac{1}{n}
\log
\operatorname{rank}
\Pi_{n,m}^{\mathrm{res},A^{c_{n,m}}}
\nonumber\\
&\quad+
\frac{b_{n,m}}{n}
\frac{1}{b_{n,m}}
\log
\sum_{
\substack{
0\le k\le b_{n,m}\\
\left|
k/b_{n,m}-\delta_{\widehat{\Pi}_m}
\right|
\le
\eta
}
}
\binom{b_{n,m}}{k}
\left|
A_{\widehat{\Pi}_m}^m
\right|^{b_{n,m}-k}
\operatorname{rank}
\Pi_{k,m,\xi}^{\mathrm{rare}}.
\end{align}
Using
\eqref{eq:IID_decoupling_global_cutoff_residual_rank_rate},
\eqref{eq:IID_decoupling_global_cutoff_rank_sum_rate}, and
\eqref{eq:decoupling_IID_block_quotient_properties},
we conclude that
\begin{align}
&\lim_{n\to\infty}
\frac{1}{n}
\log
\left|
A_{\Pi_{n,m,\eta,\xi}}^n
\right|
\nonumber\\
&=
\frac{1}{m}
\max_{
\substack{
t\in[0,1]\\
\left|
t-\delta_{\widehat{\Pi}_m}
\right|
\le
\eta
}
}
\Biggl\{
h_2(t)
+
(1-t)
\log
\left|
A_{\widehat{\Pi}_m}^m
\right|
+
t
\gamma_{m,\xi}^{\mathrm{rare}}
\Biggr\}.
\end{align}
The finitely many exceptional values of \(n\) do not affect this
limit.

\medskip
\noindent
\textbf{Proof of
\eqref{eq:IID_decoupling_global_cutoff_fixed_partial_trace_rate}.}
Suppose that
\(0<\delta_{\widehat{\Pi}_m}<1\).
For all sufficiently large \(n\),
\eqref{eq:IID_decoupling_global_cutoff_fixed_factorization},
\eqref{eq:IID_decoupling_global_cutoff_output_space}, and
\eqref{eq:IID_decoupling_global_cutoff_discarded_space}
give
\begin{align}
\left|
A_{\Pi_{n,m,\eta,\xi}}^n
\right|
&=
\left|
A_{\widehat{\Pi}_m}^m
\right|^{\ell_{n,m,\eta}}
\left|
B_{n,m,\eta,\xi}
\right|,
\label{eq:IID_decoupling_global_cutoff_fixed_factorization_dimension}
\\
\left|
M_{n,m,\eta,\xi}
\right|
&=
\left|
M_m
\right|^{\ell_{n,m,\eta}}
\left|
B_{n,m,\eta,\xi}
\right|.
\label{eq:IID_decoupling_global_cutoff_discarded_dimension}
\end{align}
Eliminating
\(\lvert B_{n,m,\eta,\xi}\rvert\)
from
\eqref{eq:IID_decoupling_global_cutoff_fixed_factorization_dimension}
and
\eqref{eq:IID_decoupling_global_cutoff_discarded_dimension}
gives
\begin{align}
\log
\left|
M_{n,m,\eta,\xi}
\right|
&=
\log
\left|
A_{\Pi_{n,m,\eta,\xi}}^n
\right|
\nonumber\\
&\quad-
\ell_{n,m,\eta}
\Biggl[
\log
\left|
A_{\widehat{\Pi}_m}^m
\right|
-
\log
\left|
M_m
\right|
\Biggr].
\label{eq:IID_decoupling_global_cutoff_discarded_input_dimension_identity}
\end{align}
By
\eqref{eq:IID_decoupling_global_cutoff_common_block_number}
and
\eqref{eq:decoupling_IID_block_quotient_properties},
\begin{align}
\lim_{n\to\infty}
\frac{
\ell_{n,m,\eta}
}{
n
}
=
\frac{
1-\delta_{\widehat{\Pi}_m}-\eta
}{
m
}.
\label{eq:IID_decoupling_global_cutoff_common_block_rate}
\end{align}
Dividing
\eqref{eq:IID_decoupling_global_cutoff_discarded_input_dimension_identity}
by \(n\), taking the limit, and using
\eqref{eq:IID_decoupling_global_cutoff_fixed_block_rank_rate}
and
\eqref{eq:IID_decoupling_global_cutoff_common_block_rate},
we obtain
\begin{align}
\lim_{n\to\infty}
\frac{1}{n}
\log
\left|
M_{n,m,\eta,\xi}
\right|
&=
\frac{1}{m}
\max_{
\substack{
t\in[0,1]\\
\left|
t-\delta_{\widehat{\Pi}_m}
\right|
\le
\eta
}
}
\Biggl\{
h_2(t)
+
(1-t)
\log
\left|
A_{\widehat{\Pi}_m}^m
\right|
+
t
\gamma_{m,\xi}^{\mathrm{rare}}
\Biggr\}
\nonumber\\
&\quad-
\frac{
1-\delta_{\widehat{\Pi}_m}-\eta
}{
m
}
\Biggl[
\log
\left|
A_{\widehat{\Pi}_m}^m
\right|
-
\log
\left|
M_m
\right|
\Biggr].
\end{align}
The finitely many exceptional definitions do not affect this limit.
This proves
\eqref{eq:IID_decoupling_global_cutoff_fixed_partial_trace_rate}.

\medskip
\noindent
\textbf{Proof of
\eqref{eq:IID_decoupling_global_cutoff_common_only_rank_rate}
and
\eqref{eq:IID_decoupling_global_cutoff_common_only_partial_trace_rate}.}
Suppose that
\(\delta_{\widehat{\Pi}_m}=0\).
By
\eqref{eq:IID_decoupling_global_cutoff_projection_common_only},
\begin{align}
\left|
A_{\Pi_{n,m,0,0}}^n
\right|
=
\left|
A_{\widehat{\Pi}_m}^m
\right|^{b_{n,m}}
\operatorname{rank}
\Pi_{n,m}^{\mathrm{res},A^{c_{n,m}}}.
\label{eq:IID_decoupling_global_cutoff_common_only_input_dimension}
\end{align}
Therefore,
\eqref{eq:IID_decoupling_global_cutoff_residual_rank_rate}
and
\eqref{eq:decoupling_IID_block_quotient_properties}
give
\begin{align}
\lim_{n\to\infty}
\frac{1}{n}
\log
\left|
A_{\Pi_{n,m,0,0}}^n
\right|
=
\frac{1}{m}
\log
\left|
A_{\widehat{\Pi}_m}^m
\right|,
\end{align}
which proves
\eqref{eq:IID_decoupling_global_cutoff_common_only_rank_rate}.

Moreover,
\eqref{eq:IID_decoupling_global_cutoff_common_only_remainder_space}
and
\eqref{eq:IID_decoupling_global_cutoff_discarded_space}
give
\begin{align}
\left|
M_{n,m,0,0}
\right|
=
\left|
M_m
\right|^{b_{n,m}}
\operatorname{rank}
\Pi_{n,m}^{\mathrm{res},A^{c_{n,m}}}.
\label{eq:IID_decoupling_global_cutoff_common_only_discarded_dimension}
\end{align}
Using again
\eqref{eq:IID_decoupling_global_cutoff_residual_rank_rate}
and
\eqref{eq:decoupling_IID_block_quotient_properties},
we obtain
\begin{align}
\lim_{n\to\infty}
\frac{1}{n}
\log
\left|
M_{n,m,0,0}
\right|
=
\frac{1}{m}
\log
\left|
M_m
\right|.
\end{align}
This proves
\eqref{eq:IID_decoupling_global_cutoff_common_only_partial_trace_rate}.

\medskip
\noindent
\textbf{Proof of
\eqref{eq:IID_decoupling_global_cutoff_projected_dimension_iterated_limit}.}
Consider first an index \(m\) for which
\begin{align}
0
<
\delta_{\widehat{\Pi}_m}
<
1.
\end{align}
The upper typical-rank bound
\eqref{eq:IID_decoupling_global_cutoff_rare_rank_bound}
and
\eqref{eq:IID_decoupling_global_cutoff_rare_rank_exponent}
give
\begin{align}
\gamma_{m,\xi}^{\mathrm{rare}}
\le
H
\left(
A_{\widehat{\Pi}_m^\perp}^m
\right)_{\rho_{\widehat{\Pi}_m^\perp}}
+
\xi.
\label{eq:IID_decoupling_global_cutoff_rare_rank_exponent_upper_proof}
\end{align}

Every sequence in
\(\mathcal T_{k,m,\xi}^{\mathrm{rare}}\)
has probability at most
\begin{align}
\exp
\left[
-k
\left(
H
\left(
A_{\widehat{\Pi}_m^\perp}^m
\right)_{\rho_{\widehat{\Pi}_m^\perp}}
-
\xi
\right)
\right].
\end{align}
Therefore,
\begin{align}
&\operatorname{Tr}
\left[
\Pi_{k,m,\xi}^{\mathrm{rare}}
\left(
\rho_{\widehat{\Pi}_m^\perp}^{A_{\widehat{\Pi}_m^\perp}^m}
\right)^{\otimes k}
\right]
\nonumber\\
&\le
\operatorname{rank}
\Pi_{k,m,\xi}^{\mathrm{rare}}
\exp
\left[
-k
\left(
H
\left(
A_{\widehat{\Pi}_m^\perp}^m
\right)_{\rho_{\widehat{\Pi}_m^\perp}}
-
\xi
\right)
\right].
\end{align}
For all sufficiently large \(k\), the trace on the left-hand side is
positive, and hence
\begin{align}
\frac{1}{k}
\log
\operatorname{rank}
\Pi_{k,m,\xi}^{\mathrm{rare}}
&\ge
H
\left(
A_{\widehat{\Pi}_m^\perp}^m
\right)_{\rho_{\widehat{\Pi}_m^\perp}}
-
\xi
\nonumber\\
&\quad+
\frac{1}{k}
\log
\operatorname{Tr}
\left[
\Pi_{k,m,\xi}^{\mathrm{rare}}
\left(
\rho_{\widehat{\Pi}_m^\perp}^{A_{\widehat{\Pi}_m^\perp}^m}
\right)^{\otimes k}
\right].
\end{align}
By
\eqref{eq:IID_decoupling_global_cutoff_rare_probability_limit},
\begin{align}
\lim_{k\to\infty}
\operatorname{Tr}
\left[
\Pi_{k,m,\xi}^{\mathrm{rare}}
\left(
\rho_{\widehat{\Pi}_m^\perp}^{A_{\widehat{\Pi}_m^\perp}^m}
\right)^{\otimes k}
\right]
=
1,
\end{align}
and therefore
\begin{align}
\lim_{k\to\infty}
\frac{1}{k}
\log
\operatorname{Tr}
\left[
\Pi_{k,m,\xi}^{\mathrm{rare}}
\left(
\rho_{\widehat{\Pi}_m^\perp}^{A_{\widehat{\Pi}_m^\perp}^m}
\right)^{\otimes k}
\right]
=
0.
\end{align}
Using
\eqref{eq:IID_decoupling_global_cutoff_rare_rank_exponent},
we obtain
\begin{align}
\gamma_{m,\xi}^{\mathrm{rare}}
\ge
H
\left(
A_{\widehat{\Pi}_m^\perp}^m
\right)_{\rho_{\widehat{\Pi}_m^\perp}}
-
\xi.
\label{eq:IID_decoupling_global_cutoff_rare_rank_exponent_lower_proof}
\end{align}

Applying
\eqref{eq:IID_decoupling_global_cutoff_fixed_block_rank_rate}
with
\(\eta=\eta_m\)
and
\(\xi=\xi_m\),
and evaluating the maximum at the admissible point
\(t=\delta_{\widehat{\Pi}_m}\), gives
\begin{align}
&\lim_{n\to\infty}
\frac{1}{n}
\log
\left|
A_{\Pi_{n,m,\eta_m,\xi_m}}^n
\right|
\nonumber\\
&\ge
\frac{1}{m}
\Biggl[
h_2
\left(
\delta_{\widehat{\Pi}_m}
\right)
+
\left(
1-\delta_{\widehat{\Pi}_m}
\right)
\log
\left|
A_{\widehat{\Pi}_m}^m
\right|
+
\delta_{\widehat{\Pi}_m}
\gamma_{m,\xi_m}^{\mathrm{rare}}
\Biggr].
\end{align}
Since
\begin{align}
H
\left(
A_{\widehat{\Pi}_m}^m
\right)_{\rho_{\widehat{\Pi}_m}}
\le
\log
\left|
A_{\widehat{\Pi}_m}^m
\right|,
\end{align}
and
\eqref{eq:IID_decoupling_global_cutoff_rare_rank_exponent_lower_proof}
gives
\begin{align}
&\lim_{n\to\infty}
\frac{1}{n}
\log
\left|
A_{\Pi_{n,m,\eta_m,\xi_m}}^n
\right|
\nonumber\\
&\ge
\frac{1}{m}
\Biggl[
h_2
\left(
\delta_{\widehat{\Pi}_m}
\right)
+
\left(
1-\delta_{\widehat{\Pi}_m}
\right)
H
\left(
A_{\widehat{\Pi}_m}^m
\right)_{\rho_{\widehat{\Pi}_m}}
+
\delta_{\widehat{\Pi}_m}
H
\left(
A_{\widehat{\Pi}_m^\perp}^m
\right)_{\rho_{\widehat{\Pi}_m^\perp}}
\nonumber\\
&\qquad-
\delta_{\widehat{\Pi}_m}\xi_m
\Biggr].
\end{align}
By
\eqref{eq:decoupling_IID_marginal_entropy_branch_decomposition},
\begin{align}
&mH(A)_\rho
\nonumber\\
&=
h_2
\left(
\delta_{\widehat{\Pi}_m}
\right)
+
\left(
1-\delta_{\widehat{\Pi}_m}
\right)
H
\left(
A_{\widehat{\Pi}_m}^m
\right)_{\rho_{\widehat{\Pi}_m}}
+
\delta_{\widehat{\Pi}_m}
H
\left(
A_{\widehat{\Pi}_m^\perp}^m
\right)_{\rho_{\widehat{\Pi}_m^\perp}}.
\end{align}
Consequently,
\begin{align}
&\lim_{n\to\infty}
\frac{1}{n}
\log
\left|
A_{\Pi_{n,m,\eta_m,\xi_m}}^n
\right|
\nonumber\\
&\ge
H(A)_\rho
-
\frac{
\delta_{\widehat{\Pi}_m}\xi_m
}{
m
}.
\label{eq:IID_decoupling_global_cutoff_iterated_dimension_lower_proof}
\end{align}

For the converse bound,
\eqref{eq:IID_decoupling_global_cutoff_fixed_block_rank_rate}
and
\eqref{eq:IID_decoupling_global_cutoff_rare_rank_exponent_upper_proof}
give
\begin{align}
&\lim_{n\to\infty}
\frac{1}{n}
\log
\left|
A_{\Pi_{n,m,\eta_m,\xi_m}}^n
\right|
\nonumber\\
&\le
\frac{\log 2}{m}
+
\frac{
1-\delta_{\widehat{\Pi}_m}+\eta_m
}{
m
}
\log
\left|
A_{\widehat{\Pi}_m}^m
\right|
\nonumber\\
&\quad+
\frac{
\delta_{\widehat{\Pi}_m}+\eta_m
}{
m
}
\left[
H
\left(
A_{\widehat{\Pi}_m^\perp}^m
\right)_{\rho_{\widehat{\Pi}_m^\perp}}
+
\xi_m
\right].
\label{eq:IID_decoupling_global_cutoff_iterated_dimension_upper_proof}
\end{align}
Here, for every admissible \(t\), we used
\begin{align}
h_2(t)
&\le
\log 2,
\\
1-t
&\le
1-\delta_{\widehat{\Pi}_m}+\eta_m,
\\
t
&\le
\delta_{\widehat{\Pi}_m}+\eta_m.
\end{align}

By
\eqref{eq:IID_decoupling_global_cutoff_admissible_parameter_sequences},
\begin{align}
0
<
\eta_m
<
\min
\left\{
\delta_{\widehat{\Pi}_m},
1-\delta_{\widehat{\Pi}_m}
\right\}.
\end{align}
It follows that
\begin{align}
\frac{
\delta_{\widehat{\Pi}_m}+\eta_m
}{
m
}
H
\left(
A_{\widehat{\Pi}_m^\perp}^m
\right)_{\rho_{\widehat{\Pi}_m^\perp}}
\le
2
\frac{
\delta_{\widehat{\Pi}_m}
}{
m
}
H
\left(
A_{\widehat{\Pi}_m^\perp}^m
\right)_{\rho_{\widehat{\Pi}_m^\perp}},
\label{eq:IID_decoupling_global_cutoff_rare_entropy_upper_control}
\end{align}
and
\begin{align}
0
\le
\frac{
\delta_{\widehat{\Pi}_m}+\eta_m
}{
m
}
\xi_m
\le
\frac{\xi_m}{m}.
\label{eq:IID_decoupling_global_cutoff_rare_typicality_upper_control}
\end{align}
Moreover,
\(\eta_m<\delta_{\widehat{\Pi}_m}\)
implies
\begin{align}
1-\delta_{\widehat{\Pi}_m}+\eta_m
<
1.
\end{align}
Thus,
\begin{align}
&\lim_{n\to\infty}
\frac{1}{n}
\log
\left|
A_{\Pi_{n,m,\eta_m,\xi_m}}^n
\right|
\nonumber\\
&\le
\frac{1}{m}
\log
\left|
A_{\widehat{\Pi}_m}^m
\right|
+
\frac{\log 2}{m}
+
2
\frac{
\delta_{\widehat{\Pi}_m}
}{
m
}
H
\left(
A_{\widehat{\Pi}_m^\perp}^m
\right)_{\rho_{\widehat{\Pi}_m^\perp}}
+
\frac{\xi_m}{m}.
\label{eq:IID_decoupling_global_cutoff_positive_rare_dimension_simplified_upper_bound}
\end{align}

Consider next an index \(m\) for which
\(\delta_{\widehat{\Pi}_m}=0\).
Then
\(\eta_m=\xi_m=0\), and
\eqref{eq:IID_decoupling_global_cutoff_common_only_rank_rate}
gives
\begin{align}
\lim_{n\to\infty}
\frac{1}{n}
\log
\left|
A_{\Pi_{n,m,0,0}}^n
\right|
=
\frac{1}{m}
\log
\left|
A_{\widehat{\Pi}_m}^m
\right|.
\label{eq:IID_decoupling_global_cutoff_common_only_dimension_identity}
\end{align}
Since
\((\rho^A)^{\otimes m}\)
is supported on
\(\mathcal H^{A_{\widehat{\Pi}_m}^m}\),
\begin{align}
mH(A)_\rho
=
H(A^m)_{\rho^{\otimes m}}
\le
\log
\left|
A_{\widehat{\Pi}_m}^m
\right|.
\end{align}
Hence,
\begin{align}
\lim_{n\to\infty}
\frac{1}{n}
\log
\left|
A_{\Pi_{n,m,0,0}}^n
\right|
\ge
H(A)_\rho.
\label{eq:IID_decoupling_global_cutoff_iterated_dimension_lower_proof_common_only}
\end{align}

Consequently, for every \(m\),
\begin{align}
&\lim_{n\to\infty}
\frac{1}{n}
\log
\left|
A_{\Pi_{n,m,\eta_m,\xi_m}}^n
\right|
\nonumber\\
&\ge
H(A)_\rho
-
\frac{
\delta_{\widehat{\Pi}_m}\xi_m
}{
m
}.
\label{eq:IID_decoupling_global_cutoff_unified_dimension_lower_bound}
\end{align}
Similarly, for every \(m\),
\begin{align}
&\lim_{n\to\infty}
\frac{1}{n}
\log
\left|
A_{\Pi_{n,m,\eta_m,\xi_m}}^n
\right|
\nonumber\\
&\le
\frac{1}{m}
\log
\left|
A_{\widehat{\Pi}_m}^m
\right|
+
\begin{cases}
\displaystyle
\frac{\log 2}{m}
+
2
\frac{
\delta_{\widehat{\Pi}_m}
}{
m
}
H
\left(
A_{\widehat{\Pi}_m^\perp}^m
\right)_{\rho_{\widehat{\Pi}_m^\perp}}
+
\frac{\xi_m}{m},
&
0<\delta_{\widehat{\Pi}_m}<1,
\\[4mm]
0,
&
\delta_{\widehat{\Pi}_m}=0.
\end{cases}
\label{eq:IID_decoupling_global_cutoff_unified_dimension_upper_bound}
\end{align}

Proposition~\ref{prop:decoupling_IID_finite_entropy_spectral_cutoffs}
gives
\begin{align}
\lim_{m\to\infty}
\frac{1}{m}
\log
\left|
A_{\widehat{\Pi}_m}^m
\right|
&=
H(A)_\rho,
\\
\lim_{m\to\infty}
\begin{cases}
\displaystyle
\frac{
\delta_{\widehat{\Pi}_m}
}{
m
}
H
\left(
A_{\widehat{\Pi}_m^\perp}^m
\right)_{\rho_{\widehat{\Pi}_m^\perp}},
&
0<\delta_{\widehat{\Pi}_m}<1,
\\[4mm]
0,
&
\delta_{\widehat{\Pi}_m}=0
\end{cases}
&=
0.
\end{align}
Moreover,
\eqref{eq:IID_decoupling_global_cutoff_parameter_sequence_limits}
gives
\begin{align}
\lim_{m\to\infty}
\frac{
\delta_{\widehat{\Pi}_m}\xi_m
}{
m
}
&=
0,
&
\lim_{m\to\infty}
\frac{\xi_m}{m}
&=
0,
\end{align}
and
\begin{align}
\lim_{m\to\infty}
\frac{\log 2}{m}
=
0.
\end{align}
It follows that
\begin{align}
&\lim_{m\to\infty}
\begin{cases}
\displaystyle
\frac{\log 2}{m}
+
2
\frac{
\delta_{\widehat{\Pi}_m}
}{
m
}
H
\left(
A_{\widehat{\Pi}_m^\perp}^m
\right)_{\rho_{\widehat{\Pi}_m^\perp}}
+
\frac{\xi_m}{m},
&
0<\delta_{\widehat{\Pi}_m}<1,
\\[4mm]
0,
&
\delta_{\widehat{\Pi}_m}=0
\end{cases}
\nonumber\\
&=
0.
\end{align}
Therefore,
\eqref{eq:IID_decoupling_global_cutoff_unified_dimension_lower_bound}
and
\eqref{eq:IID_decoupling_global_cutoff_unified_dimension_upper_bound}
yield
\begin{align}
\lim_{m\to\infty}
\lim_{n\to\infty}
\frac{1}{n}
\log
\left|
A_{\Pi_{n,m,\eta_m,\xi_m}}^n
\right|
=
H(A)_\rho,
\end{align}
which proves
\eqref{eq:IID_decoupling_global_cutoff_projected_dimension_iterated_limit}.

\medskip
\noindent
\textbf{Proof of
\eqref{eq:IID_decoupling_global_cutoff_prescribed_partial_trace_rate}.}
For every \(m\) such that
\(0<\delta_{\widehat{\Pi}_m}<1\),
equations
\eqref{eq:IID_decoupling_global_cutoff_fixed_partial_trace_rate}
and
\eqref{eq:IID_decoupling_global_cutoff_fixed_block_rank_rate},
with
\(\eta=\eta_m\)
and
\(\xi=\xi_m\), give
\begin{align}
&\lim_{n\to\infty}
\frac{1}{n}
\log
\left|
M_{n,m,\eta_m,\xi_m}
\right|
\nonumber\\
&=
\lim_{n\to\infty}
\frac{1}{n}
\log
\left|
A_{\Pi_{n,m,\eta_m,\xi_m}}^n
\right|
\nonumber\\
&\quad-
\frac{
1-\delta_{\widehat{\Pi}_m}-\eta_m
}{
m
}
\Biggl[
\log
\left|
A_{\widehat{\Pi}_m}^m
\right|
-
\log
\left|
M_m
\right|
\Biggr].
\label{eq:IID_decoupling_global_cutoff_positive_rare_discarded_rate_identity}
\end{align}
For every \(m\) such that
\(\delta_{\widehat{\Pi}_m}=0\),
one has
\(\eta_m=\xi_m=0\), and
\eqref{eq:IID_decoupling_global_cutoff_common_only_rank_rate}
and
\eqref{eq:IID_decoupling_global_cutoff_common_only_partial_trace_rate}
give
\begin{align}
&\lim_{n\to\infty}
\frac{1}{n}
\log
\left|
M_{n,m,0,0}
\right|
\nonumber\\
&=
\lim_{n\to\infty}
\frac{1}{n}
\log
\left|
A_{\Pi_{n,m,0,0}}^n
\right|
\nonumber\\
&\quad-
\frac{1}{m}
\Biggl[
\log
\left|
A_{\widehat{\Pi}_m}^m
\right|
-
\log
\left|
M_m
\right|
\Biggr].
\label{eq:IID_decoupling_global_cutoff_common_only_discarded_rate_identity}
\end{align}
Thus, for every \(m\),
\begin{align}
&\lim_{n\to\infty}
\frac{1}{n}
\log
\left|
M_{n,m,\eta_m,\xi_m}
\right|
\nonumber\\
&=
\lim_{n\to\infty}
\frac{1}{n}
\log
\left|
A_{\Pi_{n,m,\eta_m,\xi_m}}^n
\right|
\nonumber\\
&\quad-
\left(
1-\delta_{\widehat{\Pi}_m}-\eta_m
\right)
\Biggl[
\frac{1}{m}
\log
\left|
A_{\widehat{\Pi}_m}^m
\right|
-
\frac{1}{m}
\log
\left|
M_m
\right|
\Biggr].
\label{eq:IID_decoupling_global_cutoff_unified_discarded_rate_identity}
\end{align}

By
\eqref{eq:IID_decoupling_global_cutoff_projected_dimension_iterated_limit},
\begin{align}
\lim_{m\to\infty}
\lim_{n\to\infty}
\frac{1}{n}
\log
\left|
A_{\Pi_{n,m,\eta_m,\xi_m}}^n
\right|
=
H(A)_\rho.
\end{align}
Moreover,
Proposition~\ref{prop:decoupling_IID_finite_entropy_spectral_cutoffs}
and
\eqref{eq:IID_decoupling_global_cutoff_parameter_sequence_limits}
give
\begin{align}
\lim_{m\to\infty}
\left(
1-\delta_{\widehat{\Pi}_m}-\eta_m
\right)
&=
1,
\\
\lim_{m\to\infty}
\frac{1}{m}
\log
\left|
A_{\widehat{\Pi}_m}^m
\right|
&=
H(A)_\rho,
\\
\lim_{m\to\infty}
\frac{1}{m}
\log
\left|
M_m
\right|
&=
q.
\end{align}
Consequently,
\begin{align}
&\lim_{m\to\infty}
\left(
1-\delta_{\widehat{\Pi}_m}-\eta_m
\right)
\Biggl[
\frac{1}{m}
\log
\left|
A_{\widehat{\Pi}_m}^m
\right|
-
\frac{1}{m}
\log
\left|
M_m
\right|
\Biggr]
\nonumber\\
&=
H(A)_\rho-q.
\end{align}
Taking the limit \(m\to\infty\) in
\eqref{eq:IID_decoupling_global_cutoff_unified_discarded_rate_identity}
therefore gives
\begin{align}
\lim_{m\to\infty}
\lim_{n\to\infty}
\frac{1}{n}
\log
\left|
M_{n,m,\eta_m,\xi_m}
\right|
&=
H(A)_\rho
-
\left(
H(A)_\rho-q
\right)
\nonumber\\
&=
q.
\end{align}
This proves
\eqref{eq:IID_decoupling_global_cutoff_prescribed_partial_trace_rate}.
\end{proof}

\subsection{Bounds on infinite-dimensional IID decoupling}
\label{subsec:infinite_dimensional_IID_decoupling_achievability}

We construct the infinite-dimensional IID decoupling protocol by
applying the finite-dimensional one-shot decoupling theorem in
Section~\ref{sec:finite_input_infinite_reference_one_shot_decoupling}
to the fixed-count common-block subsystem constructed in
Subsection~\ref{subsec:high_probability_finite_rank_IID_decoupling},
while including all remaining components in the discarded system.
These remaining components contribute to the discarded-system dimension
for each fixed block length \(m\); however, in
Theorem~\ref{thm:IID_decoupling_high_probability_global_cutoffs},
we will control this contribution to show that it does not alter the target first-order rate after taking the limit \(m\to\infty\).

Fix
\begin{align}
q
\in
\left[
0,
H(A)_\rho
\right],
\end{align}
and choose
\begin{align}
\left\{
\widehat{\Pi}_m^{A^m},
E_m,
M_m,
V_m^{A_{\widehat{\Pi}_m}^m\to E_mM_m}
\right\}_{m=1}^{\infty}
\end{align}
as in
Proposition~\ref{prop:decoupling_IID_finite_entropy_spectral_cutoffs}
for this value of \(q\).
In particular,
\(\widehat{\Pi}_m^{A^m}\),
\(\delta_{\widehat{\Pi}_m}\),
\(\mathcal H^{A_{\widehat{\Pi}_m}^m}\), and
\(\rho_{\widehat{\Pi}_m}^{A_{\widehat{\Pi}_m}^mR^m}\)
are defined in
\eqref{eq:decoupling_IID_spectral_cutoff_projection},
\eqref{eq:decoupling_IID_spectral_cutoff_error},
\eqref{eq:decoupling_IID_projected_input_space}, and
\eqref{eq:decoupling_IID_normalized_projected_block_state},
respectively, and
\begin{align}
V_m^{A_{\widehat{\Pi}_m}^m\to E_mM_m}
:
\mathcal H^{A_{\widehat{\Pi}_m}^m}
\longrightarrow
\mathcal H^{E_m}
\otimes
\mathcal H^{M_m}
\end{align}
is the unitary isomorphism in
\eqref{eq:IID_decoupling_prescribed_common_factorization}.

Fix
\begin{align}
m
\in
\{1,2,\ldots\}.
\end{align}
If
\begin{align}
0
<
\delta_{\widehat{\Pi}_m}
<
1,
\end{align}
fix
\begin{align}
0
<
\eta
<
\min
\left\{
\delta_{\widehat{\Pi}_m},
1-\delta_{\widehat{\Pi}_m}
\right\},
\qquad
\xi
>
0.
\label{eq:IID_decoupling_achievability_positive_rare_parameters}
\end{align}
If
\begin{align}
\delta_{\widehat{\Pi}_m}
=
0,
\end{align}
set
\begin{align}
\eta
=
\xi
=
0,
\end{align}
as in
\eqref{eq:eta_xi_zero},
and use the common-only definitions in
\eqref{eq:IID_decoupling_common_only_pattern_set}--%
\eqref{eq:IID_decoupling_global_cutoff_common_only_remainder_space}.

Let
\(\Pi_{n,m,\eta,\xi}^{A^n}\),
\(\delta_{\Pi_{n,m,\eta,\xi}}\),
\(\mathcal H^{A_{\Pi_{n,m,\eta,\xi}}^n}\),
\(\rho_{\Pi_{n,m,\eta,\xi}}^{
A_{\Pi_{n,m,\eta,\xi}}^nR^n}\),
and
\(\rho_{\Pi_{n,m,\eta,\xi}}^{R^n}\)
be defined in
\eqref{eq:IID_decoupling_global_cutoff_projection_fixed_block},
\eqref{eq:IID_decoupling_global_cutoff_projection_common_only},
\eqref{eq:IID_decoupling_global_cutoff_error},
\eqref{eq:IID_decoupling_global_cutoff_input_space},
\eqref{eq:IID_decoupling_global_cutoff_normalized_state}, and
\eqref{eq:IID_decoupling_global_cutoff_reference_marginal},
respectively.
Let
\(\ell_{n,m,\eta}\)
be defined in
\eqref{eq:IID_decoupling_global_cutoff_common_block_number}
when
\(0<\delta_{\widehat{\Pi}_m}<1\),
and in
\eqref{eq:IID_decoupling_common_only_fixed_count}
when
\(\delta_{\widehat{\Pi}_m}=0\).
Similarly, let
\(\mathcal H^{B_{n,m,\eta,\xi}}\)
be the remainder space defined in
\eqref{eq:IID_decoupling_global_cutoff_remainder_space}
when
\(0<\delta_{\widehat{\Pi}_m}<1\),
and in
\eqref{eq:IID_decoupling_global_cutoff_common_only_remainder_space}
when
\(\delta_{\widehat{\Pi}_m}=0\).

For all sufficiently large \(n\) in the case
\(0<\delta_{\widehat{\Pi}_m}<1\), and for every \(n\ge m\) in the case
\(\delta_{\widehat{\Pi}_m}=0\), the unitary isomorphism
\begin{align}
V_{n,m,\eta,\xi}^{\mathrm{fix}}
:
\mathcal H^{A_{\Pi_{n,m,\eta,\xi}}^n}
\longrightarrow
\left(
\mathcal H^{A_{\widehat{\Pi}_m}^m}
\right)^{\otimes\ell_{n,m,\eta}}
\otimes
\mathcal H^{B_{n,m,\eta,\xi}}
\end{align}
defined in
\eqref{eq:IID_decoupling_global_cutoff_fixed_factorization}
identifies the globally projected input space as
\begin{align}
\mathcal H^{A_{\Pi_{n,m,\eta,\xi}}^n}
\simeq
\left(
\mathcal H^{A_{\widehat{\Pi}_m}^m}
\right)^{\otimes\ell_{n,m,\eta}}
\otimes
\mathcal H^{B_{n,m,\eta,\xi}}.
\label{eq:IID_decoupling_achievability_fixed_common_factor}
\end{align}
The first tensor factor in
\eqref{eq:IID_decoupling_achievability_fixed_common_factor}
is the fixed-count common tensor factor to which the
finite-dimensional one-shot decoupling theorem will be applied.
The remainder system
\(B_{n,m,\eta,\xi}\)
contains the pattern-sector label, the remaining common blocks, the
rare weakly typical component, and the residual component.

The finite-dimensional output and discarded systems
\(E_{n,m,\eta,\xi}\)
and
\(M_{n,m,\eta,\xi}\)
are defined in
\eqref{eq:IID_decoupling_global_cutoff_output_space}
and
\eqref{eq:IID_decoupling_global_cutoff_discarded_space},
respectively.
Thus,
\begin{align}
\mathcal H^{E_{n,m,\eta,\xi}}
&=
\left(
\mathcal H^{E_m}
\right)^{\otimes\ell_{n,m,\eta}},
\label{eq:IID_decoupling_achievability_output_space}
\\
\mathcal H^{M_{n,m,\eta,\xi}}
&=
\left(
\mathcal H^{M_m}
\right)^{\otimes\ell_{n,m,\eta}}
\otimes
\mathcal H^{B_{n,m,\eta,\xi}}.
\label{eq:IID_decoupling_achievability_discarded_space}
\end{align}
In particular, the entire remainder system
\(B_{n,m,\eta,\xi}\)
is included in the discarded system
\(M_{n,m,\eta,\xi}\).

The unitary isomorphism
\begin{align}
V_{n,m,\eta,\xi}^{
A_{\Pi_{n,m,\eta,\xi}}^n
\to
E_{n,m,\eta,\xi}M_{n,m,\eta,\xi}
}
:
\mathcal H^{A_{\Pi_{n,m,\eta,\xi}}^n}
\longrightarrow
\mathcal H^{E_{n,m,\eta,\xi}}
\otimes
\mathcal H^{M_{n,m,\eta,\xi}}
\end{align}
is defined in
\eqref{eq:IID_decoupling_global_cutoff_partial_trace_factorization}
by applying
\(\left(V_m^{A_{\widehat{\Pi}_m}^m\to E_mM_m}\right)^{
\otimes\ell_{n,m,\eta}}\)
to the fixed-count common tensor factor in
\eqref{eq:IID_decoupling_achievability_fixed_common_factor}
and acting as the identity on
\(B_{n,m,\eta,\xi}\).
The corresponding partial-trace channel
\begin{align}
\mathcal T_{n,m,\eta,\xi}^{
A_{\Pi_{n,m,\eta,\xi}}^n
\to
E_{n,m,\eta,\xi}
}
\end{align}
is defined in
\eqref{eq:IID_decoupling_global_cutoff_partial_trace_channel}.

For
\begin{align}
U
\in
\operatorname{U}
\left(
\left(
\mathcal H^{A_{\widehat{\Pi}_m}^m}
\right)^{\otimes\ell_{n,m,\eta}}
\right),
\end{align}
define the induced unitary on the globally projected input space by
\begin{align}
U_{n,m,\eta,\xi}[U]
\coloneqq
\left(
V_{n,m,\eta,\xi}^{\mathrm{fix}}
\right)^\dagger
\left(
U
\otimes
I^{B_{n,m,\eta,\xi}}
\right)
V_{n,m,\eta,\xi}^{\mathrm{fix}}.
\label{eq:IID_decoupling_achievability_induced_unitary}
\end{align}
Thus, the same unitary \(U\) acts on the fixed-count common tensor
factor in every typical common--rare pattern sector, while the
remainder system \(B_{n,m,\eta,\xi}\) is left unchanged.

As defined in \eqref{eq:IID_decoupling_output_state}, the resulting output state, under the identifications
\begin{align}
\Pi_n^{A^n}
&=
\Pi_{n,m,\eta,\xi}^{A^n},
&
V_n
&=
V_{n,m,\eta,\xi},
\nonumber\\
\mathcal T_n
&=
\mathcal T_{n,m,\eta,\xi},
&
U_n
&=
U_{n,m,\eta,\xi}[U],
\label{eq:IID_decoupling_achievability_protocol_identification}
\end{align}
is
\begin{align}
&\tau_{
n,U_{n,m,\eta,\xi}[U]
}^{
E_{n,m,\eta,\xi}R^n
}
\nonumber\\
&\coloneqq
\left(
\mathcal T_{n,m,\eta,\xi}
\otimes
\id^{R^n}
\right)
\Biggl[
\left(
U_{n,m,\eta,\xi}[U]
\otimes
I^{R^n}
\right)
\rho_{\Pi_{n,m,\eta,\xi}}^{
A_{\Pi_{n,m,\eta,\xi}}^nR^n
}
\left(
U_{n,m,\eta,\xi}[U]^\dagger
\otimes
I^{R^n}
\right)
\Biggr].
\label{eq:IID_decoupling_achievability_output_state}
\end{align}
By
\eqref{eq:IID_decoupling_output_reference},
its reference marginal is
\begin{align}
\tau_{
n,U_{n,m,\eta,\xi}[U]
}^{R^n}
=
\rho_{\Pi_{n,m,\eta,\xi}}^{R^n}.
\label{eq:IID_decoupling_achievability_output_reference}
\end{align}

Using the one-shot decoupling theorem for finite-dimensional input and
separable reference systems in
Theorem~\ref{thm:finite_input_infinite_reference_decoupling},
we first show that the relative-entropy decoupling error of this protocol
with respect to
\(\pi^{E_{n,m,\eta,\xi}}
\otimes
\rho_{\Pi_{n,m,\eta,\xi}}^{R^n}\)
decays exponentially.

\begin{proposition}[Exponential decoupling in the projected subspace]
\label{prop:IID_decoupling_achievability_fixed_common_factor}
Let
\(\mathcal H^A\)
and
\(\mathcal H^R\)
be separable and possibly infinite-dimensional, and let
\begin{align}
\rho^{AR}
\in
\operatorname{D}
\left(
\mathcal H^A
\otimes
\mathcal H^R
\right)
\end{align}
satisfy
\begin{align}
H(A)_\rho
<
\infty.
\end{align}
Fix
\begin{align}
q
\in
\left[
0,
H(A)_\rho
\right],
\end{align}
and choose
\begin{align}
\left\{
\widehat{\Pi}_m^{A^m},
E_m,
M_m,
V_m^{A_{\widehat{\Pi}_m}^m\to E_mM_m}
\right\}_{m=1}^{\infty}
\end{align}
as in
Proposition~\ref{prop:decoupling_IID_finite_entropy_spectral_cutoffs}
for this value of \(q\).

Fix
\begin{align}
m
\in
\{1,2,\ldots\}.
\end{align}
If
\(0<\delta_{\widehat{\Pi}_m}<1\),
fix \(\eta\) and \(\xi\) satisfying
\eqref{eq:IID_decoupling_achievability_positive_rare_parameters}; if
\(\delta_{\widehat{\Pi}_m}=0\),
set
\(\eta
=
\xi
=
0\) as in \eqref{eq:eta_xi_zero}.
Let
\(\ell_{n,m,\eta}\),
\(E_{n,m,\eta,\xi}\), and
\(\rho_{\Pi_{n,m,\eta,\xi}}^{R^n}\)
be defined in
\eqref{eq:IID_decoupling_global_cutoff_common_block_number},
\eqref{eq:IID_decoupling_common_only_fixed_count},
\eqref{eq:IID_decoupling_global_cutoff_output_space}, and
\eqref{eq:IID_decoupling_global_cutoff_reference_marginal},
as applicable.

Suppose that
\begin{align}
\sup_{0<s\le 1}
\widetilde H_{1+s}
\left(
A_{\widehat{\Pi}_m}^m
\middle|
R^m
\right)_{\rho_{\widehat{\Pi}_m}}
+
2\log
\left|
M_m
\right|
-
\log
\left|
A_{\widehat{\Pi}_m}^m
\right|
>
0.
\label{eq:IID_decoupling_achievability_positive_Renyi_gap}
\end{align}
Then, for all sufficiently large \(n\),
\begin{align}
&\mathbb E_{
U
\sim
\mathrm{Haar}
\left(
\operatorname{U}
\left(
\left(
\mathcal H^{A_{\widehat{\Pi}_m}^m}
\right)^{\otimes\ell_{n,m,\eta}}
\right)
\right)
}
\left[
D
\left(
\tau_{
n,U_{n,m,\eta,\xi}[U]
}^{
E_{n,m,\eta,\xi}R^n
}
\middle\|
\pi^{E_{n,m,\eta,\xi}}
\otimes
\rho_{\Pi_{n,m,\eta,\xi}}^{R^n}
\right)
\right]
\nonumber\\
&\le
\inf_{0<s\le 1}
C_{
\left|
A_{\widehat{\Pi}_m}^m
\right|^{\ell_{n,m,\eta}}
}
G_s
\exp
\Biggl[
-s
\ell_{n,m,\eta}
\Biggl(
\widetilde H_{1+s}
\left(
A_{\widehat{\Pi}_m}^m
\middle|
R^m
\right)_{\rho_{\widehat{\Pi}_m}}
+
2\log
\left|
M_m
\right|
-
\log
\left|
A_{\widehat{\Pi}_m}^m
\right|
\Biggr)
\Biggr],
\label{eq:IID_decoupling_achievability_projected_error_exponential_bound}
\end{align}
where
\(U_{n,m,\eta,\xi}[U]\)
and
\(\tau_{n,U_{n,m,\eta,\xi}[U]}^{
E_{n,m,\eta,\xi}R^n}\)
are defined in
\eqref{eq:IID_decoupling_achievability_induced_unitary}
and
\eqref{eq:IID_decoupling_achievability_output_state},
respectively, while
\(C_{|A|}\),
\(G_s\), and
\(\widetilde H_{1+s}(A|R)_\rho\)
are defined in
\eqref{eq:decoupling_dimension_dependent_constant},
\eqref{eq:decoupling_Gs_definition}, and
\eqref{eq:sandwiched_conditional_entropy_downarrow},
respectively.
\end{proposition}

\begin{proof}
Fix \(n\) sufficiently large that the factorization in
\eqref{eq:IID_decoupling_achievability_fixed_common_factor}
is available.

Consider first the partial-trace channel from
\(\mathcal H^{A_{\widehat{\Pi}_m}^m}\)
to
\(\mathcal H^{E_m}\)
induced by the unitary isomorphism
\(V_m^{A_{\widehat{\Pi}_m}^m\to E_mM_m}\).
Denote its normalized Choi state by
\begin{align}
\omega_m^{A_{\widehat{\Pi}_m}^{m\prime}E_m},
\end{align}
where
\(A_{\widehat{\Pi}_m}^{m\prime}\)
is an isomorphic copy of
\(A_{\widehat{\Pi}_m}^m\).
Under the corresponding identification
\begin{align}
\mathcal H^{A_{\widehat{\Pi}_m}^{m\prime}}
\simeq
\mathcal H^{E_m'}
\otimes
\mathcal H^{M_m'},
\end{align}
one has, up to the canonical ordering of tensor factors,
\begin{align}
\omega_m^{A_{\widehat{\Pi}_m}^{m\prime}E_m}
&=
\pi^{M_m'}
\otimes
\Phi^{E_m'E_m},
\label{eq:IID_decoupling_achievability_block_Choi_state}
\\
\omega_m^{E_m}
&=
\pi^{E_m}.
\label{eq:IID_decoupling_achievability_block_Choi_output}
\end{align}
It follows from
\eqref{eq:IID_decoupling_achievability_block_Choi_state},
\eqref{eq:sandwiched_conditional_entropy_downarrow}, and
\eqref{eq:IID_decoupling_prescribed_common_factorization}
that, for every \(0<s\le 1\),
\begin{align}
\widetilde H_{1+s}
\left(
A_{\widehat{\Pi}_m}^{m\prime}
\middle|
E_m
\right)_{\omega_m}
&=
\log
\left|
M_m
\right|
-
\log
\left|
E_m
\right|
\nonumber\\
&=
2\log
\left|
M_m
\right|
-
\log
\left|
A_{\widehat{\Pi}_m}^m
\right|.
\label{eq:IID_decoupling_achievability_partial_trace_Choi_entropy}
\end{align}

The normalized Choi state of the
\(\ell_{n,m,\eta}\)-fold partial-trace channel induced by
\begin{align}
\left(
V_m^{A_{\widehat{\Pi}_m}^m\to E_mM_m}
\right)^{\otimes\ell_{n,m,\eta}}
\end{align}
is
\(\omega_m^{\otimes\ell_{n,m,\eta}}\).
By additivity of the sandwiched conditional R\'enyi entropy under
tensor products,
\begin{align}
&\widetilde H_{1+s}
\left(
\left(
A_{\widehat{\Pi}_m}^{m\prime}
\right)^{\ell_{n,m,\eta}}
\middle|
E_m^{\ell_{n,m,\eta}}
\right)_{\omega_m^{\otimes\ell_{n,m,\eta}}}
\nonumber\\
&=
\ell_{n,m,\eta}
\left(
2\log
\left|
M_m
\right|
-
\log
\left|
A_{\widehat{\Pi}_m}^m
\right|
\right),
\label{eq:IID_decoupling_achievability_tensor_power_Choi_entropy}
\end{align}
and
\begin{align}
&\widetilde H_{1+s}
\left(
\left(
A_{\widehat{\Pi}_m}^m
\right)^{\ell_{n,m,\eta}}
\middle|
R^{m\ell_{n,m,\eta}}
\right)_{
\rho_{\widehat{\Pi}_m}^{\otimes\ell_{n,m,\eta}}
}
\nonumber\\
&=
\ell_{n,m,\eta}
\widetilde H_{1+s}
\left(
A_{\widehat{\Pi}_m}^m
\middle|
R^m
\right)_{\rho_{\widehat{\Pi}_m}}.
\label{eq:IID_decoupling_achievability_tensor_power_input_entropy}
\end{align}

For
\begin{align}
U
\in
\operatorname{U}
\left(
\left(
\mathcal H^{A_{\widehat{\Pi}_m}^m}
\right)^{\otimes\ell_{n,m,\eta}}
\right),
\end{align}
define the output obtained from the fixed-count common tensor factor
by
\begin{align}
&\widehat{\tau}_{n,m,\eta;U}^{
E_m^{\ell_{n,m,\eta}}
R^{m\ell_{n,m,\eta}}
}
\nonumber\\
&\coloneqq
\operatorname{Tr}_{M_m^{\ell_{n,m,\eta}}}
\Biggl[
\left(
\left(
V_m^{A_{\widehat{\Pi}_m}^m\to E_mM_m}
\right)^{\otimes\ell_{n,m,\eta}}
U
\otimes
I^{R^{m\ell_{n,m,\eta}}}
\right)
\left(
\rho_{\widehat{\Pi}_m}^{A_{\widehat{\Pi}_m}^mR^m}
\right)^{\otimes\ell_{n,m,\eta}}
\nonumber\\
&\hspace{28mm}\times
\left(
U^\dagger
\left(
\left(
V_m^{A_{\widehat{\Pi}_m}^m\to E_mM_m}
\right)^{\otimes\ell_{n,m,\eta}}
\right)^\dagger
\otimes
I^{R^{m\ell_{n,m,\eta}}}
\right)
\Biggr].
\label{eq:IID_decoupling_achievability_common_factor_output}
\end{align}

Applying
Theorem~\ref{thm:finite_input_infinite_reference_decoupling}
to
\begin{align}
\left(
\rho_{\widehat{\Pi}_m}^{A_{\widehat{\Pi}_m}^mR^m}
\right)^{\otimes\ell_{n,m,\eta}}
\end{align}
and to the
\(\ell_{n,m,\eta}\)-fold partial-trace channel induced by
\begin{align}
\left(
V_m^{A_{\widehat{\Pi}_m}^m\to E_mM_m}
\right)^{\otimes\ell_{n,m,\eta}},
\end{align}
and using
\eqref{eq:IID_decoupling_achievability_block_Choi_output},
\eqref{eq:IID_decoupling_achievability_tensor_power_Choi_entropy}, and
\eqref{eq:IID_decoupling_achievability_tensor_power_input_entropy},
gives, for every \(0<s\le 1\),
\begin{align}
&\mathbb E_{
U
\sim
\mathrm{Haar}
\left(
\operatorname{U}
\left(
\left(
\mathcal H^{A_{\widehat{\Pi}_m}^m}
\right)^{\otimes\ell_{n,m,\eta}}
\right)
\right)
}
\left[
D
\left(
\widehat{\tau}_{n,m,\eta;U}^{
E_m^{\ell_{n,m,\eta}}
R^{m\ell_{n,m,\eta}}
}
\middle\|
\left(
\pi^{E_m}
\right)^{\otimes\ell_{n,m,\eta}}
\otimes
\left(
\rho_{\widehat{\Pi}_m}^{R^m}
\right)^{\otimes\ell_{n,m,\eta}}
\right)
\right]
\nonumber\\
&\le
C_{
\left|
A_{\widehat{\Pi}_m}^m
\right|^{\ell_{n,m,\eta}}
}
G_s
\exp
\Biggl[
-s\ell_{n,m,\eta}
\Biggl(
\widetilde H_{1+s}
\left(
A_{\widehat{\Pi}_m}^m
\middle|
R^m
\right)_{\rho_{\widehat{\Pi}_m}}
+
2\log
\left|
M_m
\right|
-
\log
\left|
A_{\widehat{\Pi}_m}^m
\right|
\Biggr)
\Biggr].
\label{eq:IID_decoupling_achievability_common_block_one_shot_bound}
\end{align}

We next compare the fixed-count common-factor output with the output
of the globally projected protocol.
Under the unitary isomorphism
\(V_{n,m,\eta,\xi}^{\mathrm{fix}}\)
in
\eqref{eq:IID_decoupling_global_cutoff_fixed_factorization},
the remainder system
\(B_{n,m,\eta,\xi}\)
is an orthogonal direct sum over the typical common--rare pattern
sectors.
For each
\begin{align}
x^{b_{n,m}}
\in
\mathcal T_{n,m,\eta},
\end{align}
let
\begin{align}
\Pi_{x^{b_{n,m}}}^{
B_{n,m,\eta,\xi},\mathrm{pat}
}
\end{align}
denote the orthogonal projection onto the corresponding pattern
summand of
\(\mathcal H^{B_{n,m,\eta,\xi}}\).
In the case
\(\delta_{\widehat{\Pi}_m}=0\),
the pattern set is understood to consist only of
\(0^{b_{n,m}}\).

By
\eqref{eq:IID_decoupling_global_cutoff_discarded_space},
the system
\(B_{n,m,\eta,\xi}\)
is included in the discarded system
\(M_{n,m,\eta,\xi}\).
Consequently, for every operator \(X\) on
\begin{align}
\mathcal H^{E_{n,m,\eta,\xi}}
\otimes
\mathcal H^{M_m^{\ell_{n,m,\eta}}}
\otimes
\mathcal H^{B_{n,m,\eta,\xi}}
\otimes
\mathcal H^{R^n},
\end{align}
the partial trace over
\(M_{n,m,\eta,\xi}\)
annihilates every off-diagonal pattern block:
\begin{align}
&\operatorname{Tr}_{M_{n,m,\eta,\xi}}
\Biggl[
\left(
I^{
E_{n,m,\eta,\xi}
M_m^{\ell_{n,m,\eta}}
}
\otimes
\Pi_{x^{b_{n,m}}}^{
B_{n,m,\eta,\xi},\mathrm{pat}
}
\otimes
I^{R^n}
\right)
X
\nonumber\\
&\hspace{22mm}\times
\left(
I^{
E_{n,m,\eta,\xi}
M_m^{\ell_{n,m,\eta}}
}
\otimes
\Pi_{y^{b_{n,m}}}^{
B_{n,m,\eta,\xi},\mathrm{pat}
}
\otimes
I^{R^n}
\right)
\Biggr]
\nonumber\\
&=
0
\qquad
\text{for }
x^{b_{n,m}}
\ne
y^{b_{n,m}}.
\label{eq:IID_decoupling_achievability_off_diagonal_pattern_vanishing}
\end{align}
Likewise, the partial trace over the projected input system
annihilates the off-diagonal pattern blocks in the reference
marginal.

For each diagonal pattern sector, the product structure of
\((\rho^{AR})^{\otimes n}\)
and the construction of
\(\Pi_{n,m,\eta,\xi}^{A^n}\)
imply that the selected
\(\ell_{n,m,\eta}\)
common blocks and their corresponding reference systems are in the
normalized state
\begin{align}
\left(
\rho_{\widehat{\Pi}_m}^{A_{\widehat{\Pi}_m}^mR^m}
\right)^{\otimes\ell_{n,m,\eta}},
\end{align}
independently of the pattern.
The remaining reference systems are in a normalized state that may
depend on the pattern.

Consequently, there exist probabilities
\begin{align}
\left\{
p_{x^{b_{n,m}}}
\right\}_{
x^{b_{n,m}}
\in
\mathcal T_{n,m,\eta}
},
\end{align}
normalized states
\begin{align}
\left\{
\sigma_{x^{b_{n,m}}}
\right\}_{
x^{b_{n,m}}
\in
\mathcal T_{n,m,\eta}
}
\end{align}
on the remaining reference systems, and unitary permutations
\begin{align}
\left\{
U_{x^{b_{n,m}}}^{R,\mathrm{perm}}
\right\}_{
x^{b_{n,m}}
\in
\mathcal T_{n,m,\eta}
}
\end{align}
of the reference tensor factors such that, for every \(U\),
\begin{align}
&\tau_{
n,U_{n,m,\eta,\xi}[U]
}^{E_{n,m,\eta,\xi}R^n}
\nonumber\\
&=
\sum_{
x^{b_{n,m}}
\in
\mathcal T_{n,m,\eta}
}
p_{x^{b_{n,m}}}
\left(
I^{E_{n,m,\eta,\xi}}
\otimes
U_{x^{b_{n,m}}}^{R,\mathrm{perm}}
\right)
\left[
\widehat{\tau}_{n,m,\eta;U}^{
E_m^{\ell_{n,m,\eta}}
R^{m\ell_{n,m,\eta}}
}
\otimes
\sigma_{x^{b_{n,m}}}
\right]
\nonumber\\
&\hspace{35mm}\times
\left(
I^{E_{n,m,\eta,\xi}}
\otimes
\left(
U_{x^{b_{n,m}}}^{R,\mathrm{perm}}
\right)^\dagger
\right),
\label{eq:IID_decoupling_achievability_global_output_sector_decomposition}
\end{align}
whereas
\begin{align}
&\pi^{E_{n,m,\eta,\xi}}
\otimes
\rho_{\Pi_{n,m,\eta,\xi}}^{R^n}
\nonumber\\
&=
\sum_{
x^{b_{n,m}}
\in
\mathcal T_{n,m,\eta}
}
p_{x^{b_{n,m}}}
\left(
I^{E_{n,m,\eta,\xi}}
\otimes
U_{x^{b_{n,m}}}^{R,\mathrm{perm}}
\right)
\Biggl[
\left(
\pi^{E_m}
\right)^{\otimes\ell_{n,m,\eta}}
\otimes
\left(
\rho_{\widehat{\Pi}_m}^{R^m}
\right)^{\otimes\ell_{n,m,\eta}}
\otimes
\sigma_{x^{b_{n,m}}}
\Biggr]
\nonumber\\
&\hspace{35mm}\times
\left(
I^{E_{n,m,\eta,\xi}}
\otimes
\left(
U_{x^{b_{n,m}}}^{R,\mathrm{perm}}
\right)^\dagger
\right).
\label{eq:IID_decoupling_achievability_global_reference_sector_decomposition}
\end{align}

Joint convexity of the quantum relative entropy, invariance under
unitary conjugation, and additivity under tensor products give, for
every \(U\),
\begin{align}
&D
\left(
\tau_{
n,U_{n,m,\eta,\xi}[U]
}^{E_{n,m,\eta,\xi}R^n}
\middle\|
\pi^{E_{n,m,\eta,\xi}}
\otimes
\rho_{\Pi_{n,m,\eta,\xi}}^{R^n}
\right)
\nonumber\\
&\le
\sum_{
x^{b_{n,m}}
\in
\mathcal T_{n,m,\eta}
}
p_{x^{b_{n,m}}}
D
\Biggl(
\widehat{\tau}_{n,m,\eta;U}^{
E_m^{\ell_{n,m,\eta}}
R^{m\ell_{n,m,\eta}}
}
\otimes
\sigma_{x^{b_{n,m}}}
\Biggm\|
\left(
\pi^{E_m}
\right)^{\otimes\ell_{n,m,\eta}}
\otimes
\left(
\rho_{\widehat{\Pi}_m}^{R^m}
\right)^{\otimes\ell_{n,m,\eta}}
\otimes
\sigma_{x^{b_{n,m}}}
\Biggr)
\nonumber\\
&=
D
\left(
\widehat{\tau}_{n,m,\eta;U}^{
E_m^{\ell_{n,m,\eta}}
R^{m\ell_{n,m,\eta}}
}
\middle\|
\left(
\pi^{E_m}
\right)^{\otimes\ell_{n,m,\eta}}
\otimes
\left(
\rho_{\widehat{\Pi}_m}^{R^m}
\right)^{\otimes\ell_{n,m,\eta}}
\right).
\label{eq:IID_decoupling_achievability_global_to_common_error_bound}
\end{align}
Taking the expectation over the Haar-random unitary \(U\) in
\eqref{eq:IID_decoupling_achievability_global_to_common_error_bound}
and applying
\eqref{eq:IID_decoupling_achievability_common_block_one_shot_bound}
therefore yields, for every \(0<s\le 1\),
\begin{align}
&\mathbb E_{
U
\sim
\mathrm{Haar}
\left(
\operatorname{U}
\left(
\left(
\mathcal H^{A_{\widehat{\Pi}_m}^m}
\right)^{\otimes\ell_{n,m,\eta}}
\right)
\right)
}
\left[
D
\left(
\tau_{
n,U_{n,m,\eta,\xi}[U]
}^{E_{n,m,\eta,\xi}R^n}
\middle\|
\pi^{E_{n,m,\eta,\xi}}
\otimes
\rho_{\Pi_{n,m,\eta,\xi}}^{R^n}
\right)
\right]
\nonumber\\
&\le
C_{
\left|
A_{\widehat{\Pi}_m}^m
\right|^{\ell_{n,m,\eta}}
}
G_s
\exp
\Biggl[
-s\ell_{n,m,\eta}
\Biggl(
\widetilde H_{1+s}
\left(
A_{\widehat{\Pi}_m}^m
\middle|
R^m
\right)_{\rho_{\widehat{\Pi}_m}}
+
2\log
\left|
M_m
\right|
-
\log
\left|
A_{\widehat{\Pi}_m}^m
\right|
\Biggr)
\Biggr].
\end{align}
Taking the infimum over \(0<s\le 1\) proves
\eqref{eq:IID_decoupling_achievability_projected_error_exponential_bound}.
\end{proof}

We now choose the block length, the cutoff parameters, and the common partial-trace rate diagonally.
This yields a single sequence of protocols whose error relative to the original reference marginal vanishes.

\begin{theorem}[Achievability of infinite-dimensional IID decoupling]
\label{thm:IID_decoupling_original_direct}
Let
\(\mathcal H^A\)
and
\(\mathcal H^R\)
be separable and possibly infinite-dimensional, and let
\begin{align}
\rho^{AR}
\in
\operatorname{D}
\left(
\mathcal H^A
\otimes
\mathcal H^R
\right)
\end{align}
satisfy
\begin{align}
H(A)_\rho
<
\infty.
\end{align}
Then there exist nonzero finite-rank projections
\begin{align}
\left\{
\Pi_n^{A^n}
\right\}_{n=1}^{\infty},
\end{align}
finite-dimensional systems
\begin{align}
\left\{
E_n
\right\}_{n=1}^{\infty},
\qquad
\left\{
M_n
\right\}_{n=1}^{\infty},
\end{align}
unitary isomorphisms
\begin{align}
V_n^{A_{\Pi_n}^n\to E_nM_n}
:
\mathcal H^{A_{\Pi_n}^n}
\longrightarrow
\mathcal H^{E_n}
\otimes
\mathcal H^{M_n},
\end{align}
and unitaries
\begin{align}
U_n
\in
\operatorname{U}
\left(
\mathcal H^{A_{\Pi_n}^n}
\right)
\end{align}
such that, with
\(\delta_{\Pi_n}\),
\(\mathcal H^{A_{\Pi_n}^n}\),
\(\lvert A_{\Pi_n}^n\rvert\),
\(\mathcal T_n^{A_{\Pi_n}^n\to E_n}\), and
\(\tau_{n,U_n}^{E_nR^n}\)
defined in
\eqref{eq:IID_decoupling_projection_error},
\eqref{eq:IID_decoupling_projected_input_space},
\eqref{eq:IID_decoupling_projected_input_dimension},
\eqref{eq:IID_decoupling_partial_trace_channel}, and
\eqref{eq:IID_decoupling_output_state},
respectively, one has
\begin{align}
\lim_{n\to\infty}
\delta_{\Pi_n}
&=
0,
\label{eq:IID_decoupling_original_direct_projection_limit}
\\
\lim_{n\to\infty}
D
\left(
\tau_{n,U_n}^{E_nR^n}
\middle\|
\pi^{E_n}
\otimes
\left(
\rho^R
\right)^{\otimes n}
\right)
&=
0,
\label{eq:IID_decoupling_original_direct_error_limit}
\\
\lim_{n\to\infty}
\frac{1}{n}
\log
\left|
A_{\Pi_n}^n
\right|
&=
H(A)_\rho,
\label{eq:IID_decoupling_original_direct_rank_limit}
\\
\lim_{n\to\infty}
\frac{1}{n}
\log
\left|
M_n
\right|
&=
\frac{1}{2}
I(A:R)_\rho.
\label{eq:IID_decoupling_original_direct_discarded_rate}
\end{align}
\end{theorem}

\begin{proof}
By
\eqref{eq:mutual_information_bounds_finite_marginal},
\begin{align}
0
\le
\frac{1}{2}
I(A:R)_\rho
\le
H(A)_\rho.
\label{eq:IID_decoupling_achievability_mutual_information_range}
\end{align}
We distinguish the cases in which the second inequality is strict and
in which equality holds.

\medskip
\noindent
\textbf{The case
\(\frac{1}{2}I(A:R)_\rho<H(A)_\rho\).}
Choose a sequence
\(\{\varepsilon_j\}_{j=1}^{\infty}\)
such that
\begin{align}
0
<
\varepsilon_j
<
\frac{1}{2},
\qquad
\lim_{j\to\infty}
\varepsilon_j
=
0,
\label{eq:IID_decoupling_achievability_epsilon_limit}
\end{align}
and
\begin{align}
4\varepsilon_j
<
H(A)_\rho
-
\frac{1}{2}
I(A:R)_\rho
\label{eq:IID_decoupling_achievability_epsilon_range}
\end{align}
for every \(j\).
Define
\begin{align}
q_j
\coloneqq
\frac{1}{2}
I(A:R)_\rho
+
4\varepsilon_j.
\label{eq:IID_decoupling_achievability_qj}
\end{align}
Then
\begin{align}
0
<
q_j
<
H(A)_\rho,
\qquad
\lim_{j\to\infty}
q_j
=
\frac{1}{2}
I(A:R)_\rho.
\label{eq:IID_decoupling_achievability_qj_limit}
\end{align}

For every \(j\), apply
Proposition~\ref{prop:decoupling_IID_finite_entropy_spectral_cutoffs}
with \(q=q_j\).
To distinguish the constructions obtained for different values of
\(j\), denote the resulting objects by
\begin{align}
\left\{
\widehat{\Pi}_m^{(j),A^m},
E_m^{(j)},
M_m^{(j)},
V_m^{(j)}
\right\}_{m=1}^{\infty}.
\end{align}
Thus,
\begin{align}
V_m^{(j)}
:
\mathcal H^{A_{\widehat{\Pi}_m^{(j)}}^m}
\longrightarrow
\mathcal H^{E_m^{(j)}}
\otimes
\mathcal H^{M_m^{(j)}}
\end{align}
is a unitary isomorphism, and
\begin{align}
\lim_{m\to\infty}
\frac{1}{m}
H
\left(
A_{\widehat{\Pi}_m^{(j)}}^m
\middle|
R^m
\right)_{\rho_{\widehat{\Pi}_m^{(j)}}}
&=
H(A|R)_\rho,
\label{eq:IID_decoupling_achievability_j_conditional_limit}
\\
\lim_{m\to\infty}
\frac{1}{m}
\log
\left|
A_{\widehat{\Pi}_m^{(j)}}^m
\right|
&=
H(A)_\rho,
\label{eq:IID_decoupling_achievability_j_rank_limit}
\\
\lim_{m\to\infty}
\frac{1}{m}
\log
\left|
M_m^{(j)}
\right|
&=
q_j.
\label{eq:IID_decoupling_achievability_j_discarded_limit}
\end{align}

For each \(j\), choose admissible parameter sequences
\begin{align}
\left\{
\eta_m^{(j)}
\right\}_{m=1}^{\infty},
\qquad
\left\{
\xi_m^{(j)}
\right\}_{m=1}^{\infty}
\end{align}
satisfying
\eqref{eq:IID_decoupling_global_cutoff_admissible_parameter_sequences}
and
\eqref{eq:IID_decoupling_global_cutoff_parameter_sequence_limits}
for the \(j\)-th spectral-cutoff construction.

Using
\eqref{eq:IID_decoupling_achievability_j_conditional_limit}--%
\eqref{eq:IID_decoupling_achievability_j_discarded_limit},
\eqref{eq:IID_decoupling_global_cutoff_projected_dimension_iterated_limit},
and
\eqref{eq:IID_decoupling_global_cutoff_prescribed_partial_trace_rate},
choose a strictly increasing sequence of block lengths
\(\{m_j\}_{j=1}^{\infty}\)
such that
\begin{align}
\left|
\frac{1}{m_j}
H
\left(
A_{\widehat{\Pi}_{m_j}^{(j)}}^{m_j}
\middle|
R^{m_j}
\right)_{\rho_{\widehat{\Pi}_{m_j}^{(j)}}}
-
H(A|R)_\rho
\right|
&\le
\varepsilon_j,
\label{eq:IID_decoupling_achievability_block_conditional_approximation}
\\
\left|
\frac{1}{m_j}
\log
\left|
A_{\widehat{\Pi}_{m_j}^{(j)}}^{m_j}
\right|
-
H(A)_\rho
\right|
&\le
\varepsilon_j,
\label{eq:IID_decoupling_achievability_block_rank_approximation}
\\
\left|
\frac{1}{m_j}
\log
\left|
M_{m_j}^{(j)}
\right|
-
q_j
\right|
&\le
\varepsilon_j,
\label{eq:IID_decoupling_achievability_block_discarded_approximation}
\end{align}
and, upon setting
\begin{align}
\eta_j
&\coloneqq
\eta_{m_j}^{(j)},
&
\xi_j
&\coloneqq
\xi_{m_j}^{(j)},
\end{align}
such that
\begin{align}
\Biggl|
\lim_{n\to\infty}
\frac{1}{n}
\log
\left|
A_{\Pi_{n,m_j,\eta_j,\xi_j}^{(j)}}^n
\right|
-
H(A)_\rho
\Biggr|
&\le
\varepsilon_j,
\label{eq:IID_decoupling_achievability_fixed_rank_rate_choice}
\\
\Biggl|
\lim_{n\to\infty}
\frac{1}{n}
\log
\left|
M_{n,m_j,\eta_j,\xi_j}^{(j)}
\right|
-
q_j
\Biggr|
&\le
\varepsilon_j.
\label{eq:IID_decoupling_achievability_fixed_discarded_rate_choice}
\end{align}
Here and below, the superscript \((j)\) on an object with indices
\((n,m_j,\eta_j,\xi_j)\) indicates that the object is constructed
from the \(j\)-th spectral-cutoff family.

Equations
\eqref{eq:IID_decoupling_achievability_qj},
\eqref{eq:IID_decoupling_achievability_block_conditional_approximation},
\eqref{eq:IID_decoupling_achievability_block_rank_approximation}, and
\eqref{eq:IID_decoupling_achievability_block_discarded_approximation}
give
\begin{align}
&\frac{1}{m_j}
\Biggl[
H
\left(
A_{\widehat{\Pi}_{m_j}^{(j)}}^{m_j}
\middle|
R^{m_j}
\right)_{\rho_{\widehat{\Pi}_{m_j}^{(j)}}}
+
2
\log
\left|
M_{m_j}^{(j)}
\right|
-
\log
\left|
A_{\widehat{\Pi}_{m_j}^{(j)}}^{m_j}
\right|
\Biggr]
\nonumber\\
&\ge
H(A|R)_\rho
-
\varepsilon_j
+
2q_j
-
2\varepsilon_j
-
H(A)_\rho
-
\varepsilon_j
\nonumber\\
&=
4\varepsilon_j
>
0,
\label{eq:IID_decoupling_achievability_positive_von_Neumann_gap}
\end{align}
where we used
\begin{align}
I(A:R)_\rho
=
H(A)_\rho
-
H(A|R)_\rho.
\end{align}

Since
\(A_{\widehat{\Pi}_{m_j}^{(j)}}^{m_j}\)
is finite-dimensional and
\(R^{m_j}\)
is separable,
Lemma~\ref{lem:conditional_sandwiched_order_one_limit}
gives
\begin{align}
&\lim_{s\downarrow0}
\widetilde H_{1+s}
\left(
A_{\widehat{\Pi}_{m_j}^{(j)}}^{m_j}
\middle|
R^{m_j}
\right)_{\rho_{\widehat{\Pi}_{m_j}^{(j)}}}
\nonumber\\
&=
H
\left(
A_{\widehat{\Pi}_{m_j}^{(j)}}^{m_j}
\middle|
R^{m_j}
\right)_{\rho_{\widehat{\Pi}_{m_j}^{(j)}}}.
\end{align}
Consequently, for every \(j\), there exists
\(s_j\in(0,1]\)
such that
\begin{align}
&\widetilde H_{1+s_j}
\left(
A_{\widehat{\Pi}_{m_j}^{(j)}}^{m_j}
\middle|
R^{m_j}
\right)_{\rho_{\widehat{\Pi}_{m_j}^{(j)}}}
+
2
\log
\left|
M_{m_j}^{(j)}
\right|
-
\log
\left|
A_{\widehat{\Pi}_{m_j}^{(j)}}^{m_j}
\right|
>
0.
\label{eq:IID_decoupling_achievability_positive_Renyi_gap_diagonal}
\end{align}
Thus, the hypothesis of
Proposition~\ref{prop:IID_decoupling_achievability_fixed_common_factor}
is satisfied for each \(j\).

For each fixed \(j\),
Theorem~\ref{thm:IID_decoupling_high_probability_global_cutoffs}
gives
\begin{align}
\lim_{n\to\infty}
\delta_{\Pi_{n,m_j,\eta_j,\xi_j}^{(j)}}
=
0.
\label{eq:IID_decoupling_achievability_fixed_projection_limit}
\end{align}
Moreover,
Proposition~\ref{prop:IID_decoupling_achievability_fixed_common_factor}
and
\eqref{eq:IID_decoupling_achievability_positive_Renyi_gap_diagonal},
together with
\begin{align}
\ell_{n,m_j,\eta_j}
&=
\Theta(n),
\label{eq:IID_decoupling_achievability_diagonal_common_block_scaling}
\\
C_{
\left|
A_{\widehat{\Pi}_{m_j}^{(j)}}^{m_j}
\right|^{\ell_{n,m_j,\eta_j}}
}
&=
O(1),
\label{eq:IID_decoupling_achievability_diagonal_prefactor_bound}
\end{align}
imply that
\begin{align}
&\lim_{n\to\infty}
\mathbb E_{
U
\sim
\mathrm{Haar}
\left(
\operatorname{U}
\left(
\left(
\mathcal H^{A_{\widehat{\Pi}_{m_j}^{(j)}}^{m_j}}
\right)^{
\otimes
\ell_{n,m_j,\eta_j}
}
\right)
\right)
}
\Biggl[
D
\Biggl(
\tau_{
n,
U_{n,m_j,\eta_j,\xi_j}^{(j)}[U]
}^{
E_{n,m_j,\eta_j,\xi_j}^{(j)}R^n
}
\Biggm\|
\pi^{E_{n,m_j,\eta_j,\xi_j}^{(j)}}
\otimes
\rho_{\Pi_{n,m_j,\eta_j,\xi_j}^{(j)}}^{R^n}
\Biggr)
\Biggr]
=
0,
\label{eq:IID_decoupling_achievability_fixed_expected_error_limit}
\end{align}
where
\eqref{eq:IID_decoupling_achievability_diagonal_common_block_scaling}
follows from
\eqref{eq:IID_decoupling_global_cutoff_common_block_rate}
when
\(0<\delta_{\widehat{\Pi}_{m_j}^{(j)}}<1\),
and from
\eqref{eq:IID_decoupling_common_only_fixed_count}
and
\eqref{eq:decoupling_IID_block_quotient_properties}
when
\(\delta_{\widehat{\Pi}_{m_j}^{(j)}}=0\), and
\eqref{eq:IID_decoupling_achievability_diagonal_prefactor_bound}
follows from
\eqref{eq:decoupling_dimension_dependent_constant}.

Since the expectation is finite, for every sufficiently large
\(n\), there exists a unitary
\begin{align}
\widetilde U_{n,j}
\in
\operatorname{U}
\left(
\left(
\mathcal H^{A_{\widehat{\Pi}_{m_j}^{(j)}}^{m_j}}
\right)^{
\otimes
\ell_{n,m_j,\eta_j}
}
\right)
\end{align}
such that
\begin{align}
&D
\Biggl(
\tau_{
n,
U_{n,m_j,\eta_j,\xi_j}^{(j)}
\left[
\widetilde U_{n,j}
\right]
}^{
E_{n,m_j,\eta_j,\xi_j}^{(j)}R^n
}
\Biggm\|
\pi^{E_{n,m_j,\eta_j,\xi_j}^{(j)}}
\otimes
\rho_{\Pi_{n,m_j,\eta_j,\xi_j}^{(j)}}^{R^n}
\Biggr)
\nonumber\\
&\le
\mathbb E_{
U
\sim
\mathrm{Haar}
\left(
\operatorname{U}
\left(
\left(
\mathcal H^{A_{\widehat{\Pi}_{m_j}^{(j)}}^{m_j}}
\right)^{
\otimes
\ell_{n,m_j,\eta_j}
}
\right)
\right)
}
\Biggl[
D
\Biggl(
\tau_{
n,
U_{n,m_j,\eta_j,\xi_j}^{(j)}[U]
}^{
E_{n,m_j,\eta_j,\xi_j}^{(j)}R^n
}
\Biggm\|
\pi^{E_{n,m_j,\eta_j,\xi_j}^{(j)}}
\otimes
\rho_{\Pi_{n,m_j,\eta_j,\xi_j}^{(j)}}^{R^n}
\Biggr)
\Biggr].
\label{eq:IID_decoupling_achievability_fixed_unitary_choice}
\end{align}
For each fixed \(j\), choose
\(\widetilde U_{n,j}\)
arbitrarily for the remaining finitely many values of \(n\).
It follows from
\eqref{eq:IID_decoupling_achievability_fixed_expected_error_limit}
and
\eqref{eq:IID_decoupling_achievability_fixed_unitary_choice}
that
\begin{align}
&\lim_{n\to\infty}
D
\Biggl(
\tau_{
n,
U_{n,m_j,\eta_j,\xi_j}^{(j)}
\left[
\widetilde U_{n,j}
\right]
}^{
E_{n,m_j,\eta_j,\xi_j}^{(j)}R^n
}
\Biggm\|
\pi^{E_{n,m_j,\eta_j,\xi_j}^{(j)}}
\otimes
\rho_{\Pi_{n,m_j,\eta_j,\xi_j}^{(j)}}^{R^n}
\Biggr)
=
0.
\label{eq:IID_decoupling_achievability_fixed_projected_error_limit}
\end{align}

Choose a strictly increasing sequence
\(\{n_j\}_{j=1}^{\infty}\)
such that
\begin{align}
n_j
\ge
m_j
\end{align}
and such that, for every \(n\ge n_j\),
\begin{align}
\delta_{\Pi_{n,m_j,\eta_j,\xi_j}^{(j)}}
&\le
\varepsilon_j,
\label{eq:IID_decoupling_achievability_diagonal_projection_choice}
\\
D
\Biggl(
\tau_{
n,
U_{n,m_j,\eta_j,\xi_j}^{(j)}
\left[
\widetilde U_{n,j}
\right]
}^{
E_{n,m_j,\eta_j,\xi_j}^{(j)}R^n
}
\Biggm\|
\pi^{E_{n,m_j,\eta_j,\xi_j}^{(j)}}
\otimes
\rho_{\Pi_{n,m_j,\eta_j,\xi_j}^{(j)}}^{R^n}
\Biggr)
&\le
\varepsilon_j,
\label{eq:IID_decoupling_achievability_diagonal_projected_error_choice}
\\
\left|
\frac{1}{n}
\log
\left|
A_{\Pi_{n,m_j,\eta_j,\xi_j}^{(j)}}^n
\right|
-
H(A)_\rho
\right|
&\le
2\varepsilon_j,
\label{eq:IID_decoupling_achievability_diagonal_rank_choice}
\\
\left|
\frac{1}{n}
\log
\left|
M_{n,m_j,\eta_j,\xi_j}^{(j)}
\right|
-
q_j
\right|
&\le
2\varepsilon_j.
\label{eq:IID_decoupling_achievability_diagonal_discarded_choice}
\end{align}
The existence of such a sequence follows from
\eqref{eq:IID_decoupling_achievability_fixed_rank_rate_choice},
\eqref{eq:IID_decoupling_achievability_fixed_discarded_rate_choice},
\eqref{eq:IID_decoupling_achievability_fixed_projection_limit}, and
\eqref{eq:IID_decoupling_achievability_fixed_projected_error_limit}.

For
\begin{align}
n_j
\le
n
<
n_{j+1},
\end{align}
define
\begin{align}
\Pi_n^{A^n}
&\coloneqq
\Pi_{n,m_j,\eta_j,\xi_j}^{(j),A^n},
\\
\mathcal H^{E_n}
&\coloneqq
\mathcal H^{E_{n,m_j,\eta_j,\xi_j}^{(j)}},
\\
\mathcal H^{M_n}
&\coloneqq
\mathcal H^{M_{n,m_j,\eta_j,\xi_j}^{(j)}},
\\
V_n
&\coloneqq
V_{n,m_j,\eta_j,\xi_j}^{(j)},
\\
U_n
&\coloneqq
U_{n,m_j,\eta_j,\xi_j}^{(j)}
\left[
\widetilde U_{n,j}
\right].
\label{eq:IID_decoupling_achievability_diagonal_protocol_definition}
\end{align}
For the finitely many values \(n<n_1\), choose
\(\Pi_n^{A^n}\)
to be a rank-one spectral projection corresponding to a positive
eigenvalue of
\((\rho^A)^{\otimes n}\),
take
\begin{align}
\mathcal H^{E_n}
&\coloneqq
\mathcal H^{A_{\Pi_n}^n},
&
\mathcal H^{M_n}
&\coloneqq
\mathbb C,
\end{align}
choose the canonical unitary isomorphism \(V_n\), and set
\begin{align}
U_n
=
I^{A_{\Pi_n}^n}.
\end{align}
These finitely many choices do not affect any asymptotic statement.

Equations
\eqref{eq:IID_decoupling_achievability_epsilon_limit},
\eqref{eq:IID_decoupling_achievability_diagonal_projection_choice}, and
\eqref{eq:IID_decoupling_achievability_diagonal_rank_choice}
give
\begin{align}
\lim_{n\to\infty}
\delta_{\Pi_n}
&=
0,
\\
\lim_{n\to\infty}
\frac{1}{n}
\log
\left|
A_{\Pi_n}^n
\right|
&=
H(A)_\rho.
\end{align}
This proves
\eqref{eq:IID_decoupling_original_direct_projection_limit}
and
\eqref{eq:IID_decoupling_original_direct_rank_limit}.

For
\(n_j\le n<n_{j+1}\),
\eqref{eq:IID_decoupling_original_error_from_projected_error},
\eqref{eq:IID_decoupling_achievability_diagonal_projection_choice}, and
\eqref{eq:IID_decoupling_achievability_diagonal_projected_error_choice}
give
\begin{align}
&0
\le
D
\left(
\tau_{n,U_n}^{E_nR^n}
\middle\|
\pi^{E_n}
\otimes
\left(
\rho^R
\right)^{\otimes n}
\right)
\nonumber\\
&\le
\varepsilon_j
-
\log
\left(
1-\varepsilon_j
\right).
\label{eq:IID_decoupling_achievability_original_error_diagonal_bound}
\end{align}
Since the right-hand side converges to zero as \(j\to\infty\),
\eqref{eq:IID_decoupling_original_direct_error_limit}
follows.

Finally, whenever
\(n_j\le n<n_{j+1}\),
\eqref{eq:IID_decoupling_achievability_qj} and
\eqref{eq:IID_decoupling_achievability_diagonal_discarded_choice}
give
\begin{align}
\left|
\frac{1}{n}
\log
\left|
M_n
\right|
-
\frac{1}{2}
I(A:R)_\rho
\right|
&\le
\left|
\frac{1}{n}
\log
\left|
M_n
\right|
-
q_j
\right|
+
\left|
q_j
-
\frac{1}{2}
I(A:R)_\rho
\right|
\nonumber\\
&\le
2\varepsilon_j
+
4\varepsilon_j
\nonumber\\
&=
6\varepsilon_j.
\label{eq:IID_decoupling_achievability_diagonal_discarded_convergence}
\end{align}
Therefore, the discarded-system rate has a limit and satisfies
\begin{align}
\lim_{n\to\infty}
\frac{1}{n}
\log
\left|
M_n
\right|
=
\frac{1}{2}
I(A:R)_\rho,
\label{eq:IID_decoupling_achievability_discarded_rate_identity}
\end{align}
which proves \eqref{eq:IID_decoupling_original_direct_discarded_rate}.

\medskip
\noindent
\textbf{The case
\(\frac{1}{2}I(A:R)_\rho=H(A)_\rho\).}
Apply
Proposition~\ref{prop:decoupling_IID_finite_entropy_spectral_cutoffs}
with
\begin{align}
q
=
H(A)_\rho,
\end{align}
and let
\(\{\widehat{\Pi}_n^{A^n}\}_{n=1}^{\infty}\)
be the resulting block-length-\(n\) spectral projections.
Define the final protocol projections in this case by
\begin{align}
\Pi_n^{A^n}
\coloneqq
\widehat{\Pi}_n^{A^n}
\end{align}
for every \(n\in\{1,2,\ldots\}\).
Then
\begin{align}
\lim_{n\to\infty}
\delta_{\Pi_n}
&=
0,
\label{eq:IID_decoupling_achievability_boundary_projection_limit}
\\
\lim_{n\to\infty}
\frac{1}{n}
\log
\left|
A_{\Pi_n}^n
\right|
&=
H(A)_\rho.
\label{eq:IID_decoupling_achievability_boundary_rank_limit}
\end{align}

For every \(n\), set
\begin{align}
\mathcal H^{E_n}
&\coloneqq
\mathbb C,
&
\mathcal H^{M_n}
&\coloneqq
\mathcal H^{A_{\Pi_n}^n},
\end{align}
let
\begin{align}
V_n
:
\mathcal H^{A_{\Pi_n}^n}
\longrightarrow
\mathcal H^{E_n}
\otimes
\mathcal H^{M_n}
\end{align}
be the canonical unitary isomorphism, and choose
\begin{align}
U_n
=
I^{A_{\Pi_n}^n}.
\end{align}
Since the output system \(E_n\) is one-dimensional,
\begin{align}
\tau_{n,U_n}^{E_nR^n}
=
\pi^{E_n}
\otimes
\rho_{\Pi_n}^{R^n},
\end{align}
and hence
\begin{align}
D
\left(
\tau_{n,U_n}^{E_nR^n}
\middle\|
\pi^{E_n}
\otimes
\rho_{\Pi_n}^{R^n}
\right)
=
0.
\end{align}
It follows from
\eqref{eq:IID_decoupling_original_error_from_projected_error} that
\begin{align}
0
&\le
D
\left(
\tau_{n,U_n}^{E_nR^n}
\middle\|
\pi^{E_n}
\otimes
\left(
\rho^R
\right)^{\otimes n}
\right)
\nonumber\\
&\le
-\log
\left(
1-\delta_{\Pi_n}
\right).
\end{align}
Together with
\eqref{eq:IID_decoupling_achievability_boundary_projection_limit},
this proves
\eqref{eq:IID_decoupling_original_direct_projection_limit}
and
\eqref{eq:IID_decoupling_original_direct_error_limit}.
Equation
\eqref{eq:IID_decoupling_achievability_boundary_rank_limit}
proves
\eqref{eq:IID_decoupling_original_direct_rank_limit}.

Moreover,
\begin{align}
\left|
M_n
\right|
=
\left|
A_{\Pi_n}^n
\right|,
\end{align}
and therefore
\begin{align}
\lim_{n\to\infty}
\frac{1}{n}
\log
\left|
M_n
\right|
&=
H(A)_\rho
\nonumber\\
&=
\frac{1}{2}
I(A:R)_\rho.
\end{align}
Thus,
\eqref{eq:IID_decoupling_original_direct_discarded_rate}
also holds in the equality case.
\end{proof}

\subsection{Optimality of infinite-dimensional IID decoupling}
\label{subsec:converse_infinite_dimensional_IID_decoupling}

We prove converse bounds for arbitrary sequences of finite-rank
projections and arbitrary partial-trace decoupling protocols defined
on the corresponding projected input spaces.
In particular, no commutation relation between the projections and
the IID input marginals is assumed.

We first provide converse bounds that hold for every sequence of
finite-rank projections whose success probabilities converge to one.
These bounds apply, in particular, to the projections constructed in
Theorem~\ref{thm:IID_decoupling_high_probability_global_cutoffs}.

\begin{proposition}[Converse for high-probability finite-rank
projections of IID states]
\label{prop:IID_decoupling_high_probability_projection_converse}
Let
\(\mathcal H^A\)
and
\(\mathcal H^R\)
be separable and possibly infinite-dimensional, and let
\begin{align}
\rho^{AR}
\in
\operatorname{D}
\left(
\mathcal H^A
\otimes
\mathcal H^R
\right)
\end{align}
satisfy
\begin{align}
H(A)_\rho
<
\infty.
\end{align}
For every
\(n\in\{1,2,\ldots\}\),
let
\(\Pi_n^{A^n}\)
be a nonzero finite-rank projection satisfying
\begin{align}
\delta_{\Pi_n}
<
1,
\end{align}
where
\(\Pi_n^{\perp A^n}\),
\(\delta_{\Pi_n}\),
\(\mathcal H^{A_{\Pi_n}^n}\),
\(\lvert A_{\Pi_n}^n\rvert\),
\(\rho_{\Pi_n}^{A_{\Pi_n}^nR^n}\),
\(\rho_{\Pi_n}^{A_{\Pi_n}^n}\), and
\(\rho_{\Pi_n}^{R^n}\)
are defined in
\eqref{eq:IID_decoupling_projection_complement},
\eqref{eq:IID_decoupling_projection_error},
\eqref{eq:IID_decoupling_projected_input_space},
\eqref{eq:IID_decoupling_projected_input_dimension},
\eqref{eq:IID_decoupling_projected_state},
\eqref{eq:IID_decoupling_projected_input_marginal}, and
\eqref{eq:IID_decoupling_projected_reference},
respectively.
Whenever
\(\delta_{\Pi_n}>0\),
let
\(\mathcal H^{A_{\Pi_n^\perp}^n}\),
\(\rho_{\Pi_n^\perp}^{A_{\Pi_n^\perp}^nR^n}\), and
\(\rho_{\Pi_n^\perp}^{A_{\Pi_n^\perp}^n}\)
be defined in
\eqref{eq:IID_decoupling_complementary_input_space},
\eqref{eq:IID_decoupling_complementary_state}, and
\eqref{eq:IID_decoupling_complementary_input_marginal},
respectively.
When
\(\delta_{\Pi_n}=0\),
use the convention in
\eqref{eq:IID_decoupling_zero_complementary_convention}.

Suppose that
\begin{align}
\lim_{n\to\infty}
\delta_{\Pi_n}
=
0.
\label{eq:IID_decoupling_converse_projection_success_assumption}
\end{align}
Then
\begin{align}
\lim_{n\to\infty}
\frac{1}{n}
H
\left(
A_{\Pi_n}^n
\right)_{\rho_{\Pi_n}}
&=
H(A)_\rho,
\label{eq:IID_decoupling_converse_projected_entropy_limit}
\\
\lim_{n\to\infty}
\frac{
\delta_{\Pi_n}
}{
n
}
H
\left(
A_{\Pi_n^\perp}^n
\right)_{\rho_{\Pi_n^\perp}}
&=
0,
\label{eq:IID_decoupling_converse_weighted_complementary_entropy_limit}
\\
\liminf_{n\to\infty}
\frac{1}{n}
\log
\left|
A_{\Pi_n}^n
\right|
&\ge
H(A)_\rho,
\label{eq:IID_decoupling_converse_projection_rank_bound_lemma}
\\
\liminf_{n\to\infty}
\frac{1}{n}
I
\left(
A_{\Pi_n}^n:R^n
\right)_{\rho_{\Pi_n}}
&\ge
I(A:R)_\rho.
\label{eq:IID_decoupling_converse_projected_mutual_information_bound}
\end{align}
\end{proposition}

\begin{proof}
By
\eqref{eq:IID_decoupling_converse_projection_success_assumption},
\begin{align}
\lim_{n\to\infty}
\left(
1-\delta_{\Pi_n}
\right)
=
1.
\label{eq:IID_decoupling_converse_projection_success_probability_limit}
\end{align}

\medskip
\noindent
\textbf{Lower bound in
\eqref{eq:IID_decoupling_converse_projected_entropy_limit}.}
Suppose first that
\begin{align}
H(A)_\rho
>
0,
\end{align}
and fix
\begin{align}
0
<
\gamma
<
H(A)_\rho.
\end{align}

Throughout this part, operators on
\(\mathcal H^{A_{\Pi_n}^n}\)
are canonically identified with their zero extensions to
\(\mathcal H^{A^n}\).
Taking the \(A_{\Pi_n}^n\)-marginal in
\eqref{eq:IID_decoupling_projected_state}
gives
\begin{align}
\left(
1-\delta_{\Pi_n}
\right)
\rho_{\Pi_n}^{A_{\Pi_n}^n}
=
\Pi_n^{A^n}
\left(
\rho^A
\right)^{\otimes n}
\Pi_n^{A^n}.
\label{eq:IID_decoupling_converse_compressed_marginal}
\end{align}

For a self-adjoint operator \(X\), we use the notation
\begin{align}
\left\{
X>0
\right\}
\coloneqq
\mathbf 1_{(0,\infty)}(X)
\label{eq:IID_decoupling_converse_positive_spectral_projection_notation}
\end{align}
for the spectral projection of \(X\) onto its strictly positive
spectral subspace.
Define
\begin{align}
\Pi_{n,\gamma}^{\mathrm{src},A^n}
\coloneqq
\left\{
\left(
\rho^A
\right)^{\otimes n}
-
\exp[-n\gamma]
I^{A^n}
>
0
\right\}.
\label{eq:IID_decoupling_converse_source_large_eigenvalue_projection}
\end{align}
Thus,
\(\Pi_{n,\gamma}^{\mathrm{src},A^n}\)
is the spectral projection of
\((\rho^A)^{\otimes n}\)
corresponding to eigenvalues strictly larger than
\(\exp[-n\gamma]\).
Since
\((\rho^A)^{\otimes n}\)
is trace class and the threshold is strictly positive,
\(\Pi_{n,\gamma}^{\mathrm{src},A^n}\)
has finite rank.

Set
\begin{align}
\beta_n(\gamma)
\coloneqq
\operatorname{Tr}
\left[
\Pi_{n,\gamma}^{\mathrm{src},A^n}
\left(
\rho^A
\right)^{\otimes n}
\right].
\label{eq:IID_decoupling_converse_source_large_eigenvalue_mass}
\end{align}
An eigenvalue
\(\lambda_{j^n}^{A^n}\)
of
\((\rho^A)^{\otimes n}\)
satisfies
\begin{align}
\lambda_{j^n}^{A^n}
>
\exp[-n\gamma]
\end{align}
if and only if
\begin{align}
-\frac{1}{n}
\log
\lambda_{j^n}^{A^n}
<
\gamma.
\end{align}
Therefore, in terms of the IID information density introduced in
\eqref{eq:decoupling_IID_average_information_density},
\begin{align}
\beta_n(\gamma)
=
\Pr
\left\{
\overline{\imath}_{A,n}(J^n)
<
\gamma
\right\}.
\label{eq:IID_decoupling_converse_source_large_eigenvalue_probability}
\end{align}
By
\eqref{eq:decoupling_IID_information_density_mean_convergence},
the random variable
\(\overline{\imath}_{A,n}(J^n)\)
converges in mean, and hence in probability, to
\(H(A)_\rho\).
Since
\(\gamma<H(A)_\rho\),
it follows that
\begin{align}
\lim_{n\to\infty}
\beta_n(\gamma)
=
0.
\label{eq:IID_decoupling_converse_source_large_eigenvalue_mass_limit}
\end{align}

By the definition of
\(\Pi_{n,\gamma}^{\mathrm{src},A^n}\),
the restriction of
\((\rho^A)^{\otimes n}\)
to its complementary spectral subspace satisfies
\begin{align}
&\left(
I^{A^n}
-
\Pi_{n,\gamma}^{\mathrm{src},A^n}
\right)
\left(
\rho^A
\right)^{\otimes n}
\left(
I^{A^n}
-
\Pi_{n,\gamma}^{\mathrm{src},A^n}
\right)
\nonumber\\
&\le
\exp[-n\gamma]
\left(
I^{A^n}
-
\Pi_{n,\gamma}^{\mathrm{src},A^n}
\right).
\label{eq:IID_decoupling_converse_source_complementary_spectral_bound}
\end{align}
Moreover,
\(\Pi_{n,\gamma}^{\mathrm{src},A^n}\)
commutes with
\((\rho^A)^{\otimes n}\).
Consequently,
\begin{align}
\left(
\rho^A
\right)^{\otimes n}
&\le
\Pi_{n,\gamma}^{\mathrm{src},A^n}
\left(
\rho^A
\right)^{\otimes n}
\Pi_{n,\gamma}^{\mathrm{src},A^n}
+
\exp[-n\gamma]
I^{A^n}.
\label{eq:IID_decoupling_converse_source_spectral_threshold_bound}
\end{align}

Compressing
\eqref{eq:IID_decoupling_converse_source_spectral_threshold_bound}
by
\(\Pi_n^{A^n}\),
using
\eqref{eq:IID_decoupling_converse_compressed_marginal},
and dividing by
\(1-\delta_{\Pi_n}>0\),
gives
\begin{align}
\rho_{\Pi_n}^{A_{\Pi_n}^n}
&\le
\frac{
\Pi_n^{A^n}
\Pi_{n,\gamma}^{\mathrm{src},A^n}
\left(
\rho^A
\right)^{\otimes n}
\Pi_{n,\gamma}^{\mathrm{src},A^n}
\Pi_n^{A^n}
}{
1-\delta_{\Pi_n}
}
\nonumber\\
&\quad+
\frac{
\exp[-n\gamma]
}{
1-\delta_{\Pi_n}
}
\Pi_n^{A^n}.
\label{eq:IID_decoupling_converse_compressed_spectral_threshold_bound}
\end{align}

Define
\begin{align}
\Pi_{n,\gamma}^{\mathrm{cmp},A^n}
\coloneqq
\left\{
\rho_{\Pi_n}^{A_{\Pi_n}^n}
-
\frac{
2\exp[-n\gamma]
}{
1-\delta_{\Pi_n}
}
\Pi_n^{A^n}
>
0
\right\}.
\label{eq:IID_decoupling_converse_compressed_large_eigenvalue_projection}
\end{align}
Thus,
\(\Pi_{n,\gamma}^{\mathrm{cmp},A^n}\)
is the spectral projection of
\(\rho_{\Pi_n}^{A_{\Pi_n}^n}\)
corresponding to eigenvalues strictly larger than
\begin{align}
\frac{
2\exp[-n\gamma]
}{
1-\delta_{\Pi_n}
}.
\end{align}
Since
\(\rho_{\Pi_n}^{A_{\Pi_n}^n}\)
is supported on
\(\mathcal H^{A_{\Pi_n}^n}\),
one has
\begin{align}
\Pi_{n,\gamma}^{\mathrm{cmp},A^n}
\le
\Pi_n^{A^n}.
\label{eq:IID_decoupling_converse_compressed_projection_domination}
\end{align}
By the definition of
\(\Pi_{n,\gamma}^{\mathrm{cmp},A^n}\),
\begin{align}
\frac{
2\exp[-n\gamma]
}{
1-\delta_{\Pi_n}
}
\operatorname{Tr}
\left[
\Pi_{n,\gamma}^{\mathrm{cmp},A^n}
\right]
\le
\operatorname{Tr}
\left[
\Pi_{n,\gamma}^{\mathrm{cmp},A^n}
\rho_{\Pi_n}^{A_{\Pi_n}^n}
\right].
\label{eq:IID_decoupling_converse_compressed_large_eigenvalue_rank_lower_bound}
\end{align}

On the other hand,
\eqref{eq:IID_decoupling_converse_compressed_spectral_threshold_bound}
and
\eqref{eq:IID_decoupling_converse_compressed_projection_domination}
give
\begin{align}
&\operatorname{Tr}
\left[
\Pi_{n,\gamma}^{\mathrm{cmp},A^n}
\rho_{\Pi_n}^{A_{\Pi_n}^n}
\right]
\nonumber\\
&\le
\frac{1}{
1-\delta_{\Pi_n}
}
\operatorname{Tr}
\left[
\Pi_{n,\gamma}^{\mathrm{cmp},A^n}
\Pi_{n,\gamma}^{\mathrm{src},A^n}
\left(
\rho^A
\right)^{\otimes n}
\Pi_{n,\gamma}^{\mathrm{src},A^n}
\right]
\nonumber\\
&\quad+
\frac{
\exp[-n\gamma]
}{
1-\delta_{\Pi_n}
}
\operatorname{Tr}
\left[
\Pi_{n,\gamma}^{\mathrm{cmp},A^n}
\right]
\nonumber\\
&\le
\frac{
\beta_n(\gamma)
}{
1-\delta_{\Pi_n}
}
+
\frac{
\exp[-n\gamma]
}{
1-\delta_{\Pi_n}
}
\operatorname{Tr}
\left[
\Pi_{n,\gamma}^{\mathrm{cmp},A^n}
\right].
\label{eq:IID_decoupling_converse_compressed_large_eigenvalue_mass_intermediate}
\end{align}
The last inequality follows from
\begin{align}
0
\le
\Pi_{n,\gamma}^{\mathrm{cmp},A^n}
\le
I^{A^n},
\end{align}
the positivity of
\begin{align}
\Pi_{n,\gamma}^{\mathrm{src},A^n}
\left(
\rho^A
\right)^{\otimes n}
\Pi_{n,\gamma}^{\mathrm{src},A^n},
\end{align}
and the definition
\eqref{eq:IID_decoupling_converse_source_large_eigenvalue_mass}.

Combining
\eqref{eq:IID_decoupling_converse_compressed_large_eigenvalue_rank_lower_bound}
and
\eqref{eq:IID_decoupling_converse_compressed_large_eigenvalue_mass_intermediate}
yields
\begin{align}
\exp[-n\gamma]
\operatorname{Tr}
\left[
\Pi_{n,\gamma}^{\mathrm{cmp},A^n}
\right]
\le
\beta_n(\gamma).
\label{eq:IID_decoupling_converse_compressed_large_eigenvalue_rank_bound}
\end{align}
Substituting
\eqref{eq:IID_decoupling_converse_compressed_large_eigenvalue_rank_bound}
back into
\eqref{eq:IID_decoupling_converse_compressed_large_eigenvalue_mass_intermediate}
gives
\begin{align}
\operatorname{Tr}
\left[
\Pi_{n,\gamma}^{\mathrm{cmp},A^n}
\rho_{\Pi_n}^{A_{\Pi_n}^n}
\right]
\le
\frac{
2\beta_n(\gamma)
}{
1-\delta_{\Pi_n}
}.
\label{eq:IID_decoupling_converse_compressed_large_eigenvalue_mass}
\end{align}
Equations
\eqref{eq:IID_decoupling_converse_projection_success_probability_limit},
\eqref{eq:IID_decoupling_converse_source_large_eigenvalue_mass_limit},
and
\eqref{eq:IID_decoupling_converse_compressed_large_eigenvalue_mass}
therefore imply
\begin{align}
\lim_{n\to\infty}
\operatorname{Tr}
\left[
\Pi_{n,\gamma}^{\mathrm{cmp},A^n}
\rho_{\Pi_n}^{A_{\Pi_n}^n}
\right]
=
0.
\label{eq:IID_decoupling_converse_projected_large_eigenvalue_mass_limit}
\end{align}

Let
\begin{align}
\nu_{n,0}
\ge
\nu_{n,1}
\ge
\cdots
\ge
0
\end{align}
be the eigenvalues of
\(\rho_{\Pi_n}^{A_{\Pi_n}^n}\),
counted with multiplicity and extended by zeros.
By the definition of
\(\Pi_{n,\gamma}^{\mathrm{cmp},A^n}\),
\begin{align}
\operatorname{Tr}
\left[
\Pi_{n,\gamma}^{\mathrm{cmp},A^n}
\rho_{\Pi_n}^{A_{\Pi_n}^n}
\right]
=
\sum_{
k:
\nu_{n,k}
>
\frac{
2\exp[-n\gamma]
}{
1-\delta_{\Pi_n}
}
}
\nu_{n,k}.
\label{eq:IID_decoupling_converse_projected_large_eigenvalue_mass_expansion}
\end{align}
Since
\begin{align}
\sum_{k=0}^{\infty}
\nu_{n,k}
=
1,
\end{align}
it follows that
\begin{align}
&\sum_{
k:
\nu_{n,k}
\le
\frac{
2\exp[-n\gamma]
}{
1-\delta_{\Pi_n}
}
}
\nu_{n,k}
\nonumber\\
&=
1
-
\operatorname{Tr}
\left[
\Pi_{n,\gamma}^{\mathrm{cmp},A^n}
\rho_{\Pi_n}^{A_{\Pi_n}^n}
\right].
\label{eq:IID_decoupling_converse_projected_small_eigenvalue_mass}
\end{align}

For every index satisfying
\begin{align}
\nu_{n,k}
\le
\frac{
2\exp[-n\gamma]
}{
1-\delta_{\Pi_n}
},
\end{align}
one has
\begin{align}
-\log
\nu_{n,k}
\ge
n\gamma
-
\log 2
+
\log
\left(
1-\delta_{\Pi_n}
\right).
\label{eq:IID_decoupling_converse_small_eigenvalue_information_bound}
\end{align}
Therefore,
\begin{align}
H
\left(
A_{\Pi_n}^n
\right)_{\rho_{\Pi_n}}
&=
\sum_{k=0}^{\infty}
\nu_{n,k}
\left(
-\log
\nu_{n,k}
\right)
\nonumber\\
&\ge
\sum_{
k:
\nu_{n,k}
\le
\frac{
2\exp[-n\gamma]
}{
1-\delta_{\Pi_n}
}
}
\nu_{n,k}
\left[
n\gamma
-
\log 2
+
\log
\left(
1-\delta_{\Pi_n}
\right)
\right]
\nonumber\\
&=
\left[
1
-
\operatorname{Tr}
\left[
\Pi_{n,\gamma}^{\mathrm{cmp},A^n}
\rho_{\Pi_n}^{A_{\Pi_n}^n}
\right]
\right]
\nonumber\\
&\qquad\times
\left[
n\gamma
-
\log 2
+
\log
\left(
1-\delta_{\Pi_n}
\right)
\right].
\label{eq:IID_decoupling_converse_projected_entropy_lower_bound}
\end{align}
Dividing
\eqref{eq:IID_decoupling_converse_projected_entropy_lower_bound}
by \(n\), and using
\eqref{eq:IID_decoupling_converse_projection_success_assumption}
and
\eqref{eq:IID_decoupling_converse_projected_large_eigenvalue_mass_limit},
gives
\begin{align}
\liminf_{n\to\infty}
\frac{1}{n}
H
\left(
A_{\Pi_n}^n
\right)_{\rho_{\Pi_n}}
\ge
\gamma.
\end{align}
Since this holds for every
\begin{align}
0
<
\gamma
<
H(A)_\rho,
\end{align}
taking
\(\gamma\uparrow H(A)_\rho\)
yields
\begin{align}
\liminf_{n\to\infty}
\frac{1}{n}
H
\left(
A_{\Pi_n}^n
\right)_{\rho_{\Pi_n}}
\ge
H(A)_\rho.
\label{eq:IID_decoupling_converse_projected_entropy_lower_limit}
\end{align}

When
\(H(A)_\rho=0\),
\eqref{eq:IID_decoupling_converse_projected_entropy_lower_limit}
follows directly from nonnegativity of the entropy.

\medskip
\noindent
\textbf{Upper bound in
\eqref{eq:IID_decoupling_converse_projected_entropy_limit}.}
Define the self-adjoint unitary
\begin{align}
U_n^{A^n,\mathrm{pin}}
\coloneqq
\Pi_n^{A^n}
-
\Pi_n^{\perp A^n},
\label{eq:IID_decoupling_converse_pinching_unitary}
\end{align}
and the pinched marginal state
\begin{align}
\overline\rho_n^{A^n}
&\coloneqq
\Pi_n^{A^n}
\left(
\rho^A
\right)^{\otimes n}
\Pi_n^{A^n}
\nonumber\\
&\quad+
\Pi_n^{\perp A^n}
\left(
\rho^A
\right)^{\otimes n}
\Pi_n^{\perp A^n}
\nonumber\\
&=
\frac{1}{2}
\left(
\rho^A
\right)^{\otimes n}
+
\frac{1}{2}
U_n^{A^n,\mathrm{pin}}
\left(
\rho^A
\right)^{\otimes n}
U_n^{A^n,\mathrm{pin}}.
\label{eq:IID_decoupling_converse_pinched_input_marginal}
\end{align}
Concavity of the entropy, invariance under unitary conjugation, and
the binary-mixture upper bound give
\begin{align}
nH(A)_\rho
\le
H(A^n)_{\overline\rho_n}
\le
nH(A)_\rho
+
\log 2.
\label{eq:IID_decoupling_converse_pinched_entropy_bounds}
\end{align}

Since the two branches in
\eqref{eq:IID_decoupling_converse_pinched_input_marginal}
have orthogonal supports, the block-diagonal entropy decomposition
gives
\begin{align}
H(A^n)_{\overline\rho_n}
&=
h_2
\left(
\delta_{\Pi_n}
\right)
+
\left(
1-\delta_{\Pi_n}
\right)
H
\left(
A_{\Pi_n}^n
\right)_{\rho_{\Pi_n}}
\nonumber\\
&\quad+
\delta_{\Pi_n}
H
\left(
A_{\Pi_n^\perp}^n
\right)_{\rho_{\Pi_n^\perp}}.
\label{eq:IID_decoupling_converse_pinched_entropy_decomposition}
\end{align}
When
\(\delta_{\Pi_n}=0\),
the final term is defined to be zero according to
\eqref{eq:IID_decoupling_zero_complementary_convention}.
In particular, every positive-probability branch entropy appearing in
\eqref{eq:IID_decoupling_converse_pinched_entropy_decomposition}
is finite.

Dropping the nonnegative first and third terms in
\eqref{eq:IID_decoupling_converse_pinched_entropy_decomposition}
and using
\eqref{eq:IID_decoupling_converse_pinched_entropy_bounds}
gives
\begin{align}
\left(
1-\delta_{\Pi_n}
\right)
H
\left(
A_{\Pi_n}^n
\right)_{\rho_{\Pi_n}}
\le
nH(A)_\rho
+
\log 2.
\label{eq:IID_decoupling_converse_projected_entropy_finite_upper_bound}
\end{align}
Together with
\eqref{eq:IID_decoupling_converse_projection_success_probability_limit},
this yields
\begin{align}
\limsup_{n\to\infty}
\frac{1}{n}
H
\left(
A_{\Pi_n}^n
\right)_{\rho_{\Pi_n}}
\le
H(A)_\rho.
\label{eq:IID_decoupling_converse_projected_entropy_upper_limit}
\end{align}
Combining
\eqref{eq:IID_decoupling_converse_projected_entropy_lower_limit}
and
\eqref{eq:IID_decoupling_converse_projected_entropy_upper_limit}
proves
\eqref{eq:IID_decoupling_converse_projected_entropy_limit}.

\medskip
\noindent
\textbf{Proof of
\eqref{eq:IID_decoupling_converse_weighted_complementary_entropy_limit}.}
Solving
\eqref{eq:IID_decoupling_converse_pinched_entropy_decomposition}
for the complementary contribution and using
\eqref{eq:IID_decoupling_converse_pinched_entropy_bounds}
gives
\begin{align}
0
&\le
\delta_{\Pi_n}
H
\left(
A_{\Pi_n^\perp}^n
\right)_{\rho_{\Pi_n^\perp}}
\nonumber\\
&\le
nH(A)_\rho
+
\log 2
-
\left(
1-\delta_{\Pi_n}
\right)
H
\left(
A_{\Pi_n}^n
\right)_{\rho_{\Pi_n}}.
\label{eq:IID_decoupling_converse_complementary_entropy_upper_bound}
\end{align}
Dividing
\eqref{eq:IID_decoupling_converse_complementary_entropy_upper_bound}
by \(n\), and using
\eqref{eq:IID_decoupling_converse_projection_success_probability_limit}
and
\eqref{eq:IID_decoupling_converse_projected_entropy_limit},
proves
\eqref{eq:IID_decoupling_converse_weighted_complementary_entropy_limit}.

\medskip
\noindent
\textbf{Proof of
\eqref{eq:IID_decoupling_converse_projection_rank_bound_lemma}.}
Since
\(\rho_{\Pi_n}^{A_{\Pi_n}^n}\)
is supported on the finite-dimensional space
\(\mathcal H^{A_{\Pi_n}^n}\),
the entropy--dimension bound gives
\begin{align}
H
\left(
A_{\Pi_n}^n
\right)_{\rho_{\Pi_n}}
\le
\log
\left|
A_{\Pi_n}^n
\right|.
\label{eq:IID_decoupling_converse_projected_entropy_dimension_bound}
\end{align}
Dividing
\eqref{eq:IID_decoupling_converse_projected_entropy_dimension_bound}
by \(n\), taking the lower limit, and using
\eqref{eq:IID_decoupling_converse_projected_entropy_limit}
proves
\eqref{eq:IID_decoupling_converse_projection_rank_bound_lemma}.

\medskip
\noindent
\textbf{Proof of
\eqref{eq:IID_decoupling_converse_projected_mutual_information_bound}.}
By
\eqref{eq:mutual_information_bounds_finite_marginal},
\begin{align}
I(A^n:R^n)_{\rho^{\otimes n}}
\le
2H(A^n)_{\rho^{\otimes n}}
=
2nH(A)_\rho
<
\infty.
\label{eq:IID_decoupling_converse_tensor_power_mutual_information_finite}
\end{align}

Define the pinched bipartite state
\begin{align}
\overline\rho_n^{A^nR^n}
&\coloneqq
\left(
\Pi_n^{A^n}
\otimes
I^{R^n}
\right)
\left(
\rho^{AR}
\right)^{\otimes n}
\left(
\Pi_n^{A^n}
\otimes
I^{R^n}
\right)
\nonumber\\
&\quad+
\left(
\Pi_n^{\perp A^n}
\otimes
I^{R^n}
\right)
\left(
\rho^{AR}
\right)^{\otimes n}
\left(
\Pi_n^{\perp A^n}
\otimes
I^{R^n}
\right)
\nonumber\\
&=
\frac{1}{2}
\left(
\rho^{AR}
\right)^{\otimes n}
\nonumber\\
&\quad+
\frac{1}{2}
\left(
U_n^{A^n,\mathrm{pin}}
\otimes
I^{R^n}
\right)
\left(
\rho^{AR}
\right)^{\otimes n}
\left(
U_n^{A^n,\mathrm{pin}}
\otimes
I^{R^n}
\right).
\label{eq:IID_decoupling_converse_pinched_bipartite_state}
\end{align}
Its \(A^n\)-marginal is
\(\overline\rho_n^{A^n}\)
defined in
\eqref{eq:IID_decoupling_converse_pinched_input_marginal}.
Hence,
\eqref{eq:IID_decoupling_converse_pinched_entropy_bounds}
and
\eqref{eq:mutual_information_bounds_finite_marginal}
imply
\begin{align}
I(A^n:R^n)_{\overline\rho_n}
<
\infty.
\label{eq:IID_decoupling_converse_pinched_mutual_information_finite}
\end{align}

Introduce a binary classical register \(X\), and define
\begin{align}
\overline\rho_n^{XA^nR^n}
&\coloneqq
\frac{1}{2}
\ket{0}\bra{0}^X
\otimes
\left(
\rho^{AR}
\right)^{\otimes n}
\nonumber\\
&\quad+
\frac{1}{2}
\ket{1}\bra{1}^X
\otimes
\left(
U_n^{A^n,\mathrm{pin}}
\otimes
I^{R^n}
\right)
\left(
\rho^{AR}
\right)^{\otimes n}
\left(
U_n^{A^n,\mathrm{pin}}
\otimes
I^{R^n}
\right).
\label{eq:IID_decoupling_converse_pinching_classical_extension}
\end{align}
By construction,
\begin{align}
\operatorname{Tr}_X
\left[
\overline\rho_n^{XA^nR^n}
\right]
=
\overline\rho_n^{A^nR^n}.
\label{eq:IID_decoupling_converse_pinching_classical_extension_marginal}
\end{align}

The two conditional states in
\eqref{eq:IID_decoupling_converse_pinching_classical_extension}
have the same \(R^n\)-marginal.
Consequently,
\begin{align}
I(X:R^n)_{\overline\rho_n}
=
0.
\label{eq:IID_decoupling_converse_pinching_classical_reference_independence}
\end{align}
Using the mutual-information chain rule
\eqref{eq:mutual_information_chain_rule},
the classical conditioning identity, additivity under tensor products,
and invariance under local unitaries, we obtain
\begin{align}
I(XA^n:R^n)_{\overline\rho_n}
&=
I(X:R^n)_{\overline\rho_n}
+
I(A^n:R^n|X)_{\overline\rho_n}
\nonumber\\
&=
I(A^n:R^n|X)_{\overline\rho_n}
\nonumber\\
&=
nI(A:R)_\rho.
\label{eq:IID_decoupling_converse_pinching_mutual_information_first_order}
\end{align}
Applying the chain rule
\eqref{eq:mutual_information_chain_rule}
in the opposite order and using
\eqref{eq:IID_decoupling_converse_pinching_classical_extension_marginal}
gives
\begin{align}
I(XA^n:R^n)_{\overline\rho_n}
&=
I(A^n:R^n)_{\overline\rho_n}
+
I(X:R^n|A^n)_{\overline\rho_n}.
\label{eq:IID_decoupling_converse_pinching_mutual_information_second_expansion}
\end{align}
Since \(X\) is binary and classical,
\begin{align}
0
\le
I(X:R^n|A^n)_{\overline\rho_n}
\le
H(X)_{\overline\rho_n}
=
\log 2.
\label{eq:IID_decoupling_converse_pinching_conditional_mutual_information_bound}
\end{align}
Combining
\eqref{eq:IID_decoupling_converse_pinching_mutual_information_first_order},
\eqref{eq:IID_decoupling_converse_pinching_mutual_information_second_expansion},
and
\eqref{eq:IID_decoupling_converse_pinching_conditional_mutual_information_bound}
gives
\begin{align}
I(A^n:R^n)_{\overline\rho_n}
\ge
nI(A:R)_\rho
-
\log 2.
\label{eq:IID_decoupling_converse_pinched_mutual_information_lower_bound}
\end{align}

Introduce a binary classical register \(Y\).
When
\(\delta_{\Pi_n}>0\),
define
\begin{align}
\widetilde\rho_n^{YA^nR^n}
&\coloneqq
\left(
1-\delta_{\Pi_n}
\right)
\ket{0}\bra{0}^Y
\otimes
\rho_{\Pi_n}^{A_{\Pi_n}^nR^n}
\nonumber\\
&\quad+
\delta_{\Pi_n}
\ket{1}\bra{1}^Y
\otimes
\rho_{\Pi_n^\perp}^{A_{\Pi_n^\perp}^nR^n},
\label{eq:IID_decoupling_converse_sector_state_positive_complement}
\end{align}
where the conditional states are canonically embedded into the
orthogonal subspaces
\(\mathcal H^{A_{\Pi_n}^n}\)
and
\(\mathcal H^{A_{\Pi_n^\perp}^n}\)
of
\(\mathcal H^{A^n}\).
When
\(\delta_{\Pi_n}=0\),
define
\begin{align}
\widetilde\rho_n^{YA^nR^n}
\coloneqq
\ket{0}\bra{0}^Y
\otimes
\rho_{\Pi_n}^{A_{\Pi_n}^nR^n},
\label{eq:IID_decoupling_converse_sector_state_zero_complement}
\end{align}
again using the canonical embedding into
\(\mathcal H^{A^n}\).
In either case,
\begin{align}
\operatorname{Tr}_Y
\left[
\widetilde\rho_n^{YA^nR^n}
\right]
=
\overline\rho_n^{A^nR^n}.
\label{eq:IID_decoupling_converse_sector_state_marginal}
\end{align}

The sector label \(Y\) is determined by the orthogonal support sectors
of the \(A^n\)-system.
Equivalently,
\(\widetilde\rho_n^{YA^nR^n}\)
is obtained from
\(\overline\rho_n^{A^nR^n}\)
by a local isometry on \(A^n\).
Therefore, invariance of mutual information under local isometries
gives
\begin{align}
I(A^n:R^n)_{\overline\rho_n}
=
I(YA^n:R^n)_{\widetilde\rho_n}.
\label{eq:IID_decoupling_converse_sector_isometric_mutual_information_identity}
\end{align}
Applying the chain rule
\eqref{eq:mutual_information_chain_rule}
to the right-hand side of
\eqref{eq:IID_decoupling_converse_sector_isometric_mutual_information_identity}
and using the classical conditioning identity gives
\begin{align}
I(A^n:R^n)_{\overline\rho_n}
&=
I(Y:R^n)_{\widetilde\rho_n}
+
\left(
1-\delta_{\Pi_n}
\right)
I
\left(
A_{\Pi_n}^n:R^n
\right)_{\rho_{\Pi_n}}
\nonumber\\
&\quad+
\delta_{\Pi_n}
I
\left(
A_{\Pi_n^\perp}^n:R^n
\right)_{\rho_{\Pi_n^\perp}}.
\label{eq:IID_decoupling_converse_mutual_information_branch_decomposition}
\end{align}
When
\(\delta_{\Pi_n}=0\),
the final term is understood to be zero according to
\eqref{eq:IID_decoupling_zero_complementary_convention}.

Since \(Y\) is classical,
\begin{align}
I(Y:R^n)_{\widetilde\rho_n}
\le
H(Y)_{\widetilde\rho_n}
=
h_2
\left(
\delta_{\Pi_n}
\right).
\label{eq:IID_decoupling_converse_sector_mutual_information_bound}
\end{align}
Furthermore, when
\(\delta_{\Pi_n}>0\),
the complementary marginal entropy is finite by
\eqref{eq:IID_decoupling_converse_pinched_entropy_decomposition}.
Hence,
\eqref{eq:mutual_information_bounds_finite_marginal}
gives
\begin{align}
I
\left(
A_{\Pi_n^\perp}^n:R^n
\right)_{\rho_{\Pi_n^\perp}}
\le
2
H
\left(
A_{\Pi_n^\perp}^n
\right)_{\rho_{\Pi_n^\perp}}.
\label{eq:IID_decoupling_converse_complementary_mutual_information_bound}
\end{align}

Combining
\eqref{eq:IID_decoupling_converse_pinched_mutual_information_lower_bound},
\eqref{eq:IID_decoupling_converse_mutual_information_branch_decomposition},
\eqref{eq:IID_decoupling_converse_sector_mutual_information_bound}, and
\eqref{eq:IID_decoupling_converse_complementary_mutual_information_bound}
yields
\begin{align}
&\left(
1-\delta_{\Pi_n}
\right)
I
\left(
A_{\Pi_n}^n:R^n
\right)_{\rho_{\Pi_n}}
\nonumber\\
&\ge
nI(A:R)_\rho
-
\log 2
-
h_2
\left(
\delta_{\Pi_n}
\right)
\nonumber\\
&\quad-
2\delta_{\Pi_n}
H
\left(
A_{\Pi_n^\perp}^n
\right)_{\rho_{\Pi_n^\perp}}.
\label{eq:IID_decoupling_converse_projected_mutual_information_finite_bound}
\end{align}
Dividing
\eqref{eq:IID_decoupling_converse_projected_mutual_information_finite_bound}
by \(n\), using
\eqref{eq:IID_decoupling_converse_projection_success_probability_limit},
\eqref{eq:IID_decoupling_converse_weighted_complementary_entropy_limit},
and
\begin{align}
0
\le
h_2
\left(
\delta_{\Pi_n}
\right)
\le
\log 2,
\end{align}
and taking the lower limit proves
\eqref{eq:IID_decoupling_converse_projected_mutual_information_bound}.
\end{proof}

We now prove the converse bound on the rates achievable by an
arbitrary partial-trace protocol for infinite-dimensional decoupling.

\begin{theorem}[Converse for infinite-dimensional IID decoupling]
\label{thm:IID_decoupling_converse}
Let
\(\mathcal H^A\)
and
\(\mathcal H^R\)
be separable and possibly infinite-dimensional, and let
\begin{align}
\rho^{AR}
\in
\operatorname{D}
\left(
\mathcal H^A
\otimes
\mathcal H^R
\right)
\end{align}
satisfy
\begin{align}
H(A)_\rho
<
\infty.
\end{align}
For every
\(n\in\{1,2,\ldots\}\),
let
\(\Pi_n^{A^n}\)
be an arbitrary nonzero finite-rank projection satisfying
\begin{align}
\delta_{\Pi_n}
<
1,
\end{align}
where
\(\delta_{\Pi_n}\),
\(\mathcal H^{A_{\Pi_n}^n}\),
\(\lvert A_{\Pi_n}^n\rvert\),
\(\rho_{\Pi_n}^{A_{\Pi_n}^nR^n}\), and
\(\rho_{\Pi_n}^{R^n}\)
are defined in
\eqref{eq:IID_decoupling_projection_error},
\eqref{eq:IID_decoupling_projected_input_space},
\eqref{eq:IID_decoupling_projected_input_dimension},
\eqref{eq:IID_decoupling_projected_state}, and
\eqref{eq:IID_decoupling_projected_reference},
respectively.
Suppose that
\begin{align}
\lim_{n\to\infty}
\delta_{\Pi_n}
=
0.
\label{eq:IID_decoupling_converse_projection_error_assumption}
\end{align}
For every \(n\), let
\(\mathcal H^{E_n}\)
and
\(\mathcal H^{M_n}\)
be finite-dimensional, let
\begin{align}
V_n^{A_{\Pi_n}^n\to E_nM_n}
:
\mathcal H^{A_{\Pi_n}^n}
\longrightarrow
\mathcal H^{E_n}
\otimes
\mathcal H^{M_n}
\end{align}
be a unitary isomorphism as in
\eqref{eq:IID_decoupling_partial_trace_factorization},
and consider an arbitrary unitary operator
\begin{align}
U_n
\in
\operatorname{U}
\left(
\mathcal H^{A_{\Pi_n}^n}
\right).
\end{align}
Let
\(\mathcal T_n^{A_{\Pi_n}^n\to E_n}\)
and
\(\tau_{n,U_n}^{E_nR^n}\)
be defined in
\eqref{eq:IID_decoupling_partial_trace_channel}
and
\eqref{eq:IID_decoupling_output_state},
respectively.

Suppose that the output decouples from the original reference
marginal, in the sense that
\begin{align}
\lim_{n\to\infty}
D
\left(
\tau_{n,U_n}^{E_nR^n}
\middle\|
\pi^{E_n}
\otimes
\left(
\rho^R
\right)^{\otimes n}
\right)
=
0.
\label{eq:IID_decoupling_converse_original_error_assumption}
\end{align}
Then
\begin{align}
\liminf_{n\to\infty}
\frac{1}{n}
\log
\left|
A_{\Pi_n}^n
\right|
&\ge
H(A)_\rho,
\label{eq:IID_decoupling_converse_projection_rank_rate}
\\
\liminf_{n\to\infty}
\frac{1}{n}
\log
\left|
M_n
\right|
&\ge
\frac{1}{2}
I(A:R)_\rho.
\label{eq:IID_decoupling_converse_discarded_rate}
\end{align}
\end{theorem}

\begin{proof}
The projection-rank bound
\eqref{eq:IID_decoupling_converse_projection_rank_rate}
follows from
Proposition~\ref{prop:IID_decoupling_high_probability_projection_converse},
specifically from
\eqref{eq:IID_decoupling_converse_projection_rank_bound_lemma}.

It remains to prove
\eqref{eq:IID_decoupling_converse_discarded_rate}.
Define the dilated output state by
\begin{align}
&\tau_{n,U_n}^{E_nM_nR^n}
\nonumber\\
&\coloneqq
\left(
V_n
U_n
\otimes
I^{R^n}
\right)
\rho_{\Pi_n}^{A_{\Pi_n}^nR^n}
\left(
U_n^\dagger
V_n^\dagger
\otimes
I^{R^n}
\right).
\label{eq:IID_decoupling_converse_dilated_output}
\end{align}
By
\eqref{eq:IID_decoupling_partial_trace_channel}
and
\eqref{eq:IID_decoupling_output_state},
\begin{align}
\tau_{n,U_n}^{E_nR^n}
=
\operatorname{Tr}_{M_n}
\left[
\tau_{n,U_n}^{E_nM_nR^n}
\right].
\label{eq:IID_decoupling_converse_dilated_output_marginal}
\end{align}
Invariance of mutual information under local unitaries and unitary
isomorphisms gives
\begin{align}
I
\left(
E_nM_n:R^n
\right)_{\tau_{n,U_n}}
=
I
\left(
A_{\Pi_n}^n:R^n
\right)_{\rho_{\Pi_n}}.
\label{eq:IID_decoupling_converse_dilated_mutual_information_identity}
\end{align}

By
\eqref{eq:IID_decoupling_output_reference},
\begin{align}
\tau_{n,U_n}^{R^n}
=
\rho_{\Pi_n}^{R^n}.
\label{eq:IID_decoupling_converse_output_reference_identity}
\end{align}
The relative-entropy chain-rule identity
\eqref{eq:IID_decoupling_original_projected_error_decomposition}
therefore gives
\begin{align}
&D
\left(
\tau_{n,U_n}^{E_nR^n}
\middle\|
\pi^{E_n}
\otimes
\left(
\rho^R
\right)^{\otimes n}
\right)
\nonumber\\
&=
D
\left(
\tau_{n,U_n}^{E_nR^n}
\middle\|
\pi^{E_n}
\otimes
\rho_{\Pi_n}^{R^n}
\right)
+
D
\left(
\rho_{\Pi_n}^{R^n}
\middle\|
\left(
\rho^R
\right)^{\otimes n}
\right).
\label{eq:IID_decoupling_converse_original_error_decomposition}
\end{align}
Both terms on the right-hand side of
\eqref{eq:IID_decoupling_converse_original_error_decomposition}
are nonnegative.
Consequently,
\begin{align}
&D
\left(
\tau_{n,U_n}^{E_nR^n}
\middle\|
\pi^{E_n}
\otimes
\rho_{\Pi_n}^{R^n}
\right)
\nonumber\\
&\le
D
\left(
\tau_{n,U_n}^{E_nR^n}
\middle\|
\pi^{E_n}
\otimes
\left(
\rho^R
\right)^{\otimes n}
\right).
\label{eq:IID_decoupling_converse_projected_error_bounded_by_original_error}
\end{align}

Using
\eqref{eq:IID_decoupling_converse_output_reference_identity},
the product-reference decomposition of relative entropy gives
\begin{align}
&D
\left(
\tau_{n,U_n}^{E_nR^n}
\middle\|
\pi^{E_n}
\otimes
\rho_{\Pi_n}^{R^n}
\right)
\nonumber\\
&=
D
\left(
\tau_{n,U_n}^{E_nR^n}
\middle\|
\tau_{n,U_n}^{E_n}
\otimes
\rho_{\Pi_n}^{R^n}
\right)
+
D
\left(
\tau_{n,U_n}^{E_n}
\middle\|
\pi^{E_n}
\right)
\nonumber\\
&=
I
\left(
E_n:R^n
\right)_{\tau_{n,U_n}}
+
D
\left(
\tau_{n,U_n}^{E_n}
\middle\|
\pi^{E_n}
\right).
\label{eq:IID_decoupling_converse_projected_error_mutual_information_decomposition}
\end{align}
Nonnegativity of relative entropy, together with
\eqref{eq:IID_decoupling_converse_projected_error_bounded_by_original_error}
and
\eqref{eq:IID_decoupling_converse_projected_error_mutual_information_decomposition},
therefore yields
\begin{align}
I
\left(
E_n:R^n
\right)_{\tau_{n,U_n}}
\le
D
\left(
\tau_{n,U_n}^{E_nR^n}
\middle\|
\pi^{E_n}
\otimes
\left(
\rho^R
\right)^{\otimes n}
\right).
\label{eq:IID_decoupling_converse_output_mutual_information_bound}
\end{align}

The mutual-information chain rule
\eqref{eq:mutual_information_chain_rule}
gives
\begin{align}
I
\left(
E_nM_n:R^n
\right)_{\tau_{n,U_n}}
&=
I
\left(
E_n:R^n
\right)_{\tau_{n,U_n}}
\nonumber\\
&\quad+
I
\left(
M_n:R^n
\middle|
E_n
\right)_{\tau_{n,U_n}}.
\label{eq:IID_decoupling_converse_output_mutual_information_chain_rule}
\end{align}

Since \(M_n\) is finite-dimensional,
\eqref{eq:conditional_entropy_bounds_finite_marginal}
gives
\begin{align}
-\log
\left|
M_n
\right|
&\le
H
\left(
M_n
\middle|
E_n
\right)_{\tau_{n,U_n}}
\le
\log
\left|
M_n
\right|,
\label{eq:IID_decoupling_converse_first_conditional_entropy_dimension_bound}
\\
-\log
\left|
M_n
\right|
&\le
H
\left(
M_n
\middle|
E_nR^n
\right)_{\tau_{n,U_n}}
\le
\log
\left|
M_n
\right|.
\label{eq:IID_decoupling_converse_second_conditional_entropy_dimension_bound}
\end{align}
Consequently,
\begin{align}
I
\left(
M_n:R^n
\middle|
E_n
\right)_{\tau_{n,U_n}}
&=
H
\left(
M_n
\middle|
E_n
\right)_{\tau_{n,U_n}}
-
H
\left(
M_n
\middle|
E_nR^n
\right)_{\tau_{n,U_n}}
\nonumber\\
&\le
2
\log
\left|
M_n
\right|.
\label{eq:IID_decoupling_converse_conditional_mutual_information_dimension_bound}
\end{align}

Combining
\eqref{eq:IID_decoupling_converse_dilated_mutual_information_identity},
\eqref{eq:IID_decoupling_converse_output_mutual_information_bound},
\eqref{eq:IID_decoupling_converse_output_mutual_information_chain_rule},
and
\eqref{eq:IID_decoupling_converse_conditional_mutual_information_dimension_bound}
gives
\begin{align}
I
\left(
A_{\Pi_n}^n:R^n
\right)_{\rho_{\Pi_n}}
&\le
D
\left(
\tau_{n,U_n}^{E_nR^n}
\middle\|
\pi^{E_n}
\otimes
\left(
\rho^R
\right)^{\otimes n}
\right)
\nonumber\\
&\quad+
2
\log
\left|
M_n
\right|.
\label{eq:IID_decoupling_converse_finite_discarded_rate_bound}
\end{align}
Equivalently,
\begin{align}
\frac{1}{n}
\log
\left|
M_n
\right|
&\ge
\frac{1}{2n}
I
\left(
A_{\Pi_n}^n:R^n
\right)_{\rho_{\Pi_n}}
\nonumber\\
&\quad-
\frac{1}{2n}
D
\left(
\tau_{n,U_n}^{E_nR^n}
\middle\|
\pi^{E_n}
\otimes
\left(
\rho^R
\right)^{\otimes n}
\right).
\label{eq:IID_decoupling_converse_normalized_finite_discarded_rate_bound}
\end{align}

Taking the lower limit in
\eqref{eq:IID_decoupling_converse_normalized_finite_discarded_rate_bound},
and using
\eqref{eq:IID_decoupling_converse_projected_mutual_information_bound}
from
Proposition~\ref{prop:IID_decoupling_high_probability_projection_converse}
together with
\eqref{eq:IID_decoupling_converse_original_error_assumption},
gives
\begin{align}
\liminf_{n\to\infty}
\frac{1}{n}
\log
\left|
M_n
\right|
\ge
\frac{1}{2}
I(A:R)_\rho.
\end{align}
This is
\eqref{eq:IID_decoupling_converse_discarded_rate}.
\end{proof}

\section{Infinite-dimensional IID quantum state merging}
\label{sec:infinite_dimensional_IID_quantum_state_merging}

In this section, we analyze IID quantum state merging for normalized
pure states on separable and possibly infinite-dimensional Hilbert
spaces.
In Subsection~\ref{subsec:formulation_infinite_dimensional_IID_state_merging},
we define the task of infinite-dimensional IID quantum state merging.
In
Subsection~\ref{subsec:achievability_infinite_dimensional_IID_quantum_state_merging},
we apply the IID partial-trace decoupling construction from
Section~\ref{sec:infinite_dimensional_IID_decoupling}.
The finite-rank projections, finite-dimensional factorizations, and
unitaries are chosen jointly by the direct decoupling theorem.
The unitary randomization acts only on the selected common
components.

\subsection{Formulation of infinite-dimensional IID quantum state merging}
\label{subsec:formulation_infinite_dimensional_IID_state_merging}

We formulate quantum state merging for separable and possibly
infinite-dimensional systems.
Whereas quantum state merging was originally formulated as a task for
entanglement-assisted local operations and classical communication (LOCC)
\cite{horodecki2005partial,horodecki2007quantum}, we consider a fully
quantum formulation with one-way quantum communication from Alice to
Bob.
In the finite-dimensional setting, this formulation is also known as
the fully quantum Slepian--Wolf protocol
\cite{abeyesinghe2009mother}.
In this setting, no preshared entanglement and no communication other than the
transmission of a finite-dimensional quantum message system are
available.
In particular, no separate classical message may be transmitted.
Any classical label, measurement outcome, or flag that must be
communicated in a particular implementation must therefore be encoded
into mutually orthogonal states of the quantum message system and is
included in its dimension.
Alice and Bob may use arbitrary local ancillary systems initialized
in fixed pure states, local randomness, and local systems that are
discarded after their respective operations.
These local resources are incorporated into the encoding and decoding
channels below.

Let
\begin{align}
\ket{\psi}^{ABR}
\in
\mathcal H^A
\otimes
\mathcal H^B
\otimes
\mathcal H^R
\label{eq:IID_state_merging_source_state}
\end{align}
be a normalized pure state, where
\(\mathcal H^A\),
\(\mathcal H^B\), and
\(\mathcal H^R\)
are separable and possibly infinite-dimensional.
Define
\begin{align}
\psi^{ABR}
\coloneqq
\ket{\psi}\bra{\psi}^{ABR}.
\label{eq:IID_state_merging_source_density_operator}
\end{align}
Alice initially holds \(A\), Bob initially holds \(B\), and \(R\) is
an inaccessible reference system.

Let
\begin{align}
\mathcal H^{A'}
\simeq
\mathcal H^A
\label{eq:IID_state_merging_transferred_system_identification}
\end{align}
with a fixed identification.
Using this identification, let
\(\ket{\psi}^{A'BR}\)
denote the vector corresponding to
\(\ket{\psi}^{ABR}\), and define
\begin{align}
\psi^{A'BR}
\coloneqq
\ket{\psi}\bra{\psi}^{A'BR}.
\label{eq:IID_state_merging_transferred_source_state}
\end{align}
Throughout this section, we assume that
\begin{align}
H(A)_\psi
<
\infty,
\label{eq:IID_state_merging_finite_entropy_assumption}
\end{align}
which may be viewed as an energy constraint, as discussed in
\eqref{eq:finite_entropy_finite_energy_condition}.

For every
\(n\in\{1,2,\ldots\}\),
let
\(E_n\)
and
\(E_n'\)
be finite-dimensional systems held by Alice and Bob, respectively, at
the end of the protocol.
Fix an identification
\begin{align}
\mathcal H^{E_n'}
\simeq
\mathcal H^{E_n},
\end{align}
and let
\begin{align}
\Phi_n^{E_nE_n'}
\end{align}
be the normalized maximally entangled state of Schmidt rank
\begin{align}
\left|
E_n
\right|
=
\left|
E_n'
\right|.
\label{eq:IID_state_merging_output_entanglement_rank}
\end{align}
The system \(E_nE_n'\) represents entanglement generated by the
protocol.
No entangled state shared between Alice and Bob is supplied as an
input resource.
The ideal target state is
\begin{align}
\Phi_n^{E_nE_n'}
\otimes
\left(
\psi^{A'BR}
\right)^{\otimes n}.
\label{eq:IID_state_merging_target_state}
\end{align}
Thus, the protocol is required both to transfer Alice's source systems
\(A^n\) coherently to the systems \(A'^n\) held by Bob and to generate
a maximally entangled state between the residual systems \(E_n\) held
by Alice and \(E_n'\) held by Bob.

An \(n\)-copy protocol for quantum state merging consists of a
finite-dimensional quantum message system \(M_n\), an encoding channel
implemented by Alice, transmission of \(M_n\) from Alice to Bob, and a
decoding channel implemented by Bob.
The encoding channel is a completely positive and trace-preserving
linear map
\begin{align}
\mathcal E_n^{
A^n
\to
E_nM_n
}
:
\operatorname{T}
\left(
\mathcal H^{A^n}
\right)
\longrightarrow
\operatorname{T}
\left(
\mathcal H^{E_n}
\otimes
\mathcal H^{M_n}
\right),
\label{eq:IID_state_merging_encoding_channel}
\end{align}
and the decoding channel is a completely positive and trace-preserving
linear map
\begin{align}
\mathcal D_n^{
M_nB^n
\to
A'^nB^nE_n'
}
:
\operatorname{T}
\left(
\mathcal H^{M_n}
\otimes
\mathcal H^{B^n}
\right)
\longrightarrow
\operatorname{T}
\left(
\mathcal H^{A'^n}
\otimes
\mathcal H^{B^n}
\otimes
\mathcal H^{E_n'}
\right).
\label{eq:IID_state_merging_decoding_channel}
\end{align}
No communication other than the transmission of \(M_n\) from Alice
to Bob occurs between
\(\mathcal E_n\)
and
\(\mathcal D_n\).

The adjoints of the channels in
\eqref{eq:IID_state_merging_encoding_channel}
and
\eqref{eq:IID_state_merging_decoding_channel}
are normal unital completely positive maps on the corresponding
bounded-operator spaces.
By applying the infinite-dimensional Stinespring dilation theorem to
these adjoint maps
\cite[Theorem~34.7]{conway2025course},
one can choose separable local systems \(L_n\) and \(L_n'\) to be discarded by Alice and Bob, respectively, after the corresponding isometric implementations.
There exist isometries
\begin{align}
V_{\mathcal E,n}^{
A^n
\to
E_nM_nL_n
}
:
\mathcal H^{A^n}
\longrightarrow
\mathcal H^{E_n}
\otimes
\mathcal H^{M_n}
\otimes
\mathcal H^{L_n}
\label{eq:IID_state_merging_encoding_isometry}
\end{align}
and
\begin{align}
V_{\mathcal D,n}^{
M_nB^n
\to
A'^nB^nE_n'L_n'
}
:
\mathcal H^{M_n}
\otimes
\mathcal H^{B^n}
\longrightarrow
\mathcal H^{A'^n}
\otimes
\mathcal H^{B^n}
\otimes
\mathcal H^{E_n'}
\otimes
\mathcal H^{L_n'}
\label{eq:IID_state_merging_decoding_isometry}
\end{align}
satisfying
\begin{align}
V_{\mathcal E,n}^\dagger
V_{\mathcal E,n}
&=
I^{A^n},
&
V_{\mathcal D,n}^\dagger
V_{\mathcal D,n}
&=
I^{M_nB^n},
\label{eq:IID_state_merging_isometry_conditions}
\end{align}
such that the encoding and decoding channels are recovered by
discarding \(L_n\) and \(L_n'\), respectively:
\begin{align}
\mathcal E_n(X)
&=
\operatorname{Tr}_{L_n}
\left[
V_{\mathcal E,n}
X
V_{\mathcal E,n}^\dagger
\right],
\label{eq:IID_state_merging_encoding_isometric_implementation}
\\
\mathcal D_n(Y)
&=
\operatorname{Tr}_{L_n'}
\left[
V_{\mathcal D,n}
Y
V_{\mathcal D,n}^\dagger
\right].
\label{eq:IID_state_merging_decoding_isometric_implementation}
\end{align}
The systems \(L_n\) and \(L_n'\) are local environments.
They are neither communicated nor retained as output systems and are not included in the communication or entanglement rates.

The final state produced by the protocol is
\begin{align}
&\zeta_{\mathcal E_n,\mathcal D_n}^{
E_nA'^nB^nE_n'R^n
}
\nonumber\\
&\coloneqq
\left(
\id^{E_nR^n}
\otimes
\mathcal D_n^{
M_nB^n
\to
A'^nB^nE_n'
}
\right)
\Biggl[
\left(
\mathcal E_n^{
A^n
\to
E_nM_n
}
\otimes
\id^{B^nR^n}
\right)
\left[
\left(
\psi^{ABR}
\right)^{\otimes n}
\right]
\Biggr],
\label{eq:IID_state_merging_final_state}
\end{align}
where canonical reorderings of tensor factors are understood.

The error in quantum state merging is quantified by the purified distance
\begin{align}
P
\Biggl(
\zeta_{\mathcal E_n,\mathcal D_n}^{
E_nA'^nB^nE_n'R^n
},
\Phi_n^{E_nE_n'}
\otimes
\left(
\psi^{A'BR}
\right)^{\otimes n}
\Biggr).
\label{eq:IID_state_merging_error}
\end{align}
A sequence of protocols achieves IID quantum state merging if
\begin{align}
&\lim_{n\to\infty}
P
\Biggl(
\zeta_{\mathcal E_n,\mathcal D_n}^{
E_nA'^nB^nE_n'R^n
},
\Phi_n^{E_nE_n'}
\otimes
\left(
\psi^{A'BR}
\right)^{\otimes n}
\Biggr)
=
0.
\label{eq:IID_state_merging_convergence_criterion}
\end{align}

The quantum communication cost of the \(n\)-copy protocol is defined
by
\begin{align}
q_n
\coloneqq
\log
\left|
M_n
\right|.
\label{eq:IID_state_merging_quantum_communication_cost}
\end{align}
Since no entanglement is initially supplied, the amount of
entanglement generated by the protocol is defined by
\begin{align}
e_n
\coloneqq
\log
\left|
E_n
\right|
=
\log
\left|
E_n'
\right|.
\label{eq:IID_state_merging_net_entanglement_cost}
\end{align}
The total cost is defined as the quantum communication cost minus the
generated-entanglement yield:
\begin{align}
q_n-e_n
&\coloneqq
\log
\left|
M_n
\right|
-
\log
\left|
E_n
\right|.
\label{eq:IID_state_merging_total_cost}
\end{align}
Although the quantum communication cost \(q_n\) is nonnegative, the total cost \(q_n-e_n\) can be negative when the amount of entanglement generated by the protocol exceeds its quantum communication cost. Operationally, \(q_n-e_n\) represents the net entanglement consumption when the transmission of a quantum system of dimension \(\exp[q_n]\) is implemented via quantum teleportation: teleportation consumes a maximally entangled state of Schmidt rank \(\exp[q_n]\), whereas the protocol generates one of Schmidt rank \(\exp[e_n]\).

A quantum-communication and total-cost rate pair
\begin{align}
\left(
q,
q-e
\right)
\in
[0,\infty)
\times
\mathbb R
\end{align}
is achievable if there exists a sequence of protocols satisfying
\eqref{eq:IID_state_merging_convergence_criterion} and
\begin{align}
\liminf_{n\to\infty}
\max
\left\{
\frac{q_n}{n}-q,
\frac{q_n-e_n}{n}-(q-e)
\right\}
\le
0.
\label{eq:IID_state_merging_achievable_joint_rate}
\end{align}
Note that this definition implies the existence of a subsequence of protocols along which both rate constraints are satisfied asymptotically.
We define the achievable rate region for IID quantum state merging by
\begin{align}
\mathcal R(\psi)
\coloneqq
\left\{
\left(
q,
q-e
\right)
\in
[0,\infty)
\times
\mathbb R:
\left(
q,
q-e
\right)
\text{ is achievable}
\right\}.
\label{eq:IID_state_merging_achievable_region_definition}
\end{align}

\subsection{Bounds on infinite-dimensional IID quantum state merging}
\label{subsec:achievability_infinite_dimensional_IID_quantum_state_merging}

We construct an infinite-dimensional IID quantum state merging
protocol from the infinite-dimensional IID partial-trace decoupling
protocol established in
Section~\ref{sec:infinite_dimensional_IID_decoupling}.
As in the finite-dimensional state merging in
Refs.~\cite{horodecki2005partial,horodecki2007quantum,abeyesinghe2009mother},
the retained output system \(E_n\) of the decoupling protocol becomes
Alice's share of the maximally entangled state generated by the state
merging protocol, whereas the discarded output system \(M_n\) is
transmitted to Bob as the quantum message.
Once \(E_n\) is decoupled from the reference system, Uhlmann's
theorem~\cite{UHLMANN1976273}
provides a decoding isometry on Bob's side that coherently reconstructs
Alice's source systems and Bob's share of the generated entanglement.
No additional classical communication is introduced in this
conversion.
In the construction underlying
Theorem~\ref{thm:IID_decoupling_original_direct},
the common--rare pattern sector and every finite-dimensional factor
not retained in \(E_n\) are contained in the discarded system \(M_n\).
These systems are therefore transmitted as part of the quantum
message but do not affect the asymptotically achievable rate.

Define the \(AR\)-marginal of the source state by
\begin{align}
\psi^{AR}
\coloneqq
\operatorname{Tr}_B
\left[
\psi^{ABR}
\right].
\label{eq:IID_state_merging_AR_marginal}
\end{align}
Its \(A\)-marginal coincides with that of
\(\psi^{ABR}\), and hence the finite-entropy assumption
\eqref{eq:IID_state_merging_finite_entropy_assumption}
gives
\begin{align}
H(A)_{\psi^{AR}}
=
H(A)_\psi
<
\infty.
\label{eq:IID_state_merging_AR_finite_entropy}
\end{align}
Applying
Theorem~\ref{thm:IID_decoupling_original_direct}
to \(\psi^{AR}\), we obtain nonzero finite-rank projections
\begin{align}
\left\{
\Pi_n^{A^n}
\right\}_{n=1}^{\infty},
\end{align}
finite-dimensional systems
\begin{align}
\left\{
E_n
\right\}_{n=1}^{\infty},
\qquad
\left\{
M_n
\right\}_{n=1}^{\infty},
\end{align}
unitary isomorphisms
\begin{align}
V_n^{A_{\Pi_n}^n\to E_nM_n}
:
\mathcal H^{A_{\Pi_n}^n}
\longrightarrow
\mathcal H^{E_n}
\otimes
\mathcal H^{M_n},
\label{eq:IID_state_merging_decoupling_factorization}
\end{align}
and unitaries
\begin{align}
U_n
\in
\operatorname{U}
\left(
\mathcal H^{A_{\Pi_n}^n}
\right).
\label{eq:IID_state_merging_decoupling_unitary}
\end{align}

For every
\(n\in\{1,2,\ldots\}\),
define
\begin{align}
\Pi_n^{\perp A^n}
&\coloneqq
I^{A^n}
-
\Pi_n^{A^n},
\label{eq:IID_state_merging_projection_complement}
\\
\delta_{\Pi_n}
&\coloneqq
\operatorname{Tr}
\left[
\Pi_n^{\perp A^n}
\left(
\psi^A
\right)^{\otimes n}
\right],
\label{eq:IID_state_merging_projection_error}
\\
\mathcal H^{A_{\Pi_n}^n}
&\coloneqq
\Pi_n^{A^n}
\mathcal H^{A^n},
\label{eq:IID_state_merging_projected_space}
\\
\left|
A_{\Pi_n}^n
\right|
&\coloneqq
\operatorname{rank}
\Pi_n^{A^n},
\label{eq:IID_state_merging_projected_dimension}
\\
\mathcal H^{A_{\Pi_n^\perp}^n}
&\coloneqq
\Pi_n^{\perp A^n}
\mathcal H^{A^n}.
\label{eq:IID_state_merging_complementary_space}
\end{align}
The projections supplied by
Theorem~\ref{thm:IID_decoupling_original_direct}
satisfy
\begin{align}
\delta_{\Pi_n}
<
1
\label{eq:IID_state_merging_projection_error_strict_bound}
\end{align}
for every \(n\).

Define the normalized projected pure state by
\begin{align}
\ket{\psi_{\Pi_n}}^{
A_{\Pi_n}^nB^nR^n
}
\coloneqq
\frac{
\left(
\Pi_n^{A^n}
\otimes
I^{B^nR^n}
\right)
\ket{\psi}^{\otimes n}
}{
\sqrt{
1-\delta_{\Pi_n}
}
},
\label{eq:IID_state_merging_projected_pure_state}
\end{align}
and its density operator by
\begin{align}
\psi_{\Pi_n}^{
A_{\Pi_n}^nB^nR^n
}
\coloneqq
\ket{\psi_{\Pi_n}}
\bra{\psi_{\Pi_n}}^{
A_{\Pi_n}^nB^nR^n
}.
\label{eq:IID_state_merging_projected_pure_density_operator}
\end{align}
Its \(A_{\Pi_n}^nR^n\)- and \(R^n\)-marginals are
\begin{align}
\psi_{\Pi_n}^{A_{\Pi_n}^nR^n}
&\coloneqq
\operatorname{Tr}_{B^n}
\left[
\psi_{\Pi_n}^{
A_{\Pi_n}^nB^nR^n
}
\right]
\nonumber\\
&=
\frac{
\left(
\Pi_n^{A^n}
\otimes
I^{R^n}
\right)
\left(
\psi^{AR}
\right)^{\otimes n}
\left(
\Pi_n^{A^n}
\otimes
I^{R^n}
\right)
}{
1-\delta_{\Pi_n}
},
\label{eq:IID_state_merging_projected_AR_state}
\\
\psi_{\Pi_n}^{R^n}
&\coloneqq
\operatorname{Tr}_{A_{\Pi_n}^n}
\left[
\psi_{\Pi_n}^{A_{\Pi_n}^nR^n}
\right].
\label{eq:IID_state_merging_projected_reference_state}
\end{align}

The partial-trace channel associated with the unitary isomorphism in
\eqref{eq:IID_state_merging_decoupling_factorization}
is
\begin{align}
\mathcal T_n^{A_{\Pi_n}^n\to E_n}
\left(
X
\right)
\coloneqq
\operatorname{Tr}_{M_n}
\left[
V_n
X
V_n^\dagger
\right].
\label{eq:IID_state_merging_decoupling_partial_trace_channel}
\end{align}
The corresponding decoupling output is
\begin{align}
&\tau_{n,U_n}^{E_nR^n}
\nonumber\\
&\coloneqq
\left(
\mathcal T_n^{A_{\Pi_n}^n\to E_n}
\otimes
\id^{R^n}
\right)
\left[
\left(
U_n
\otimes
I^{R^n}
\right)
\psi_{\Pi_n}^{A_{\Pi_n}^nR^n}
\left(
U_n^\dagger
\otimes
I^{R^n}
\right)
\right].
\label{eq:IID_state_merging_decoupling_output_state}
\end{align}
Since
\(\mathcal T_n^{A_{\Pi_n}^n\to E_n}\)
is trace-preserving,
\begin{align}
\tau_{n,U_n}^{R^n}
=
\psi_{\Pi_n}^{R^n}.
\label{eq:IID_state_merging_decoupling_output_reference}
\end{align}

Theorem~\ref{thm:IID_decoupling_original_direct} gives
\begin{align}
\lim_{n\to\infty}
\delta_{\Pi_n}
&=
0,
\label{eq:IID_state_merging_decoupling_projection_limit}
\\
\lim_{n\to\infty}
D
\left(
\tau_{n,U_n}^{E_nR^n}
\middle\|
\pi^{E_n}
\otimes
\left(
\psi^R
\right)^{\otimes n}
\right)
&=
0,
\label{eq:IID_state_merging_decoupling_error_limit}
\\
\lim_{n\to\infty}
\frac{1}{n}
\log
\left|
A_{\Pi_n}^n
\right|
&=
H(A)_\psi,
\label{eq:IID_state_merging_decoupling_projection_rank_limit}
\\
\lim_{n\to\infty}
\frac{1}{n}
\log
\left|
M_n
\right|
&=
\frac{1}{2}
I(A:R)_\psi.
\label{eq:IID_state_merging_decoupling_message_rate}
\end{align}

For every \(n\), choose a finite-dimensional system \(E_n'\) with a
fixed identification
\begin{align}
\mathcal H^{E_n'}
\simeq
\mathcal H^{E_n},
\label{eq:IID_state_merging_achievability_output_copy}
\end{align}
and let
\begin{align}
\Phi_n^{E_nE_n'}
\label{eq:IID_state_merging_achievability_entangled_state}
\end{align}
be the normalized maximally entangled state of Schmidt rank
\begin{align}
\left|
E_n
\right|
=
\left|
E_n'
\right|,
\end{align}
consistently with
\eqref{eq:IID_state_merging_output_entanglement_rank}.

We next extend the projected decoupling operation to a
trace-preserving encoding channel on the entire, possibly
infinite-dimensional input space
\(\mathcal H^{A^n}\).
Define Alice's local environment by
\begin{align}
\mathcal H^{L_n}
\coloneqq
\mathbb C
\ket{0}^{L_n}
\oplus
\mathcal H^{A_{\Pi_n^\perp}^n},
\label{eq:IID_state_merging_Alice_discarded_system}
\end{align}
where
\(\ket{0}^{L_n}\)
is a unit vector spanning the first summand.
Since
\(\mathcal H^{A^n}\)
is separable,
\(\mathcal H^{A_{\Pi_n^\perp}^n}\)
and
\(\mathcal H^{L_n}\)
are separable.
Let
\begin{align}
V_{\Pi_n^\perp,n}^{
A_{\Pi_n^\perp}^n
\to
L_n
}
:
\mathcal H^{A_{\Pi_n^\perp}^n}
\longrightarrow
\mathcal H^{L_n}
\label{eq:IID_state_merging_failure_isometry}
\end{align}
be the canonical isometry into the second orthogonal summand.
Choose unit vectors
\begin{align}
\ket{0}^{E_n}
&\in
\mathcal H^{E_n},
&
\ket{0}^{M_n}
&\in
\mathcal H^{M_n}.
\label{eq:IID_state_merging_failure_output_vectors}
\end{align}
Identifying
\(V_nU_n\Pi_n^{A^n}\)
with its zero extension from
\(\mathcal H^{A_{\Pi_n}^n}\)
to
\(\mathcal H^{A^n}\),
define
\begin{align}
V_{\mathcal E,n}^{A^n\to E_nM_nL_n}
&\coloneqq
\left(
V_n
U_n
\Pi_n^{A^n}
\right)
\otimes
\ket{0}^{L_n}
\nonumber\\
&\quad+
\ket{0}^{E_n}
\otimes
\ket{0}^{M_n}
\otimes
V_{\Pi_n^\perp,n}
\Pi_n^{\perp A^n}.
\label{eq:IID_state_merging_direct_encoding_isometry}
\end{align}
The two terms in
\eqref{eq:IID_state_merging_direct_encoding_isometry}
have orthogonal ranges in
\(\mathcal H^{L_n}\).
Moreover,
\(V_n\),
\(U_n\), and
\(V_{\Pi_n^\perp,n}\)
are isometric on their respective domains.
Consequently,
\begin{align}
V_{\mathcal E,n}^\dagger
V_{\mathcal E,n}
&=
\Pi_n^{A^n}
+
\Pi_n^{\perp A^n}
\nonumber\\
&=
I^{A^n}.
\label{eq:IID_state_merging_direct_encoding_isometric}
\end{align}
Thus,
\eqref{eq:IID_state_merging_direct_encoding_isometry}
defines an isometry on the entire input space.
Define Alice's encoding channel by
\begin{align}
\mathcal E_n^{A^n\to E_nM_n}
\left(
X
\right)
\coloneqq
\operatorname{Tr}_{L_n}
\left[
V_{\mathcal E,n}
X
V_{\mathcal E,n}^\dagger
\right],
\label{eq:IID_state_merging_direct_encoding_channel}
\end{align}
where the complementary projection branch is stored entirely in Alice's
discarded local environment \(L_n\); therefore, it requires neither an additional quantum message nor a separate classical success or failure flag.

We next construct Bob's decoding channel.
Under the fixed identification
\(\mathcal H^{A'}\simeq\mathcal H^A\)
in
\eqref{eq:IID_state_merging_transferred_system_identification},
let
\begin{align}
\mathcal H^{A_{\Pi_n}^{n\prime}}
\simeq
\mathcal H^{A_{\Pi_n}^n}
\label{eq:IID_state_merging_projected_transferred_space}
\end{align}
be the corresponding finite-dimensional subspace of
\(\mathcal H^{A'^n}\).
Let
\begin{align}
\ket{\psi_{\Pi_n}}^{
A_{\Pi_n}^{n\prime}B^nR^n
}
\label{eq:IID_state_merging_projected_transferred_vector}
\end{align}
denote the relabeled copy of the vector in
\eqref{eq:IID_state_merging_projected_pure_state}, and define
\begin{align}
\psi_{\Pi_n}^{
A_{\Pi_n}^{n\prime}B^nR^n
}
\coloneqq
\ket{\psi_{\Pi_n}}
\bra{\psi_{\Pi_n}}^{
A_{\Pi_n}^{n\prime}B^nR^n
}.
\label{eq:IID_state_merging_projected_transferred_state}
\end{align}

The vector
\begin{align}
\left(
V_n
U_n
\otimes
I^{B^nR^n}
\right)
\ket{\psi_{\Pi_n}}^{
A_{\Pi_n}^nB^nR^n
}
\label{eq:IID_state_merging_projected_post_encoding_vector}
\end{align}
is a purification of
\(\tau_{n,U_n}^{E_nR^n}\)
with purifying system \(M_nB^n\).
On the other hand,
\begin{align}
\ket{\Phi_n}^{E_nE_n'}
\otimes
\ket{\psi_{\Pi_n}}^{
A_{\Pi_n}^{n\prime}B^nR^n
}
\label{eq:IID_state_merging_projected_target_purification}
\end{align}
is a purification of
\begin{align}
\pi^{E_n}
\otimes
\psi_{\Pi_n}^{R^n}
\label{eq:IID_state_merging_projected_target_marginal}
\end{align}
with purifying system
\(E_n'A_{\Pi_n}^{n\prime}B^n\).
Therefore, for these purifications, Uhlmann's theorem for density operators on separable Hilbert
spaces~\cite{UHLMANN1976273}
applies.
After adjoining a separable local environment \(L_n'\) of sufficiently
large dimension, the partial isometry supplied by Uhlmann's theorem
can be extended to an isometry on the entire space
\(\mathcal H^{M_n}\otimes\mathcal H^{B^n}\).
Consequently, there exist a separable Hilbert space
\(\mathcal H^{L_n'}\),
a unit vector
\(\ket{0}^{L_n'}\), and an isometry
\begin{align}
V_{\mathcal D,n}^{M_nB^n\to A'^nB^nE_n'L_n'}
:
\mathcal H^{M_n}
\otimes
\mathcal H^{B^n}
\longrightarrow
\mathcal H^{A'^n}
\otimes
\mathcal H^{B^n}
\otimes
\mathcal H^{E_n'}
\otimes
\mathcal H^{L_n'}
\label{eq:IID_state_merging_direct_decoding_isometry}
\end{align}
satisfying
\begin{align}
V_{\mathcal D,n}^\dagger
V_{\mathcal D,n}
=
I^{M_nB^n},
\label{eq:IID_state_merging_direct_decoding_isometric}
\end{align}
such that
\begin{align}
&P
\Biggl(
\left(
I^{E_nR^n}
\otimes
V_{\mathcal D,n}
\right)
\Biggl[
\left(
V_n
U_n
\otimes
I^{B^nR^n}
\right)
\psi_{\Pi_n}^{
A_{\Pi_n}^nB^nR^n
}
\left(
U_n^\dagger
V_n^\dagger
\otimes
I^{B^nR^n}
\right)
\Biggr]
\left(
I^{E_nR^n}
\otimes
V_{\mathcal D,n}^\dagger
\right),
\nonumber\\
&\hspace{12mm}
\Phi_n^{E_nE_n'}
\otimes
\psi_{\Pi_n}^{
A_{\Pi_n}^{n\prime}B^nR^n
}
\otimes
\ket{0}\bra{0}^{L_n'}
\Biggr)
\nonumber\\
&=
P
\left(
\tau_{n,U_n}^{E_nR^n},
\pi^{E_n}
\otimes
\psi_{\Pi_n}^{R^n}
\right),
\label{eq:IID_state_merging_direct_Uhlmann_identity}
\end{align}
where tensor factors are understood up to canonical reorderings.
Define Bob's decoding channel by
\begin{align}
\mathcal D_n^{M_nB^n\to A'^nB^nE_n'}
\left(
Y
\right)
\coloneqq
\operatorname{Tr}_{L_n'}
\left[
V_{\mathcal D,n}
Y
V_{\mathcal D,n}^\dagger
\right].
\label{eq:IID_state_merging_direct_decoding_channel}
\end{align}

For comparison with the original IID source, define the output of the
protocol on the normalized projected source by
\begin{align}
&\zeta_{\mathcal E_n,\mathcal D_n,\Pi_n}^{
E_nA'^nB^nE_n'R^n
}
\nonumber\\
&\coloneqq
\left(
\id^{E_nR^n}
\otimes
\mathcal D_n^{M_nB^n\to A'^nB^nE_n'}
\right)
\Biggl[
\left(
\mathcal E_n^{A^n\to E_nM_n}
\otimes
\id^{B^nR^n}
\right)
\left[
\psi_{\Pi_n}^{A_{\Pi_n}^nB^nR^n}
\right]
\Biggr],
\label{eq:IID_state_merging_direct_projected_output_state}
\end{align}
where
\(\psi_{\Pi_n}^{A_{\Pi_n}^nB^nR^n}\)
is canonically embedded into
\begin{align}
\mathcal H^{A^n}
\otimes
\mathcal H^{B^n}
\otimes
\mathcal H^{R^n}.
\end{align}

For this protocol, the resulting asymptotic achievability statement is as follows.

\begin{theorem}[Infinite-dimensional IID quantum
state merging]
\label{thm:infinite_dimensional_IID_state_merging_direct}
Let
\(\mathcal H^A\),
\(\mathcal H^B\), and
\(\mathcal H^R\)
be separable and possibly infinite-dimensional Hilbert spaces, and let
\begin{align}
\ket{\psi}^{ABR}
\in
\mathcal H^A
\otimes
\mathcal H^B
\otimes
\mathcal H^R
\end{align}
be a normalized pure state satisfying
\begin{align}
H(A)_\psi
<
\infty.
\end{align}
Let
\(\psi^{ABR}\)
be defined in
\eqref{eq:IID_state_merging_source_density_operator},
let
\(\mathcal H^{A'}\simeq\mathcal H^A\)
be the fixed identification in
\eqref{eq:IID_state_merging_transferred_system_identification},
and let
\(\psi^{A'BR}\)
be defined in
\eqref{eq:IID_state_merging_transferred_source_state}.

Then there exist sequences of finite-dimensional systems
\begin{align}
\left\{
E_n
\right\}_{n=1}^{\infty},
\qquad
\left\{
E_n'
\right\}_{n=1}^{\infty},
\qquad
\left\{
M_n
\right\}_{n=1}^{\infty},
\end{align}
with fixed identifications
\begin{align}
\mathcal H^{E_n'}
\simeq
\mathcal H^{E_n},
\qquad
n\in\{1,2,\ldots\},
\end{align}
and sequences of completely positive and trace-preserving maps
\begin{align}
\left\{
\mathcal E_n^{A^n\to E_nM_n}
\right\}_{n=1}^{\infty},
\qquad
\left\{
\mathcal D_n^{M_nB^n\to A'^nB^nE_n'}
\right\}_{n=1}^{\infty},
\end{align}
of the forms specified in
\eqref{eq:IID_state_merging_encoding_channel}
and
\eqref{eq:IID_state_merging_decoding_channel},
respectively, such that, for the normalized maximally entangled states
\(\Phi_n^{E_nE_n'}\)
specified in
\eqref{eq:IID_state_merging_output_entanglement_rank},
the final states
\(\zeta_{\mathcal E_n,\mathcal D_n}\)
defined in
\eqref{eq:IID_state_merging_final_state},
and the costs
\(q_n\),
\(e_n\), and
\(q_n-e_n\)
defined in
\eqref{eq:IID_state_merging_quantum_communication_cost},
\eqref{eq:IID_state_merging_net_entanglement_cost}, and
\eqref{eq:IID_state_merging_total_cost},
respectively, one has
\begin{align}
\lim_{n\to\infty}
P
\Biggl(
\zeta_{\mathcal E_n,\mathcal D_n}^{
E_nA'^nB^nE_n'R^n
},
\Phi_n^{E_nE_n'}
\otimes
\left(
\psi^{A'BR}
\right)^{\otimes n}
\Biggr)
&=
0,
\label{eq:IID_state_merging_direct_error_limit}
\\
\lim_{n\to\infty}
\frac{q_n}{n}
&=
\frac{1}{2}
I(A:R)_\psi,
\label{eq:IID_state_merging_direct_Q_rate}
\\
\lim_{n\to\infty}
\frac{q_n-e_n}{n}
&=
H(A|B)_\psi.
\label{eq:IID_state_merging_direct_total_rate}
\end{align}
\end{theorem}

\begin{proof}
Apply
Theorem~\ref{thm:IID_decoupling_original_direct}
to the state
\(\psi^{AR}\)
defined in
\eqref{eq:IID_state_merging_AR_marginal}.
The finite-entropy condition required by that theorem follows from
\eqref{eq:IID_state_merging_AR_finite_entropy}.
Let
\(\Pi_n^{A^n}\),
\(E_n\),
\(M_n\),
\(V_n^{A_{\Pi_n}^n\to E_nM_n}\), and
\(U_n\)
be the resulting projections, finite-dimensional systems, unitary
isomorphisms, and unitaries introduced in
\eqref{eq:IID_state_merging_decoupling_factorization} and
\eqref{eq:IID_state_merging_decoupling_unitary}.
Choose
\(E_n'\)
and
\(\Phi_n^{E_nE_n'}\)
as in
\eqref{eq:IID_state_merging_achievability_output_copy} and
\eqref{eq:IID_state_merging_achievability_entangled_state}.
Let
\(\mathcal E_n^{A^n\to E_nM_n}\)
and
\(\mathcal D_n^{M_nB^n\to A'^nB^nE_n'}\)
be the completely positive and trace-preserving maps defined in
\eqref{eq:IID_state_merging_direct_encoding_channel} and
\eqref{eq:IID_state_merging_direct_decoding_channel}, respectively.
Thus, the systems and maps appearing in the statement are well
defined for every
\(n\in\{1,2,\ldots\}\).

\medskip
\noindent
\textbf{Proof of
\eqref{eq:IID_state_merging_direct_error_limit}.}
By
\eqref{eq:IID_state_merging_decoupling_output_reference},
the \(R^n\)-marginal of
\(\tau_{n,U_n}^{E_nR^n}\)
is
\(\psi_{\Pi_n}^{R^n}\).
Hence, applying the relative-entropy decomposition
\eqref{eq:IID_decoupling_original_projected_error_decomposition}
with
\(\rho^{AR}=\psi^{AR}\)
gives
\begin{align}
&D
\left(
\tau_{n,U_n}^{E_nR^n}
\middle\|
\pi^{E_n}
\otimes
\left(
\psi^R
\right)^{\otimes n}
\right)
\nonumber\\
&=
D
\left(
\tau_{n,U_n}^{E_nR^n}
\middle\|
\pi^{E_n}
\otimes
\psi_{\Pi_n}^{R^n}
\right)
+
D
\left(
\psi_{\Pi_n}^{R^n}
\middle\|
\left(
\psi^R
\right)^{\otimes n}
\right).
\label{eq:IID_state_merging_decoupling_relative_entropy_decomposition}
\end{align}
Both terms on the right-hand side of
\eqref{eq:IID_state_merging_decoupling_relative_entropy_decomposition}
are nonnegative.
Therefore,
\eqref{eq:IID_state_merging_decoupling_error_limit}
implies
\begin{align}
\lim_{n\to\infty}
D
\left(
\tau_{n,U_n}^{E_nR^n}
\middle\|
\pi^{E_n}
\otimes
\psi_{\Pi_n}^{R^n}
\right)
=
0.
\label{eq:IID_state_merging_projected_decoupling_error_limit}
\end{align}
Using the relative-entropy bound on the purified distance in
\eqref{eq:decoupling_relative_entropy_controls_PD}, we obtain
\begin{align}
\lim_{n\to\infty}
P
\left(
\tau_{n,U_n}^{E_nR^n},
\pi^{E_n}
\otimes
\psi_{\Pi_n}^{R^n}
\right)
=
0.
\label{eq:IID_state_merging_projected_decoupling_PD_limit}
\end{align}

On the normalized projected source state
\(\psi_{\Pi_n}^{A_{\Pi_n}^nB^nR^n}\),
the complementary term in the encoding isometry
\eqref{eq:IID_state_merging_direct_encoding_isometry}
vanishes.
Consequently, Alice's encoder acts on this state through
\(V_nU_n\), with the environment \(L_n\) in the fixed state
\(\ket{0}^{L_n}\).
The result of Uhlmann's theorem in
\eqref{eq:IID_state_merging_direct_Uhlmann_identity}, followed by the
partial trace over \(L_n'\), therefore gives
\begin{align}
&P
\Biggl(
\zeta_{\mathcal E_n,\mathcal D_n,\Pi_n}^{
E_nA'^nB^nE_n'R^n
},
\Phi_n^{E_nE_n'}
\otimes
\psi_{\Pi_n}^{
A_{\Pi_n}^{n\prime}B^nR^n
}
\Biggr)
\nonumber\\
&\le
P
\left(
\tau_{n,U_n}^{E_nR^n},
\pi^{E_n}
\otimes
\psi_{\Pi_n}^{R^n}
\right),
\label{eq:IID_state_merging_projected_protocol_error}
\end{align}
where we used monotonicity of the purified distance under the partial
trace over \(L_n'\).

We next compare the normalized projected source with the original IID
source.
By
\eqref{eq:IID_state_merging_projected_pure_state},
\begin{align}
\left(
\bra{\psi}
\right)^{\otimes n}
\ket{\psi_{\Pi_n}}
=
\sqrt{
1-\delta_{\Pi_n}
}.
\end{align}
Since both states are pure, it follows that
\begin{align}
P
\left(
\left(
\psi^{ABR}
\right)^{\otimes n},
\psi_{\Pi_n}^{
A_{\Pi_n}^nB^nR^n
}
\right)
=
\sqrt{
\delta_{\Pi_n}
}.
\label{eq:IID_state_merging_original_projected_distance}
\end{align}

Let
\(\mathcal N_n\)
denote the completely positive and trace-preserving map consisting of
Alice's encoding channel, transmission of \(M_n\), and Bob's decoding
channel, with the identity acting on \(R^n\).
By
\eqref{eq:IID_state_merging_final_state} and
\eqref{eq:IID_state_merging_direct_projected_output_state},
\begin{align}
\mathcal N_n
\left[
\left(
\psi^{ABR}
\right)^{\otimes n}
\right]
&=
\zeta_{\mathcal E_n,\mathcal D_n},
\\
\mathcal N_n
\left[
\psi_{\Pi_n}^{A_{\Pi_n}^nB^nR^n}
\right]
&=
\zeta_{\mathcal E_n,\mathcal D_n,\Pi_n},
\end{align}
where the projected state is understood through its canonical
embedding into
\(\mathcal H^{A^n}\otimes\mathcal H^{B^n}\otimes\mathcal H^{R^n}\).
Monotonicity of the purified distance under
\(\mathcal N_n\)
and
\eqref{eq:IID_state_merging_original_projected_distance}
give
\begin{align}
P
\left(
\zeta_{\mathcal E_n,\mathcal D_n},
\zeta_{\mathcal E_n,\mathcal D_n,\Pi_n}
\right)
\le
\sqrt{
\delta_{\Pi_n}
}.
\label{eq:IID_state_merging_output_projection_distance}
\end{align}

Under the fixed identification
\(\mathcal H^{A'}\simeq\mathcal H^A\),
the overlap between
\(\ket{\psi_{\Pi_n}}^{
A_{\Pi_n}^{n\prime}B^nR^n
}\)
and
\((\ket{\psi}^{A'BR})^{\otimes n}\)
is also
\(\sqrt{1-\delta_{\Pi_n}}\).
Since tensoring both arguments with the same normalized state preserves
the purified distance, we have
\begin{align}
&P
\Biggl(
\Phi_n^{E_nE_n'}
\otimes
\psi_{\Pi_n}^{
A_{\Pi_n}^{n\prime}B^nR^n
},
\Phi_n^{E_nE_n'}
\otimes
\left(
\psi^{A'BR}
\right)^{\otimes n}
\Biggr)
\nonumber\\
&=
\sqrt{
\delta_{\Pi_n}
}.
\label{eq:IID_state_merging_target_projection_distance}
\end{align}

Applying the triangle inequality to
\eqref{eq:IID_state_merging_output_projection_distance},
\eqref{eq:IID_state_merging_projected_protocol_error}, and
\eqref{eq:IID_state_merging_target_projection_distance}, we obtain
\begin{align}
&P
\Biggl(
\zeta_{\mathcal E_n,\mathcal D_n}^{
E_nA'^nB^nE_n'R^n
},
\Phi_n^{E_nE_n'}
\otimes
\left(
\psi^{A'BR}
\right)^{\otimes n}
\Biggr)
\nonumber\\
&\le
P
\left(
\zeta_{\mathcal E_n,\mathcal D_n},
\zeta_{\mathcal E_n,\mathcal D_n,\Pi_n}
\right)
\nonumber\\
&\quad+
P
\Biggl(
\zeta_{\mathcal E_n,\mathcal D_n,\Pi_n}^{
E_nA'^nB^nE_n'R^n
},
\Phi_n^{E_nE_n'}
\otimes
\psi_{\Pi_n}^{
A_{\Pi_n}^{n\prime}B^nR^n
}
\Biggr)
\nonumber\\
&\quad+
P
\Biggl(
\Phi_n^{E_nE_n'}
\otimes
\psi_{\Pi_n}^{
A_{\Pi_n}^{n\prime}B^nR^n
},
\Phi_n^{E_nE_n'}
\otimes
\left(
\psi^{A'BR}
\right)^{\otimes n}
\Biggr)
\nonumber\\
&\le
P
\left(
\tau_{n,U_n}^{E_nR^n},
\pi^{E_n}
\otimes
\psi_{\Pi_n}^{R^n}
\right)
+
2
\sqrt{
\delta_{\Pi_n}
}.
\label{eq:IID_state_merging_direct_finite_error_bound}
\end{align}
The first term on the right-hand side converges to zero by
\eqref{eq:IID_state_merging_projected_decoupling_PD_limit}, and the
second term converges to zero by
\eqref{eq:IID_state_merging_decoupling_projection_limit}.
This proves
\eqref{eq:IID_state_merging_direct_error_limit}.

\medskip
\noindent
\textbf{Proof of
\eqref{eq:IID_state_merging_direct_Q_rate}.}
By
\eqref{eq:IID_state_merging_quantum_communication_cost},
\begin{align}
q_n
=
\log
\left|
M_n
\right|.
\end{align}
Therefore,
\eqref{eq:IID_state_merging_decoupling_message_rate}
gives
\begin{align}
\lim_{n\to\infty}
\frac{q_n}{n}
=
\frac{1}{2}
I(A:R)_\psi.
\end{align}
This proves
\eqref{eq:IID_state_merging_direct_Q_rate}.

\medskip
\noindent
\textbf{Proof of
\eqref{eq:IID_state_merging_direct_total_rate}.}
Since
\(V_n^{A_{\Pi_n}^n\to E_nM_n}\)
in
\eqref{eq:IID_state_merging_decoupling_factorization}
is a unitary isomorphism between finite-dimensional Hilbert spaces,
\begin{align}
\left|
A_{\Pi_n}^n
\right|
=
\left|
E_n
\right|
\left|
M_n
\right|.
\label{eq:IID_state_merging_decoupling_dimension_identity}
\end{align}
Consequently,
\begin{align}
\log
\left|
M_n
\right|
-
\log
\left|
E_n
\right|
=
2
\log
\left|
M_n
\right|
-
\log
\left|
A_{\Pi_n}^n
\right|.
\label{eq:IID_state_merging_total_cost_dimension_identity}
\end{align}
By
\eqref{eq:IID_state_merging_net_entanglement_cost} and
\eqref{eq:IID_state_merging_total_cost},
\begin{align}
q_n-e_n
=
\log
\left|
M_n
\right|
-
\log
\left|
E_n
\right|.
\end{align}
Combining this identity with
\eqref{eq:IID_state_merging_total_cost_dimension_identity},
\eqref{eq:IID_state_merging_decoupling_projection_rank_limit}, and
\eqref{eq:IID_state_merging_decoupling_message_rate}, we obtain
\begin{align}
\lim_{n\to\infty}
\frac{q_n-e_n}{n}
&=
I(A:R)_\psi
-
H(A)_\psi
\nonumber\\
&=
-
H(A|R)_\psi
\nonumber\\
&=
H(A|B)_\psi.
\label{eq:IID_state_merging_total_cost_rate_derivation}
\end{align}
The second equality follows from
\begin{align}
I(A:R)_\psi
=
H(A)_\psi
-
H(A|R)_\psi,
\end{align}
which is well defined because
\(H(A)_\psi<\infty\), and the final equality follows from the
pure-state conditional-entropy duality in
\eqref{eq:conditional_entropy_duality_pure_state_finite_entropy_A}.
This proves
\eqref{eq:IID_state_merging_direct_total_rate}.
\end{proof}

\begin{remark}[Converse for IID quantum state merging]
\label{rem:converse_state_merging}
When
\(\mathcal H^A\otimes\mathcal H^B\)
is finite-dimensional, the asymptotically achievable rates in
Theorem~\ref{thm:infinite_dimensional_IID_state_merging_direct}
are optimal by the converse bounds in
Refs.~\cite{horodecki2005partial,
horodecki2007quantum,abeyesinghe2009mother,Datta_2011}.
Roughly speaking, the converse bound on the quantum communication
rate compares the mutual information between Bob and the reference system before and after state merging, namely \(I(B:R)_\psi\) for the initial state \(\ket{\psi}^{ABR}\) and \(I(A'B:R)_\psi\) for the ideal final state \(\ket{\psi}^{A'BR}\).
For the converse bound on the total cost, one may simulate the quantum communication via quantum teleportation and compare the entanglement across the bipartition \(AR:B\) before state merging with that across \(R:A'B\) after state merging, namely,
\(E(AR:B)_\psi
=
H(B)_\psi
\)
for
\(\ket{\psi}^{ABR}\),
and
\(
E(R:A'B)_\psi
=
H(A'B)_\psi
=
H(AB)_\psi
\)
for
\(\ket{\psi}^{A'BR}\),
where
\(E(X:Y)_\phi\)
denotes the entanglement entropy of a pure state
\(\phi^{XY}\)
across the bipartition
\(X:Y\).
When \(\mathcal H^A\otimes\mathcal H^B\) is infinite-dimensional, the finiteness of the entropic quantities appearing in these arguments is not guaranteed by the present assumption
\(H(A)_\psi<\infty\).
The conventional converse arguments in
Refs.~\cite{horodecki2005partial,
horodecki2007quantum}
may therefore involve the indeterminate form
\(\infty-\infty\), and hence, extending the converse arguments would require alternative techniques or additional assumptions.

If one additionally assumes
\begin{align}
H(AB)_\psi
<
\infty,
\end{align}
then the entropic quantities appearing above are finite.
This assumption, however, constrains only the initial state and, consequently,
the ideal target state.
In a converse proof, the actual output of an arbitrary protocol is
required only to approach the ideal target state in purified or trace
distance, and such an approximation does not imply a corresponding energy
constraint on the actual output in an infinite-dimensional system.
Consequently, continuity bounds requiring energy constraints on both
the ideal and actual states are still not directly applicable.

The same asymmetry arises in the converse proof for the entanglement
cost of infinite-dimensional states in
Ref.~\cite[Section~5.1]{yamasaki2025entanglement}.
The analysis therein circumvents this issue by using the one-sided semicontinuity bound for the entanglement of formation in
Ref.~\cite[Proposition~4B]{shirokov2023close},
or its refinement in
Ref.~\cite[Proposition~10B]{shirokov2025alicki},
which requires an energy constraint only on the ideal target state.
Under the additional assumption
\(H(AB)_\psi<\infty\),
it appears possible to adapt a similar one-sided technique to control
the entropic quantities appearing in the converse for
infinite-dimensional quantum state merging, even when the protocol is
allowed to have an asymptotically vanishing error.
We leave the detailed derivation of such an extension for future work.

Nevertheless, the contribution of the present work is to construct a protocol that applies regardless of whether
\(\mathcal H^A\otimes\mathcal H^B\)
is finite- or infinite-dimensional and, in the finite-dimensional setting, provably achieves the optimal rates.
\end{remark}

\bibliography{citation}

\end{document}